\documentclass[12pt,twoside]{preprint}
\usepackage{comment}
\usepackage{times}
\usepackage{hyperref}
\usepackage{breakurl}
\usepackage{amssymb}
\usepackage{amsmath}
\usepackage{mhsymb}
\usepackage{mhequ}
\usepackage{mhenvs}
\usepackage{mhfig}
\usepackage{microtype}
\usepackage{mathrsfs}
\usepackage{booktabs}
\usepackage{tikz}
\usepackage{wasysym}
\usepackage{centernot}

\usetikzlibrary{shapes}

\newtheorem{assumption}[lemma]{Assumption}

\newcommand{\eqlaw}{\stackrel{\mbox{\tiny law}}{=}}

\def\PPi{\boldsymbol{\Pi}}

\def\emptyset{\mathop{\centernot\ocircle}}

\hypersetup{
citecolor=blue,
colorlinks=true,
urlcolor= blue,
linkcolor= black, 
bookmarksopen=true, 
pdftitle={Convergence to KPZ},
}

\def\CCE{\mathbb{E}}
\def\CCV{\mathbb{V}}
\def\CCG{\mathbb{G}}
\def\CCT{\mathbb{T}}

\def\E{\mathbf{E}}

\def\DD{\mathscr{D}}
\def\MM{\mathscr{M}}

\def\FF{\mathscr{F}}
\def\Ren{\mathscr{R}}
\def\LL{\mathscr{L}}

\def\Eps{\mathcal{E}}

\def\PR{\boldsymbol{R}}
\def\DeltaM{\Delta^{\!M}}
\def\DeltaW{\Delta^{\!\!\scriptscriptstyle\mathrm{Wick}}}
\def\Wick{M^{\!\scriptscriptstyle\mathrm{Wick}}}
\def\RWick{\CR^{\!\scriptscriptstyle\mathrm{Wick}}}
\def\WickL{L^{\!\scriptscriptstyle\mathrm{Wick}}}
\def\Wickh{\hat M^{\!\scriptscriptstyle\mathrm{Wick}}}
\def\W{{\scriptscriptstyle\mathrm{W}}}
\def\Deltap{\Delta^{\!+}}
\def\hDeltaM{\hat\Delta^{\!M}}
\def\s{\mathfrak{s}}

\def\BB{\mathbf{B}}
\def\RR{\mathfrak{R}}
\def\Hopf{\mathrm{HC}}
\def\H{\mathrm{H}}
\def\Z{\mathfrak{M}}
\def\ZZ{\mathbf{Z}}
\def\D{\mathbf{D}}

\def\Amin{{A^{\!-}}}
\def\ex{{\mathrm{ex}}}
\def\Bad{\mathscr{B}}
\def\Op{\mathscr Y}

\def\PPP{}
\def\v{\nu}
\def\w{\omega}
\def\u{\upsilon}

\def\${|\!|\!|}

\definecolor{darkgreen}{rgb}{0.1,0.6,0.1}
\definecolor{darkblue}{rgb}{0.1,0,0.7}
\def\J{\mathscr{I}}
\def\EE{\mathscr{E}}

\def\bbar#1{\bar{\bar #1}}
\def\n{\mathbf{n}}

\makeatletter
\pgfdeclareshape{crosscircle}
{
  \inheritsavedanchors[from=circle] 
  \inheritanchorborder[from=circle]
  \inheritanchor[from=circle]{north}
  \inheritanchor[from=circle]{north west}
  \inheritanchor[from=circle]{north east}
  \inheritanchor[from=circle]{center}
  \inheritanchor[from=circle]{west}
  \inheritanchor[from=circle]{east}
  \inheritanchor[from=circle]{mid}
  \inheritanchor[from=circle]{mid west}
  \inheritanchor[from=circle]{mid east}
  \inheritanchor[from=circle]{base}
  \inheritanchor[from=circle]{base west}
  \inheritanchor[from=circle]{base east}
  \inheritanchor[from=circle]{south}
  \inheritanchor[from=circle]{south west}
  \inheritanchor[from=circle]{south east}
  \inheritbackgroundpath[from=circle]
  \foregroundpath{
    \centerpoint%
    \pgf@xc=\pgf@x%
    \pgf@yc=\pgf@y%
    \pgfutil@tempdima=\radius%
    \pgfmathsetlength{\pgf@xb}{\pgfkeysvalueof{/pgf/outer xsep}}%
    \pgfmathsetlength{\pgf@yb}{\pgfkeysvalueof{/pgf/outer ysep}}%
    \ifdim\pgf@xb<\pgf@yb%
      \advance\pgfutil@tempdima by-\pgf@yb%
    \else%
      \advance\pgfutil@tempdima by-\pgf@xb%
    \fi%
    \pgfpathmoveto{\pgfpointadd{\pgfqpoint{\pgf@xc}{\pgf@yc}}{\pgfqpoint{-0.707107\pgfutil@tempdima}{-0.707107\pgfutil@tempdima}}}
    \pgfpathlineto{\pgfpointadd{\pgfqpoint{\pgf@xc}{\pgf@yc}}{\pgfqpoint{0.707107\pgfutil@tempdima}{0.707107\pgfutil@tempdima}}}
    \pgfpathmoveto{\pgfpointadd{\pgfqpoint{\pgf@xc}{\pgf@yc}}{\pgfqpoint{-0.707107\pgfutil@tempdima}{0.707107\pgfutil@tempdima}}}
    \pgfpathlineto{\pgfpointadd{\pgfqpoint{\pgf@xc}{\pgf@yc}}{\pgfqpoint{0.707107\pgfutil@tempdima}{-0.707107\pgfutil@tempdima}}}
  }
}
\makeatother

\colorlet{symbols}{blue}     
\colorlet{testcolor}{green!60!black}

\def\symbol#1{\textcolor{symbols}{#1}}
\def\1{\mathbf{\symbol{1}}}

\def\Proj{\mathrm{Proj}}

\usetikzlibrary{calc}
\usetikzlibrary{shapes.misc}
\usetikzlibrary{shapes.symbols}
\usetikzlibrary{snakes}
\usetikzlibrary{decorations}
\usetikzlibrary{decorations.markings}

\tikzset{
	root/.style={circle,fill=testcolor,inner sep=0pt, minimum size=2mm},
	dot/.style={circle,fill=black,inner sep=0pt, minimum size=1mm},
	var/.style={circle,fill=black!10,draw=black,inner sep=0pt, minimum size=2mm},
	dotred/.style={circle,fill=black!50,inner sep=0pt, minimum size=2mm},
	generic/.style={semithick,shorten >=1pt,shorten <=1pt},
	dist/.style={ultra thick,testcolor,shorten >=1pt,shorten <=1pt},
	testfcn/.style={ultra thick,testcolor,shorten >=1pt,shorten <=1pt,<-},
	testfcnx/.style={ultra thick,testcolor,shorten >=1pt,shorten <=1pt,<-,
		postaction={decorate,decoration={markings,mark=at position 0.6 with {\drawx}}}},
	kernelprime/.style={semithick,shorten >=1pt,shorten <=1pt,densely dashed,->},
	kernelprimex/.style={semithick,shorten >=1pt,shorten <=1pt,densely dashed,->,
		postaction={decorate,decoration={markings,mark=at position 0.4 with {\drawx}}}},
	kernel/.style={semithick,shorten >=1pt,shorten <=1pt,->},
	kernel1/.style={->,semithick,shorten >=1pt,shorten <=1pt,postaction={decorate,decoration={markings,mark=at position 0.45 with {\draw[-] (0,-0.1) -- (0,0.1);}}}},
	keps/.style={semithick,densely dashed,shorten >=1pt,shorten <=1pt,->},
	dots/.style={semithick,dotted,shorten >=1pt,shorten <=1pt},
	Deps/.style={semithick,draw=black!25,fill=black!25,shorten >=1pt,shorten <=1pt,->},
	kbase/.style={semithick,dotted,shorten >=1pt,shorten <=1pt,->},
	multx/.style={shorten >=1pt,shorten <=1pt,
		postaction={decorate,decoration={markings,mark=at position 0.5 with {\drawx}}}},
	kernelx/.style={semithick,shorten >=1pt,shorten <=1pt,->,
		postaction={decorate,decoration={markings,mark=at position 0.4 with {\drawx}}}},
	kernel1/.style={->,semithick,shorten >=1pt,shorten <=1pt,postaction={decorate,decoration={markings,mark=at position 0.45 with {\draw[-] (0,-0.1) -- (0,0.1);}}}},
	kernel2/.style={->,semithick,shorten >=1pt,shorten <=1pt,postaction={decorate,decoration={markings,mark=at position 0.45 with {\draw[-] (0.05,-0.1) -- (0.05,0.1);\draw[-] (-0.05,-0.1) -- (-0.05,0.1);}}}},
	kernelBig/.style={semithick,shorten >=1pt,shorten <=1pt,decorate, decoration={zigzag,amplitude=1.5pt,segment length = 3pt,pre length=2pt,post length=2pt}},
	rho/.style={dotted,semithick,shorten >=1pt,shorten <=1pt},
	renorm/.style={shape=circle,fill=white,inner sep=1pt},
	labl/.style={shape=rectangle,fill=white,inner sep=1pt},
	xi/.style={circle,fill=symbols!10,draw=symbols,inner sep=0pt,minimum size=1.2mm},
	xix/.style={crosscircle,fill=symbols!10,draw=symbols,inner sep=0pt,minimum size=1.2mm},
	xib/.style={circle,fill=symbols!10,draw=symbols,inner sep=0pt,minimum size=1.6mm},
	xibx/.style={crosscircle,fill=symbols!10,draw=symbols,inner sep=0pt,minimum size=1.6mm},
	not/.style={circle,fill=symbols,draw=symbols,inner sep=0pt,minimum size=0.5mm},
	>=stealth,
	graydot/.style={circle,fill=gray,inner sep=0pt, minimum size=1mm},
	zero/.style={circle,inner sep=0pt, minimum size=1mm, draw},
	kernelprimeeps/.style={densely dashed, semithick,shorten >=1pt,shorten <=1pt},
	smalldot/.style={circle,fill=black,draw=black, solid,inner sep=0pt,minimum size=0.5mm},
	}
\makeatletter
\def\DeclareSymbol#1#2#3{\expandafter\gdef\csname MH@symb@#1\endcsname{\tikz[baseline=#2,scale=0.15,draw=black]{#3}}
\expandafter\gdef\csname MH@symb@#1s\endcsname{\scalebox{0.7}{\tikz[baseline=#2,scale=0.15,draw=black]{#3}}}}
\def\<#1>{\csname MH@symb@#1\endcsname}
\makeatother

\DeclareSymbol{1}{0}{\draw[white] (-.4,0) -- (.4,0); \draw (0,0)  -- (0,1.2) node[smalldot] {};}
\DeclareSymbol{2}{0}{\draw (-0.5,1.2) node[smalldot] {} -- (0,0) -- (0.5,1.2) node[smalldot] {};}
\DeclareSymbol{11}{0}{\draw (0,1.8) node[smalldot] {} -- (-0.7,0.9) -- (0,0) -- (0.7,1) node[smalldot] {};}
\DeclareSymbol{10}{0}{\draw (0,1.8) node[smalldot] {} -- (-0.8,0.9) -- (0,0);}
\DeclareSymbol{21}{0.7}{\draw (-1,1.8) node[smalldot] {} -- (-0.5,0.9); \draw (0,1.8) node[smalldot] {} -- (-0.5,0.9) -- (0,0) -- (0.5,0.9) node[smalldot] {};}
\DeclareSymbol{20}{0.7}{\draw (-1,1.8) node[smalldot] {} -- (-0.5,0.9); \draw (0,1.8) node[smalldot] {} -- (-0.5,0.9) -- (-0.5,0);}
\DeclareSymbol{210}{1.1}{\draw (-0.5,2.4) node[smalldot] {} -- (-1,1.6);\draw (-1.5,2.4) node[smalldot] {} -- (-0.5,0.8); \draw (0,1.6) node[smalldot] {} -- (-0.5,0.8) -- (-0.5,0);}
\DeclareSymbol{211}{1.1}{\draw (-0.5,2.4) node[smalldot] {} -- (-1,1.6);\draw (-1.5,2.4) node[smalldot] {} -- (-0.5,0.8); \draw (0,1.6) node[smalldot] {} -- (-0.5,0.8) -- (0,0) -- (0.5,0.8) node[smalldot] {};}
\DeclareSymbol{22}{0.7}{\draw (-1.5,1.8) node[smalldot] {} -- (-1,0.9) -- (0,0) -- (1,0.9) -- (1.5,1.8) node[smalldot] {};
\draw (-0.5,1.8) node[smalldot] {} -- (-1,0.9);\draw (0.5,1.8) node[smalldot] {} -- (1,0.9);}

\DeclareSymbol{j1}{0.5}{\draw (-1, 2.4) node[smalldot] {} -- (-0.5,1.6) node[smalldot] {} -- (0, 2.4) node[smalldot] {}; \draw (-0.5,1.6) -- (0,0.8)  node[smalldot] {} -- (0.5, 1.6) node[smalldot] {}; \draw (0,0.8)
-- (0.5,0) node[smalldot] {};}
\DeclareSymbol{j2}{0.5}{\draw (0,0) node[smalldot] {} -- (0.5,0.8) node[smalldot] {};}
\DeclareSymbol{j3}{1.3}{\draw (-1, 2.4) node[smalldot] {} -- (-0.5,1.6) node[smalldot] {} -- (0, 2.4) node[smalldot] {}; \draw (-0.5,1.6) -- (0,0.8)  node[smalldot] {} -- (0.5, 1.6) node[smalldot] {}; }
\DeclareSymbol{j5}{0.5}{ \draw (0,0.8)  node[smalldot] {} -- (0.5, 1.6) node[smalldot] {}; \draw (0,0.8)
-- (0.5,0) node[smalldot] {};}

\DeclareSymbol{Xi22}{0.5}{\draw (0,0) node[xi] {} -- (-1,1) node[not] {} -- (0,2) node[xi] {};}

\DeclareSymbol{Xi2}{-2}{\draw (0,-0.25) node[xi] {} -- (-1,1) node[xi] {};}
\DeclareSymbol{Xi3}{0}{\draw (0,0) node[xi] {} -- (-1,1) node[xi] {} -- (0,2) node[xi] {};}
\DeclareSymbol{Xi4}{2}{\draw (0,0) node[xi] {} -- (-1,1) node[xi] {} -- (0,2) node[xi] {} -- (-1,3) node[xi] {};}
\DeclareSymbol{Xi2X}{-2}{\draw (0,-0.25) node[xi] {} -- (-1,1) node[xix] {};}
\DeclareSymbol{XXi2}{-2}{\draw (0,-0.25) node[xix] {} -- (-1,1) node[xi] {};}

\DeclareSymbol{IXi2}{0}{\draw (0,-0.25) node[not] {} -- (-1,1) node[xi] {} -- (0,2) node[xi] {};}
\DeclareSymbol{IXi^2}{-1}{\draw (-1,1) node[xi] {} -- (0,0) node[not] {} -- (1,1) node[xi] {};}

\DeclareSymbol{XiX}{-2.8}{\node[xibx] {};}
\DeclareSymbol{Xi}{-2.8}{\node[xib] {};}
\DeclareSymbol{IXiX}{-1}{\draw (0,-0.25) node[not] {} -- (-1,1) node[xix] {};}

\DeclareSymbol{Xi3b}{-1}{\draw (-1,1) node[xi] {} -- (0,0) node[xi] {} -- (1,1) node[xi] {};}

\DeclareSymbol{IXi3}{2}{\draw (0,-0.25) node[not] {} -- (-1,1) node[xi] {} -- (0,2) node[xi] {} -- (-1,3) node[xi] {};}
\DeclareSymbol{IXi}{-2}{\draw (0,-0.25) node[not] {} -- (-1,1) node[xi] {};}
\DeclareSymbol{XiI}{-2}{\draw (0,-0.25) node[xi] {} -- (-1,1) node[not] {};}

\DeclareSymbol{Xi4b}{-1}{\draw(0,1.5) node[xi] {} -- (0,0); \draw (-1,1) node[xi] {} -- (0,0) node[xi] {} -- (1,1) node[xi] {};}
\DeclareSymbol{Xi4c}{0}{\draw (0,1) -- (0.8,2.2) node[xi] {};\draw (0,-0.25) node[xi] {} -- (0,1) node[xi] {} -- (-0.8,2.2) node[xi] {};}
\DeclareSymbol{Xi4d}{-4.5}{\draw (0,-1.5) node[not] {} -- (0,0); \draw (-1,1) node[xi] {} -- (0,0) node[xi] {} -- (1,1) node[xi] {};}
\DeclareSymbol{Xi4e}{0}{\draw (0,2) node[xi] {} -- (-1,1) node[xi] {} -- (0,0) node[xi] {} -- (1,1) node[xi] {};}

\begin{document}

\title{A class of growth models rescaling to KPZ}
\author{Martin Hairer$^1$, Jeremy Quastel$^2$}
\institute{University of Warwick, \email{M.Hairer@Warwick.ac.uk}
\and University of Toronto, \email{quastel@math.toronto.edu}}
\maketitle

\begin{abstract}
We consider a large class of $1+1$-dimensional continuous
interface growth models and we show that, in both the weakly asymmetric 
and the intermediate disorder regimes, these models converge
to Hopf-Cole solutions to the KPZ equation. 
\end{abstract}

\setcounter{tocdepth}{2}
\microtypesetup{protrusion=false}
\tableofcontents
\microtypesetup{protrusion=true}

\section{Introduction}

The Kardar-Parisi-Zhang equation is formally given by
\begin{equ}[e:KPZ]
\d_t h^{(\lambda)} = \d_x^2 h^{(\lambda)} + \lambda \bigl(\d_x h^{(\lambda)}\bigr)^2 + \xi\;,
\end{equ}
where $\xi$ denotes space-time white noise and $\lambda \in \R$ is a parameter
describing the strength of its ``asymmetry''.
Equation \eref{e:KPZ} should be interpreted either via the Hopf-Cole transform \cite{MR1462228} as
\begin{equ}\label{def:hc}
h_{\Hopf}^{(\lambda)} \eqdef {1\over \lambda} \log Z^{(\lambda)}\;.
\end{equ}
where $Z^{(\lambda)}$ is the continuous \cite{MR876085}, strictly positive \cite{Carl}   It\^o solution of the multiplicative stochastic heat equation
\begin{equ}\label{eq:SHE}
dZ^{(\lambda)} = \d_x^2 Z^{(\lambda)} + \lambda Z^{(\lambda)}\,dW\;,\qquad
Z^{(\lambda)}(0) = Z_0\;,
\end{equ}
where $Z_0=\exp (\lambda h_0)$
with $W$ an $L^2$-cylindrical Wiener process, $\scal{W_t - W_s, \phi} = \xi(\phi \otimes \one_{[s,t]})$
or equivalently by using the theory exposed in
\cite{KPZ,Regularity}. 
It has been conjectured (see \cite{MR1317994,MR1462228,MR3176353} for a number of results in
this direction) 
that the KPZ equation 
has a ``universal'' character in the sense that any one-dimensional model of surface growth should
converge to it provided that it has the following features:
\begin{claim}
\item There is a microscopic smoothing mechanism.
\item The system has microscopic fluctuations with short-range correlations.
\item The system has some ``lateral growth'' mechanism in the sense that the growth speed depends in a nontrivial way on the slope.
\item At the microscopic scale, the strengths of the growth and fluctuation mechanisms 
are well separated: either the growth mechanism dominates (intermediate disorder) or the fluctuations dominate (weak asymmetry).
\end{claim} 
Only some progress has been made toward a rigorous mathematical understanding of this claim.
The only discrete microscopic models for which convergence to the KPZ equation has been established
rigorously in general are the height function of  asymmetric  exclusion processes in the weakly asymmetric limit \cite{MR1462228},
 \cite{MR2796514}, \cite{Dembo:2013aa}, qTASEP \cite{MR3152785,Corwin:2015aa} and the free energy of directed random polymers in the intermediate disorder regime \cite{AKQ}, \cite{MQR}.  In  \cite{MR3176353} it was shown that a wide class of asymmetric particle models with product invariant measures converge to \emph{energy solutions} of the KPZ equation when
 started in  
 equilibrium.  A slightly stronger version of these equilibrium energy solutions were shown to be unique in \cite{Gubinelli:2015cr}.  In the  continuous 
 setting \cite{Funaki:2014aa} consider the KPZ equation with  non-linearity smoothed out so that a smoothed out Brownian motion is invariant, and show, again, that in equilibrium it converges to KPZ.  In all these cases, including the last two, the proof goes through the Hopf-Cole transformation, and relies on the result satisfying a manageable version of \eqref{eq:SHE}.
This is avoided in the regularity structures approach \cite{KPZ, Regularity} which, in principle, allows for many different types of 
regularization of the quadratic KPZ equation or stochastic heat equation \cite{Etienne}.  At the present time it is however 
restricted to finite volume.

 Substantial progress has also been made recently in the understanding of the conjectured 
long time scaling limit of the KPZ equation itself, which is expected to be the scaling
limit for  this whole class of microscopic interface growth models \cite{spohn, MR2784327,MR2796514,MR3152785}.
Note that the type of well-posedness and approximation results considered here, or in \cite{Regularity}, even when they
are global, do not have much to say about large time, which presently can only be probed through exact calculations.

In this article, we consider continuous growth models of the type      
\begin{equ}[e:growth]
\d_t h = \d_x^2 h + \eps F(\d_x h) + \delta \eta\;,
\end{equ}
where $F$ is an even function, which we will often take to be a polynomial, 
modelling the growth mechanism, $\eta$ is a 
\textit{smooth} space-time Gaussian process
modelling the microscopic fluctuations, and $\eps$, $\delta$ are two parameters. 
The two regimes alluded to earlier correspond to $\eps \approx 1$ and 
$\delta \ll 1$ (intermediate disorder),
as well as $\eps \ll 1$ and $\delta \approx 1$ (weak asymmetry).
It is important to note that these two regimes are
\textit{not} equivalent, i.e.\ it is not possible to turn one regime into the
other by a simple change of variables.   What is usually done is to formally expand
\begin{equ}
F(s) = F(0) + F'(0) s + \tfrac12 F''(0) s^2 + \cdots
\end{equ}
The first two terms in the expansion can be removed by simple height and spatial shifts and 
 one argues that the model is then approximated by the quadratic KPZ equation  \eqref{e:KPZ} with $\lambda=  \tfrac12 F''(0)$
 \cite{hhz, krug-spohn}.

Our main result is that for a wide class of nonlinearities $F$
and correlation functions for $\eta$, the appropriate rescaling of 
\eref{e:growth} (as a function of the small parameter $\eps$ or $\delta$ depending on
the regime considered) converges to the KPZ equation \eref{e:KPZ} for a suitable
value of the parameter $\lambda$. While this result is to some extent expected in view of the above discussion,
the precise analysis uncovers some surprising facts:
\begin{claim}
\item In the weakly asymmetric regime, the 
value $\lambda$ obtained for the limiting equation is \textit{not} the
one that one would guess by formally rescaling the equation and neglecting
all terms with a positive power of the small parameter. In particular,
one generically has $\lambda \neq 0$ even if the polynomial $F$ has no quadratic term.
\item In the intermediate disorder regime, if we consider $F$ with $F''(0) = 0$ but
$F^{(4)}(0) \neq 0$ (say) then, as expected, the limit obtained under the
``na\"\i ve rescaling'' is given by the additive (linear) stochastic heat equation \eqref{e:KPZ} with $\lambda=0$.
However, by considering larger scales, one again recovers the KPZ equation with a non-trivial $\lambda$!
\end{claim}

To understand the need for the separation of scales, let us
consider the problem of trying to make sense of \eqref{e:growth}
with $\eps=\delta=1$, when $\eta$ is space-time white noise.  The 
natural approach is to replace $\eta$ by an approximate white noise $\xi^{(\gamma)}$ 
which is smooth on some small scale $\gamma>0$ and attempt to identify a limit of 
\begin{equ}
\d_t h_\gamma = \d_x^2 h_\gamma +  F(\d_x h_\gamma) +\xi^{(\gamma)}\;.
\end{equ}
In the KPZ case, $F(u) =  u^2$, 
the non-linear term does indeed converge to a non-trivial field, at the simplest level in the sense of convergence of  the space-time covariance, after renormalization by
subtraction of a diverging constant.    On the other hand,
if one takes  a higher order non-linearity such as $F(v) =v^4$, the
renormalization by constants cannot help: The space-time covariance of the non-linear field simply diverges as $\gamma^{-2}$.  A possible route might be to renormalize by subtracting 
quadratic terms.  For example, one could try to take a limit of 
\begin{equ}\label{gammaeq}
\d_t h_\gamma = \d_x^2 h_\gamma +  [(\d_x h_\gamma)^4- c_{2,\gamma}(\d_x h_\gamma)^2+ c_{1,\gamma}] + \xi^{(\gamma)}\;,
\end{equ}
with precisely chosen $c_{1,\gamma}$ and $c_{2,\gamma}$.  The model is supercritical, and on 
large scales one expects such system to be diffusive, i.e.\ to exhibit Gaussian 
fluctuations \cite{spohn}.  
On our scales the non-linear term still diverges, in fact, it 
is just a divergent multiple of space-time white noise, as
can be seen by considering instead the critically adjusted model
\begin{equ}
\d_t h_\gamma = \d_x^2 h_\gamma +  \gamma^{1/2}[(\d_x h_\gamma)^4- c_{2,\gamma}(\d_x h_\gamma)^2+ c_{1,\gamma}] + \xi^{(\gamma)}\;.
\end{equ} 
Although we know of no  proof, it is possible to convince oneself that the limit as $\gamma\searrow 0$ should just be the free field
$
\d_t h = \d_x^2 h  + a\xi$
with a new $a>1$, suggesting that the solution of \eqref{gammaeq} is essentially the solution of the free equation multiplied by $\gamma^{-1/2}$.  The lesson is that the only non-trivial limits 
in \eqref{e:growth} are going to come from fine tuning $\eps$ and
$\delta$ with the scale of decay of covariance of the forcing noise. This leads ultimately to two choices, the intermediate disorder, and weakly asymmetric limits.

In order to state our results precisely, we need to describe briefly the function spaces we are working in.  We would like  our initial conditions to have the typical regularity of the KPZ equation, which is $\CC^\alpha$ for $\alpha < {1\over 2}$, where for $\alpha\in (0,1)$ the H\"older norm is given by
\begin{equ}
\|h\|_\alpha =  \|h\|_{L^\infty} + \sup_{x \neq y} {|h(x) - h(y)| \over |x-y|^\alpha} .
\end{equ}
But even for  $\eps = 1$ and without the noise,
it is not at all straightforward to control solutions to \eref{e:weakAsym}, 
even for short times,
by exploiting the regularisation properties of the associated fixed point map. The
only tool we really have at our disposal is the maximum principle (see, for example, \cite{ben-artzi}) but it is not clear how
one can combine this with the type of analytic estimates essential in the theory of
regularity structures.  

So we define H\"older spaces $\CC^{\gamma,\alpha}_\eps$ for $\alpha \in (0,1)$ and $\gamma \in (1,2)$ 
of functions which
are $\CC^\alpha$
at ``large scales'' (ie. larger than $\eps$) and  $\CC^\gamma$ at ``small scales'', by setting
\begin{equ}[e:defh]
\|h\|_{\gamma,\alpha;\eps} = \|h\|_{\alpha} + \sup_{x \neq y \atop |x-y| \le \eps} {|h'(x) - h'(y)| \over \eps^{\alpha-\gamma}|x-y|^{\gamma-1}}\;. 
\end{equ}
This norm makes such a statement quantitative, typically in the context of a sequence of functions
$h^{(\eps)} \in \CC^{\gamma,\alpha}_\eps$ with uniformly bounded norms. For $\eps = 0$, one does of course
recover the usual $\alpha$-H\"older norms.
The natural way of comparing an element $ \bar h \in \CC^{\alpha}$ with an element $ h \in \CC^{\gamma,\alpha}_\eps$
is given by
\begin{equs}[e:diffHolder]
\|h;\bar h\|_{\gamma,\alpha;\eps} &= \|h - \bar h\|_{\gamma} 
+ \sup_{x \neq y \atop |x-y| \le \eps} {| h'(x) - h'(y)| \over \eps^{\alpha-\gamma}|x-y|^{\gamma-1}} + \sup_x {|h'(x)| \over \eps^{\alpha-1}}\;. \qquad
\end{equs}
Note that we do not impose a supremum bound of order $\eps^{\alpha -1}$
on $h'$ in \eqref{e:defh} because such a  bound follows automatically from 
$\|h\|_{\gamma,\alpha;\eps} \le 1$.

\subsection{Intermediate disorder scaling}

Let us first consider the intermediate disorder regime. In this case, 
\begin{equ}[e:weakNoise]
\d_t h = \d_x^2 h +F(\d_x h) + \eps^{1\over 2} \eta\;,
\end{equ}
where  $\eps$ will always be a small positive parameter. 
Setting $\tilde h(x,t) = h(\eps^{-1}x,\eps^{-2}t)$, we obtain for the rescaled process the equation
\begin{equ}
\d_t \tilde h = \d_x^2 \tilde h + \eps^{-2}F(\eps \d_x \tilde h) + \xi^{(\eps)}\;,
\end{equ}
where $\xi^{(\eps)}(x,t) = \eps^{-3/2}\eta(\eps^{-1}x,\eps^{-2}t)$ is a stochastic process that approximates space-time white noise on scales 
larger than $\eps$.
Expanding $f$ in a Taylor series around $0$, we formally obtain
\begin{equ}[e:weakNoiseRescaled]
\d_t \tilde h = \d_x^2 \tilde h + {a_0\over \eps^2 } + a_1(\d_x \tilde h)^2 + \CO\bigl(\eps^2 (\d_x \tilde h)^4\bigr) + \xi^{(\eps)}\;,
\end{equ}
which strongly suggests that the scaling limit of this equation as $\eps \to 0$ (modulo a height shift which
has the effect of adjusting the value of $a_0$) is given by the KPZ equation \cite{KPZOrig}.

It also raises the question of what happens if the quadratic part of $F$ vanishes. 
Under the scaling given above, it seems intuitively clear that one simply converges toward
the ``trivial'' limit given by the additive stochastic heat equation. On the other hand, 
one might look at different scalings and consider
$\tilde h(x,t) = \eps^\beta h(\eps^{-\alpha}x,\eps^{-2\alpha}t)$ for some exponents $\alpha$ and $\beta$
to be determined. Inserting this into \eref{e:weakNoise}, we obtain the rescaled equation
\begin{equ}
\d_t \tilde h = \d_x^2 \tilde h + \eps^{\beta-2\alpha}F(\eps^{\alpha-\beta} \d_x \tilde h) + \eps^{1-\alpha+2\beta \over 2}\xi^{(\eps^\alpha)}\;.
\end{equ}
In order for the noise term to converge to space-time white noise, we should choose $\beta = (\alpha-1)/2$, so that
\begin{equ}[e:weakNoiseRescaledGeneral]
\d_t \tilde h = \d_x^2 \tilde h + \eps^{-{1+3\alpha \over 2}}F(\eps^{1+\alpha \over 2} \d_x \tilde h) + \xi^{(\eps^\alpha)}\;.
\end{equ}
If $F(x) \sim x^{2p}$ around $x=0$ for some integer $p \ge 1$, this suggests that one should see a non-trivial limit by choosing
$\alpha$ such that $2p(1+\alpha) = 1+3\alpha$, i.e.\ $\alpha = (2p-1)/(3-2p)$ and that the scaling limit should be given by the equation
\begin{equ}
\d_t \tilde h = \d_x^2 \tilde h + (\d_x \tilde h)^{2p} + \tilde \xi\;,
\end{equ}
where $\tilde \xi$ denotes space-time white noise. This would to some extent 
contradict the universality of the KPZ equation. 
We immediately see a problem with this argument: when $p > 1$,
the value of $\alpha$ obtained in this way is negative, so that we do not actually look at large scales at all!
We will see  that the correct way to rescale this system in order to obtain a non-trivial
large-scale limit is to choose $\alpha = 2p-1$. With this choice, it turns out that even if $p \neq 1$,
the scaling limit obtained in this way is indeed given by the KPZ equation.

In order to fix notations, let us consider henceforth a smooth compactly supported function $\rho\colon \R^2 \to \R$
integrating to $1$ and set
\begin{equ}[e:xiEps]
\rho_\eps(t,x) = \eps^{-3} \rho(\eps^{-2}t, \eps^{-1}x)\;,
\qquad \xi^{(\eps)} = \rho_\eps * \xi\;,
\end{equ}
where ``$*$'' denotes space-time convolution and $\xi$ denotes space-time
white noise. To keep things simple, we will assume that $\rho$ is symmetric in space, $\rho(t,x) = \rho(t,-x)$\footnote{This is used in a few places such as \eqref{e:iden1} or \eqref{e:exprWick}.  Without the symmetry, one has to make further subtractions, which manifest themselves as global drifts in the resulting equation which then have to be removed by shifts, see \cite{HaoShen}.  In order not to complicate things even further, we do not pursue this here.}   Note that, in law, the field $\xi^{(\eps)}$ is obtained from $\xi_1$ as above by
a suitable parabolic rescaling:
\begin{equ}
\xi^{(\eps)}(t,s) \eqlaw \eps^{-3/2}\xi_1(\eps^{-2}t, \eps^{-1}x)\;.
\end{equ}
We furthermore define a constant $C_0$ by
\begin{equ}[e:defC0]
C_0 = \iint \bigl(P'* \rho\bigr)(t,x)^2\,dt\,dx\;,
\end{equ}
where $P$ denotes the heat kernel on $[0,2\pi)$ with periodic boundary conditions.
This constant can be rewritten using a graphical notation which will save a 
great deal of space later.
Writing \tikz[baseline=-0.1cm] \draw[keps] (0,0) -- (1,0); for the kernel $\rho_\eps * K'$,
a black dot for an integration variable, and a green dot for the value $0$,
it follows from the definition of $K$, the scaling invariance of the heat kernel 
and the fact that $\rho$ has compact support
that one has
\begin{equ}[e:defrhopicture]
\begin{tikzpicture}[baseline=10,scale=0.5]
\node at (0,2) [dot] (1) {};
\node at (0,0) [root] (0) {};

\draw[keps] (1) to[bend left=60] (0);
\draw[keps] (1) to[bend right=60] (0);
\end{tikzpicture}
= 
{C_0 \over \eps} + \CO(1)
\;.
\end{equ}

We now consider \eref{e:weakNoiseRescaledGeneral} with
$\alpha = 2p-1$. Performing the substitution $\eps^{2p-1} \mapsto \eps$, this can be rewritten as
\begin{equ}[e:intDis]
\d_t h_\eps  = \d_x^2 h_\eps + \eps^{-{3p-1\over 2p-1}} F\bigl(\eps^{p\over 2p-1} \d_x h_\eps\bigr) + \xi^{(\eps)}\;.
\end{equ}
As usual, we consider \eref{e:intDis} on a finite interval with periodic boundary conditions.
We now make use of the fact that, by assumption, $F$ is smooth and $F(u) \sim u^{2p}$ near $u = 0$, so that
one can write
\begin{equ}[e:defakF]
F(u) = \sum_{k=0}^{2p-1} a_{p+k} u^{2(p+k)} + \tilde F(u)\;,
\end{equ}
where $\tilde F$ is a smooth function such that $|\tilde F(u)| \le |u|^{6p}$ for $|u| \le 1$.
Substituting this into \eref{e:intDis}, we obtain the equation
\begin{equ}[e:rewritten]
\d_t h_\eps  = \d_x^2 h_\eps + \sum_{k=0}^{2p-1} a_{p+k} \eps^{p-1+{2pk \over 2p-1}} (\d_x h_\eps)^{2(p+k)} + \eps^{-{3p-1\over 2p-1}} \tilde F \bigl(\eps^{p\over 2p-1} \d_x h_\eps\bigr) + \xi^{(\eps)}\;.
\end{equ}
With this notation at hand, we have the following result.

\begin{theorem}\label{theo:convDisorder}
Let $p \ge 1$ be an integer, let $F \in \CC^{6p+1}$ with $F^{(2p)}(0) \neq 0$ and 
$F^{(k)}(0) = 0$ for $k < 2p$, and let $h_\eps$ be the solution to
\eref{e:intDis} with initial condition $h_0^{(\eps)} \in \CC^{\gamma,\eta}_\eps$ for 
$\gamma = 2 - {1\over 96p}$ and $\eta \in ({1\over 2} - {1\over 12p^2}, {1\over 2})$ with $\|h_0^{(\eps)}, h_0\|_{\gamma,\eta;\eps} \to 0$. Then there exists a constant $c \in \R$ such that 
$h_\eps - (C_\eps  + c)t$ converges in probability to $h_\Hopf^{(\lambda)}$ with 
initial condition $h_0$ where \begin{equ}[e:defLambdaIntermediate]
\lambda = a_p   {C_0^{p-1} (2p)! \over 2^p (p-1)!}\;,\qquad
C_\eps = \sum_{k=p}^{3p-1} a_k \eps^{-{3p - 1 - k \over 2p-1}}  {C_0^k (2k)! \over 2^k k!}\;.
\end{equ}
\end{theorem}
 
%
%
\subsection{Weakly asymmetric scaling}

Let us now consider the weakly asymmetric regime
\begin{equ}[e:weakAsym]
\d_t h = \d_x^2 h + \sqrt \eps F(\d_x h) + \xi_1\;.
\end{equ}
In this case, provided that $f$ has a non-vanishing second derivative as before,
the natural scaling is given by $\tilde h_\eps(x,t) = \eps^{1\over 2}h(\eps^{-1}x, \eps^{-2}t)$.
With such a scaling, we obtain for $\tilde h$ the equation
\begin{equ}[e:weakAsymRescaled]
\d_t \tilde h_\eps = \d_x^2 \tilde h_\eps +  \eps^{-1}  F(\eps^{1\over 2}\d_x \tilde h_\eps) + \xi^{(\eps)}\;.
\end{equ}
Formally replacing $f$ by its Taylor series as before and neglecting terms of positive order in $\eps$,
we obtain this time 
\begin{equ}
\d_t \tilde h_\eps = \d_x^2 \tilde h_\eps +  \eps^{-1}a_0  + a_1 (\d_x \tilde h_\eps)^2 + \CO\bigl(\eps (\d_x \tilde h_\eps)^4\bigr) + \xi^{(\eps)}\;.
\end{equ}
Comparing this to \eref{e:weakNoiseRescaled}, we see that now the ``error term'' is much larger,
so that it is less clear whether this still converges to the KPZ equation. It turns out that it still does,
but the ``error terms'' do not vanish in the limit. Instead, at all orders they  contribute  to the limiting asymmetry
constant $\lambda$ of the KPZ equation \eqref{e:KPZ}.

%
%
%

\begin{theorem}\label{theo:weakAsym}
Let $F\colon \R \to \R$ be an even polynomial of degree $2m$, let $\eta\in (\tfrac12 - \frac{1}{4m}, \tfrac12)$ and $\gamma=2-\tfrac1{32m}$, and let $h_0^{(\eps)}$ be a sequence of functions in $\CC_\eps^{\gamma,\eta}$ such
that there exists $h_0 \in \CC^\eta$ with $\lim_{\eps \to 0}\|h_0^{(\eps)}; h_0\|_{\gamma,\eta;\eps} = 0$ in probability.
Let $h_\eps$ be the solution to
\begin{equ}[e:approxSolution]
\d_t h_\eps = \d_x^2 h_\eps + \eps^{-1}F\bigl(\sqrt \eps \d_x h_\eps\bigr) + \xi^{(\eps)}\;,
\end{equ}
where $\xi^{(\eps)}$ is as in \eref{e:xiEps}. Let $C_0$ be given by \eref{e:defC0}, let $\mu_0$ be the
centred Gaussian measure on $\R$ with variance $C_0$, and let
\begin{equ}[e:defLambda]
\lambda = {1\over 2}\int F''(x)\,\mu_0(dx)\;,\qquad
\hat \lambda = \int F(x)\,\mu_0(dx)\;.
\end{equ}
Then there exists a constant $c$ such that, for every $T > 0$, the family of random functions
$(t,x) \mapsto h_\eps(t,x) - (\eps^{-1}\hat \lambda  + c)t$ converges to $h_\Hopf^{(\lambda)}$
in probability in $\CC^\eta([0,T] \times S^1)$.
\end{theorem}

\begin{remark}
At this stage, it seems very difficult to obtain uniform moment bounds on 
solutions to \eref{e:approxSolution} as $\eps \to 0$. Therefore, it is unrealistic 
to expect a much stronger notion of convergence than convergence in probability.
\end{remark}

\begin{remark}
We would like to emphasize again that a na\"\i ve guess would be that,
after being appropriately centred, $h_\eps$ converges to $h_\Hopf^{(\lambda)}$
with $\lambda = a_1$. It is plain from \eref{e:defLambda} that this is \textit{not} 
the case. Instead, each of the higher order terms yields a non-trivial contribution
in the limit, although they formally disappear as $\eps \to 0$ 
in \eref{e:approxSolution}. Another remark is that the constant 
$\hat \lambda$ which determines
the average speed of the interface $h_\eps$ 
is in general \textit{different} from $\eps^{-1}C_0 \lambda$,
which is what one would obtain when replacing the nonlinearity
by $\lambda (\d_x h_\eps)^2$. Finally, note that the additional constant $c$
that needs to be subtracted in order to obtain the Hopf-Cole
solution depends in a very non-trivial (actually trilinear) way on 
all of the coefficients of $P$.
\end{remark}

\begin{remark}
A piece of physics lore is that white noise is invariant for the generalized stochastic Burgers equation
\begin{equ}[e:generalizedburgers1]
\d_t u = \d_x^2 u + \d_x F\bigl( u) + \d_x\xi\;,
\end{equ}
for any polynomial $F$.  Here one simply thinks of $u=\d_x h$ and $h$ is then a solution to the polynomial
KPZ, which, as we learn in this article, simply means quadratic KPZ with a non-trivially renormalized $\lambda$.
So the invariance of the white noise for \eqref{e:generalizedburgers1} would appear to be a statement with little new content beyond the white noise invariance for the quadratic case.  It is worth remarking however that if we convolve the noise in space only: $\xi^{(\eps)} (t,x) = \int \xi(t, y) \rho_\eps(x+y) dy$ where $\rho$ is a non-negative, symmetric function of total integral $1$ and $\rho_\eps(x) =\eps^{-1}
\rho(\eps^{-1}x)$, then white noise convolved with $\rho_\eps$ is always invariant for the approximating equation
\begin{equ}[e:generalizedburgers]
\d_t u_\eps = \d_x^2 u_\eps + \d_xC_{2,\eps}\left( F\bigl( u_\eps)\right) + \d_x\xi^{(\eps)}\;,
\end{equ}
where $C_{2,\eps}f$ denotes convolution with the $\eps$-rescaling of $\rho_2 = \rho\ast\rho$ (and  also the covariance operator of $\xi^{(\eps)}$.)  This can be shown by adapting \cite[Thm. 2.1]{Funaki:2014aa}, which makes the following argument about the Burgers flow $\d_t u_\eps = \d_x C_{2,\eps}\left( F\bigl( u)\right)$ rigorous:
\begin{equs}
&\partial_t \int f(u(t))e^{-\frac{1}{2}\langle u, C^{-1}_{2,\eps}u\rangle}
 =   \int \langle \frac{\delta f}{\delta u},  
\d_xC_{2,\eps}(F(u))\rangle 
e^{-\frac{1}{2}\langle u, C^{-1}_{2,\eps}u\rangle}\\ & \qquad
=- \int f\Big\langle \frac{\delta }{\delta u} \Big( 
\d_xC_{2,\eps}(F(u))
e^{-\frac{1}{2}\langle u, C^{-1}_{2,\eps}u\rangle}\Big)\Big\rangle =0\,, 
\end{equs}
where $\langle f, g\rangle =\int fgdx$ and the last term vanishes because of the following:  By Leibniz rule
$ \frac{\delta }{\delta u}  (
\d_xCF
e^{-\frac{1}{2}\langle u, C^{-1}u\rangle})$$= $$\frac{\delta }{\delta u}  (
\d_xCF)
e^{-\frac{1}{2}\langle u, C^{-1}u\rangle} $$+ $$
\d_xCF\frac{\delta }{\delta u}  (
e^{-\frac{1}{2}\langle u, C^{-1}u\rangle})$.  The first   term $\frac{\delta }{\delta u}  
\d_xC_{2,\eps}(F(u))= C_{2,\eps}(F^{''}(u)\d_x u)= \d_xC_{2,\eps}(F^{'}(u))$ integrates to zero because
it is an exact derivative.  The second is as well, but this uses the more subtle fact that $\frac{\delta }{\delta u} e^{-\frac{1}{2}\langle u, C^{-1}_{2,\eps}u\rangle} =C^{-1}_{2,\eps}u$ and $\langle \d_xC_{2,\eps}(F(u)) , C^{-1}_{2,\eps}u\rangle= \langle F'(u) \partial_xu, u\rangle=0$ because if $G'(u)=uF'(u)$ then $\partial_x G(u) = uF'(u) \partial_xu$.
\end{remark}

\subsection{Possible generalisations}

Although the class of models \eref{e:growth} considered in this article is quite rich, 
we have placed a number of rather severe restrictions on it and it is legitimate to ask
whether they are genuinely necessary for our universality result to hold. We now
discuss a number of these restrictions and possible strategies for relaxing them.

\begin{enumerate}
\item {\bf Regularity of $F$.} In the weakly asymmetric limits we assume that $F$ is a polynomial.
  The formulation of Theorem~\ref{theo:weakAsym} suggests that this is not an essential
  assumption since the limiting values of $\lambda$  and 
  $\hat \lambda$ appearing in the statement are finite for any function (even distribution)
  $F$ which is sufficiently tame at infinity. This is a strong hint that it is probably
  sufficient to impose that $F$ satisfies a suitable growth condition and is 
  locally Lipschitz continuous. It is not clear at this
  stage however if and how the theory of regularity structures used in this article could be 
  tweaked to cover this case.  
  
  The restriction to \emph{even} polynomials is natural
  because of the lack of a preferred direction, but it is not 
  really important for our proof.  Odd polynomials 
  produce large spatial shifts, which simply add a layer of complication to the argument.  
  It is important  to note that we are not using the large scale convexity of 
  the even polynomial in any way; none of our arguments use convexity at all.
  
  \item {\bf Gaussianity of $\xi^{(\eps)}$.} At the microscopic level, there is no a priori reason for the 
  randomness to be described by Gaussian noise. One may ask whether the arguments in this article
  still hold if $\xi$ is an arbitrary smooth and stationary space-time random field with suitable
  integrability and mixing conditions. (Think of conditions similar to those considered 
  in \cite{PP12,HPP13}.) The only part of the paper where we use Gaussianity is in 
  Section~\ref{sec:convModel}. In principle, one would expect these results to hold also for suitable non-Gaussian
  noises (with the same limit). This was done in \cite{HaoShen} for the particular case when
  $F$ is quadratic, but the technique employed there should also work for the general case.
  
  \item {\bf Smoothing mechanism.} One could replace the smoothing mechanism $\d_x^2$ in \eref{e:weakAsym}
  by a more general (pseudo-)differential operator of the type $Q(i\d_x)$ for an even polynomial
  (or suitable smooth function) $Q$. Provided that $Q(0) = 0$, $Q''(0) < 0$, and 
  $\lim_{|k| \to \infty} Q(k) = -\infty$,
  one would expect essentially the same results to still hold true since the large-scale behaviour
  of the fundamental solution for $\d_t - Q(i\d_x)$ is still described by the heat kernel. 
  Unfortunately, the convergence of the rescaled fundamental solutions does not take place in a topology allowing
  to easily reuse the results of \cite{Regularity} in this case, although one would still
  expect the general theory to apply, at least for some choices of $Q$.
\item {\bf Symmetry.} Our model is symmetric for the reflection $x \mapsto -x$. This symmetry could
  be broken by considering uneven nonlinearities $F$ or, in one of the previously discussed generalisations,
  by considering asymmetric processes $\xi$ or uneven functions $Q$. The expectation is that in this
  case one should consider limits of the type $\tilde h_\eps(x-c_\eps t, t) - C_\eps t$, where the constant
  $c_\eps$ is also allowed to diverge. The correct choice of these diverging constants should however
  again lead to the Hopf-Cole solution of the KPZ equation. 
\item {\bf The ``balanced'' weakly asymmetric case.} In the weakly asymmetric case, it may happen
  that the constant $\lambda$ in \eref{e:defLambda} is equal to $0$. This situation is non-generic
  as it requires a very fine balance between all ingredients of the model (since the variance of $\mu$
  in \eref{e:defLambda} depends on the details of both the noise and the smoothing mechanism). 
  In this situation, our results imply that the limiting process
  is given by the (additive) stochastic heat equation. One might ask whether, similarly to the intermediate disorder case,
  it is then possible to consider the model on larger scales and still obtain convergence to KPZ
  (or some other non-Gaussian process). By analogy with what happens in the context of lattice gases,
  we do not expect this to be the case \cite{Diffusive}. 
  
  \item  {\bf Unbounded space.}  Our results are on a finite interval with periodic boundary conditions, and 
  extending them to the real line represents a challenge.  Recently, \cite{Hairer:2015ab} introduced weighted spaces
  allowing the extension of the results on convergence of smoothed noise approximations of the quadratic KPZ equation to the whole line.  However, these use in an essential way the Hopf-Cole transformation, which is not
  available for the non-quadratic versions considered in this article.

%
\end{enumerate}

\subsection{Standing assumptions and terminology}
\label{1point5}

Throughout the article, we will consider stochastic processes $h$ taking values in some
Banach space $\BB$. (Typically a space of periodic H\"older continuous functions on $\R$.)
Since we consider equations with polynomially growing coefficients, we allow for solutions with
a finite lifetime $\tau$ such that $\lim_{t \to \tau} \|h(t)\|_\BB = \infty$ on $\{\tau < \infty\}$.  One way of formalising this is to consider, for each $T > 0$, the space
$\bar \CC_T(\BB)$ of continuous $\BB$-valued functions $h \colon [0,T] \to \BB$ endowed with 
 a ``point at infinity'' $\infty$ for which we postulate that
\begin{equ}
d(h, \infty) = d(\infty,h) =  \Bigl(1 + \sup_{t \le T} \|h(t)\|_\BB\Bigr)^{-1}\;.
\end{equ}
For any two elements $h,\tilde h \neq \infty$,
we then set 
\begin{equ}
d(h,\tilde h) = d(h,\infty) \wedge d(\tilde h,\infty) \wedge \sup_{t \le T} \|h(t) - \tilde h(t)\|_\BB\;.
\end{equ}

For fixed $T$, we can then view a process $h$ with lifetime $\tau$ as a random variable 
in $\bar \CC_T(\BB)$ with the understanding that it is equal to $\infty$ if $\tau \le T$.
Throughout the remainder of this article, when we state that a sequence of $\BB$-valued
processes $h_\eps$ converges in probability to a limit $h$, this is a shorthand for the fact that the
corresponding $\bar \CC_T(\BB)$-valued random variables converge for every choice of 
final time $T> 0$.

Throughout the text, we will make use of the parabolic distance on space-time: 
if $z=(t,x)$ then we write $|z| = \|z\|_\s= |t|^{1/2} + |x|$.  We always work on a  domain $z\in [-1,T+1]\times S^1$ where
$S^1=[0,L)$ with periodic boundary conditions, and we will often write $\sup_z$ to mean the supremum over 
$z$ in this compact set without further comment.  The time interval here is chosen to be large enough to strictly contain $[0,T]$ where are convergence results take place.

We will also use $\lesssim$ throughout to indicate a bound of the left side by a constant multiple of the right side with a constant independent of the relevant quantities.  When necessary, these will be indicated explicitly.  
\label{standing}

\subsection*{Acknowledgements}

{\small
MH gratefully acknowledges financial support from the Leverhulme trust,
the Royal Society, the ERC, and the Institute for Advanced Study.  JQ gratefully 
acknowledges financial support from the Natural Sciences and Engineering Research 
Council of Canada, the I.~W. Killam Foundation, and the
Institute for Advanced Study.
}

\section{Methodology}

In order to prove theorems~\ref{theo:convDisorder} and \ref{theo:weakAsym}, we 
make use of the theory
of regularity structures as developed in \cite{Regularity,Notes}.
Let us rapidly recall the main features of this theory. The main
idea is to replace the usual $\CC^\gamma$ spaces of H\"older continuous functions
by  analogues $\CD^\gamma$ obtained by extending the usual 
Taylor polynomials with the  addition of a few special universal  processes 
built from the driving noise.

When trying to follow the methodology developed in \cite{Regularity},
there are two principal  obstacles that must be overcome:
\begin{enumerate}
\item In \eref{e:weakAsymRescaled}, the
parameter $\eps$ appears in two places:  In the regularisation
of the noise, and multiplying the nonlinearity. If one tries to brutally cast this into 
the framework of \cite{Regularity}, one might try to deal with arbitrary polynomial
nonlinearities andview the multiplicative $\eps$ as simply a parameter of the equation.
This is bound to fail since the KPZ equation with a higher than quadratic nonlinearity  fails to satisfy the assumption of local subcriticality
which is key to the analysis of \cite{Regularity}. 

\item Since polynomials of arbitrary degree are allowed in the right hand side
of \eref{e:weakAsym}, the number of objects that need to be explicitly controlled in the limit
$\eps \to 0$ can be very large. In the original article \cite{KPZ}, almost half of the article
was devoted to the control of only five such objects. This was substantially improved in
\cite{Regularity}, but we heavily exploited the fact that most of the objects that require
control for a solution theory to $\Phi^4_3$ can be decomposed as products of convolutions
of integral kernels, for which general bounds exist. In our case, 
we have to deal with generalised convolutions which cannot be broken into simple
convolutions and products. 
\end{enumerate}

The second of these is more of a practical nature, and Appendix~\ref{sec:bounds} contains a very general bound
which allows one to control such generalised convolutions, even in the presence of certain
renormalisation procedures. This bound is then used in Section~\ref{sec:convModel} to give
a relatively short proof of the convergence of the required objects as $\eps \to 0$.
It has also been used in the article \cite{Etienne} to control the necessary objects
to provide a Wong-Zakai theorem for a natural class of SPDEs.

 The first obstacle above is the main new conceptual difficulty.  In a sense, the main point of the regularity structure  in \cite{Regularity} is to remove the $\eps$-regularization of the noise from the problem: The equation with an arbitrary smooth noise forcing it is simply lifted to the $\eps$-independent abstract space.
In this way $\eps$ just takes the role of a parameter in the lifts.
However in the present case, the equation itself is also $\eps$ dependent.  So what we want to do is, as much as possible, separate the two $\eps$'s.
 
To accomplish this, we  build an extension  of the type of regularity structure used in \cite[Sec.~8]{Regularity}   which contains an additional abstract symbol
$\Eps$ representing the multiplicative parameter $\eps$ appearing in the nonlinearity, but \emph{not} the $\eps$ in the noise.
The resulting regularity structure is described in Section~\ref{sec:structure}
and the corresponding renormalisation procedure is described in Section~\ref{sec:renorm}.   To this regularity structure we lift the equation with an arbitrary smooth forcing noise, which does \emph{not} depend on $\eps$.

Only in Section~\ref{sec:convModel}, when we prove convergence of the models, do  we again take the special noise
depending on $\eps$  as in \eqref{e:xiEps}.  For symbols which do not contain $\Eps$ this choice is unnecessary.  But symbols containing $\Eps$ cannot, of course, converge except for this particular choice of approximating noises (or something relatively close). 

This turns out to be possible as long as the initial condition is sufficiently smooth.  However, the results would then
only be valid up to some finite (random) lifetime.  To avoid this, we use the fact that the limit can be identified with
the Hopf-Cole solution of KPZ, which we know independently is global in time.   The difficultly is that one is then  
forced to start with typical data.  This would be slightly below H\"older $1/2$, thus leading to a 
singularity which ruins our fixed point argument.  What saves us is that because of the 
regularization of the noise, the solutions
are really smoother on a small scale than this na\"ive argument suggests. In order to be able 
to exploit this, we introduce 
$\eps$-dependent versions $\CD^{\gamma,\eta}_\eps$ of the $\CD^{\gamma,\eta}$ spaces, which are generalizations to space-time modelled distributions of weighted H\"older spaces, which, 
just like the $\CC^{\gamma,\alpha}_\eps$ spaces defined in \eqref{e:defh}, measure regularity 
differently at scales above and below $\eps$. The parameter $\eta$ appearing here
allows for possible blow-up as $t \to 0$, just as in \cite[Sec.~6]{Regularity} 
(see also Section~\ref{sec:epsilon}). So we cannot completely separate the two $\eps$'s, 
although we try to do it to  the largest extent possible.

As we see here, to a  certain degree, the $\eps$ is producing a small scale cutoff in the problem, below which 
things can be thought of as smooth.  This means that multiplication by $\eps$ effectively increases the homogeneity of a function by $1$, and hence our new symbol $\CE$ acts much like an integration operator.  There are technical differences however.   In the definition of admissible models, $\Pi_x\CI\tau$ is defined in terms of $\Pi_x\tau$ but $\Pi_x \CE\tau$ is not; in fact, there is much more freedom in how it is defined.  Also, $\CE$ doesn't need to kill polynomials.  More strikingly, $\CE$ is not really even an
operator on the regularity structure $\CT$.  The reason is that while we need objects such as $\CE( (\CI'(\Xi))^4)$
to describe the right hand side of our equation (where $\Xi$ is the lift of the noise), we do \emph{not} need
$ (\CI'(\Xi))^4$, and such an object would not converge, whatever the renormalization.  

One unfortunate consequence of these observations is that it makes the structure group highly non-trivial to construct.  However, there is a nice trick.  We construct a larger regularity structure $\CT_\ex$, which \emph{does}
contain objects such as $ (\CI'(\Xi))^4$, and on which $\CE$ acts much more simply as a linear map defined
on a subspace.  On this extended regularity structure, the structure group can be constructed as in \cite{Regularity},
using the formalism of Hopf algebras.  Of course, in $\CT_\ex$, things will not converge in the end, even after 
renormalization.  But our real regularity structure, on which things do converge, is simply a sector of $\CT_\ex$,
so the structure group of $\CT_\ex$ is defined on it by restriction.

\section{Construction of the regularity structure}
\label{sec:structure}

Since the weakly asymmetric case is the more difficult one, we will treat the intermediate
disorder scaling essentially as a perturbation of the weakly asymmetric one. The 
equation of interest is then
\begin{equ}[e:approxKPZ]
\d_th_\eps = \d_x^2h_\eps +  F_\eps\bigl(\d_x h_\eps\bigr)  - C_\eps+ \xi^{(\eps)},
\end{equ}
where $\xi^{(\eps)}$ denotes a regularised version of space-time white noise, the polynomial $F_\eps$ is of the form
\begin{equ}
F_\eps(u) = \sum_{j=1}^m a_j \eps^{j-1} u^{2j}\;,
\end{equ}
for some coefficients $a_j \in \R$ and some finite degree $m \ge 1$ and $P$ is the heat kernel.
Following the methodology of \cite{Regularity}, we would like to build a regularity structure that is
sufficiently large to be able to accommodate an abstract reformulation of \eref{e:approxKPZ}
as a fixed point problem in some space $\CD$ and which is stable in the limit $\eps\searrow 0$.

\subsection{The collection of symbols}
\label{sec:symbols}

Let us first recall how the construction works for the KPZ equation, where only the term with 
$j = 1$ appears in the nonlinearity. In this case, a regularity structure is built in the following way.
We write $\CU$ for a collection of symbols, or formal expressions, that will be useful to 
describe the solution $h$ as a function of space and time, $\CU'$ for a collection of symbols useful to describe its spatial distributional 
derivative $h' = \d_x h$, and $\CV$ for
a collection of symbols useful to describe the terms $F_\eps\bigl(\d_x h_\eps\bigr)  + \xi^{(\eps)}$ on the right hand side of the KPZ equation. 
We decree that $\CU$ and $\CU'$  contain at least symbols representing the usual Taylor
polynomials, i.e. all symbols of the form $X^k$ for $k$ a two-dimensional multiindex $k=(k_1,k_2)$, $k_i\in \{0,1,2,\ldots\}$, representing time and space.

Furthermore, we introduce a symbol $\Xi\in \CV$ describing the driving noise. Finally, we introduce abstract integration maps $\CI$ and $\CI'$ that
represent integration with respect to the heat kernel and its spatial derivative respectively.
In view of the structure of the KPZ equation, it is then natural to decree that 
\begin{equs}[e:rel]
\tau, \bar \tau \in \CU' \quad&\Rightarrow\quad \tau \bar \tau \in \CV\;,\\
\tau \in \CV \quad&\Rightarrow\quad \CI(\tau) \in \CU\;,\quad \CI'(\tau) \in \CU'\;,
\end{equs}
and to \textit{define} $\CU$, $\CU'$ and $\CV$ as the smallest collection of formal expressions
such that $\Xi \in \CV$, $X^k \in \CU$, $X^k \in \CU'$, and \eref{e:rel} holds. For consistency
with \cite{Regularity}, we furthermore decree that $\CI(X^k) = \CI'(X^k) = 0$. In other words,
we only keep formal expressions that do not contain $\CI(X^k)$ or $\CI'(X^k)$ as a sub-expression.
We also decree that $\tau \bar \tau = \bar \tau \tau$ and we denote by $\CW$ the union of these collections of formal expressions:
\begin{equ}
\CW \eqdef \CU \cup \CU' \cup \CV\;.
\end{equ}

We can then associate to any formal expression $\tau$ a homogeneity $|\tau|\in \R$ 
(despite what the notation may suggest $|\tau|$ is not necessarily positive)
in the following way. For any multi-index $k = (k_0, k_1)$, we set
$|X^k| = |k| = 2k_0 + k_1$. Here, $k_0$ denotes the degree of the ``time'' variable, which we
choose to count double in order to reflect the parabolic scaling of the
heat equation. For the symbol representing the driving noise  we set 
\begin{equ}[e:regXi]
	|\Xi| = -{3\over 2} - \kappa\;,
\end{equ}
where $\kappa> 0$ is a fixed small value,
and we extend this recursively to every 
formal expression as follows:
\begin{equ}
|\tau \bar \tau| = |\tau| + |\bar \tau|\;,\qquad |\CI(\tau)| = |\tau| + 2\;,\qquad|\CI'(\tau)| = |\tau| + 1\;.
\end{equ}
With all of these expressions at hand, a simple power-counting argument (see \cite[Sec.~8]{Regularity})
yields the following crucial result.

\begin{lemma}\label{lem:subcritical}
If $\kappa < {1\over 2}$ then for 
every $\gamma \in \R$, the set $\{\tau \in \CW\,:\, |\tau| < \gamma\}$ is finite.
\end{lemma}

This is a reflection of the fact that the KPZ equation is subcritical with respect to the scaling
imposed by the linearised equation. In the context of \eref{e:approxKPZ}, one could think that 
it suffices to replace the first implication in \eref{e:rel} by
\begin{equ}
\tau_1, \ldots, \tau_{2m} \in \CU' \quad\Rightarrow\quad \tau_1 \cdots \tau_{2m} \in \CV\;.
\end{equ}
(Here we exploited the fact that $\one = X^0$ belongs to $\CU'$, so that this automatically covers
the case of products of less than $m$ terms.)  The problem with this definition is that the conclusion of
Lemma~\ref{lem:subcritical} no longer holds, so that it appears as though the theory developed
in \cite{Regularity} breaks down. This is fortunately not the case, but we have to be a little
bit more sophisticated. 

The reason why we can circumvent the problem is of course that the very singular behaviour of the 
higher powers of $\d_x h$ is precisely compensated by the powers of the small parameter $\eps$
that multiply them. It is therefore quite reasonable to expect that we can somehow encode this
into the properties of our regularity structure. The trick is to introduce an additional symbol $\CE$
besides $X$, $\Xi$, $\CI$ and $\CI'$ which symbolises the operation ``multiplication by $\eps$''.
With this new symbol at hand, we build $\CU$, $\CU'$ and $\CV$ as before, but we replace the
first implication of \eref{e:rel} by the implication
\begin{equ}[e:recursion]
\tau_1,\ldots,\tau_{2k} \in \CU' \qquad \Rightarrow \qquad \Eps^{k-1} \tau_1\cdots \tau_{2k} \in \CV\;,
\end{equ}
which we impose for every $k \in \{1,\ldots,m\}$. The product is made commutative and associative by identifying the corresponding formal expressions and making multiplication by $\one$ the identity, and $\Eps^{k}\Eps^{\ell}\tau = \Eps^{\ell+k}\tau$.
At this stage, it is very important to note
that as a consequence of our definitions, there will be formal expressions $\tau$ such that 
$\Eps \tau \in \CW$, but $\tau \not \in \CW$.  For example, $\tau = \CI'(\Xi)^4$. 
This  reflects the fact that 
$\eps (\d_x h)^4 - C_\eps$  converges weakly to a distributional limit as $\eps \to 0$ for a suitable
choice of  $C_\eps$, while $(\d_x h)^4 - C'_\eps$ diverges no matter what is $C'_\eps$.

With these notations, we then define $\CT$ as the linear span of $\CW$ and
we view the symbols $\Eps^{k-1}$ as $2k$-linear maps on $\CT$ via
\begin{equ}
(\tau_1,\ldots,\tau_{2k}) \mapsto \Eps^{k-1} \tau_1\cdots \tau_{2k}\;.
\end{equ}
We furthermore decree that the homogeneity of an element of $\CW$
obtained in this way is given by
\begin{equ}[e:homofE]
|\Eps^{k-1} \tau_1\cdots \tau_{2k}| = k-1 + \sum_i |\tau_i|\;.
\end{equ}

Elements $x \in\CT$ can be written uniquely as $
x = \sum_{\tau \in \CW} x_\tau \,\tau$, $x_\tau \in \R$ and 
with this notation, we  set
\begin{equ}[e:defxalpha]
|x|_\alpha = \sum_{|\tau| = \alpha} |x_\tau|\;,
\end{equ}
with the usual convention that $|x|_\alpha = 0$ for those  $\alpha$ where the sum is empty.

\subsection{Structure group}
\label{sec:structGroup}

We now describe the structure group $\CG$ \label{defofCG} associated to the space $\CT$. For this, we 
first introduce $\CT_+$, the free commutative algebra generated by $\CW_+$ which consists of  $X_0, X_1$ as well as the formal expressions \label{struck}
$\bigl\{\J_\ell(\tau)\,:\, \tau \in \CW \setminus \bar \CT \;, \; |\tau| + 2 > |\ell|\}$ and  \label{e:genT+}
$ \bigl\{\EE_\ell^{k}(\tau)\,:\, \tau \in \CV_{\ell,k} \bigr\}$
where $\ell$ is an arbitrary $2$-dimensional multi-index with $|\ell| = 2 \ell_0 + \ell_1$\label{parlm}, $k$ is an integer with $k \in \{1,\ldots,m-1\}$, $\bar \CT$ is the subset generated by the $X^k$, and 
$\CV_{\ell,k}$ is the subset of $\CV$ consisting of $\tau$ of the form 
$\tau_1\cdots \tau_{2k+2}$, $\tau_i \in \CU'$ with 
$|\ell| \ge \sum|\tau_i| > |\ell| - k$.
Note that for the moment, elements of $\CT_+$ are formal objects. They are only used to index
matrix elements for the linear transformations belonging to the structure group of our
regularity structure.  The scheme is as follows:  Starting from these formal objects, we will define $\Delta$ by
\eqref{e:defDelta1} and \eqref{e:defDelta}.  The structure group is then defined by \eqref{e:Gammadef}.

\begin{remark}
In principle, there is no \textit{a priori} reason to impose that
$|\ell| \ge \sum|\tau_i|$ in the second line of \eqref{e:genT+} (the analogous 
constraint does not appear for $\J_\ell(\tau)$ for example). The reason why have imposed
this here is twofold. First, it is natural in view of the canonical lift defined in 
\eqref{e:canonical} below in the sense that even if we did not impose $\EE_\ell^{k}(\tau) = 0$
for $|\ell| < |\tau|$ at the algebraic level, all of the models we ever consider in this article
involve linear forms $f_z$ over $\CT_+$ such that $f_z(\EE_\ell^{k}(\tau)) = 0$.
The second, more pragmatic, reason is that this greatly simplifies the expression
\eqref{e:RHS} which would otherwise sport a number of spurious additional terms. 
\end{remark}

With this definition at hand, we construct a linear map $\Delta \colon \CT \to \CT\otimes \CT_+$ \label{Delta} in a recursive way. 
In order to streamline notations, we shall write 
$\tau^{(1)}\otimes\tau^{(2)}$ as a shorthand for $\Delta\tau$.
(This is a slight abuse of notation, following Sweedler, since in general
$\Delta\tau$ is a linear combination of such terms. It is justified by the fact that expressions
containing the $\tau^{(i)}$ will always be linear in them.) 
We then define $\Delta$ via the identities
\begin{equ}\label{e:defDelta1}
\Delta\one=\one\otimes\one\;,\qquad
\Delta\Xi= \Xi\otimes\one\;,\qquad
\Delta X_i= X_i\otimes\one+\one\otimes X_i\;,
\end{equ}
and then recursively by the following relations: 
\minilab{e:defDelta}
\begin{equs}
\Delta \tau\overline{\tau}&= \tau^{(1)}\overline{\tau}^{(1)}\otimes \tau^{(2)}\overline{\tau}^{(2)}\;,\label{e:multProp}\\
\Delta\CI(\tau)&=\CI(\tau^{(1)})\otimes \tau^{(2)}+\sum_{\ell,k}\frac{X^\ell}{\ell!}\otimes\frac{X^k}{k!}
\J_{\ell+k}(\tau)\;,\label{e:intProp}
\\
\Delta\CI'(\tau)&=\CI'(\tau^{(1)})\otimes \tau^{(2)}+\sum_{\ell,k}\frac{X^\ell}{\ell!}\otimes\frac{X^k}{k!}
\J_{\ell+k+1}(\tau)\;,\label{e:intProp2}\\
\Delta\Eps^k(\tau)&=\Eps^k(\tau^{(1)})\otimes \tau^{(2)} + \sum_{\ell,m}\frac{X^\ell}{\ell!}\otimes\frac{X^m}{m!}
\EE^k_{\ell+m}(\tau)\;.\label{e:epsProp}
\end{equs}
Here, we write $\ell+k+1$ as a shorthand for $\ell+k+(0,1)$, where $(0,1)$ is the multiindex corresponding
to the spatial direction. We also implicitly set $\J_{k}(\tau) = 0$ 
if $|\tau| \le |k| - 2$ and $\EE^k_{\ell}(\tau) = 0$ if $|\tau| \le |\ell| - |k|$
or $|\tau| > |\ell|$ so these sums, as well as the corresponding ones in the sequel, are all finite.

Finally, we define a linear map $\DD$ on all elements of the type $X^k \CI(\tau)$ by \label{DD}
$\DD \CI(\tau) = \CI'(\tau)$, $\DD X^k = k_1 X^{k-(0,1)}$ for every $k \ge (0,1)$, $\DD \one = 0$,
and by extending it using the Leibnitz rule. 
It then follows immediately from \eqref{e:intProp} and \eqref{e:intProp2} that
$\DD$ commutes with $\Delta$ in the sense that $\Delta \DD\tau = (\DD \otimes I)\Delta \tau$.

\begin{remark}
As already mentioned before, one should really view the $\Eps^{k-1}$ as $2k$-multilinear maps. A more pedantic
way of writing the last line in the above equation  would then be
\begin{equs}
\Delta\Eps^{k-1}(\tau_1,\ldots,\tau_{2k}) &= \Eps^{k-1}(\tau_1^{(1)},\ldots,\tau_{2k}^{(1)})\otimes \tau_1^{(2)}\cdots \tau_{2k}^{(2)} \\
&\qquad + \sum_{\ell,m}\frac{X^\ell}{\ell!}\otimes\frac{X^m}{m!}
\EE^{k-1}_{\ell+m}(\tau_1,\ldots,\tau_{2k})\;.
\end{equs}
However, there is no ambiguity in the above since we implicitly used the fact that 
$\Delta$ extends to arbitrary products of elements of $\CT$ via the multiplicative property.
This abuse of notation is further justified in view of Section~\ref{sec:extended} below.
\end{remark}

For any linear functional $f \colon \CT_+ \to \R$, we can now define in a natural way
a map $\Gamma_{\!f} \colon \CT \to \CT$ by
\begin{equ}\label{e:Gammadef}
\Gamma_{\!f} \tau = (I \otimes f)\Delta \tau\;.
\end{equ}
Let now $\CG_+$ \label{defofCGplus} denote the set of all such linear functionals $f$ which are multiplicative in the sense that 
$f(\tau \bar \tau) = f(\tau)f(\bar \tau)$ for any two elements $\tau, \bar \tau \in \CT_+$. With this definition
at hand, we set
\begin{equ}
\CG = \{\Gamma_{\! f}\,:\, f \in \CG_+\}\;.
\end{equ}
It is not difficult to see that these operators are ``lower triangular'' in the sense that
\begin{equ}
\tau \in \CT_\alpha \quad\Rightarrow\quad \Gamma_f \tau - \tau \in \bigoplus_{\beta < \alpha} \CT_\beta\;,
\end{equ}
but it is
not obvious that the set $\CG$ does indeed form a group
under composition. In the case where the symbols $\Eps^k$ are absent, a proof is given in 
\cite[Sec.~8.1]{Regularity}. In our situation, we note that from a purely algebraic
point of view, the only thing that distinguishes $\Eps^k$ from an abstract integration operator
of order $k$ is that it does not annihilate polynomials. This property was however never used
in \cite[Sec.~8.1]{Regularity}. The only reason why this property was imposed in \cite{Regularity}
is the aesthetic consideration that we do not want to have a proliferation of abstract 
symbols that all encode smooth functions, as this would lead to more redundancy in the theory.

%
%

\begin{remark} 
While the symbol $\Eps$ should be thought as ``multiplication by $\eps$''
and the models we consider will typically implement this by satisfying the relation \eref{e:canonical2} below, we do not impose that relation. In particular, no
real number $\eps$ needs to be specified in general for the notion of 
an ``admissible model'' to make sense.
As a matter of fact, while there are natural limiting models with ``$\eps = 0$'' 
for which $\Pi_x \tau = 0$ whenever
$\tau$ contains at least one factor $\Eps$, there are also limiting models for which this is not the case.
\end{remark}

\begin{remark} 
We do \textit{not} impose the identity $\CI(\Eps \tau) = \Eps \CI(\tau)$, which would
in principle have been natural given the interpretation of $\Eps$ as essentially 
multiplication by $\eps$. The reason for this is that if we had done this,
then we would have run into consistency problems when trying to also impose that 
$\Eps$ increases homogeneity by $1$.
\end{remark}

\subsection{The extended regularity structure \texorpdfstring{$\CT_\ex$}{Tex}}
\label{sec:extended}

Before we proceed, we ``trim'' the regularity structure $(\CT,\CG)$ to the bare minimum
required for the right hand side of \eqref{e:abstract} to make sense as a map
from $\CD^\gamma$ into itself for $\gamma \in ({3\over 2} + \kappa, 2- (6m-2)\kappa)$.
From now on, with $\bar \CT$ the usual Taylor polynomials as before, we set
\begin{equ}\label{e:defCWbar}
\CT = \bar \CT \oplus \scal{\bar \CW}\;,\qquad \bar \CW = \bar \CU' \cup \bar \CV \cup \{\CI(\tau)\,:\, \tau \in \bar \CV\}\;,
\end{equ}
with
\begin{equs}
\bar \CU' &= \Big\{\tau\in \CU'\,:\,|\tau| < {3\over 4}\Big\}\;,\\
\bar \CV &= \Big\{\CE^{k-1}(\tau_1\cdots \tau_{2k})\,:\, k \in \{1,\ldots,m\}\,,\;\tau_i \in \bar \CU'\,,\;
\sum_{j=1}^{2k}|\tau_j|\le 0\Big\}\;,
\end{equs}
where we implicitly used the identification $\Eps^0(\tau) = \tau$.
Setting furthermore
\begin{equ}
\CU'_\ex =  \bigl\{\tau_1\cdots \tau_{2m} \;:\; \tau_i \in \bar\CU'\bigr\}\;,
\end{equ}\label{defofUprimeex}
we also define $\bar \CW_+$ \label{defofwbarplus} to consist of $\{X_0,X_1\}$, as well as those elements 
in $\CW_+$ of the form $\J_\ell(\tau)$ and $\EE_\ell^k(\bar \tau)$ for elements 
$\tau, \bar \tau \in \CW$ such that $\CI(\tau) \in \bar \CW$ and 
$\bar \tau \in \CU'_\ex$. With this definition at hand, we define $\CT_+$ as the free commutative algebra
generated by $\bar \CW_+$.
It will also be very convenient in the sequel to consider an ``extended'' regularity structure whose
structure space $\CT_\ex$ is given by
\begin{equ}
\CT_\ex = \bar \CT \oplus \scal{\CW_\ex}\;,\qquad \CW_\ex = \bar \CW \cup \CU'_\ex\;.
\end{equ}\label{defofwex}
In particular, if we extend the definition of $\Delta$ to elements in $\CU'_\ex$ by
imposing that it is multiplicative, our definitions guarantee that 
\begin{equ}[e:stable]
\Delta \colon \CT_\ex \to \CT_\ex \otimes \CT_+\;,\qquad
\Delta \colon \CT \to \CT \otimes \CT_+\;,
\end{equ}
i.e.\ both $\CT$ and $\CT_\ex$ are stable under the action of $\CG_+$,
so that $(\CT_\ex, \CG)$ is again a regularity structure
and $(\CT,\CG)$ can be viewed as a sector of $(\CT_\ex, \CG)$, i.e. a subspace
that is stable under $\CG$ and diagonal with respect to the direct sum decomposition of $\CT_\ex$.
The key point of \eqref{e:stable} is that the same space $\CT_+$ suffices to define the structure 
group for $\CT_\ex$, and therefore \textit{the structure group for 
$\CT_\ex$ is the same as the structure group for $\CT$}.

A key advantage of $\CT_\ex$, and this is why we introduce it, 
is that for $\ell \le m-1$, the maps $\Eps^\ell$ can be viewed as genuine \emph{linear}
maps defined on the subspace of $\CT_\ex$ generated by $\tau_1\cdots\tau_{2\ell+2}$, $\tau_i\in \CU'$.  
However, although $\Eps^\ell$ and $\CI$ are defined on subspaces of $\CT_\ex$, they 
are not necessarily defined for all elements of $\CT_\ex$. 
The other main advantage of $\CT_\ex$ is that all of its elements can be obtained
from the ``basic'' elements $\{\one,X_0,X_1,\Xi\}$ by application of the operators
$\CI$, $\CI'$ and $\Eps^\ell$ without ever leaving $\CT_\ex$. In fact, it is the minimal 
extension of $\CT$ with that property.

Since $\CT \subset \CT_\ex$ and since the structure groups are the same for both
regularity structures, every model\footnote{See \cite{Regularity} or Section~\ref{sec:models}.} $(\Pi,\Gamma)$ for $(\CT_\ex,\CG)$ defines
a model for $(\CT,\CG)$ by restriction. 
We will use this fact in Section~\ref{sec:canonical} by defining a model first recursively on $\CT_\ex$ and
then on $\CT$ by restriction. 
On the other hand, one should remember if $(\Pi,\Gamma)$ is a model for
the structure $(\CT,\CG)$, it does \textit{not} automatically extend to a model for
$(\CT_\ex,\CG)$. As a matter of fact, we are precisely interested in the limiting situation 
in which it does not!
Since the structure group $\CG$ is identical for both structures however, the 
family of operators $\Gamma_{zz'}$ can be viewed as acting on $\CT_\ex$ for any model
on $\CT$.  In particular,  the spaces 
$\CD^\gamma$  also make sense over $\CT_\ex$ (see Section~\ref{sec:Dgamma} below
for the definition of these spaces and their variants), even for models on $\CT$.
Since $\Eps^\ell$ can be thought of as linear operators on $\CT_\ex$,  this will
give a simple way to understand the fixed point argument.


\subsection{Admissible models}
\label{sec:models}

From now on, we also 
set $\CT = \bigoplus_{\alpha \in A\,:\, \alpha \le 2} \CT_\alpha$, which has the advantage that
$\CT$ is finite-dimensional so we do not need to worry about topologies.
In order to describe our ``polynomial-like'' objects, 
we first fix a kernel $K \colon \R^2 \to \R$ with the following properties:
\begin{enumerate}
\item The kernel $K$ is supported in $\{|z| \le 1\}$,  
$K(t,x) = 0$ for $t \le 0$, and $K(t,-x) = K(t,x)$.
\item For $z$ with $|z| < 1/2$, $K$ coincides with the heat kernel and $K$ is smooth outside of the origin.
\item For every polynomial $Q \colon \R^2 \to \R$ of parabolic degree $2$ or higher, one has
\begin{equ}[e:killPoly]
\int_{\R^2} K(t,x) Q(t,x)\,dx\,dt = 0\;.
\end{equ}
\end{enumerate}
in other words, $K$ has essentially all the properties of the heat kernel, except that it is furthermore
compactly supported and satisfies \eref{e:killPoly}.
The existence of a kernel $K$ satisfying these properties is very easy to show.

\begin{remark}
The identity \eqref{e:killPoly} is imposed only for convenience. If we didn't impose this,
then in order to be able to impose \eqref{e:admissible2} later on we would have to add symbols of the 
type $\CI(X^k)$ which would also describe smooth functions. This would introduce some rather unnatural
redundancy into the construction.
\end{remark}

Let $\CS'$ be the space of Schwartz distributions on $\R^2$ and  $\CL(\CT,\CS')$ the space
of  linear maps from $\CT$ to $\CS'$. Furthermore, given
a continuous test function $\phi\colon \R^2 \to \R$ and a point $z = (t,x) \in \R^2$, we set
\begin{equ}
\phi_z^\lambda(\bar z) = \lambda^{-3} \phi\bigl((\lambda^{-2}(\bar t - t), \lambda^{-1}(\bar x - x)\bigr)\;,
\end{equ}
where we also used the shorthand $\bar z = (\bar t, \bar x)$. Finally, we write $\CB$ for the set
of functions $\phi \colon \R^2 \to \R$ that are smooth, compactly supported in the ball of radius one,
and with their values and both first and second derivatives bounded by $1$.

Given a kernel $K$ as above, we then introduce a set $\MM$ of \textit{admissible models} \label{admissible} which
are analytical objects built upon our regularity structure $(\CT,\CG)$ that will play a role for our solutions that is
similar to that of the usual Taylor polynomials for smooth functions.
A \textit{model} (not necessarily admissible) for $\CT$ on $\R^2$  consists of a pair $(\Pi,\Gamma)$ of functions
\begin{equs}[2]\label{eq:models}
\Pi \colon \R^2 &\to \CL(\CT,\CS') \quad & \quad \Gamma \colon \R^2\times \R^2  &\to \CG \\
z &\mapsto \Pi_z & (z,\bar z) &\mapsto \Gamma_{z\bar z} 
\end{equs}
with the following properties. First, we impose that
they satisfy the analytical bounds 

\begin{equ}[e:bounds]
\bigl|\bigl(\Pi_z \tau\bigr)(\phi_z^\lambda)\bigr| \lesssim \lambda^{|\tau|}\;,\qquad
\bigl\|\CQ_\alpha \Gamma_{z\bar z}\tau \bigr\| \lesssim |z-\bar z|^{|\tau|-\alpha}\;, 
\end{equ}
uniformly over $\phi \in \CB$, $\lambda \in (0,1]$, $\tau \in \CW$, and $\alpha$ with
$\alpha \le |\tau|$, where $\CQ_\alpha$ denotes the projection onto $\CT_\alpha$. 
Also, the proportionality constants implicit in the notation $\lesssim$ are assumed to be bounded
uniformly for $z$ and $\bar z$ taking values in any compact set. 
We furthermore assume
that one has the algebraic identities
\begin{equ}[e:algebraic]
\Pi_{z} \Gamma_{z\bar z} = \Pi_{\bar z}\;,\qquad \Gamma_{z\bar z} \Gamma_{\bar z\bbar z} =\Gamma_{z\bbar z} 
\end{equ}
valid for every $z, \bar z, \bbar z$ in $\R^2$. 

\begin{remark}  It is important to note that \eqref{e:bounds} is the 
crux of the whole theory of regularity structures, providing a concrete
meaning to the 
abstract notion of homogeneity.  It is to make \eqref{e:bounds} hold
that one is forced to make the subtractions in \eqref{e:admissible3} and \eqref{e:admissible4}, 
which then produces the non-trivial algebraic structure.
\end{remark}

In this article, we will always consider \textit{admissible} models that come with some 
additional structure. Our models will actually consist of pairs $(\Pi,f)$ where
$\Pi$ is as in \eqref{eq:models} and $f \colon \R^2 \to \CG_+$ is a continuous function 
such that, if we define
\begin{equ}[e:defGamma]
\Gamma_{z\bar z} = \Gamma_{f_z}^{-1} \Gamma_{f_{\bar z}}\;,
\end{equ}
the properties \eqref{e:bounds} and \eqref{e:algebraic} are satisfied.
In other words, we assume that there exists one
\textit{single} linear map $\PPi \in \CL(\CT, \CS')$,\label{defofschwartz} where $\CS'$ is the dual of smooth functions,  such that
$\Pi_z = \PPi F_z$ for every $z$, where $F_z = \Gamma_{f_z}$.

Note also that elements of $\CG_+$ contain strictly more information than the corresponding
element of $\CG$. This is because the range of $\Delta$ on $\CT_\ex$ is actually contained
in $\CT_\ex \otimes \tilde \CT_+$, where $\tilde \CT_+ \subset \CT_+$ is the subalgebra
generated only by the $X_i$ and elements of the type $\J_\ell(\tau)$. 
The bound \eqref{e:bounds} then yields some regularity assumption on the action of $f_z$ on
$\tilde \CT_+$, but not on its action on elements of the form $\EE^k_\ell(\bar \tau)$.
The reason why we still need these elements will be clear from the construction of the operators
$\hat \CE^k$ in Section~\ref{sec:epsk}.
We will also impose more stringent bounds on $f_z(\EE^k_\ell(\bar \tau))$ in 
Section~\ref{sec:epsilon} below.


\begin{definition}\label{def:admissible}
A model $(\Pi, f)$ as above is admissible on $\CT$ if $\Pi_z \one = 1$, for every multiindex $k$,
\minilab{e:admissible}
\begin{equ}\label{e:admissible1}
\bigl(\Pi_z X^k\tau\bigr)(\bar z) = (\bar z- z)^k\bigl(\Pi_z \tau\bigr)(\bar z) \;,\qquad f_z(X^k) = (-z)^k\;,
\end{equ}
and, for every $\tau \in \CW$ with $\CI(\tau) \in \CT$, one has the identities
\minilab{e:admissible}
\begin{equs} 
f_z(\J_k \tau) &=  -\int_{\R^2} D^k K(z - \bar z)\bigl(\Pi_{z} \tau\bigr)(d\bar z) \;, \qquad |k| < |\tau|+2\;,\label{e:admissible2}\\
\bigl(\Pi_z \CI \tau\bigr)(\bar z) &=  \int_{\R^2} K(\bar z - \bbar z)\bigl(\Pi_{z} \tau\bigr)(d\bbar z) + \sum_{k} {(\bar z - z)^k \over k!} f_z(\J_k \tau) \;,\label{e:admissible3}\\
\bigl(\Pi_z \CI' \tau\bigr)(\bar z) &=  \int_{\R^2} DK(\bar z - \bbar z)\bigl(\Pi_{z} \tau\bigr)(d\bbar z) + \sum_{k} {(\bar z - z)^k \over k!} f_z(\J_{k+1} \tau) \;,\qquad \label{e:admissible4}
\end{equs}
where $D=\partial_x$ and $k+1$ means $(k_0,k_1+1)$. 
\end{definition}

Note that these definitions in particular also guarantee that 
$\bigl(\Pi_z \DD\tau\bigr)(\bar z) = \d_{\bar z}\bigl(\Pi_z \tau\bigr)(\bar z)$
for every $\tau$ in the domain of definition of $\DD$.

\begin{remark}
 Here we set $\J_k \tau = 0$ if $|k| \ge |\tau|+2$, so that the sum appearing in
\eref{e:admissible3} is always finite.
It is not clear in principle that all the integrals appearing in \eref{e:admissible} 
converge, but it turns out that the analytical conditions
\eref{e:bounds} combined with the condition $|k| < |\tau|+2$ 
guarantee that this is always the case, see \cite[Sec.~5]{Regularity}.
\end{remark}

\begin{remark}
Given an admissible model $(\Pi,f)$, we write $\$\Pi\$$ \label{triple}
for the smallest choice of 
proportionality constant in \eqref{e:bounds} with the operators $\Gamma_{z\bar z}$ given 
by \eqref{e:defGamma}. This is a slight abuse of notation since
we should rather write $\$(\Pi,f)\$$ instead, but we hope that this notation is lighter
while remaining sufficiently unambiguous.
Given any two models $(\Pi,f)$, $(\bar \Pi,\bar f)$,
we furthermore write $\$\Pi;\bar \Pi\$$ for the same quantity, but with $\Pi_z$ 
replaced by $\Pi_z - \bar \Pi_z$ and $\Gamma_{z\bar z}$
replaced by 
$\Gamma_{z\bar z}-\bar\Gamma_{z\bar z}$.
Note that these bounds are only locally uniform in general, so these norms also depend on
some underlying bounded domain in which we allow $z$ and $\bar z$ to vary. Since we
are only interested in situations with periodic boundary conditions on a bounded domain
and on a bounded time interval, this is irrelevant for the purpose of this article.
Note also that $\$\cdot\$$ is not a norm since the space $\MM$ of admissible models is
not linear. It does however behave like a norm for all practical purposes and we will refer
to as as the ``norm'' of a model.
\end{remark}

\begin{remark}\label{rem:redundant1}
Note that since $f_z \in \CG_+$, so that it is multiplicative, \eref{e:admissible1} and \eref{e:admissible2}
do specify $f_z$ on elements of the type $\J_k(\tau)$ once we know $\Pi_z$. There is therefore
quite a lot of rigidity in these definitions, which makes the mere existence of admissible models a 
highly non-trivial fact.
\end{remark}

\begin{remark}\label{rem:redundant2}
Building further on Remark~\ref{rem:redundant1}, it actually turns out that if $\Pi $ 
satisfies the {\em first} analytical bound in \eref{e:bounds} and
is such that, for $F$ defined from $\Pi$ via \eref{e:admissible2}, one has the 
identities \eref{e:admissible3} and \eref{e:admissible4}, then the second analytical bound in \eref{e:bounds}
is {\em automatically} satisfied for elements of the type $\J_k(\tau)$. 
This is a consequence of \cite[Thm.~5.14]{Regularity}. However, it is \textit{not}
automatic for terms of the type $\EE^{k-1}_\ell(\tau)$.
This is because our notion
of an ``admissible model'' does not specify any relation 
between $f_z(\EE_\ell^{k-1}(\tau_1,\ldots,\tau_{2k}))$ and the distributions $\Pi_z \tau_i$.
\end{remark}

At this point we have that $\CE$ is an abstract integration operator on the regularity structure $\CT_\ex$ and the results of 
\cite[Sections 8.1 and 8.2]{Regularity} hold for $\CE$ by repeating the proofs there for $\CI$.  These will be used
repeatedly in the sequel.  In principle, $\CE$ is not really an operator on the regularity structure $\CT$, like it is on $\CT_\ex$, however it is now defined on $\CT$ 
through the  restriction map:  If $\tau\in \CT$, then $\tau\in \CT_\ex$ since $\CT$ is a subset of $\CT_\ex$. Now $\CE\tau\in \CT_\ex$ and
the restriction of $\CE\tau$ to $\CT$ is what we will call $\CE\tau\in \CT$.

Finally, we define an analogous set $\MM_\ex$\label{cemex} of admissible models for $\CT_\ex$ on $\R^2$. 
 A model for $\CT_\ex$ is a
pair $(\Pi,F)$ of functions
$
\Pi \colon \R^2 \to \CL(\CT_\ex,\CS')$ and $F \colon \R^2  \to \CG$
satisfying  \eqref{e:algebraic} and \eqref{e:bounds} for  $\tau \in \CW_\ex$ and $\bar \tau \in \CW_{+}$, and it is 
admissible if \eqref{e:admissible1}-\eqref{e:admissible4} hold for $\tau\in \CW$.

\subsection{Definition of \texorpdfstring{$\CD^\gamma$}{Dg}}
\label{sec:Dgamma}

Given the space $\CT$ as above and $\gamma > 0$, as well as an admissible model $(\Pi,F) \in \MM$ 
we now define a space $\CD^\gamma$ of modelled distributions consisting of those functions 
$H \colon \R^2 \to \CT$ such that 
\begin{equ}[eq:distH0]
\|H\|_\gamma = \sup_{|z-\bar z| \le 1}\sup_{\alpha < \gamma} {|H(z) - \Gamma_{z\bar z} H(\bar z)|_\alpha \over |z-\bar z|^{\gamma - \alpha}} + \sup_{z,\alpha} |H(z)|_\alpha<\infty\;.
\end{equ}
Recall that, as defined in \eqref{e:defxalpha}, 
$|H(z)|_\alpha$ refers to the (Euclidean) norm of the part of $H(z)$ in $\CT_\alpha$.
Here, the arguments $z,\bar z$ are typically constrained to lie furthermore
in some fixed bounded set and we have used the shorthand $\Gamma_{z\bar z} \eqdef F_z^{-1}\circ F_{\bar z}$. Note that the space $\CD^\gamma$ depends on the underlying model! 
It is however natural to be able to also compare elements $H$ and $\bar H$ belonging to
spaces $\CD^\gamma$ based on two different models $(\Pi,f)$ and $(\bar \Pi,\bar f)$.
In this case, we write 
\begin{equs}[e:distH]
\|H;\bar H\|_\gamma = &\sup_{|z-\bar z| \le 1}\sup_{\alpha < \gamma} {|H(z) 
	- \Gamma_{z\bar z} H(\bar z) - \bar H(z) + \bar \Gamma_{z\bar z} \bar H(\bar z)|_\alpha 
	\over |z-\bar z|^{\gamma - \alpha}}\\ &\qquad+ \sup_{z,\alpha} |H(z)- \bar H(z)|_\alpha\;.
\end{equs}
This yields a ``total space'' $\MM \ltimes \CD^\gamma$  
containing all triples of the form
$(\Pi,f,H)$ with $H \in \CD^\gamma$ based on the model $(\Pi,f)$.
The distances $\$\cdot;\cdot\$$ and \eqref{e:distH} endow $\MM \ltimes \CD^\gamma$
with a metric structure.

It was then shown in \cite{Regularity} that for any $\gamma > 0$
there exists a \textit{unique} locally Lipschitz continuous 
map $\CR\colon \MM \ltimes \CD^\gamma \to \CS'$ \label{reco} with the property that
\begin{equ}
\bigl|\bigl(\CR H - \Pi_z H(z)\bigr)(\phi_z^\lambda)\bigr| \lesssim \lambda^\gamma\;,
\end{equ}
uniformly over $\phi \in \CB$, $\lambda \in (0,1]$ and locally uniformly in $z$.
The interpretation  of the ``reconstruction operator'' $\CR$ is that $H$ is really
just a local description of a ``Taylor expansion'' for the actual distribution $\CR H$.
It is straightforward to show that in the particular case where $\Pi_z \tau$ 
represents
a continuous function for every $\tau \in T$, one has the identity
\begin{equ}[e:propR]
\bigl(\CR f\bigr)(z) = \bigl(\Pi_z f(z)\bigr)(z)\;.
\end{equ}
This identity will be crucial in the sequel.

We will also make use of weighted spaces $\CD^{\gamma,\eta}$\label{def:cdgammaeta}, which essentially
consist of elements of $\CD^\gamma$ that are allowed to blow up at rate $\eta$ 
near the line $\{(t,x)\,:\, t=0\}$. For a precise definition, see \cite[Def.~6.2]{Regularity}.
In our setting, this is the set of functions $H\colon \R^2 \to \CT$ such that
\begin{equs}[e:defgammaeta]
\|H\|_{\gamma,\eta} \eqdef & \sup_{z}\sup_{\alpha < \gamma} {| H(z)|_\alpha \over |t|^{(\eta-\alpha) \wedge 0 \over 2}}
+ \sup_{ |z-\bar z| \le  \sqrt{|t| \wedge |\bar t|}}\sup_{\alpha<\gamma} {|H(z) - \Gamma_{z\bar z} H(\bar z)|_\alpha \over |z-\bar z|^{\gamma - \alpha} (|t| \wedge |\bar t|)^{\eta-\gamma \over 2}} \\&\qquad + \sup_{z,\alpha} |H(z)|_\alpha< \infty\;,
\end{equs}
where we used $t$ and $\bar t$ for the time coordinates of $z$ and $\bar z$.

\begin{remark}
Note that we do not necessarily assume that $H(z) \in \CT_{<\gamma}\eqdef \oplus_{\alpha<\gamma}\CT_\alpha$.\label{defofCTsubleg} This will be
useful especially in the case when $\gamma < 0$ which we will encounter later on.
\end{remark}

\begin{remark}
While we have defined the spaces $\CD^\gamma$ and $\CD^{\gamma,\eta}$ for admissible models with respect to $\CT$, there is of course an analogous definition for admissible models with 
respect to $\CT_\ex$.  Many of the statements in the next several sections will be true for either, and we will indicate if a specific one is being used.

In the limit $\eps\to 0$, we will obtain a model $(\Pi,F)$ on $\CT$, not on $\CT_\ex$. However, 
although $\Pi$ doesn't extend to $\CT_\ex$, the operators $\Gamma_{xy}$ \textit{do} extend to it
by multiplicativity. As a consequence, the spaces $\CD^\gamma$ and $\CD^{\gamma,\eta}$ make sense for
function with values in $\CT_\ex$, even if we are only given a model on $\CT$. If we are given
such a model, it is only when 
applying the reconstruction operator $\CR$ that it is crucial that the function be $\CT$-valued.
\end{remark}

\subsection{Canonical lift to \texorpdfstring{$\CT_\ex$}{Tex}}
\label{sec:canonical}

Given any \textit{smooth} space-time function $\zeta$ and any real number $\eps$, 
there is a canonical way of building a family of admissible models,
$\LL_\eps(\zeta) = (\Pi^{(\eps)}, f^{(\eps)})$ for the \textit{extended} regularity structure
$(\CT_\ex, \CG)$
as follows. First, we set $\Pi_z^{(\eps)} \Xi = \zeta$, independently
of $z$ and of $\eps$, and we define it on $X^k$ as in \eref{e:admissible1}. 
Then, we define $\Pi^{(\eps)}_z$ recursively by \eref{e:admissible3}
and \eref{e:admissible4}, together with the identities
\begin{equ}[e:canonical]
\bigl(\Pi_z^{(\eps)} \tau \bar \tau\bigr)(\bar z) = \bigl(\Pi_z^{(\eps)} \tau\bigr)(\bar z)
\bigl(\Pi_z^{(\eps)} \bar \tau\bigr)(\bar z)\;.
\end{equ}
as well as
\minilab{e:canonical2}
\begin{equs}
\bigl(\Pi_z^{(\eps)} \CE^{k-1}(\tau)\bigr)(\bar z) &= \eps^{k-1} \bigl(\Pi_z^{(\eps)} \tau\bigr)(\bar z)
 + \sum_{\ell} {(\bar z - z)^\ell \over \ell!} f^{(\eps)}_z \bigl(\EE_\ell^{k-1}(\tau)\bigr)\;,\label{e:canonical2A} \\
f^{(\eps)}_z \bigl(\EE_\ell^{k-1}(\tau)\bigr) &= 
 -\eps^{k-1} \bigl(D^{(\ell)} \bigl(\Pi_z^{(\eps)} \tau\bigr)\bigr)(z)\;.  \label{e:canonical2B}
\end{equs}
None of this make sense on $\CT$, which is one of the key reasons to introduce the larger
regularity structure $\CT_\ex$.
Here, the multiindex $\ell$ is furthermore constrained by imposing that
$|\tau| \le |\ell| < k-1 + |\tau|$.
Note again that in general, this definition is only guaranteed to makes 
sense if $\zeta$ is a smooth function! 
Note also that when we use this definition in practice
later on, $\zeta$ will really be given by 
some smooth approximation $\xi_{\bar\eps}$ to our space-time white noise.
It is however very important to note that $\bar\eps$  can be completely unrelated to $\eps$,
so the models $\LL_{ \eps}(\xi_{\bar\eps})$ or even $\LL_0(\xi_{\bar\eps})$ make perfect sense.
Finally, note that the definition \eqref{e:canonical2} would not even make sense on our actual regularity structure
$\CT$, because we could have $\CE^{k-1}(\tau)\in \CT$ but $\tau\not\in\CT$.

\begin{proposition}
If $\zeta$ is smooth then $\LL_\eps(\zeta)\in\MM_\ex$ for any $\eps$. 
\end{proposition}

\begin{proof}
The argument is very similar to that of \cite[Prop.~8.27]{Regularity}.
The fact that the algebraic identity \eqref{e:algebraic} is satisfied follows immediately from
our construction. The analytical bounds \eqref{e:bounds} for $\Pi_z$ follow in exactly the same
way as in \cite[Prop.~8.27]{Regularity} from the stronger bound 
\begin{equ}
|\bigl(\Pi_z \tau\bigr)(\bar z)| \lesssim |z-\bar z|^{|\tau| \wedge 0}\;,
\end{equ}
which is easily verified by induction. Writing $\gamma_{z\bar z}$ \label{defofgammazbarz} for the
element in $\CT_+^*$ such that $\Gamma_{z\bar z} = \Gamma_{\gamma_{z\bar z}}$
(recall \eqref{e:Gammadef} for these notations), the required bounds on $\Gamma_{z\bar z}$ 
are equivalent to $\gamma_{z\bar z}(\bar \tau) \lesssim |z-\bar z|^{|\bar \tau|}$.
The bounds on $\gamma_{z\bar z}(\bar \tau)$
for $\bar \tau$ of the form $\J_k(\tau)$ follow from the bounds on $\Pi_z$ 
as in \cite[Prop.~8.27]{Regularity}, so
it only remains to get a bound on $\gamma_{z\bar z}(\EE^\ell_k(\tau))$. 

In almost exactly the same way as in the proof of \cite[Prop.~8.27]{Regularity}, 
it is straightforward to set up a inductive structure on $\CT$ which
allows us to assume that all components of $\Gamma_{z\bar z}\tau$ do satisfy the required bounds.
Proceeding exactly as in the last part of the proof of \cite[Prop.~8.27]{Regularity}, 
one then obtains the identity
\begin{equ}[e:exprgammaz]
\gamma_{z\bar z}(\EE^\ell_k(\tau)) = f_{\bar z} (\EE^\ell_k(\tau)) - \sum_m {(\bar z - z)^m \over m!}
f_{z} (\EE^\ell_{k+m}(\Gamma_{z\bar z}\tau))\;.
\end{equ}
Fix $\tau$ and write $g_{\bar z}(z)$ as a shorthand for $-\eps^{\ell} \bigl(D^{(k)}(\Pi_{\bar z}^{(\eps)}\tau)\bigr)(z)$, so that $f_{\bar z} (\EE^\ell_k(\tau)) = g_{\bar z}(\bar z)$. 
It follows from our construction that the map $(z,\bar z) \mapsto g_{\bar z}(z)$ is smooth.
It then follows
from \eqref{e:canonical2} that
\begin{equ}[e:exprfe]
f_{z} (\EE^\ell_{k+m}(\Gamma_{z\bar z}\tau)) = D^{(m)} g_{\bar z}(z) \mid_{z=\bar z} - \eps^{\ell} \bigl(D^{(k+m)} \bigl(\Pi_z^{(\eps)} \,\Proj_{<|k|+|m| - \ell}\Gamma_{z\bar z}\tau\bigr)\bigr)(z)\;,
\end{equ}
where $\Proj_{<\alpha}$ denotes the orthogonal projection onto $\CT_{<\alpha} \subset \CT_\ex$.  At this stage, we note that by our induction hypothesis one has $\|\Gamma_{z\bar z} \tau\|_\alpha \lesssim |z-\bar z|^{|\tau| - \alpha}$. In particular, we can combine this with \eqref{e:exprfe} to conclude that one has
\begin{equ}[e:boundf]
\bigl|f_{z} (\EE^\ell_{k+m}(\Gamma_{z\bar z}\tau)) - D^{(m)} g_{\bar z}(z)\bigr| \lesssim |\bar z - z|^{|\tau| + \ell - |k+m|}\;.
\end{equ}
We can also combine \eqref{e:exprfe} with \eqref{e:exprgammaz}, which yields the identity
\begin{equ}[e:exprgammazbarz]
\gamma_{z\bar z}(\EE^\ell_k(\tau)) = g_{\bar z}(\bar z) - \sum_{|m| < |\tau|+\ell-|k|} {(\bar z - z)^m \over m!}
D^{(m)}g_{\bar z}(z)\;,
\end{equ}
which is bounded by a multiple of $|z-\bar z|^{|\tau| + \ell - |k|}$ 
as a consequence of usual Taylor expansion. \end{proof}

It is however very important to keep in mind that not every admissible model is obtained in this way, 
or even as a limit of such models! This will be apparent in Section~\ref{sec:renorm} 
below where we describe the renormalization group.

\begin{proposition}
If $\zeta$ is smooth then the restriction of $\PPP\LL_\eps(\zeta)$ to $\CT$ is in $\MM$ for any $\eps$. 
\end{proposition}

\subsection{Multiplication by \texorpdfstring{$\eps^{k}$}{epsk}}
\label{sec:epsk}

For any model that is constructed as the canonical lift of a smooth function as above to $\CT_\ex$,
the symbol $\Eps^{k}$ should be thought of as representing the operation of
``multiplication with $\eps^{k}$''. This is however not quite true: \eqref{e:canonical2A}
suggests that we should introduce the (model dependent) linear maps $\hat \Eps^k$ acting
on the spaces $\CD^\gamma$ by
\begin{equ}[e:defEpshat]
\bigl(\hat \Eps^k U\bigr)(z) = \Eps^k U(z) - \sum_{\ell} {X^\ell \over \ell!} f_z \bigl(\EE^k_\ell(U(z))\bigr)\;,
\end{equ}
where $f$ is determined by the underlying model on which $\CD^\gamma$ is based.
One then has the following fact where we implicitly assume that $U$
takes values in the domain of the operator $\Eps^k$.

\begin{proposition}\label{prop:multEps}
Let $\gamma \in \R$ and let $\delta = \inf\{\gamma - \alpha\,:\, \alpha \in A \cap (-\infty,\gamma)\}$. 
Then, if $U \in \CD^\gamma$, one has $\hat \Eps^k U \in \CD^{\bar \gamma}$ 
for $\bar \gamma = (\gamma + k) \wedge \delta$.
\end{proposition}

\begin{proof}
Our aim is to obtain a suitable bound on the components of 
$
\bigl(\hat \Eps^k U\bigr)(z) - \Gamma_{zz'} \bigl(\hat \Eps^k U\bigr)(z')
$.
For this, we note that one has from \eqref{e:defEpshat},
\begin{equs}
\Gamma_{zz'} \bigl(\hat \Eps^k U\bigr)(z')
= \Gamma_{zz'} \Eps^k U(z') - \sum_\ell {(X+z- z')^\ell \over \ell!} f_{z'} \bigl(\EE^k_\ell(U(z'))\bigr).
\end{equs}
Now from  the analogue for $\Eps^k $ of the proof of \cite[Thm. 8.24]{Regularity}
\begin{equs}\label{e:followsfrom}
\Gamma_{zz'} \Eps^k U(z') 
= \Eps^k \Gamma_{zz'} U(z') + \sum_{\ell,m} {X^\ell \over \ell!} {(z-z')^m \over m!} \gamma_{zz'} \bigl(\EE^k_{\ell+m}(U(z'))\bigr).
\end{equs}
At this stage, we make use of the fact that one has the identity \cite[p.~127]{Regularity}
\begin{equ}[e:idengamma]
\gamma_{zz'}(\EE^k_\ell\tau) = f_{z'}(\EE^k_\ell\tau) - \sum_m {(z'-z)^m \over m!} f_z \bigl(\EE^k_{\ell+m}\Gamma_{zz'}\tau\bigr)\;.
\end{equ}
Inserting this into the above expression and using the
binomial identity yields
\begin{equ}
\Gamma_{zz'} \bigl(\hat \Eps^k U\bigr)(z') =  \Eps^k \Gamma_{zz'} U(z') - \sum_\ell {X^\ell \over \ell!} f_z \bigl(\EE^k_\ell(\Gamma_{zz'}U(z'))\bigr)\;,
\end{equ}
so that
\begin{equs}
\bigl(\hat \Eps^k U\bigr)(z) - \Gamma_{zz'} \bigl(\hat \Eps^k U\bigr)(z')
&= \Eps^k \bigl(U(z) - \Gamma_{zz'} U(z')\bigr) \label{e:boundepsk} \\ &\qquad +  \sum_\ell {X^\ell\over \ell!} f_z \bigl(\EE^k_\ell \bigl(\Gamma_{zz'}U(z') - U(z)\bigr)\bigr)\;.
\end{equs}
The components in $\CT_\alpha$ arising from the first term are bounded
by $|z-z'|^{\gamma + k - \alpha}$ as a trivial consequence of the definition of $\CD^\gamma$
and the fact that $|\Eps^k \tau| = |\tau|+k$, so that we only need to consider the
components arising from the second
term. For this, we only need to note that these components are bounded by some multiple of $|z-z'|^\delta$ as
an immediate consequence of the definitions of $\CD^\gamma$ and $\delta$.
\end{proof}

\begin{remark}
There are two very important facts to note here. First, we do \textit{not} assume that
$\gamma > 0$. Second, the only property of $f$ that we used is that $f_z(X^k) = z^k$. In particular,
we do not need to assume that our model is the canonical model associated to a smooth function and
parameter $\eps > 0$. Actually, we do not even need to assume that it is admissible.
\end{remark}

It is then immediate from \eqref{e:canonical2} that if this model is the canonical model associated
to a smooth function as in Section~\ref{sec:canonical}, then the reconstruction operator
defined in Section~\ref{sec:Dgamma} satisfies the identity
\begin{equ}[e:propReconstr]
\CR \hat \Eps^{k-1} \bigl(U_1 \cdots U_{2k}\bigr) = \eps^{k-1} \CR U_1 \cdots \CR U_{2k}\;,
\end{equ}
so that the operation $\hat \Eps^{k-1}$ does indeed represent multiplication by $\eps^{k-1}$.
In general, if we have any model on $(\CT_\ex, \CG)$ consisting of smooth functions and satisfying
the identities \eqref{e:canonical2}, then $\CR \hat \Eps^{\ell} U = \eps^\ell \CR U$.
This remains true even in situations where \eqref{e:canonical} fails and / or when
$U \in \CD^\gamma$ for some $\gamma < 0$, provided that in the latter case one defines
$\CR U$ through the identity $(\CR U)(x) = \bigl(\Pi_x U(x)\bigr)(x)$.

\section{Abstract solution map}
\label{sec:solMap}

We start this section with a computation on $\CT_\ex$ showing that if one starts with sufficiently regular initial data, one
expects a well-posed fixed point problem in $\CD^\gamma=\CD^\gamma(\CT_\ex)$ for $\gamma>3/2$.  There are two key issues
which will have to be addressed in subsections \ref{sec:epsilon} -\ref{picit}:  1.  In order to iterate the argument to get global solutions, we will want to be able to start with less regular initial data; and, 2. We want the fixed point argument on $\CT$ itself, instead of $\CT_\ex$ where we can think of $\Eps^j$ as abstract integration operators increasing 
homogeneity by $j$.  For these reasons we will introduce spaces $\CD^{\gamma,\eta}_\eps$ in Section \ref{sec:epsilon}.
From now on, in order to simplify notations and similarly to \cite{Regularity}, we use the shortcut notation
\begin{equ}
\Psi = \CI'(\Xi)\;.
\end{equ}
We also write $\CQ_{\le 0}$ for the projection onto $\bigoplus_{\alpha \le 0}\CT_\alpha$ in $\CT_\ex$.\label{defofCQ}
Fix now some coefficients $\hat a_j$ and define the linear maps on $\CT_\ex$  given by
\begin{equs}
\hat \CF (\tau) &= \sum_{j=1}^m \hat a_j \CQ_{\le 0}\Eps^{j-1} \bigl(\CQ_{\le 0}\Psi^{2j} \tau\bigr)\;,\\ 
\hat \CF^{(n)}(\tau) &= \sum_{j=\lceil n/2\rceil}^m (2j+1-n)\cdots (2j)\hat a_j \CQ_{\le 0}\Eps^{j-1} \bigl(\CQ_{\le 0} \Psi^{2j-n} \tau\bigr)\;,\\
\hat \FF^{(n)}(\tau) &= \sum_{j=\lceil n/2\rceil}^m (2j+1-n)\cdots (2j)\hat a_j \EE_0^{j-1} \bigl(\CQ_{\le 0}\Psi^{2j-n} \tau\bigr)\;.
\end{equs}
(Of course we assume $n \le 2m$.)
We will also write $\hat \CF'$ as a shortcut for $\hat \CF^{(1)}$ and $\hat \CF''$ as a shortcut for $\hat \CF^{(2)}$.

Since the homogeneity of $\Psi$ is just below $-1/2$, and, according to \eqref{e:homofE},  $\CE^{j-1}$ increases the homogeneity by 
$j-1$, $\hat \CF$ decreases the homogeneity of its argument by just a bit more than $1$, 
$\hat \CF'$ decreases it by just a bit more than ${1\over 2}$, 
$\hat \CF''$ decreases it by a little bit more than $0$, and all the other
$\hat \CF^{(n)}$ increase the homogeneity of their argument (provided that $\kappa$ is small enough). 
From now on, we will write $\hat \CF^{(n)}\tau$ instead of $\hat \CF^{(n)}(\tau)$ 
and we will use the shorthand
\begin{equ}
\hat \CF^{(n)} = \hat \CF^{(n)}(\one)\;.
\end{equ}
Note also that $\Gamma \hat \CF^{(n)}(\tau) = \hat \CF^{(n)}(\Gamma \tau)$ for $n \le 2$, so that 
one actually has $\Pi_x \hat \CF \in \CC_\s^{-1-2m\kappa}$, 
$\Pi_x \hat \CF' \in \CC_\s^{-{1\over 2}-(2m-1)\kappa}$, etc. for every model
$(\Pi,\Gamma)$.

Denote now by $\CP$ the integration operator given by 
\begin{equ}\label{eq:integrationop}
\CP = \CK + R \CR\;,
\end{equ}
where $\CK$ is the operator defined from the kernel $K$ as in \cite[Sec.~5]{Regularity},
$\CR$ is the reconstruction operator, and $R$ is defined in \cite[Lemma~7.7]{Regularity}. 
 For suitable $\alpha > 0$, the operator
$\CP$ maps $\CD^\alpha$ to $\CD^{\alpha+2}$ as a consequence of \cite[Thm~4.7]{Regularity}. 
We also write $\one_+$ for the indicator function of the set of positive
times $\{(t,x)\,:\, t> 0\}$.
Because of its discontinuity at the origin, multiplication with $\one_+$ is not
a bounded linear operator on $\CD^\gamma$, so, as in \cite{Regularity},
one really does this on $\CD^{\gamma,\eta}$ defined at the end of Section~\ref{sec:Dgamma}.  However, the argument is only
formal at this point anyway because of the initial conditions, so we do not pursue it yet.
\footnote{Note that $\one_+$ is called $\PR^{+}$ in \cite{Regularity}}. 

With these notations at hand it is natural, just as in \cite[Sec.~9]{Regularity}, 
to associate to our problem the fixed point equation 
\begin{equ}[e:abstrFPorig]
H = \CP \one_+ \Bigl(\Xi + \sum_{j=1}^m \hat a_j \CQ_{\le 0} \hat\Eps^{j-1} \bigl(\CQ_{\le 0}(\DD H)^{2j}\bigr)\Bigr) + P h_0\;,
\end{equ} 
where $\DD$ was defined in Section~\ref{sec:structGroup}.

\begin{remark}
The reason why we are so explicit about the presence of the projection operators
$\CQ_{\le 0}$ (the analogous projections were mostly implicit in \cite{Regularity}) is
that we will end up in a situation where $(\DD H)^{2j}$ belongs to a space $\CD^{\gamma_j,\eta_j}$
with $\gamma_j < 0$ for some $j$. Projecting onto $\CQ_{<\gamma_j}$, as is done in \cite{Regularity},
would then have the effect of actually modifying the effect of the reconstruction operator on 
$\hat\Eps^{j-1} \bigl((\DD H)^{2j}\bigr)$, which is not a desirable feature.
\end{remark}

In principle, one may want to look for solutions to this problem 
in $\CD^{\gamma,\eta}$ for suitable
values of $\gamma$ and $\eta$. The remainder of this section is devoted to the 
study of \eqref{e:abstrFPorig}. Before we delve into the details, we give
a heuristic argument showing why one would expect this equation to 
have local solutions.
First, we note that \eqref{e:abstrFPorig} is of the form
\begin{equ}[e:abstract]
H = \CI \Bigl(\Xi + \sum_{j=1}^m \hat a_j  \CQ_{\le 0}\Eps^{j-1} \bigl(\CQ_{\le 0}(\DD H)^{2j}\bigr)\Bigr) + (\ldots)\;,
\end{equ}
where $(\ldots)$ denotes terms taking values in $\bar \CT$.
These additional terms arise as in \cite{Regularity} from the initial condition
and from the fact that the operator $\CP$ representing convolution with the heat kernel
is given by $\bigl(\CP f\bigr)(z) = \CI f(z) + (\ldots)$, where $(\ldots)$ denotes
again some terms taking values in $\bar \CT$. 

It follows that if we are able to solve \eqref{e:abstrFPorig} in 
$\CD^{\gamma,\eta}$ for ${3\over 2} < \gamma < 2 - (6m-2)\kappa$, then 
\textit{any} solution is necessarily of the form
\begin{equ}[e:U]
H = h\cdot\one + \CI(\Xi) + \CI(\hat \CF) + h'\cdot X + \CI(\hat \CF' \CI'(\hat \CF)) + h'\cdot\CI(\hat \CF')\;,
\end{equ}
for some continuous real-valued functions $h = h(t,x)$ and $h' = h'(t,x)$. 
Note that $h$ is \textit{not}
necessarily differentiable and that even when it is, $h'$ is \textit{not} in general the derivative of $h$ (see Section 2 of \cite{Hairer:2015aa} for an introduction and explanation of this issue).
This notation is only used by analogy with the usual Taylor expansions.  To obtain \eqref{e:U}, 
write the right hand side of \eqref{e:abstract}
first with $H=0$, then with the resulting $H$ from the left hand side substituted 
into the right hand side, etc.\ until the expression stabilises and
only components in $\bar\CT$ change from one step to the next. 
In the simpler context of the KPZ equation, this is explained in the proof of 
Proposition~15.12 of \cite{Book}.
The abstract derivative of $H$ is therefore given by
\begin{equ}[e:DU]
\DD H = \Psi + \CI'(\hat \CF) + h'\cdot \one + \CI'(\hat \CF' \CI'(\hat \CF)) + h'\,\CI'(\hat \CF')\;.
\end{equ}
Regarding the argument of $\CI$ in the right hand side
of \eref{e:abstract},
since we only keep terms of negative (or vanishing) homogeneities, it is given by
\begin{equs}
 {} & \Xi + \hat \CF + \hat \CF' \CI'(\hat \CF) + h'\cdot \hat \CF' + 
\hat \CF' \CI'(\hat \CF' \CI'(\hat \CF)) + h'\cdot \hat \CF' \CI'(\hat \CF') \qquad \label{e:RHS} \\
&\quad + {1\over 2} \hat \CF^{(2)} \bigl(\CI'(\hat \CF) + h'\cdot \one\bigr)^2 
- \sum_{n > 2} f_z \bigl(\hat
\FF^{(n)}\bigl((\CI'(\hat \CF) + h'\cdot \one)^n\bigr)\bigr)\,\one\;.
\end{equs}
The reason why no other terms of the form 
$\hat \FF^{(n)}(\cdot)$ appear in this expression is that 
$\EE_0^j(\tau) = 0$ for $\tau$ such that 
$|\tau| > 0$ (see the remark just after \eqref{e:defDelta} as well as the definition 
\eqref{e:abstrFPorig} of our fixed point problem).

As a consequence of \cite[Thm~4.7]{Regularity} and Proposition~\ref{prop:multEps}, we then note 
that if $H \in \CD^\gamma$ for $\gamma > {3\over 2} + \kappa$
then, disregarding the effect of initial conditions and provided that 
$\kappa$ is sufficiently small, the Picard 
iteration \eref{e:abstract} maps $\CD^\gamma$ into $\CD^{\gamma'}$ with
\begin{equ}
\gamma' = \gamma + {1\over 2} - (2m-1)\kappa \;.
\end{equ}
This strongly suggests that it is possible to build local fixed points of the Picard
iteration for $\kappa$ sufficiently small. It turns out that this heuristic is 
correct, although technical problems arise due to the effect of the initial condition.
The resolution of these problems is the subject of the remainder of this section.

%
%
%
%
%
%
%
\subsection{Dealing with irregular initial conditions}
\label{sec:epsilon}

There is a problem with the argument outlined above stemming from the class of initial
conditions we would like to consider. Since the solutions to the KPZ equation are $\alpha$-H\"older
continuous only for $\alpha < {1\over 2}$, we would like to have a (uniform in the small
parameter $\eps$ controlling our smoothing) solution theory for the
approximating equations that can deal with this type of initial data. The problem is that 
in this case, even for fixed $\eps$, say $\eps = 1$, and considering the deterministic equation
\begin{equ}
\d_t h = \d_x^2 h + (\d_x h)^{2m} + \zeta\;,\qquad h(0,\cdot) = h_0 \in \CC^\alpha\;,
\end{equ}
for some smooth $\zeta$, one expects the supremum norm of $\d_x h$ to develop a
singularity of order $t^{(\alpha-1)/2}$ at the origin, since this is what happens for solutions
to the heat equation.
 As a consequence, the term $(\d_x h)^{2m}$
leads to a non-integrable singularity as soon as $\alpha < 1$ and $m$ is large enough!   (\cite{ben-artzi} gives a nice survey of what is known about  the deterministic problem.)

One could of course circumvent this problem by simply postulating that the initial data is smooth
(or say Lipschitz continuous). However, in order to obtain approximation results for any fixed time
interval, one would like to exploit the global well-posedness of the limiting equation in order to
``restart'' our approximation argument (see Proposition~\ref{prop:restart}). 
Such an argument would then of course break down since the limiting solutions are only in $\CC^\alpha$
for $\alpha < {1\over 2}$. On the other hand,
it is reasonable to expect the solutions to the approximate equation to remain smooth at scales below 
$\eps$. In order to formalise this, we will introduce spaces of models / functions / modelled 
distributions that depend on a parameter $\eps \in (0,1]$, as well as their limiting counterparts for
$\eps = 0$, and we will set up suitable notions of convergence in such a context.

Recall from Section~\ref{sec:structure} that $\CU'\subset \CW$ is 
the set of all formal expressions in $\CW$ which are 
of the form $\CI'(\tau)$ for some $\tau$ in $\CW$.
For $\eps > 0$, we then define a class of $\eps$-models $\MM_\eps$ \label{epsilonmodels} 
which consist of all
admissible models $(\Pi,f)$ that furthermore satisfy the bounds 
\minilab{e:modelEpsWanted}
\begin{equs}[2]
\bigl|f_z\bigl(\EE^k_\ell(\tau)\bigr)\bigr| &\le C \eps^{|\tau| + k - |\ell|}\;,&\quad \tau &\in \CV_{\ell,k}\,:\,
|\ell| \ge |\tau| > |\ell|-k\;,\label{e:modelEpsWantedf}\\
\bigl|\bigl(\Pi_z\tau\bigr)(\phi_z^\lambda)\bigr| &\le C \lambda^{\bar \gamma} \eps^{|\tau| - \bar \gamma}\;,
&\quad \tau &\in \CU'\;,\quad \bar \gamma = 1 - {1\over 32m}\;,\qquad\label{e:modelEpsWantedPi}
\end{equs}
for some constant $C$, uniformly for $z$ belonging to an arbitrary compact set and for 
$\lambda \le \eps$. Here, $m$ is as in \eqref{e:abstrFPorig} and $\kappa$ is as
in \eqref{e:regXi}. The second bound is assumed to hold uniformly 
over all test functions $\phi \in \CB$
as in Section~\ref{sec:models} such that furthermore $\int \phi(z)\,dz = 0$.

\begin{remark}
The second bound in \eqref{e:modelEpsWanted} is non-trivial (i.e.\ not already implied by 
the definition of a model) only if $|\tau| < \bar \gamma$.
Note also that the condition on $\phi$ guarantees that $\bigl(\Pi_z\one\bigr)(\phi_z^\lambda) = 0$,
so that the bound holds trivially for all of $\bar \CT$.
\end{remark}

Note that, viewed as sets, one has of course $\MM_\eps = \MM_{\eps'}$ for
any $\eps,\eps' > 0$. However, 
they do differ at the level of the corresponding natural distance functions.
Indeed, we introduce a natural family of ``norms'' on $\MM_\eps$ by setting
$\$\Pi\$_\eps = \$\Pi\$ + \|\Pi\|_\eps$ with
\begin{equ}[e:modelEps]
\|\Pi\|_\eps = \sup_{z }\Bigl( \sup_{\tau \in \CV_{\ell,k}} \sup_{k,\ell}\eps^{|\ell| - k -|\tau|} \bigl|f_z\bigl(\EE^k_\ell(\tau)\bigr)\bigr| +  \sup_{\tau \in \CU' \atop |\tau| < \bar \gamma} \sup_{\lambda\le \eps\atop \phi}\lambda^{-\bar \gamma}\eps^{\bar \gamma+\delta-|\tau|}\bigl|\bigl(\Pi_z\tau\bigr)(\phi_z^\lambda)\bigr| \Bigr)\;,
\end{equ}
where the supremum over $\phi$ runs over the same set as above.
In particular, the restriction of the canonical lift $\LL_\eps(\zeta)$ to $\CT$ is in $\MM_\eps$ for any $\eps>0$.

\begin{remark}
We have made an abuse of notation here: Unlike for the class of models considered in \cite{Regularity}, 
there is here in general no canonical way of recovering $f$ from $\Pi$, so we should
really write $\|(\Pi,f)\|_\eps$ instead. This is because while our definition of an admissible model
imposes \eqref{e:admissible} which determines $f_z(\J_k\tau)$ in terms of $\Pi$, 
there is no analogue of this for $f_z(\Eps_k^\ell \tau)$. We do  have \eqref{e:canonical2} for the canonical lift, but this is \textit{not} preserved by our renormalisation procedure.
Furthermore, unlike \eqref{e:admissible}, it is not a continuous relation in the 
topology on models introduced in \cite{Regularity}.
\end{remark}

The natural way of comparing two elements of $\MM_\eps$ is to set
\begin{equ}
\$\Pi; \bar \Pi\$_\eps = \$\Pi; \bar \Pi\$ + \|\Pi - \bar \Pi\|_\eps\;.
\end{equ}
The point here is that we will be interested in distance bounds that are uniform in $\eps$ as
$\eps \to 0$.

We also introduce $\MM_0$ which is the subspace of $\MM$ consisting of those admissible
models that furthermore satisfy $f_z\bigl(\EE^k_\ell(\tau)\bigr) = 0$ for every 
$\tau$ and every $k$ and $\ell$. 
Since both $\MM_\eps$ and $\MM_0$ are subspaces of $\MM$, we can in principle
compare them by using the metric $\$\cdot;\cdot\$$ on $\MM$. It will also be convenient to set
up a way of comparing elements in $\MM_\eps$ with elements in $\MM_0$ in a way that takes into account
the $\eps$-dependence. This is done
by setting
\begin{equ}[e:boundPidiff]
\$\Pi; \bar \Pi\$_{\eps,0} = \$\Pi; \bar \Pi\$ + \|\Pi\|_\eps\;,
\end{equ}
for every pair of admissible models with $(\Pi, \Gamma) \in \MM_\eps$ and $(\bar\Pi,\bar\Gamma) \in \MM_0$.

\begin{remark}
One might wonder if there is a natural way of comparing elements $(\Pi,\Gamma) \in \MM_\eps$ 
with elements $(\bar \Pi,\bar \Gamma) \in \MM_{\bar \eps}$ for $0 < \bar \eps < \eps$.
For $\bar \eps > \eps / 2$ say, it is natural to view both models as belonging to $\MM_\eps$
and to use the distance $\$\cdot;\cdot\$_\eps$ defined there. For $\bar \eps < \eps / 2$ on
the other hand, it is more natural to set 
$\$\Pi, \bar \Pi\$_{\eps,\bar \eps} = \$\Pi; \bar \Pi\$ + \|\Pi\|_\eps + \|\bar \Pi\|_{\bar \eps}$.
We will however not make use of these definitions in the sequel.
\end{remark}

We similarly introduce $\eps$-dependent norms on suitable subspaces $\CD^{\gamma,\eta}_\eps$
\label{def:cep} of
the spaces $\CD^{\gamma,\eta}$ of modelled distributions previously 
introduced in \eqref{e:defgammaeta}. 
We will usually consider situations where the space $\CD^{\gamma,\eta}_\eps$ 
is built from an underlying model
belonging to $\MM_\eps$, but this is not needed in general. 
The space $\CD^{\gamma,\eta}_\eps$ consists of the elements $H \in \CD^{\gamma,\eta}$
such that the norm $\|H\|_{\gamma,\eta;\eps}$ given by
\begin{equ}[e:defgammaetaeps]
\|H\|_{\gamma,\eta;\eps} = \|H\|_{\gamma,\eta} 
+  \sup_{z}\sup_{\alpha > \eta} {| H(z)|_\alpha\over \eps^{\eta - \alpha}} 
+ \sup_{ |z-\bar z| \le  \sqrt{|t| \wedge |\bar t|}\atop|z-\bar z| \le \eps}\sup_{\alpha<\gamma} {|H(z) - \Gamma_{z\bar z} H(\bar z)|_\alpha \over |z-\bar z|^{\gamma - \alpha} \eps^{\eta-\gamma}}\;,
\end{equ}
is finite.

Note that the space $\CD^{\gamma,\eta}_0$ is nothing but $\CD^{\gamma,\eta}$.
The norms \eqref{e:defgammaetaeps}
are of course all equivalent as long as $\eps > 0$, but as $\eps \to 0$ they get closer 
and closer to the inequivalent norm $\|\cdot\|_{\gamma,\eta}$.

As before, it is natural to
compare elements $H \in \CD^{\gamma,\eta}_\eps$ with elements $\bar H \in \CD^{\gamma,\eta}_0$
by setting
\begin{equ}[e:bounddifff]
\|H;\bar H\|_{\gamma,\eta;\eps} = \|H;\bar H\|_{\gamma,\eta} 
+  \sup_{z}\sup_{\alpha> \eta} {| H(z)|_\alpha\over \eps^{\eta -\alpha}}
+ \sup_{ |z-\bar z| \le  \sqrt{|t| \wedge |\bar t|}\atop|z-\bar z| \le \eps}\sup_{\alpha<\gamma} {| H(z) - \Gamma_{z\bar z}  H(\bar z)|_\alpha \over |z-\bar z|^{\gamma - \alpha} \eps^{\eta-\gamma}}\;.
\end{equ}


\begin{remark}
As before, the fact that $\bar H$ does not appear in the second term of
\eqref{e:bounddifff} is not a typo. Indeed, for general $\bar H \in \CD^{\gamma,\eta}_0$
this supremum would in general be infinite.
\end{remark}

\subsection{Properties of the spaces \texorpdfstring{$\CD^{\gamma,\eta}_\eps$}{Dge}}

In this section, we collect some useful properties of the spaces $\CD^{\gamma,\eta}_\eps$
introduced earlier. Unless otherwise specified, we make the following standing assumptions
and abuses of notation:
\begin{claim}
\item Whenever we make a claim of the type ``if $H$ belongs to $\CD^{\gamma,\eta}_\eps$,
then $\bar H$ belongs to $\CD^{\bar \gamma,\bar \eta}_\eps$'', it is understood that the norm
of $\bar H$ can be bounded in terms of the norm of $H$, uniformly over $\eps \in [0,1]$
and over models in $\MM_\eps$ with bounded norm.
\item When comparing modelled distributions in $\CD^{\gamma,\eta}_\eps$ with
some in $\CD^{\gamma,\eta}_0$, we always assume that we are given respective models
 $(\Pi,\Gamma) \in \MM_\eps$ and $(\bar \Pi,\bar \Gamma)\in \MM_0$. 
Modelled distributions denoted by $H$, $H_1$, etc are assumed to belong
to spaces $\CD^{\gamma,\eta}_\eps$ based on 
$(\Pi,\Gamma)$, while $\bar H$, $\bar H_1$, etc belong to spaces $\CD^{\gamma,\eta}_0$
based on $(\bar \Pi,\bar \Gamma)$.
\item Whenever we write $\Phi \lesssim \Psi$ for two expressions $\Phi$ and $\Psi$ depending
on $\eps$, it is understood that there exists
a constant $C$ independent of $\eps$ such that $\Phi \le C \Psi$. For every fixed value
$\bar C > 0$, 
the constant $C$ can be chosen the same for all possible functions / models appearing in 
$\Phi$ and $\Psi$, as long as their norms are bounded by $\bar C$.
\item We implicitly assume that the modelled distributions we consider take values in
sectors such that the operations we perform are well-defined.
\item The space-time domain on which our elements are defined is given by
$[0,T] \times S^1$ for some $T \in [\eps^2, 1]$.
\end{claim}
For all practical purposes, the spaces $\CD^{\gamma,\eta}_\eps$ 
behave just like the spaces $\CD^{\gamma,\eta}$. First, we show that the
definition \eqref{e:defgammaetaeps} is somewhat redundant in the sense that the second
term is bounded by the two other terms. This shows that in many cases, it suffices to
bound the last term in \eqref{e:defgammaetaeps}.
Note that this is however \textit{not} the case
for \eqref{e:bounddifff}, which is why we chose to keep the current notations.

\begin{proposition}\label{prop:pointwise}
For $H \in \CD^{\gamma,\eta}_\eps$, the second term in
\eqref{e:defgammaetaeps} is bounded by a fixed multiple of 
the sum of the first and the last term.
\end{proposition}

\begin{proof}
Since the first term yields $\|H(t,x)\|_\ell \lesssim t^{\eta - \ell \over 2}$, 
the claimed bound is non-trivial only for $z = (t,x)$
with $0 < |t| \le \eps^2$. For such a value of $z$, one can always find a sequence
$\{z_n\}_{n \ge 0}$ such that $(z_n, z_{n+1})$ with  $|z-\bar z| \le  \sqrt{|t| \wedge |\bar t|}$ and $|z-\bar z| \le \eps \in D_\eps^{(2)}$, such that 
$|z_n - z_{n+1}| \le \eps c^n$ for some fixed $c \in (0,1)$, and such that $z_n = z$
for $n$ sufficiently large. It then suffices to
rewrite $H(z)$ as
\begin{equ}
H(z) = H(z_0) + \sum_{n\ge 0} \bigl(H(z_{n+1}) - \Gamma_{z_{n+1}z_n} H(z_{n})\bigr)
+ \sum_{n\ge 0} \bigl(\Gamma_{z_{n+1}z_n} - 1\bigr) H(z_{n})\;.
\end{equ}
The first sum is bounded by a multiple of $\eps^{\eta - \gamma} \sum_{n \ge 0} |c^n \eps|^{\gamma-\ell}$, which is the required bound. 

To bound the second sum, we proceed by ``reverse induction'' on $\ell$. Indeed, for the
largest possible value of $\ell$ less than $\gamma$, one has 
$\|\bigl(\Gamma_{z_{n+1}z_n} - 1\bigr) H(z_{n})\|_\ell = 0$, so that the required bound
holds trivially there. Assuming now that the required bound holds for all $m > \ell$, we have
\begin{equ}
\|\bigl(\Gamma_{z_{n+1}z_n} - 1\bigr) H(z_{n})\|_\ell 
\lesssim \sum_{m > \ell} |\eps c^n|^{m-\ell} \|H(z_n)\|_m 
\lesssim \sum_{m > \ell} |\eps c^n|^{m-\ell} \eps^{\eta-m}\;. 
\end{equ}
Summing again over $n$, the required bound follows.
\end{proof}

One motivation for our definitions are the following two results.
To formulate the first one, we introduce some notation. 

\begin{proposition}\label{prop:epsLift}
Let $\alpha \in (0,1)$ and $\gamma \in (1,2)$, let $h \in \CC^{\gamma,\alpha}_\eps$, and let $P h$ be the canonical lift (via its truncated Taylor series)
of the harmonic extension of $h$ (in other words, the action of the heat kernel on a function, but then interpreted in the canonical way as a modelled distribution, see \cite[(7.13)]{Regularity}.) Then, one has $P h \in \CD_\eps^{\gamma,\alpha}$ and
the bound
\begin{equ}[e:boundHarm]
\|Ph\|_{\gamma,\alpha;\eps} \le C \|h\|_{\gamma,\alpha;\eps}\;,
\end{equ}
holds uniformly over $\eps \in [0,1]$ for some $C \ge 1$. 
If furthermore $ \bar h \in \CC^{\alpha}$ and $ h \in \CC^{\gamma,\alpha}_\eps$, 
 then
\begin{equ}
\|Ph ; P\bar h\|_{\gamma,\alpha;\eps} \le C \|h ; \bar h\|_{\gamma,\alpha;\eps}\;.
\end{equ}
\end{proposition}

\begin{proof} Since $\CC^{\gamma,\alpha}_\eps \subset \CC^\alpha$ with embedding
constants uniform in $\eps$, we conclude from \cite[Lem.~7.5]{Regularity}
that we only need to bound the second term in \eref{e:bounddifff} (with $H = Gh$).
In particular, we only need to consider the case $\eps > 0$.

This in turn is nothing but the statement that the map $Ph$ is of class 
$\CC^{\gamma/2}$ in time and 
$\CC^\gamma$ in space, with norm bounded by $\eps^{\alpha-\gamma}$.
This in turn follows from classical properties of the heat kernel, combined 
with the fact that the $\CC^\gamma$-norm of $h$ is bounded by 
$\eps^{\alpha-\gamma}$ by assumption.

To obtain the bound on $\|Ph ; P\bar h\|_{\gamma,\alpha;\eps}$, we only need to bound
the last two terms in \eqref{e:bounddifff} in terms of the last two terms in 
\eqref{e:diffHolder}. This follows again immediately from the properties of the heat kernel.
\end{proof}

We also have the following result, where $\CU$ is as in Section~\ref{sec:symbols},
$\scal{\CU}$ denotes its linear span in $\CT$, and $\bar \gamma$ is as in \eqref{e:modelEpsWanted},
so that in particular $\bar \gamma >0$.

\begin{proposition}\label{prop:restart}
Let $\alpha \le {1\over 2}-{3\kappa \over 2}$, let $\gamma = 1+\bar \gamma$, 
and let $H_\eps \in \CD^\gamma$ 
with values in $\scal{\CU}$, based
on some model $\Pi^{(\eps)} \in \MM_\eps$.
Then, for every $t$ such that $[t-\eps^2, t+\eps^2] \subset [0,T]$, the function $h_t^{(\eps)} = (\CR H_\eps)(t,\cdot)$ belongs to 
$\CC_\eps^{\gamma,\alpha}$ and one has
\begin{equ}
\|h_t^{(\eps)}\|_{\gamma,\alpha;\eps} \le C \|H_\eps\|_\gamma\,\$\Pi^{(\eps)}\$_\eps\;,
\end{equ}
for some constant $C$ independent of $\eps \in (0,1]$. Furthermore, 
given $H \in \CD^\gamma$ with values in $\scal{\CU}$, based
on some model $\Pi \in \MM_0$, the function $h_t = (\CR H)(t,\cdot)$
belongs to $\CC^\alpha$ and one has the bound
\begin{equ}
\|h_t; h_t^{(\eps)}\|_{\gamma,\alpha;\eps} \le C \|H;H_\eps\|_\gamma\,\bigl(\$\Pi\$ + \$\Pi^{(\eps)}\$_\eps\bigr)
+ \$\Pi^{(\eps)};\Pi\$_{\eps,0} \bigl(\|H\|_\gamma+\|H_\eps\|_\gamma\bigr) \;.
\end{equ}
\end{proposition}

\begin{proof}
Let $\alpha_0 = |\CI(\Xi)| = {1\over 2}-\kappa$ be the homogeneity of the element 
of lowest non-zero homogeneity 
in $\CU$. It then follows from \cite[Prop.~3.28]{Regularity} that $\CR H_\eps$ is a continuous
function with $\CR H_\eps \in \CC^{\alpha_0}$ (with parabolic space-time scaling) and, 
since $\alpha < \alpha_0$, that
\begin{equ}
\|h_t^{(\eps)}\|_{\alpha} \lesssim \|H_\eps\|_\gamma\,\$\Pi^{(\eps)}\$\;,
\end{equ}
so that it only remains to obtain the 
bound on the last term in \eqref{e:defh}. Setting $\tilde h_t^{(\eps)} = \d_x h_t^{(\eps)} = \bigl(\CR \DD H_\eps\bigr)(t,\cdot)$, 
we will prove the stronger fact that $\tilde h_t^{(\eps)}$ is a continuous function such that
\begin{equ}[e:wantedMicroBound]
\sup_{z \neq \bar z \atop |z-\bar z| \le \eps} {\eps^{\gamma-\alpha}|\tilde h_t^{(\eps)}(z) - \tilde h_t^{(\eps)}(\bar z)| \over |z-\bar z|^{\gamma-1}} \lesssim \|H_\eps\|_\gamma\,\$\Pi^{(\eps)}\$_\eps \;,
\end{equ}
where the supremum runs over $z$ and $\bar z$ in $[t-\eps^2/4, t+\eps^2 / 4] \times S^1$
and $|z|$ denotes the parabolic distance.

As in \cite[Thm~6.5]{MR1228209}, the left hand side in \eqref{e:wantedMicroBound} is 
bounded, up a factor independent of $\eps$, by the quantity
\begin{equ}[e:wantedBoundleps]
\sup_{\phi} \sup_{\lambda < \eps} \sup_{z} \lambda^{1-\gamma}\eps^{\gamma-\alpha} \bigl|\tilde h_t^{(\eps)}(\phi_z^\lambda)\bigr|\;,
\end{equ}
where the first supremum runs over all space-time test functions $\phi \in \CB$ integrating to $0$,
the supremum over $z$ runs over $[t-\eps^2/2, t+\eps^2/2] \times S^1$, and $\tilde h_t^{(\eps)}$ 
is interpreted as a distribution.

In order to obtain the required bound on \eqref{e:wantedBoundleps} note that, as a consequence of 
\cite[Lem~6.7]{Regularity}, one has for $\lambda < \eps$ the bound
\begin{equ}
\bigl|\bigl(\tilde h_t^{(\eps)} - \Pi_z^{(\eps)} \DD H_\eps(z) \bigr)(\phi_z^\lambda)\bigr| \lesssim \lambda^{\gamma-1} \|H_\eps\|_\gamma\,\$\Pi^{(\eps)}\$
\le \eps^{\alpha-\gamma}\lambda^{\gamma-1} \|H_\eps\|_\gamma\,\$\Pi^{(\eps)}\$_\eps\;,
\end{equ}
where we used the fact that $\gamma > \alpha$ and $\eps < 1$ to obtain second inequality.
Furthermore, it follows from \eqref{e:modelEps}, combined with the facts that
$\phi$ integrates to $0$ and $\bar \gamma = \gamma-1$, that 
\begin{equ}
\bigl|\bigl(\Pi_z^{(\eps)} \DD H_\eps(z) \bigr)(\phi_z^\lambda)\bigr| \lesssim \lambda^{\gamma-1} \eps^{\alpha_0-\gamma - \delta} \|H_\eps\|_\gamma\,\$\Pi^{(\eps)}\$_\eps
\le  \lambda^{\gamma-1} \eps^{\alpha - \gamma} \|H_\eps\|_\gamma\,\$\Pi^{(\eps)}\$_\eps\;,
\end{equ}
where $\delta = \kappa / 2$ as above.
Here, the second inequality follows from the fact that $\alpha \le \alpha_0 - \delta$ by assumption. 
Combining both of these bounds, the required bound on $\|h_t^{(\eps)}\|_{\gamma,\alpha;\eps}$ 
follows at once. The bound on $\|h_t; h_t^{(\eps)}\|_{\gamma,\alpha;\eps}$ then follows
in the same way.
\end{proof}

\subsection{Operations in \texorpdfstring{$\CD^{\gamma,\eta}_\eps$}{Dge}}

We now show how the basic operations required for our purposes behave in these spaces.
First, we have the following
bound on the abstract derivatives of modelled distributions:

\begin{proposition}\label{prop:derSing}
Let  $H \in \CD^{\gamma,\eta}_\eps$ for some $\gamma > 1$ and $\eta \in \R$. 
Then, $\DD H \in \CD^{\gamma-1, \eta-1}_\eps$. Furthermore, one has
$\|\DD H;\DD \bar H\|_{\gamma-1,\eta-1;\eps}
\lesssim \|H;\bar H\|_{\gamma,\eta;\eps}$.
\end{proposition}

\begin{proof}
Immediate from the definitions.
\end{proof}


We also have a bound on their products:

\begin{proposition}\label{prop:multSing}
Let $H_1 \in \CD^{\gamma_1,\eta_1}_\eps(V^{(1)})$ and $H_2 \in \CD^{\gamma_2,\eta_2}_\eps(V^{(2)})$ for two sectors
$V^{(1)}$ and $V^{(2)}$ with respective regularities $\alpha_1$ and $\alpha_2$, such that 
a product satisfying the properties \cite[Def.~4.1\,\&\,4.6]{Regularity} is defined on $V^{(1)}\times V^{(2)}$.
Let furthermore
$\gamma = (\gamma_1 + \alpha_2) \wedge (\gamma_2 + \alpha_1)$ and assume that $\gamma_i > \alpha_i$.
Then, the function $H = H_1\, H_2$ belongs to
$\CD^{\gamma, \eta}_\eps$ with $\eta = (\eta_1 + \alpha_2)\wedge (\eta_2+\alpha_1) \wedge (\eta_1+\eta_2)$.

Furthermore,  
writing $H = H_1 \, H_2$ and $\bar H = \bar H_1 \, \bar H_2$, 
one has the bound
\begin{equ}[e:diffprod]
\|H;\bar H\|_{\gamma,\eta;\eps}
\lesssim \|H_1;\bar H_1\|_{\gamma_1,\eta_1;\eps} + \|H_2;\bar H_2\|_{\gamma_2,\eta_2;\eps}
+ \$\Pi;\bar \Pi\$\;.
\end{equ}
\end{proposition}

\begin{proof}
The proof is identical to that of \cite[Prop.~6.12]{Regularity}. The only difference
is that when bounding $H(z) - \Gamma_{z\bar z}H(\bar z)$ one replaces 
$\|z;\bar z\|_P$ by $\eps + \sqrt{|t| \wedge |t'|}$ throughout.
\end{proof}

\begin{remark}
Note that we did not assume that $\gamma > 0$! In particular, unlike in 
\cite{Regularity}, we do not compose the product with a projection onto 
$\CT_{<\gamma}$. 
\end{remark}


Writing $\CQ_{<\alpha} \colon \CT \to \CT$ for the 
projection onto $\CT_{<\alpha}$, we also see that such a projection
leaves the space $\CD_\eps^{\gamma,\eta}$ invariant.\label{defofQle}

\begin{proposition}\label{prop:proj}
Let $F \in \CD_\eps^{\gamma,\eta}$ with $\eta \le \gamma$ and let $\alpha \ge \gamma$. Then, 
one has again $\CQ_{<\alpha} F \in \CD_\eps^{\gamma,\eta}$.
\end{proposition}

\begin{proof}
It is sufficient to show that one actually has 
$F_\alpha \eqdef \CQ_{\alpha} F \in \CD_\eps^{\gamma,\eta}$ for every
$\alpha \ge \gamma$. It follows from the definitions that
$|F_\alpha(z)|\lesssim (|t| + \eps^2)^{(\eta - \alpha)/2}$. As a consequence,
for $\beta < \gamma$ (so in particular also $\beta < \alpha$) and 
for $|z-\bar z| \le \sqrt{|t|\wedge |\bar t|}$, one has
\begin{equs}
|F_\alpha(z) - \Gamma_{z\bar z}F_\alpha(\bar z)|_\beta
&\lesssim |z-\bar z|^{\alpha-\beta} |F_\alpha(\bar z)|
\lesssim |z-\bar z|^{\alpha-\beta} (|t| + \eps^2)^{(\eta - \alpha)/2}\\
&\lesssim |z-\bar z|^{\gamma-\beta} (|t| + \eps^2)^{(\eta - \gamma)/2}\;,
\end{equs}
thus yielding the required bound.
\end{proof}

The following proposition shows how these spaces behave under the action of the
integral operator $\CK$ defined in \eqref{eq:integrationop},
\begin{proposition}\label{prop:intSing}
Let $V$ be a sector of regularity $\alpha$
and let $H \in \CD^{\gamma,\eta}_\eps(V)$
with $-2 < \eta < \gamma \wedge\alpha$. 
Then, provided that $\gamma \not \in \N$ and $\eta \not \in \mathbf{Z}$,
one has $\CK H \in \CD^{\bar \gamma,\bar \eta}_\eps$ with 
$\bar \gamma = \gamma + 2$ and $\bar \eta = \eta+2$.
Furthermore, one has the bound
\begin{equ}[e:diffIntSing]
\|\CK H ; \CK \bar H\|_{\bar\gamma,\bar \eta;\eps} \lesssim 
\|H ; \bar H\|_{\gamma,\eta;\eps} + \$\Pi;\bar \Pi\$\;.
\end{equ}
\end{proposition}

\begin{proof}
In view of \cite[Prop.~6.16]{Regularity} and Proposition~\ref{prop:pointwise}, we only
need to bound the last term in \eqref{e:defgammaetaeps} with $H$
replaced by $\CK H$.

This bound follows immediately from the definitions for the
components of $\CK H$ that are not proportional to the Taylor monomials, so we
only need to consider the latter, i.e.\ we need to show that 
\begin{equ}
\|\CK H(z) - \Gamma_{z\bar z} \CK H(\bar z)\|_\ell \lesssim |z-\bar z|^{\bar \gamma - \ell} \eps^{\bar \eta- \bar \gamma}\;,
\end{equ}
for integer values of $\ell$ and for $(z,\bar z) \in D_\eps^{(2)}$.

The proof of this fact follows
the proof of \cite[Prop.~6.16]{Regularity} \textit{mutatis mutandis}, so we do
not reproduce it here.
The only difference is that all the expressions $\|x,y\|_P$ appearing there
are now replaced by $\eps$.
\end{proof}

\begin{remark}
All conclusions of Proposition~\ref{prop:intSing} still hold if $\CK$ is replaced
by $\CP$. 
\end{remark}

Note that in all the results so far, we never used the fact that the models
actually belong to $\MM_\eps$ rather than just $\MM$. This is somewhat explicit in the
fact that the bounds \eqref{e:diffprod} and \eqref{e:diffIntSing} depend on
$\$\Pi;\bar \Pi\$$ rather than on $\$\Pi;\bar \Pi\$_{\eps,0}$.
Furthermore, up to now, while we have seen that the spaces $\CD^{\gamma,\eta}_\eps$
do not behave any ``worse'' than the spaces $\CD^{\gamma,\eta}$, they do not behave
any ``better'' either, so it may seem unclear at this stage why we introduced them.

The final property of these spaces that we use is their behaviour under the 
operation $\hat \Eps^k$ introduced in Section~\ref{sec:epsk}. At this stage it
is absolutely essential to use the spaces $\CD^{\gamma,\eta}_\eps$ and models in $\MM_\eps$
since the corresponding property would simply be false otherwise.

\begin{proposition}\label{prop:epsSing}
Let $H \in \CD^{\gamma,\eta}_\eps$ with $\gamma > -k$ based on a model $\Pi$ in $\MM_\eps$ and set
\begin{equ}
\bar \gamma = \delta\;,\qquad \bar \eta = \eta + k\;,
\end{equ}
with $\delta = (\gamma + k) \wedge \inf_{\alpha \in A \cap [-k,\gamma)} (\gamma-\alpha)$.
Then, one has $\hat \Eps^k H \in \CD^{\bar \gamma,\bar \eta}_\eps$. Furthermore,
for $\bar H \in \CD^{\gamma,\eta}$ based on a model $\bar \Pi$ in $\MM_0$, one has the bound
\begin{equ}
\|\hat \Eps^k H ; \hat \Eps^k \bar H\|_{\bar\gamma,\bar \eta;\eps} \lesssim 
\|H ; \bar H\|_{\gamma,\eta;\eps} + \$\Pi;\bar \Pi\$_{\eps,0}\;,
\end{equ}
with a proportionality constant depending on $\$\Pi\$_{\eps}
+ \$\bar \Pi\$$, but not explicitly on $\eps$.
\end{proposition}

\begin{proof}
Setting $g = \hat \Eps^k H$, it then follows from \eqref{e:boundepsk} that 
\begin{equ}
g(z) - \Gamma_{zz'} g(z')
= \Eps^k \bigl(H(z) - \Gamma_{zz'} H(z')\bigr) +  f_z \bigl(\EE^k_0 \bigl(\Gamma_{zz'}H(z') - H(z)\bigr)\bigr)\,\one\;.
\end{equ}
For the components other than the one multiplying $\one$, the required bounds follow at once,
provided that $\bar \gamma \le \gamma + k$ and $\bar \eta \le \eta + k$. 
Regarding the component multiplying $\one$, it follows from the definitions of $\CD^{\gamma,\eta}_\eps$  and $\MM_\eps$
that the terms arising from components of $\Gamma_{zz'}H(z') - H(z)$ proportional
to $\tau$ are bounded by
\begin{equ}[e:boundDiffEps]
\eps^{|\tau| + k} |z-z'|^{\gamma-|\tau|} (\eps + \sqrt{|t|})^{\eta-\gamma}\;,
\end{equ}
where $t$ is the time component of $z$ and we only consider pairs $z,z'$ such that
$|z-z'|^2 \le |t|/2$, say. If $|z-z'| \le \eps$, then this 
bound gets worse for larger values of $|\tau|$.
By the definition of $\delta$ the largest value that arises is given by at most
$|\tau| = \gamma - \delta$. It follows that the requested bound holds, provided that
$\bar \gamma \le \delta$ and $\bar \eta \le \delta + \eta - \gamma$.
For $|z-z'| \ge \eps$, the bound \eqref{e:boundDiffEps} is worse for small values of
$\tau$. Since the smallest possible value of $\tau$ contributing to it is $|\tau| = \delta-k$,
this expression is bounded by $|z-z'|^{\gamma+k} (\eps + \sqrt{|t|})^{\eta-\gamma}$.
Since furthermore we only consider pairs $z$, $z'$ such that 
$|z-z'| \le \eps + \sqrt{|t|}$, this is also bounded by a multiple
of $|z-z'|^{\delta} (\eps + \sqrt{|t|})^{\eta+k-\delta}$ as required.

We now turn to the pointwise bound on $g$. For the components not multiplying $\one$,
it is immediate to see that the required bound holds as soon as $\bar \eta \le \eta + k$.
The component multiplying $\one$ is given by $f_z(\EE_0^k(H(z)))$. Again, the worst
available bound is on the component of $H(z)$ multiplying $\tau$ with $|\tau| = \gamma - \delta$, for which we obtain a bound of the type
\begin{equ}
|\scal{g(z),\one}| \lesssim f_z(\EE_0^k(\tau)) \, (\eps + \sqrt{|t|})^{(\eta - |\tau|)\wedge 0} \|H\|_{\gamma,\eta;\eps}\;.
\end{equ}
At this stage, we make use of the assumption that the underlying model belongs
to $\MM_\eps$, which guarantees that 
\begin{equ}[e:boundfEpsk]
|f_z(\EE_0^k(\tau))| \lesssim \eps^{|\tau| + k}\;.
\end{equ}
Since only terms with $|\tau| + k > 0$ contribute 
(see the remark following \eqref{e:defDelta}), we conclude that
\begin{equ}[e:boundgg]
|\scal{g(z),\one}| \lesssim (\eps + \sqrt{|t|})^{(\eta + k)\wedge 0} \|H\|_{\gamma,\eta;\eps}
\le (\eps + \sqrt{|t|})^{\bar \eta \wedge 0} \|H\|_{\gamma,\eta;\eps}\;,
\end{equ}
provided that $\bar \eta < \eta + k$, which is the required bound.

It remains to bound $\|\hat \Eps^k H ; \hat \Eps^k \bar H\|_{\bar\gamma,\bar \eta;\eps}$.
For this, the bounds on $\|\hat \Eps^k H ; \hat \Eps^k \bar H\|_{\bar\gamma,\bar \eta}$
follow in the same way as above. The bound on the second term in 
\eqref{e:bounddifff} also follows in the same way, noting that
it only requires the bounds \eqref{e:boundfEpsk} which in turn are controlled by 
$\$\Pi;\bar \Pi\$_{\eps,0}$ as a consequence of \eqref{e:boundPidiff} and \eqref{e:modelEps}.
\end{proof}

\subsection{Picard iteration and convergence}
\label{picit}

We now show that the ``abstract'' fixed point problem associated to 
our equation is uniformly well-behaved 
in the spaces $\CD^{\gamma,\eta}_\eps$ for suitable values of $\gamma$ and $\eta$.
(This is precisely what motivates our choice of definitions for $\CD^{\gamma,\eta}_\eps$
in the first place.)
More precisely, we have the following result.

\begin{theorem}\label{theo:FP}
Let $m \ge 1$, 
$\eta  \in ({1\over 2} - {1\over 4m},{1\over 2})$, $\eps \in [0,\eps_0]$,
and let $\kappa > 0$ be sufficiently small (depending only on $m$ and $\eta$). 
Let furthermore $\gamma = 2 - \nu$ with $\nu = 1/(32m)$, 
and consider the fixed point equation
\begin{equ}[e:abstrFP]
H= \CP \one_+ \Bigl(\Xi + \sum_{j=1}^m \hat a_j \CQ_{\le 0}\hat\Eps^{j-1} \bigl(\CQ_{\le 0}(\DD H)^{2j}\bigr)\Bigr) + P h_0\;,\end{equ} 
for some $h_0 \in \CC^{\gamma,\eta}_\eps$.
Then, for $\eps \le \eps_0$ with $\eps_0$ and the final time $T>0$ sufficiently small and 
for any model in $\MM_\eps$, there exists a unique solution to \eref{e:abstrFP} in 
$\CD^{\gamma,\eta}_\eps$. Furthermore, the time $T$ can be chosen uniformly 
over bounded sets of initial conditions in  $\CC_\eps^{\gamma,\eta}$, over bounded
sets in $\MM_\eps$, over bounded sets in the
space of parameters $\hat a_1,\ldots,\hat a_m$, and over $\eps \in [0,\eps_0]$.

Let $h_0^{(\eps)} \in \CC^{\gamma,\eta}_\eps$ be a sequence of elements such that there exists
$h_0 \in \CC^\eta$ with $\lim_{\eps \to 0} \|h_0;h_0^{(\eps)}\|_{\gamma,\eta;\eps} = 0$, 
and let $\Pi^{(\eps)} \in \MM_\eps$ be a
sequence of models such that there exists 
$\Pi \in \MM_0$ with $\lim_{\eps \to 0} \$\Pi^{(\eps)};\Pi\$_{\eps,0} = 0$.
Let $T > 0$ be fixed and assume that $H \in \CD^{\gamma,\eta}_0$ solves \eref{e:abstrFP} 
with model $\Pi$
up to some terminal time $T > 0$. Then, 
for $\eps > 0$ small enough, there exists a unique solution $H_\eps\in \CD^{\gamma,\eta}_\eps$ to \eref{e:abstrFP} 
with initial condition $h_0^{(\eps)}$
and model $\Pi_\eps$ up to time $T$, 
and $\lim_{\eps \to 0}\|H^{(\eps)};H\|_{\gamma,\eta;\eps} = 0$.
\end{theorem}

\begin{proof}
We first prove that the fixed point problem
\eref{e:abstrFP} can be solved locally with dependencies of the local existence time
that are uniform in $\eps$, provided that both the initial condition and the underlying
model are controlled in the corresponding $\eps$-dependent norms.
We consider \eref{e:abstrFP}
as a fixed point argument in $\CD^{\gamma,\eta}_\eps$. In other words, we
show that if we denote by $\CM$ the map
\begin{equ}[e:FPequ]
\CM(H) = \CP \one_+ \Bigl(\Xi + \sum_{j=1}^m \hat a_j \CQ_{\le 0}\hat \Eps^{j-1}\bigl( \CQ_{\le 0}(\DD H)^{2j}\bigr)\Bigr)+ P h_0^{(\eps)}\;,
\end{equ}
then, for sufficiently small values of the final time $T$ and uniformly in the
stated data, $\CM$ is a contraction mapping the centred ball of large enough radius $R$ in $\CD^{\gamma,\eta}_\eps$
into the ball of radius $R/2$.
Additional details, in particular the proof that solutions can be continued uniquely until
the explosion time in $\CC^\eta_\eps$, can be found in \cite[Sec.~7]{Regularity}.

Regarding the term $P h^{(\eps)}_0$, it follows from Proposition~\ref{prop:epsLift}, 
combined with our assumptions
on the initial conditions, that it belongs to $\CD^{\gamma,\eta}_\eps$, uniformly 
over $\eps \in [0,1]$, and that $\|P h_0^{(\eps)}; P h_0\|_{\gamma,\eta;\eps} \to 0$
as $\eps \to 0$.

Combining Propositions~\ref{prop:derSing}, \ref{prop:multSing} and \ref{prop:epsSing}, 
we conclude that if we set
\begin{equ}
\gamma_1 = \gamma - {1\over 2} - j - \kappa(2j-1)\;,\qquad \eta_1 = 2j(\eta-1)\;,
\end{equ}
then the map $H \mapsto  (\DD H)^{2j}$ is continuous from
$\CD^{\gamma,\eta}_\eps$ into $\CD^{\gamma_1,\eta_1}_\eps$. Note that
 $\gamma_1$ is negative as soon as $j \ge 2$, so that by Proposition~\ref{prop:proj}
the map $H \mapsto \CQ_{\le 0} (\DD H)^{2j}$ is also continuous from
$\CD^{\gamma,\eta}_\eps$ into $\CD^{\gamma_1,\eta_1}_\eps$ as soon as $j \ge 2$.
For $j = 1$, it turns out that one actually has $\CQ_{\le 0} (\DD H)^{2} = \CQ_{<\gamma_1} (\DD H)^{2}$
as a consequence of the fact that $\gamma_1 < {1\over 2} - {1\over 32m}$ and the homogeneities
appearing in $\CT_\ex$ are arbitrarily close (from below) to half-integers when $\kappa$ is small,
so that this term also belongs to $\CD^{\gamma_1,\eta_1}_\eps$. 


Since the homogeneities of elements of $\CW$ with homogeneity smaller than $2$ (say) 
are all of the form ${k\over 2} - \ell\kappa$
for $k$ and $\ell$ some integers with $\ell$ bounded by some fixed multiple of $m$,
we can apply Proposition~\ref{prop:epsSing} with $\delta = {1\over 2} - 2\nu$ 
provided that
we choose $\kappa$ sufficiently small. As a consequence, we see that
$H \mapsto  \hat \Eps^{j-1} \bigl(\CQ_{\le 0}(\DD H)^{2j}\bigr)$ is continuous from 
$\CD^{\gamma,\eta}_\eps$ into $\CD^{\delta,\eta_2}_\eps$ with 
\begin{equ}
\eta_2 = j(2\eta-1) + \kappa(2j-1) + {1\over 2} + \delta - \gamma
= j(2\eta-1) + \kappa(2j-1) - 1 - \nu \;.
\end{equ}
In order to be able to apply Proposition~\ref{prop:intSing}, we would like to 
guarantee that $\eta_2 > -2$. Provided that $\kappa$ is sufficiently small, this is
the case if $j(2\eta-1) > -1+ 2\nu$ for $j \le m$ which, keeping in mind
our choice of $\nu$, is guaranteed by the condition $\eta > {1\over 2} - {1\over 4m}$.


It then follows from Propositions~\ref{prop:intSing} and \ref{prop:proj} that, 
again provided that $\kappa$ is chosen
sufficiently small, there exists $\theta > 0$ such that 
$\CP \one_+ \CQ_{\le 0}\hat \Eps^{j-1} \bigl(\CQ_{\le 0}(\DD H)^{2j}\bigr)$ belongs to $\CD_\eta^{\gamma, \eta + \theta}$,
provided that 
\begin{equ}
j(2\eta-1) + \kappa(2j-1) + 1 - \nu \ge \eta + \theta\;.
\end{equ}
This is the case if $\eta (2j-1) > j-1 + 2\nu$ for $j=1,\ldots,m$, which in turn
is again guaranteed by the assumption that $\eta > {1\over 2} - {1\over 4m}$.
Since the heat kernel is non-anticipative, we actually know a little bit more:
as a consequence of \cite[Thm~7.1, Lem~7.3]{Regularity}, we know that
\begin{equ}
\|\CP \one_+ H\|_{\gamma,\eta} \le C T^\theta \|H\|_{\delta,\eta_2;\eps}\;,
\end{equ}
where $T$ denotes the length of the time interval over which the norms are taken.
As a consequence of our definitions, we then conclude that there exists a constant $C$
such that one has the bound
\begin{equ}
\|\CP \one_+ H\|_{\gamma,\eta;\eps} \le C (T+\eps)^\theta \|H\|_{\delta,\eta_2;\eps}\;.
\end{equ}
Combining these remarks, we see that for every $K > 1$ there exists a final time $T$ and
a constant $\eps_0$ such that, for all $\eps \in [0,\eps_0]$, the map $\CM$ defined in \eref{e:FPequ} 
maps the ball of radius $K$ in $\CD^{\gamma,\eta}_\eps$ into itself and is a contraction there, provided that the 
underlying model $\Pi \in \MM_\eps$ satisfies $\$\Pi\$_\eps \le K$ and that the initial condition
$h_0^{(\eps)}$ satisfies $\|h_0^{(\eps)}\|_{\eta,\eps} \le K/(2C)$ for $C$ as in \eqref{e:boundHarm}.

We now turn to the second part of the statement, namely the question of convergence
as $\eps \to 0$. We denote by $\CM_T$ the fixed point map given in \eref{e:FPequ}, where we make explicit 
the dependency on the terminal time $T$, and we write $\CM_T^{(\eps)}$ for the same map, but with
initial condition $h_0^{(\eps)} \in \CC^\eta_\eps$ and with respect to some model $\Pi^{(\eps)} \in \MM_\eps$.
Collecting all of the previously obtained estimates, we see that for $H \in \CD^{\gamma,\eta}_0$
and $H^{(\eps)} \in \CD^{\gamma,\eta}_\eps$, as well as corresponding models $\Pi \in \MM_0$ and $\Pi^{(\eps)} \in \MM_\eps$,
the fixed point map $\CM$ satisfies the bound
\begin{equ}
\|\CM_T^{(\eps)}(H^{(\eps)});\CM_T(H)\|_\eps 
\lesssim (T + \eps)^\theta \|H^{(\eps)};H\|_{\gamma,\eta;\eps} + \|\Pi^{(\eps)};\Pi\|_\eps + \|h_0^{(\eps)};h_0\|_{\eta;\eps} \;,
\end{equ}
where the proportionality constant is uniform over $T,\eps$ sufficiently small, as well as underlying models, initial conditions,
and modelled distributions $H$, $H^{(\eps)}$ belonging to a ball of fixed radius in the corresponding ``norms''.
It immediately follows that for sufficiently small final time $T$, one has
\begin{equ}[e:boundSolDist]
\|H^{(\eps)};H\|_{\gamma,\eta;\eps} \lesssim \|\Pi^{(\eps)};\Pi\|_\eps + \|h_0^{(\eps)};h_0\|_{\eta;\eps} \;.
\end{equ}
It remains to show that if $H$ is a solution to \eref{e:FPequ} up to some specified final time $T$, then 
the corresponding fixed point problem for $\CM_T^{(\eps)}$ also has a solution up to the same time $T$, provided that
$\eps$ is small enough, and the two underlying models and initial conditions are sufficiently close.
This is not completely trivial since it may well happen that $T$ is sufficiently large so that $\CM_T$ is no longer 
a contraction.

In view of \eref{e:boundSolDist}, it suffices to obtain a bound on the
solution, as well as the difference between solutions, at positive times in 
the same spaces $\CC^\eta_\eps$ that we choose
our initial condition in, so that we can iterate the bounds \eref{e:boundSolDist}. 
(See also the construction of maximal solutions in \cite[Prop.~7.11]{Regularity} which shows that a restarted solution
is again a solution of the original fixed point problem.)
This on the other hand immediately follows from Proposition~\ref{prop:restart}.
\end{proof}

To conclude this section, let us mention a straightforward way in which the 
solution map constructed in Theorem~\ref{theo:FP} actually relates to a PDE problem. 
Recall that, given any smooth
(actually continuous is enough) function $\zeta$,
the construction of Section~\ref{sec:canonical} yields a family
of maps $\LL_\eps\colon \CC^\infty \to \MM$  lifting $\zeta$ to an 
admissible model $(\Pi,\Gamma) = \LL_\eps(\zeta)$.
The following result is then immediate:

\begin{proposition}\label{prop:solutionCanonical}
Let $h_0 \in \CC^\gamma$ with $\gamma$ as in Theorem~\ref{theo:FP} and, given 
$\eps \in \R$ and $\zeta \in \CC^0$, let $H\in \CD^{\gamma,\eta}_\eps$ 
be the local solution to \eref{e:abstrFP} given by Theorem~\ref{theo:FP}
for the restriction to $\CT$ of the canonical model $\LL_\eps(\zeta)$.
Then, the function $h = \CR H$ is the classical (local) solution to the PDE
\begin{equ}
\d_t h = \d_x^2 h + \sum_{j=1}^m \eps^{j-1} \hat a_j (\d_x h)^{2j} + \zeta\;.
\end{equ}
\end{proposition}

\begin{proof}
Applying the reconstruction operator to both sides of \eref{e:abstrFP} and using the 
facts that the model $\LL_\eps(\zeta)$ is admissible, that $\CR \CP \one_+ = P * \one_+\CR$ (see \cite[Section 4]{Regularity}), and that $\CR \CQ_{\le 0} H = \CR H$, we see that
\begin{equ}
h = P * \one_+ \Bigl(\zeta + \sum_{j=1}^m \hat a_j \CR\bigl(\hat \Eps^{j-1}\bigl(\CQ_{\le 0} (\DD H)^{2j}\bigr)\bigr)\Bigr) + P h_0\;,
\end{equ}
where $\one_+$ denotes the indicator function of the set $\{t \ge 0\}$.
The claim now follows from the fact that the reconstruction operator
obtain for the model $\LL_\eps(\zeta)$ satisfies
\begin{equ}
\CR \bigl(\hat \Eps^{j-1}\bigl(\CQ_{\le 0} (\DD H)^{2j}\bigr)\bigr) = \eps^{j-1} (\d_x h)^{2j}\;,
\end{equ}
as a consequence of \eref{e:propReconstr} which holds on $\CT$ by restriction.
\end{proof}

\begin{remark}
Note that the parameter $\eps$ only enters in the construction of the model $\LL_\eps(\zeta)$.
In particular, the solution map built in  Theorem~\ref{theo:FP} does not itself have any
knowledge of $\eps$. This is the crucial feature of our construction that then allows us
to send $\eps$ to $0$ in a ``transparent'' way.
\end{remark}

\section{Renormalisation}
\label{sec:renorm}

The purpose of this section is to build a family of transformations on the space $\MM$ of all
admissible models for the regularity structure $(\CT,\CG)$ (as opposed to $(\CT_\ex,\CG)$ where we would not find 
any convergent renormalized model.)
These transformations will be of the type
\begin{equ}[e:newModel]
\hat \Pi_x \tau = \bigl(\Pi_x \otimes f_x\bigr)\DeltaW M_0 \tau\;,\qquad 
\hat f_x(\sigma) = f_x(\Wickh \sigma)\;,
\end{equ}
where $M_0 \colon \CT \to \CT$, $\Wickh\colon \CT_+ \to \CT_+$, and 
$\DeltaW\colon \CT \to \CT \otimes \CT_+$ are linear maps with additional properties
guaranteeing that $(\hat \Pi, \hat f)$ is again an admissible model.
Of course, we could also have just defined one single map instead of the composition $\DeltaW M_0$, but it turns 
out that the effects of the two factors are easier to analyse separately.

\subsection{Renormalisation of the average speed}

We start by discussing the map $M_0$ since this is easier to define. At the level of the
equation, the effect of $M_0$ will simply be to add a constant term to the right hand side.
Denote by $\Bad \subset \CT$ the set of canonical basis vectors that are of one of
the following two types:
\begin{equs}[e:tauBB]
\tau &= \Eps^\ell\bigl(\Psi^{2\ell} \CI'(\Eps^m(\Psi^{2m+2}))\CI'(\Eps^n(\Psi^{2n+2}))\bigr)\;, \\
\tau &= \Eps^\ell\bigl(\Psi^{2\ell+1} \CI'(\Eps^m(\Psi^{2m+1}\CI'(\Eps^n(\Psi^{2n+2}))))\bigr)\;,
\end{equs}
where $\ell, m, n \ge 0$ are positive integers. Note that in both cases
one has $|\tau| = -2(\ell+m+n+2)\kappa$. For any $\tau \in \Bad$, we then define
$L_\tau\colon \CT \to \CT$ by setting $L_\tau \tau = \one$ and $L_\tau \bar \tau = 0$ for every
canonical basis vector $\bar \tau \neq \tau$.

Finally, given constants $C_\tau \in \R$, we set
\begin{equ}[e:defM0]
M_0 = \exp\Big(-\sum_{\tau \in\Bad} C_\tau L_\tau\Big) = 1 -  \sum_{\tau \in\Bad} C_\tau L_\tau\;.
\end{equ}
This defines a  map
$(\Pi,f) \mapsto (\hat \Pi,\hat f)$ on models $(\Pi,f)\in(\CT,\CG)$  by $\hat \Pi_z \tau = \Pi_z M_0\tau$ and
$\hat f_z = f_z$, taking reconstruction operator $\CR$ associated to $(\Pi,f)$ to
$\hat \CR$ associated to $(\hat \Pi, \hat f)$.
These enjoys the following properties:

\begin{proposition}\label{prop:commuteMGamma}\begin{enumerate}
\item
For every $\Gamma \in \CG$ and  $\tau \in \CT$,  $M_0 \Gamma \tau = \Gamma M_0 \tau$;\label{p1}
\item $M_0 \CI'(\tau) = \CI'(\tau)$;\label{p2}
\item  \label{p3}The map
$(\Pi,f) \mapsto (\hat \Pi,\hat f)$  is continuous on the space of all models for $(\CT,\CG)$ and maps the 
space $\MM$ of \textit{admissible} models into itself;
\item  \label{p4} $\hat \CR H = \CR H - \sum_{\tau \in \Bad} C_\tau u_\tau$.
\end{enumerate}
\end{proposition}

\begin{proof}
Note first that $\Delta L_\tau \tau = \one \otimes \one = (L_\tau \otimes 1)\Delta \tau$.
Furthermore, for any $\bar \tau \in \CT$, one has $\Delta \bar \tau = \bar \tau \otimes \one
+ \sum \bar \tau^{(1)}\otimes \bar \tau^{(2)}$ with $|\tau^{(1)}|<|\tau|$ and by checking the few cases
of $\bar\tau\in\CT$ with $|\bar\tau|>0$ we see that  
$\bar \tau^{(1)}\not\in\Bad$ for any$\bar\tau\in\CT$. It immediately follows that if $\bar \tau \neq \tau$, one has $(L_\tau \otimes 1)\Delta \bar \tau = 0$,
thus concluding the proof of \ref{p1}.
\ref{p2} follows from the definition of  $M_0$ since  $\CI'(\tau) \not\in \Bad$.
\ref{p3}  follows from $1$ together with \cite[Prop.~2.30]{Regularity}.
Let now $H \in \CD^\gamma$ be such that, for every $\tau \in \Bad$, the corresponding
coefficient $u_\tau$ of $H$ is constant. 
 Then, it immediately follows from
\eqref{e:defM0} and the definition of $\CR$ that one has the identity \ref{p4}.
\end{proof}

\subsection{Wick renormalisation}
\label{sec:renormOp}

We now describe maps $\Wickh\colon \CT_+ \to \CT_+$, and 
$\DeltaW\colon \CT \to \CT \otimes \CT_+$ corresponding to  \emph{Wick renormalisation} with respect to the
Gaussian structure generated by solutions to the linearised equation.
Here the extended regularity structure $\CT_\ex$ is particularly useful.  The way the maps $\Wickh$ and $\DeltaW$ are constructed is to first build them on $\CT_\ex$ and then define them on $\CT$ simply by
restriction.  The key defining properties on the renormalization group,
that $(\hat \Pi, \hat f)$ defined through \eqref{e:newModel} is in $\MM$ and that $\DeltaW \tau = \tau \otimes \one + \sum \hat\tau^{(1)}\otimes \hat\tau^{(2)}$,
with $|\hat\tau^{(1)}| > \tau$,
are inherited by descent from $\CT_\ex$, since $\CT$ is a sector of $\CT_\ex$.  Hence it suffices to construct $\Wickh\colon \CT_+ \to \CT_+$, and 
$\DeltaW\colon \CT_\ex \to \CT_\ex \otimes \CT_+$

We  first build  an associated map $\Wick\colon \CT_\ex \to \CT_\ex$ depending on a parameter $C^\W \in \R$ \label{seew}
by setting
\begin{equ}[e:defWick]
\Wick = \exp(-C^\W \WickL)\;,
\end{equ}\label{defWick}
where the generator $\WickL$  iterates over every occurrence of the
sub-expression $\Psi^2$ and sends it to $\one$.
More formally, 
\minilab{e:L0}
\begin{equ}[e:basicL0]
\WickL\Xi = \WickL \one = 0\;,\qquad \WickL \Psi^{j} = \binom{j}{2} \Psi^{j-2}\;,
\end{equ}
for every $j \ge 2$. This is extended to  $\CT_\ex$ by imposing
the \emph{Leibniz rule},
\minilab{e:L0}
\begin{equ}[e:propL0]
\WickL (\tau \CI'(\bar\tau)) = \WickL(\tau) \CI'(\bar\tau) + \tau \CI'(\WickL\bar\tau)\;,
\end{equ}
as well as the commutation relations
\minilab{e:L0}
\begin{equ}[e:commute]
\WickL \CI'(\tau) = \CI'(\WickL \tau)\;,
\quad \WickL \Eps^\ell(\tau) = \Eps^\ell(\WickL \tau)\;,
\quad \WickL (X^\ell \tau) = X^\ell(\WickL \tau)\;,
\end{equ}
for any two formal expressions $\tau$ and $\bar \tau$ with $\bar \tau \neq \Xi$.
Since all elements of $\CT_\ex$ can be obtained in this way, this
defines $\WickL$ uniquely. In particular, these definitions imply that
\begin{equ}[e:Hermite]
\Wick \Psi^m = \H_m(\Psi,C^\W)\;,
\end{equ}
where $\H_m(x,c)$ \label{defofhermite} denote the generalised Hermite polynomials given by 
$\H_2(x,c) = x^2 -c$, $\H_4(x,c) = x^4 - 6 cx^2 + 3c^2$, etc.

Denote now by $\RR_0$ the set of all linear maps 
$M \colon \CT_\ex \to \CT_\ex$ which fix $\Xi$ and $\one$ and commute with the 
abstract integration operators $\CI$, $\CI'$ and $\Eps^\ell$.
Recall then from \cite[Sec.~8]{Regularity} that if $M \in \RR_0$, then one can uniquely associate to it maps
$\DeltaM\colon \CT_\ex \to \CT_\ex\otimes \CT_+$ and $\hat M\colon \CT_+ \to \CT_+$
satisfying the properties
\begin{equs}[e:defining]
\hat M \J_k &= \CM(\J_k \otimes 1) \DeltaM\;,\\
\hat M \EE^\ell_k &= \CM(\EE^\ell_k \otimes 1) \DeltaM\;,\\
(1 \otimes \CM)(\Delta \otimes 1)\DeltaM &= (M \otimes \hat M) \Delta\;,\\
\hat M(\sigma_1\sigma_2) &= (\hat M \sigma_1)(\hat M \sigma_2)\;,\qquad \hat M X^k = X^k\;,
\end{equs}
where $\CM \colon \CT_+ \otimes \CT_+ \to \CT_+$ denotes the product in the Hopf algebra $\CT_+$.

\begin{remark}
At first sight, our regularity structure appears not to be exactly of the 
type considered in \cite[Sec.~8]{Regularity}.
However, it follows from \eqref{e:intProp} that $\Eps^\ell$ is nothing but an abstract integration
map of order $\ell$ on $\CT_\ex$. It is then straightforward to verify that the results 
of that section still apply \textit{mutatis mutandis}
to the present situation.
\end{remark}

We then define the renormalisation group $\RR$ for $\CT_\ex$ as follows:

\begin{definition}\label{def:renorm}
A linear map $M \in \RR_0$ belongs to $\RR$ if the associated map $\DeltaM$
is such that $\DeltaM \tau = \tau \otimes \one + \sum \tau_M^{(1)}\otimes \tau_M^{(2)}$,
for some elements $\tau_M^{(i)}$ satisfying $|\tau_M^{(1)}| > \tau$.
\end{definition}

\begin{remark}
The definition of $\RR$ given here does not appear to match the definition given in
\cite[Def.~8.41]{Regularity}, where we also imposed a similar condition on a second operator $\hDeltaM$
built from $M$. It turns out however that Definition~\ref{def:renorm}
actually implies that second condition, as we show in the appendix.
\end{remark}

With these definitions at hand, given $M \in \RR$, we can use it to 
build a map $(\Pi, f) \mapsto (\Pi^M, f^M)$ mapping admissible models to admissible models by setting
\begin{equ}
\Pi_z^M = (\Pi_z \otimes f_z) \DeltaM \;, \qquad f_z^M = f_z \circ \hat M\;,
\end{equ}
see \cite[Thm~8.44]{Regularity}.
It is furthermore straightforward to verify that if an admissible model 
$(\Pi,f)$ consists of smooth functions satisfying the
identity \eqref{e:canonical2} then, as a consequence of the second identity in \eqref{e:defining},
the renormalised model $(\Pi^M,f^M)$ is also guaranteed to satisfy this identity.
The remainder of this section is devoted to the proof that the map $\Wick$ given
in \eqref{e:defWick} does indeed belong to $\RR$.
In order to do this, we first make a few general considerations.
Given a linear map $M \colon \CT_\ex \to \CT_\ex$ in $\RR_0$,
we first show the following result.

\begin{proposition}\label{prop:renormInt}
Let $M \in \RR_0$ and let $\DeltaM$ and $\hat M$ be the unique maps satisfying \eqref{e:defining}. 
Let $\tau$ be a canonical basis element of $\CT_\ex$, and let $\DeltaM \tau = \tau^{(1)}_M \otimes \tau^{(2)}_M$
(with summation implicit) be such that $|\tau^{(1)}_M| \ge |\tau|$. Then, one has
\begin{equs}
\DeltaM \Eps^\ell(\tau) &= (\Eps^\ell \otimes 1)\DeltaM \tau - \sum_{|k| > |\tau| + \ell} {X^k \over k!} \otimes \EE^\ell_k(\tau_M^{(1)})\tau_M^{(2)}\;, \\
\DeltaM \CI'(\tau) &= (\CI' \otimes 1)\DeltaM \tau - \sum_{|k| > |\tau| + 1} {X^k \over k!} \otimes \J_{k+1}(\tau_M^{(1)})\tau_M^{(2)}\;, 
\end{equs}
and similarly for $\DeltaM \CI(\tau)$.
\end{proposition}

\begin{proof}
We use the shorthand $D = (1 \otimes \CM)(\Delta \otimes 1)$.
We only give a proof for $\CI’(\tau)$. The proofs for $\CI'(\tau)$ and $\Eps^\ell(\tau)$ 
are identical since these
operators have exactly the same algebraic properties.
Combining \eqref{e:intProp} with the first identity in \eqref{e:defining} and the fact 
that $\CI’$ and $M$ commute by assumption, we obtain the identity
\begin{equs}
(M \otimes \hat M)\Delta \CI’(\tau) &= (\CI’ M\otimes \hat M)\Delta \tau + \sum_{|k+\ell| < |\tau|+1}
{X^k \over k!} \otimes {X^\ell\over \ell!} \hat M \J_{k+\ell+1}(\tau)\\
&=(\CI’ \otimes 1)D\DeltaM \tau + \sum_{|k+\ell| < |\tau|+1}
{X^k \over k!} \otimes {X^\ell\over \ell!} \J_{k+\ell+1}(\tau^{(1)}_M)\tau^{(2)}_M\;.
\end{equs}
On the other hand, using again \eqref{e:intProp}, we also have the identity
\begin{equs}
D(\CI’ \otimes 1)\DeltaM \tau  &= 
(\CI’\otimes 1)D\DeltaM \tau
+ \sum_{k,\ell} {X^k \over k!}\otimes {X^\ell \over \ell!} \CM(\J_{k+\ell+1} \otimes 1)\DeltaM \tau \\
& = (\CI’\otimes 1)D\DeltaM \tau
+ \sum_{|k+\ell| < |\tau_M^{(1)}| + 1} {X^k \over k!}\otimes {X^\ell \over \ell!} \J_{k+\ell+1}(\tau^{(1)}_M)\tau^{(2)}_M\;,
\end{equs}
so that, since $|\tau^{(1)}_M| \ge |\tau|$ by assumption, one has
\begin{equ}[e:DIDeltaM]
D(\CI’ \otimes 1)\DeltaM \tau  = (M \otimes \hat M)\Delta \CI’(\tau) 
+ \sum_{|k+\ell| > |\tau| + 1} {X^k \over k!}\otimes {X^\ell \over \ell!} \J_{k+\ell+1}(\tau^{(1)}_M)\tau^{(2)}_M\;.
\end{equ}
At this stage we note that, if $\{\tau_k\}$ is any collection of elements of $\CT_\ex$ indexed
by the multiindex $k$, then it follows from the action of $\Delta$ on $X^m$ that one has the identity
\begin{equ}
D \Bigl({X^m \over m!} \otimes \tau_m\Bigr)
= \sum_{k + \ell = m} {X^k \over k!} \otimes {X^\ell \over \ell!}\tau_m\;.
\end{equ}
Combining this with \eqref{e:DIDeltaM}, we conclude that 
\begin{equ}
D(\CI’ \otimes 1)\DeltaM \tau  = (M \otimes \hat M)\Delta \CI’(\tau) + D\sum_{|k| > |\tau| + 1} {X^k \over k!}\otimes \J_{k+1}(\tau^{(1)}_M)\tau^{(2)}_M\;.
\end{equ}
Since furthermore $(M \otimes \hat M)\Delta \CI’(\tau) = D\DeltaM \CI’ \tau$ by the definition
\eqref{e:defining} of
$\DeltaM$ and since the linear map $D$ is invertible (it differs from the identity by a nilpotent
operator), the claim follows at once.
\end{proof}

\begin{proposition}\label{prop:prodRule}
Let $M \in \RR_0$, let $k \ge 0$ and let $V_0,\ldots, V_k$ be sectors of $\CT_\ex$ 
such that, if $\tau_i \in V_i$, then $\tau_0 \cdots \tau_k \in \CT_\ex$ and
$M(\tau_0 \cdots \tau_k) = (M\tau_0)\cdots (M\tau_k)$.
Then, one also has $\DeltaM (\tau_0 \cdots \tau_k) = (\DeltaM \tau_0)\cdots (\DeltaM \tau_k)$.
\end{proposition}

\begin{proof}
Let $\tau_i \in V_i$ as in the statement and set $\tau = \tau_0 \cdots \tau_k$.
Since $\hat M$ is a multiplicative morphism, 
it follows from our assumption that
\begin{equ}[e:multM]
(M \otimes \hat M)\Delta \tau = \prod_{i=0}^k (M \otimes \hat M)\Delta\tau_i\;.
\end{equ}
Since $\DeltaM \tau = D^{-1}(M \otimes \hat M)\Delta \tau$ (with $D$ as above) and since 
$D$ is a multiplicative morphism,
the claim follows at once by applying $D^{-1}$ to both sides of \eqref{e:multM}.
\end{proof}

We then have

\begin{proposition}\label{prop:renormProd}
Let $\Wick$ be as above,
 let $\DeltaW$ and $\Wickh$ be the corresponding maps satisfying \eqref{e:defining},
and let $\tau \in \CT_\ex$ be a canonical
basis vector of the form
\begin{equ}[e:decompTau]
\tau = \Psi^m \prod_{i=1}^k \CI'(\tau_i)\;,
\end{equ}
where $k,m \ge 0$, 
and the $\tau_i$ are canonical basis vectors with $\tau_i \neq\Xi$.
Then, one has
\begin{equ}[e:defDeltaW]
\DeltaW\tau = (\Wick\Psi^m \otimes \one) \prod_{i=1}^k \DeltaW \CI'(\tau_i)\;.
\end{equ}
\end{proposition}

\begin{proof}
We first note the following very important fact. By the construction of $\CT_\ex$, if
$\CI'(\tau) \in \CT_\ex$ with $\tau \neq \Xi$, then $\tau$ cannot contain any factor $\Xi$
by the construction of $\CT_\ex$.
Therefore, by construction, $\WickL \tau_i$ does not contain any summand
proportional to $\Xi$ either. 
As a consequence of the
``Leibnitz rule'' satisfied by the $L_j$, this then shows that, for every $p \ge 0$,
\begin{equ}
(\WickL)^p \tau = \sum_{p_0+\ldots+p_k = p} \bigl((\WickL)^{p_0}\Psi^m\bigr) \prod_{i=1}^k \CI'((\WickL)^{p_i}\tau_i)\;,
\end{equ}
which in particular implies that 
\begin{equ}[e:actionM]
\Wick\tau = \bigl(\Wick\tau_{\ell,m,n}\bigr) \prod_{i=1}^k \CI'(\Wick\tau_i)\;.
\end{equ}

Similarly, one verifies that
if one writes $\Delta \tau_i = \tau_i^{(1)} \otimes \tau_i^{(2)}$ (with an implicit
summation over such terms), then none of the terms $\tau_i^{(1)}$ can be equal to $\Xi$.
Applying the definition of $\Delta$, one also verifies that the linear span of the vectors
$\Psi^m$ is stable under the action of the structure group $\CG$. Combining 
these observations, we see that Proposition~\ref{prop:prodRule} applies, so that
\begin{equ}
\DeltaW\tau = \bigl(\DeltaW\Psi^m\bigr) \prod_{i=1}^k \DeltaW\CI'(\tau_i)\;.
\end{equ}
The fact that $\DeltaW\Psi^m = (\Wick\Psi^m \otimes \one)$ can easily be verified ``by hand''
from \eqref{e:defining}.
\end{proof}

As a corollary of these two results, it is now easy to show that $\Wick \in \RR$. 

\begin{corollary}
One has $\Wick \in \RR$.
\end{corollary}

\begin{proof}
As a consequence of the construction of $\CT_\ex$ given in Section~\ref{sec:structure}, 
we see that every one of its basis elements can be built from
$\Xi$ by making use of the operations 
$\tau \mapsto \CI(\tau)$, $\tau \mapsto \CI'(\tau)$, $\tau \mapsto \Eps^\ell(\tau)$,
$\tau \mapsto X^\ell \tau$,
as well as $(\tau_1,\ldots,\tau_k) \mapsto \Psi^m \prod_{i=1}^k \CI'(\tau_i)$
with $\tau_i \neq\Xi$.
Since $\DeltaW \Xi = \Xi \otimes \one$ and since the upper triangular structure of $\DeltaW$
is preserved under all of these operations by Propositions~\ref{prop:renormInt} and 
\ref{prop:renormProd}, the claim follows.
\end{proof}

\subsection{Renormalised equations}

Let now $(\Pi,f) = \PPP\LL_\eps (\zeta)$, where $\zeta$ is a continuous function 
and the canonical lift $\LL_\eps$ is as in Section~\ref{sec:canonical}. 
We furthermore consider the renormalised model $(\hat \Pi, \hat f)$ given
by \eqref{e:newModel} with $M_0$ and $\Wick$ as in \eqref{e:defM0} and \eqref{e:defWick}.
In particular, $\Wick$ depends on the renormalisation constant $C^\W$ while
$M_0$ depends on  a collection of renormalisation constants $C_\tau$.

The aim of this section is to 
show that if $H$ solves the abstract fixed point problem \eref{e:abstrFP} for the model $(\hat \Pi, \hat f)$, then $h = \hat \CR H$, where $\hat \CR$ the reconstruction operator associated to the 
renormalised model, can be identified with the solution to a modified PDE.
In order to derive this new equation, we combine the explicit abstract form of the solutions
with the product formula given by Proposition~\ref{prop:renormProd}.
The result is the following, where $\CC$ denotes the space of continuous functions on
$\R \times S^1$:

\begin{proposition}\label{prop:solutionRenormalised}
Let $h_0 \in \CC^1$ and, given $\eps \in \R$ and $\zeta \in \CC^0$, 
let $H\in\CD^{\gamma,\eta}_\eps\subset \CD^{\gamma,\eta}$
be the local solution to \eref{e:abstrFP} given by Theorem~\ref{theo:FP} for 
the renormalised model $(\hat \Pi, \hat f)$ obtained from $\PPP\LL_\eps (\zeta)$ in the way 
described above.
Then, there exists a constant $c$ such that the 
function $h = \hat \CR H$ is the classical (local) solution to the PDE
\begin{equ}
\d_t h = \d_x^2 h + \sum_{j=1}^m \eps^{j-1} \hat a_j {\H}_{2j}\bigl(\d_x h, C^\W) + c + \zeta\;,
\end{equ}
with initial condition $h_0$.
\end{proposition}

\begin{remark}
The constant $c$ is a suitable linear combination of the constants $C_\tau$ 
appearing in the definition \eqref{e:defM0} of $M_0$, with coefficients 
depending on the constants $\hat a_j$.
In principle, one can derive an explicit expression for it, but this expression does not
seem to be of particular interest. The only important fact is that if we
write
\begin{equ}
\tau_1 = \Psi \CI'(\Psi \CI'(\Psi^2))\;,\qquad
\tau_2 = \CI'(\Psi^2)^2\;,
\end{equ}
then the corresponding renormalisation constants $c_{\tau_1}$ and $c_{\tau_2}$ only ever arise as a multiple
of $4 c_{\tau_1} + c_{\tau_2}$. This is important since, as we will see 
in Theorem~\ref{theo:logs} below, 
these renormalisation constants need to be chosen to diverge logarithmically as $\eps \to 0$
and the particular form of this linear combination guarantees that these logarithmic divergencies
cancel out and are therefore not visible in the renormalised equations. 
\end{remark}

\begin{proof}
As in the proof of Proposition~\ref{prop:solutionCanonical}, we use the fact that the 
renormalised model is admissible to conclude that, when applying $\hat \CR$ to both sides of
\eref{e:abstrFP}, the function $h = \hat \CR H$ satisfies the identity
\begin{equ}[e:FPconcr]
h = P * \one_+ \Bigl(\zeta + \sum_{j=1}^m \hat a_j \hat \CR \bigl(\hat \Eps^{j-1}\bigl(\CQ_{\le 0} (\DD H)^{2j})\bigr)\Bigr) + P h_0\;.
\end{equ}
At this stage, the proofs diverge since it is no longer the case that $\hat \CR$ preserves the usual
product. The only fact that we can use is that $(\hat \CR F)(z) = \bigl(\hat \Pi_z F(z)\bigr)(z)$,
combined with the definition of the renormalised model $\hat \Pi$.

Denoting by $\RWick$\label{rwick} the reconstruction operator associated to the model
$(\Pi_x \otimes f_x)\DeltaW$, then it follows immediately from \eqref{e:propR} and \eqref{e:newModel} that one has the identity
\begin{equ}[e:defRhat]
\hat \CR U = \RWick M_0 U\;.
\end{equ}
Furthermore, as a consequence of the first identity in \eqref{e:defining} combined with 
\eqref{e:canonical2}, 
one has the identity
\begin{equ}[e:propWick]
\bigl(\RWick \hat \Eps^\ell(U)\bigr)(z)  = \eps^\ell \bigl(\RWick U\bigr)(z)\;,
\end{equ}
provided that the underlying model $(\Pi,f)$ is of the form $\LL_\eps(\zeta)$ for some smooth $\zeta$.
(Note though that this identity fails in general if we were to replace $\RWick$ by $\hat \CR$.)

It follows from the fact that $\DD\CP F$ differs from $\CI' F$ by a Taylor polynomial at each point
that if $H$ is the solution to \eqref{e:abstrFP}, then one can write
\begin{equ}
\DD H(z) = \Psi + U(z)\;,
\end{equ}
where the remainder $U$ only contains components proportional to either $\one$, $X$, or
$\CI'(\tau)$ with $\tau \neq \Xi$.
In particular, none of the components belongs to $\Bad$, so that one has the identity
\begin{equ}
\bigl(\hat \CR \DD H\bigr)(z) =\bigl(\RWick \DD H\bigr)(z) = \bigl(\Pi_z \Psi\bigr)(z) + \bigl((\Pi_z \otimes f_z)\DeltaW U(z)\bigr)(z)\;.
\end{equ}
On the other hand, for $\ell \ge 0$, 
we can apply the reconstruction operator to $\hat \Eps^\ell \bigl((\DD H)^{2\ell+2}\bigr)$
and combine \eqref{e:defRhat} with \eqref{e:propWick} and the definition of $M_0$ thus yielding
\begin{equs}[e:multEps]
\bigl(\hat \CR \hat \Eps^\ell \bigl(\CQ_{\le 0}(\DD H)^{2\ell+2}\bigr)\bigr)(z) &= 
\bigl(\RWick \hat \Eps^\ell \bigl(\CQ_{\le 0}(\DD H)^{2\ell+2}\bigr)\bigr)(z) + c \\
&= \eps^\ell \bigl(\RWick (\DD H)^{2\ell+2}\bigr)(z) + c\;,
\end{equs}
for some constant $c$. It thus remains to compute $\RWick (\DD H)^{m}$ for arbitrary $m$. As a consequence
of Proposition~\ref{prop:renormProd} and \eqref{e:Hermite}, we have
\begin{equs}
\DeltaW (\DD H(z))^{m} &= 
\sum_{k+\ell = m} \binom{m}{k} (\Wick \Psi^k \otimes \one) \bigl(\DeltaW U(z)\bigr)^\ell \\
&= \sum_{k+\ell = m} \binom{m}{k} (\H_k(\Psi,C^\W) \otimes \one) \bigl(\DeltaW U(z)\bigr)^\ell\;.
\end{equs}
At this stage, we use the fact that since our original model originates from a canonical
lift by assumption, it has the property that $\Pi_x \tau \bar \tau = \Pi_x \tau \, \Pi_x \bar \tau$.
Applying $\Pi_z \otimes f_z$ to both sides of this equality and combining this with the fact 
that $f_z$ is also multiplicative, we conclude that 
\begin{equs}
\bigl(\RWick (\DD H(z))^{m}\bigr)(z) &= \sum_{k+\ell = m} \binom{m}{k} \H_k\bigl(\bigl(\Pi_z \Psi\bigr)(z),C^\W\bigr) \bigl((\Pi_z \otimes f_z)\DeltaW U(z)\bigr)(z)^\ell  \\
&= \H_m\bigl((\RWick \DD H)(z), C^\W\bigr)\;.
\end{equs}
Combining this with \eqref{e:multEps} and \eqref{e:FPconcr}, the claim follows.
\end{proof}

\section{Convergence of the models}
\label{sec:convModel}

In this section, we now show how the renormalisation maps from the previous section
can be used to renormalise the models built from regularisations of space-time white
noise. From now on, we will use a graphical shorthand notation 
similar to the one used in \cite{KPZ} for symbols $\tau \in \CW$ which do not contain
the symbol $\CE$: dots represent the symbol $\Xi$, 
lines denote the operator $\CI'$, and the joining of symbols by their roots denotes their
product. For example, one has $\<1> = \CI'(\Xi) = \Psi$,
$\<2> = \Psi^2$, $\<21> = \Psi \CI'(\Psi^2)$, etc.
We will also assume from now on that $(\CT,\CG)$ has been truncated in the way specified in
the beginning of Section~\ref{sec:extended}

With the same graphical notations, we also define two additional renormalisation 
constants 
\begin{equ}[e:defRenormConst]
C_2^{(\eps)} = \;
\begin{tikzpicture}[baseline=0.6cm,scale=0.35]
\node at (0,0) [root] (0) {};
\node at (0,2) [dot] (1) {};
\node at (-2,2) [dot] (2) {};
\node at (0,4) [dot] (4) {};
\node at (2,3) [dot] (5) {};

\draw[kernel] (1) to (0);
\draw[keps] (2) to (0);
\draw[keps] (5) to (1);
\draw[kernel] (4) to (1);
\draw[keps] (5) to (4);
\draw[keps] (2) to (4);
\end{tikzpicture}\;,\qquad
C_3^{(\eps)} = \;
\begin{tikzpicture}[baseline=0.6cm,scale=0.35]
\node at (0,0) [root] (0) {};
\node at (-2,2) [dot] (1) {};
\node at (2,2) [dot] (2) {};
\node at (0,2) [dot] (3) {};
\node at (0,4) [dot] (4) {};

\draw[kernel] (1) to (0);
\draw[kernel] (2) to (0);
\draw[keps] (3) to (1);
\draw[keps] (3) to (2);
\draw[keps] (4) to (1);
\draw[keps] (4) to (2);
\end{tikzpicture}\;,
\end{equ}
where a plain arrow represents the kernel $K'$.
We will see in Section~\ref{sec:logs} below that these
two constants diverge logarithmically as $\eps \to 0$, but this is not important at the moment.
We also set $C_0^{(\eps)}$ to be the left hand side of \eqref{e:defrhopicture} and, for all $\tau \in \Bad \setminus \{\<22>,\<211>\}$  defined in the prelude to \eqref{e:tauBB}, we set 
\begin{equ}[e:defCtaueps]
C_\tau^{(\eps)} = \E \bigl(\PPi^{(\eps)}\Wick \tau\bigr) (0)\;,
\end{equ}
where $\Wick$ is the map defined in \eqref{e:defWick} with $C^\W = C_0^{(\eps)}$ and 
$\PPi^{(\eps)}\colon \CT \to \CC$ denotes the linear map 
defined recursively by
$\PPi^{(\eps)} \Xi = \xi^{(\eps)}$ and 
\begin{equ}
\PPi^{(\eps)} \Eps^k(\tau) = \eps^k \PPi^{(\eps)} \tau\;,\quad
\PPi^{(\eps)} \CI'(\tau) = K' * \PPi^{(\eps)} \tau\;,\quad
\PPi^{(\eps)} \tau \bar \tau = \bigl(\PPi^{(\eps)} \tau \bigr)\bigl(\PPi^{(\eps)} \bar\tau \bigr)\;.
\end{equ}
Note that the functions $\PPi^{(\eps)}\tau$ are stationary, so the choice of the evaluation
at $0$ in \eqref{e:defCtaueps} is irrelevant.

We then define a map 
$M^{(\eps)}$ acting on the space of admissible models by
\begin{equ}[e:defMeps]
M^{(\eps)}\colon (\Pi,f) \mapsto (\hat \Pi, \hat f)\;,
\end{equ}
with $(\hat \Pi, \hat f)$ as in \eqref{e:newModel}, where we set
\begin{equ}
C_{\<211s>} = 2C_2^{(\eps)}\;,\qquad C_{\<22s>} = 2C_3^{(\eps)}\;,
\end{equ}
as well as $C_\tau = C_\tau^{(\eps)}$ for $\tau \in \Bad \setminus \{\<22>,\<211>\}$.
With these notations at hand, the following is then the 
main result of this section.

\begin{theorem}\label{theo:convModel}
Let $\xi^{(\eps)}$ be as in \eref{e:xiEps} and consider the sequence of models 
on $\CT$ given by
\begin{equ}
\Z_\eps = M^{(\eps)}\PPP \LL_\eps(\xi^{(\eps)})\;.
\end{equ}
Then, there exists a random model $\Z$ such that $\$\Z_\eps; \Z\$_\eps \to 0$ in probability
as $\eps \to 0$. 
Furthermore, the limiting model $\Z = (\hat \Pi, \hat f)$ is independent of the choice of mollifier $\rho$
and it satisfies $\hat\Pi_z\tau = 0$ for every symbol $\tau$ containing at least one
occurrence of $\Eps$.
\end{theorem}

Before we turn to the proof of Theorem~\ref{theo:convModel}, we give a criterion
allowing to verify whether a sequence of models converges in $\MM_\eps$.

\subsection{A convergence criterion}

The following result is very useful. Here, we fix a sufficiently regular wavelet basis / 
multiresolution analysis with compactly supported elements
and we reuse the notation of \cite[Sec.~3.1]{Regularity}. In particular, $\Psi$
is a finite set of functions in $\CB$ such that the wavelet basis is obtained
by translations and rescalings of elements in $\Psi$ (we use the notation $\Psi$  to be consistent with \cite{Regularity}.  It should not be confused with the shorthand for $\CI'(\Xi)$ used elsewhere in the paper). Here, we follow the usual convention,
so $\psi_z^n$ denotes a wavelet basis function at level $n$ (scale $2^{-n}$) 
centred at some point $z$ in the level $n$ dyadic set $\Lambda^n$. We normalise 
these basis functions so that their $L^2$ norm (not the $L^1$ norm as before!) equals $1$.

Recall also that our definition of the spaces $\MM_\eps$ involves the constant
$\bar\gamma = 1-{1\over 32m}$ as defined in \eqref{e:modelEpsWantedPi}.

\begin{proposition}\label{prop:criterion}
Let $(\Pi^{(\eps)}, f^{(\eps)})$ be a family of models for the regularity structure $(\CT,\CG)$
converging to a limiting model $(\Pi,f)$ in the sense
that $\lim_{\eps\to 0}\$\Pi^{(\eps)} ; \Pi\$ = 0$.
Assume that, for some $\delta > 0$, one has
\begin{equ}[e:assumfeps]
\bigl|f_z^{(\eps)}\bigl(\EE^k_\ell(\tau)\bigr)\bigr| \le C \eps^{|\tau| + k - |\ell| + {\delta}}\;,
\end{equ} for $\tau$ and $k,\ell$ as in  \eqref{e:modelEpsWantedf}, and that furthermore
for $\tau \in \CU'$
\begin{equ}[e:assumPieps]
\bigl|\bigl(\Pi_z^{(\eps)} \tau\bigr)(\psi_z^n)\bigr| \le C 2^{-{3n \over 2} - \bar \gamma n} \eps^{|\tau| - \bar \gamma + \delta}\;,
\end{equ}
for every $n \ge 0$, every $z \in \Lambda^n$, and every $\psi \in \Psi$.
Then, one has $\lim_{\eps\to 0}\$\Pi^{(\eps)} ; \Pi\$_{\eps,0} = 0$.
\end{proposition}

\begin{proof}
We only need to show that, for any test function $\eta \in \CB$ with $\int \eta(z)\,dz = 0$,
one has
\begin{equ}
\bigl|\bigl(\Pi_z^{(\eps)} \tau\bigr)(\eta_z^\lambda)\bigr| \lesssim \lambda^{\bar \gamma} \eps^{|\tau| - \bar \gamma + \delta}\;,
\end{equ}
provided that $\lambda \le \eps$, since this will then guarantee that
$\|\Pi^{(\eps)}\|_\eps \lesssim \eps^{\delta}$.
We fix $N \ge 0$ such that $2^{-N} \le \eps \le 2^{1-N}$ and we write
\begin{equ}[e:decompeta]
\eta_z^\lambda = \sum_{z'\in \Lambda^N} A^N_{z'} \phi_{z'}^N + \sum_{n \ge N} \sum_{\psi\in \Psi} A^{n,\psi}_{z'} \psi_{z'}^n\;. 
\end{equ}
It is then a simple consequence of the scaling properties of these objects that one has the bounds
\begin{equ}
|A^N_{z'}| \lesssim 2^{{3N \over 2}} (\lambda / 2^{-N})\;,\qquad 
|A^{n,\psi}_{z'}| \lesssim 
\left\{\begin{array}{cl}
	2^{{3n \over 2}} (\lambda / 2^{-n}) & \text{if $2^{-n} \ge \lambda$,} \\
	2^{{3n \over 2}} (2^{-n}/\lambda)^{3} & \text{otherwise.}
\end{array}\right.
\end{equ}
Note here that the factor $2^{{3n \over 2}}$ comes from the fact that the functions
$\psi_z^n$ and $\phi_z^n$ appearing in \eqref{e:decompeta} are normalised in 
$L^2$ rather than in $L^1$. Furthermore, the factor $\lambda / 2^{-n}$ appearing in 
the first two bounds is a consequence of the fact that $\eta$ integrates to $0$ by assumption
and the wavelet basis is sufficiently regular ($\CC^2$ is enough).

We furthermore obtain from \eqref{e:assumPieps} and the fact that $\Pi^{(\eps)}$ converges to a limit 
(and therefore is bounded in $\MM$, uniformly in $\eps$) the bound
\begin{equs}
\bigl|\bigl(\Pi_{z}^{(\eps)} \tau\bigr)(\psi_{z'}^n)\bigr| &= \bigl|\bigl(\Pi_{z'}^{(\eps)} \Gamma_{z'z}^{(\eps)}\tau\bigr)(\psi_{z'}^n)\bigr| \lesssim
2^{-{3n\over 2}}\sum_{\alpha} |z-z'|^{|\tau|-\alpha} 2^{-\bar \gamma n} \eps^{\alpha-\bar \gamma- {\delta \over 2}} \label{e:boundPipsi}\\
&\lesssim 2^{-{3n\over 2}}\sum_{\alpha} (\lambda + 2^{-n})^{|\tau|-\alpha} 2^{-\bar \gamma n} \eps^{\alpha-\bar \gamma+ \delta}
\lesssim 2^{-{3n\over 2}} 2^{-\bar \gamma n} \eps^{|\tau|-\bar \gamma+ \delta}\;,
\end{equs}
where $\alpha$ runs over all homogeneities less or equal to $|\tau|$ appearing in
$\CU'$ and the last inequality is a consequence of the fact that we only consider $\lambda$ and $n$ 
such that $\lambda + 2^{-n}\lesssim \eps$. (The term corresponding to $\one$, for which
\eqref{e:assumPieps} does not hold in principle, does not contribute since $\psi_{z'}^n$ integrates to $0$.)
The same bound (with $n$ replaced by $N$) can also be obtained for 
$\bigl|\bigl(\Pi_{z}^{(\eps)} \tau\bigr)(\phi_{z'}^N)\bigr|$. (In that case $\phi_{z'}^N$ does not 
integrate to $0$, but since $2^{-N} \approx \eps$, the contribution arising from
$\one$ is of order $2^{-{3N\over 2}}\eps^{|\tau|}$ which is in particular bounded by 
the right hand side of \eqref{e:boundPipsi}.)

It remains to note that, for fixed $n$, the number of non-vanishing values of $A_{z'}^{n,\psi}$
(or $A_{z'}^{N}$)
is of order $1$ if $2^{-n} \ge \lambda$ and of order $(\lambda / 2^{-n})^3$ otherwise. 
Combining all of these bounds and using the fact that $\bar \gamma \in (0,1)$, we finally obtain
\begin{equ}
\bigl|\bigl(\Pi_{z}^{(\eps)} \tau\bigr)(\eta_{z}^n)\bigr|
\lesssim \sum_{\lambda \le 2^{-n} \le \eps}\lambda 2^{(1-\bar \gamma)n} \eps^{|\tau|-\bar \gamma- {\delta \over 2}} + \sum_{2^{-n} \le \lambda} 2^{-\bar \gamma n} \eps^{|\tau|-\bar \gamma- {\delta \over 2}} \lesssim \lambda^{\bar \gamma} \eps^{|\tau|-\bar \gamma- {\delta \over 2}}\;,
\end{equ}
as required.
\end{proof}

As a consequence, we obtain the following Kolmogorov-type convergence criterion.

\begin{proposition}\label{prop:momentBounds}
Let $(\CT,\CG)$ be the regularity structure built in Section~\ref{sec:structure}
and let $\hat \Pi^{(\eps)}$ be as in Theorem~\ref{theo:convModel}. Assume that there exists $\delta > 0$ such
that, for every test function
$\eta \in \CB$, every $\tau \in \bar \CW$ with $|\tau| < 0$, every $x \in \R^2$ and 
every $\lambda \in (0,1]$ 
there exists a random variable $\bigl(\hat \Pi_z\tau\bigr)(\eta_z^\lambda)$
such that
\minilab{e:requiredBounds}
\begin{equ}[e:wantedKPZ]
\E \bigl|\bigl(\hat \Pi^{(\eps)}_z\tau\bigr)(\eta_z^\lambda)\bigr|^2
\lesssim \lambda^{2|\tau|+\delta}\;,\qquad 
\E \bigl|\bigl(\hat \Pi_z\tau - \hat \Pi^{(\eps)}_z\tau\bigr)(\eta_z^\lambda)\bigr|^2
\lesssim \eps^\delta \lambda^{2|\tau|+\delta}\;.
\end{equ}
Assume furthermore that, for $ \tau$ with  $\CE^k(\tau) \in  \CW_+$, one has
\minilab{e:requiredBounds}
\begin{equ}[e:assumf]
\E\, \bigl|D_z^\ell \hat f_z^{(\eps)}\bigl(\EE^k_0(\tau)\bigr)\bigr| \lesssim \eps^{|\tau| + k - |\ell| + \delta}\;,
\end{equ}
and that, for $\tau \in \CU'$, one has the bound
\minilab{e:requiredBounds}
\begin{equ}[e:assumPi]
\E\,\bigl|\bigl(\hat \Pi_z^{(\eps)} \tau\bigr)(\eta_z^\lambda)\bigr| \lesssim \lambda^{\bar \gamma+\delta} \eps^{|\tau| - \bar \gamma+\delta}\;,
\end{equ}
for $\lambda \le \eps$ and for test functions $\eta$ that integrate to $0$.
Then, there exists a random model $(\hat \Pi, \hat f) \in \MM_0$ such that
$\$\hat \Pi^{(\eps)}; \hat \Pi\$_\eps \to 0$ in probability as $\eps \to 0$.
\end{proposition}

\begin{proof}
The proof goes in two steps: first, we show that there is a limiting model
$(\hat \Pi, \hat f)$ such that $\$\hat \Pi^{(\eps)};\hat \Pi\$ \to 0$ in probability,
and then we show that $\|\hat \Pi^{(\eps)}\|_\eps \to 0$ in probability. 
If we restrict ourselves to the sector $\CT_- \subset \CT$ spanned
by basis vectors in $\CW$ with negative (or vanishing) homogeneity, the first step
follows in exactly the same way as in \cite{Regularity,Etienne},
using \cite[Thm~10.7]{Regularity}.
This by itself does however not yet yield convergence on all of $\CT$. The reason 
for this is that it contains basis vectors of the form $\bar \tau = \Eps^k(\tau)$ with
$|\Eps^k(\tau)| > 0$. These do not satisfy the assumptions of \cite[Prop.~3.31]{Regularity} since 
one does not have any \textit{a priori} control over the components of
$\Gamma_{zz'}\bar \tau$. (Unlike in \cite{Regularity,Etienne} where, for vectors of the
form $\bar \tau = \CI(\tau)$, such a control was given by \cite[Thm~5.14]{Regularity}.) 

Note now that, by the definition \eqref{e:defCWbar}, all of the vectors of the form 
$\tau = \Eps^k(\bar \tau)$
appearing in $\bar \CW$ have $|\bar \tau| < 0$ (or $\bar \tau = \one$, but this case is
trivial). By simple inspection, we see that those
vectors such that furthermore $|\tau| > 0$ are necessarily of the form
\begin{equ}[e:specialTau]
\tau = \Eps^{j-1}(\bar \tau)\;,\qquad \bar \tau = \Psi^{2j-n} \CI'(\Eps^{k_1-1}\Psi^{2k_i})\cdots \CI'(\Eps^{k_\ell-1}\Psi^{2k_\ell})\;,
\end{equ}
with $n > 2$, $j \in \{\lceil n/2\rceil,\ldots,m\}$, and $k_i \in \{1,\ldots,m\}$.
At this stage, we note that since $|\CI'(\Eps^{k-1}\Psi^{2k})| < 0$, $|\Eps^{k-1}\Psi^{2k}| < 0$, and $|\Psi| < 0$, the structure group
acts trivially on $\bar \tau$, so that one has the identity
\begin{equ}[e:equalfeps]
f_z^{(\eps)}(\EE^k_\ell(\tau)) = D_z^\ell f_z\bigl(\EE^k_0(\tau)\bigr)\;.
\end{equ}
Setting $g(z) = f_z\bigl(\EE^k_0(\tau)\bigr)$ as a shorthand, it then follows
from \eqref{e:exprgammazbarz} that
\begin{equ}[e:boundgamma]
\gamma_{z\bar z}(\EE^k_\ell(\tau)) = g(\bar z) - \sum_{|m| < |\tau|+k-|\ell|} {(\bar z - z)^m \over m!}
D^{(m)}g(z)\;.
\end{equ}
It follows immediately from the Kolmogorov continuity test combined with \eqref{e:equalfeps} 
that the bound \eqref{e:assumf} implies not only that \eqref{e:assumfeps} holds, but also
that the required convergence of $\hat \Pi^{(\eps)}$ holds on every element $\tau$ of the form
\eqref{e:specialTau}. Through \eqref{e:boundgamma}, it also yields the missing bound
on $\Gamma_{z\bar z}$ on all of $\CT$.
From this point on, the proof that
$\$\hat \Pi^{(\eps)}; \hat \Pi\$ \to 0$ in probability as $\eps \to 0$ proceeds in 
exactly the same fashion as the proof of \cite[Thm~10.7]{Regularity}.

Since we have furthermore already shown that \eqref{e:assumfeps} holds, it only remains to show
that \eqref{e:assumPieps} holds as well. This however follows immediately from 
\eqref{e:assumPi}, using the 
equivalence of moments of random variables belonging to a fixed Wiener chaos in 
the same way as in \cite[Thm~10.7]{Regularity}, combined with the fact that
wavelet basis functions do indeed integrate to $0$.
\end{proof}

\subsection{Proof of Theorem~\ref{theo:convModel}}

\begin{proof}
As a consequence of Proposition~\ref{prop:momentBounds}, we only need to show 
that the bounds \eqref{e:requiredBounds} hold.
We start with the proof that \eqref{e:wantedKPZ} holds. Actually, as a consequence of
\cite[Thm~5.14]{Regularity}, we only require these bounds for symbols that are not of the
form $\CI(\tau)$ or $\CI'(\tau)$. Furthermore, it suffices to show
\eqref{e:wantedKPZ} for $z = 0$ by translation invariance, and most
of this section is devoted to this proof. We first consider those basis vectors 
that do not contain the symbol $\Eps$. 

\subsubsection[Term Psi^2]{The case \texorpdfstring{$\tau = \<2>$}{Psi^2}}

The first non-trivial symbol is given by $\tau = \<2>$. In order to represent 
the random variable $\bigl(\hat \Pi_0^{(\eps)}\tau\bigr)(\phi_0^\lambda)$ for some
test function $\phi$, we make use of the following graphical notation, which is 
essentially the same as in \cite{Etienne}.
Elements belonging to the $k$th Wiener chaos are represented by kernels with $k$ space-time 
arguments, via the map $f \mapsto I_k(f)$ described in \cite[Ch.~1.1.2]{Nualart}.
We will sometimes represent such a kernel by a graph which contains $k$ distinguished 
vertices of the type \tikz[baseline=-3] \node [var] {};, 
each of them representing one of the arguments of the kernel. A special vertex 
\tikz[baseline=-3] \node [root] {}; represents the origin $0$. All other vertices
represent integration variables.

Each line then represents a kernel, with 
\tikz[baseline=-0.1cm] \draw[kernel] (0,0) to (1,0);
representing the kernel $K'$, 
\tikz[baseline=-0.1cm] \draw[keps] (0,0) to (1,0);
representing the kernel $K_\eps' = \rho_\eps * K'$, and
\tikz[baseline=-0.1cm] \draw[testfcn] (1,0) to (0,0);
representing a generic test function $\phi_0^\lambda$ rescaled to 
scale $\lambda$.
Whenever we draw a barred arrow 
\tikz[baseline=-0.1cm] \draw[kernel1] (0,0) to (1,0);
this represents a factor $K'(\bar z- z) - K'(-z)$, where
$z$ and $\bar z$ are the coordinates of the starting and end
point respectively.

With these graphical notations at hand, we have the following expression for the
unrenormalised model $\Pi_0^{(\eps)}$:
\begin{equ}
\bigl(\Pi_0^{(\eps)}\<2>\bigr)(\phi_0^\lambda)
= 
\begin{tikzpicture}[scale=0.35,baseline=0.3cm]
	\node at (0,-1)  [root] (root) {};
	\node at (0,1)  [dot] (int) {};
	\node at (-1,3.5)  [var] (left) {};
	\node at (1,3.5)  [var] (right) {};
	
	\draw[testfcn] (int) to  (root);
	
	\draw[keps] (left) to (int);
	\draw[keps] (right) to (int);
\end{tikzpicture}\;
+ \;
\begin{tikzpicture}[scale=0.35,baseline=0.3cm]
	\node at (0,-1)  [root] (root) {};
	\node at (0,1)  [dot] (int) {};
	\node at (0,3.5) [dot] (top) {};
	
	\draw[testfcn] (int) to  (root);
	
	\draw[keps] (top) to[bend left=60] (int); 
	\draw[keps] (top) to[bend right=60] (int); 
\end{tikzpicture}\;.
\end{equ}
Here, via the correspondence explained above, the first term represents the element 
$I_2(f)$ of the second Wiener chaos associated to the kernel
\begin{equ}
f(z_1,z_2) = \int \phi_0^\lambda(z)K_\eps'(z-z_1)K_\eps'(z-z_2)\,dz\;,
\end{equ}
while the second term represents the constant
\begin{equ}
\int\int \phi_0^\lambda(z)\bigl(K_\eps'(z-\bar z)\bigr)^2\,d\bar z\,dz\;.
\end{equ}
All variables $z$, $\bar z$, etc appearing in these expressions denote space-time variables.

At this stage, we realise that for $\eps$ sufficiently small, the second term is 
identical to $C_0^{(\eps)} \int \phi_0^\lambda(z)\,dz = C_0^{(\eps)}\bigl(\Pi_0^{(\eps)}\one\bigr)(\phi_0^\lambda)$. As a consequence, this term cancels out exactly in the definition of $\hat \Pi_0^{(\eps)}\<2>$
and we have 
\begin{equ}[e:exprPi2]
\bigl(\hat \Pi_0^{(\eps)}\<2>\bigr)(\phi_0^\lambda)
= 
\begin{tikzpicture}[scale=0.35,baseline=0.3cm]
	\node at (0,-1)  [root] (root) {};
	\node at (0,1)  [dot] (int) {};
	\node at (-1,3.5)  [var] (left) {};
	\node at (1,3.5)  [var] (right) {};
	
	\draw[testfcn] (int) to  (root);
	
	\draw[keps] (left) to (int);
	\draw[keps] (right) to (int);
\end{tikzpicture}\;.
\end{equ}
We now argue that we can find random variables $\bigl(\Pi_0\<2>\bigr)(\phi_0^\lambda)$
so that the bound \eqref{e:wantedKPZ} does indeed hold. Note first that as a consequence 
of \cite[Ch.~1.1.2]{Nualart} and of the fact that symmetrisation
is a projection in $L^2$, a random variable $X$ belonging to the $k$th homogeneous Wiener chaos and 
represented by a kernel $K_X$ satisfies $\E X^2 \le k! \|K_X\|^2_{L^2}$. As a consequence
of \eqref{e:exprPi2}, one therefore has the bound
\begin{equ}[e:bound2]
\E |\bigl(\hat \Pi_0^{(\eps)}\<2>\bigr)(\phi_0^\lambda)|^2
\le 2\;
\begin{tikzpicture}[scale=0.35,baseline=0.25cm]
	\node at (0,-1)  [root] (root) {};
	\node at (-2,1)  [dot] (intl) {};
	\node at (2,1)  [dot] (intr) {};
	\node at (0,1)  [dot] (top1) {};
	\node at (0,3)  [dot] (top2) {};
	
	\draw[testfcn] (intl) to  (root);
	\draw[testfcn] (intr) to  (root);
	
	\draw[keps] (top1) to (intl);
	\draw[keps] (top1) to (intr);
	\draw[keps] (top2) to (intl);
	\draw[keps] (top2) to (intr);
\end{tikzpicture}\;.
\end{equ}
Furthermore, using the explicit form of the heat kernel, one can verify that the kernel $K_\eps'$
satisfies  
\begin{equ}
\sup_{\eps \in (0,1]} \|K_\eps'\|_{2;p} < \infty\;,
\end{equ}
where $\|\cdot\|_{\alpha;p}$ is given by \eqref{e:boundK} below.
(In particular, it also satisfies the same bound with $2$ replaced by $2+\kappa/2$.)
The right hand side of \eqref{e:bound2} is therefore precisely of the form 
$\CI_\lambda^\CG(K)$ for some
collection of kernels $K$ satisfying the assumptions of 
Section~\ref{sec:bounds} uniformly over $\eps \in (0,1]$
and for the labelled graph $\CG$ given by
\begin{equ}[e:labelledG]
\CG = 
\begin{tikzpicture}[scale=0.4,baseline=0.3cm]
	\node at (0,-1)  [root] (root) {};
	\node at (-3,1)  [dot] (intl) {};
	\node at (3,1)  [dot] (intr) {};
	\node at (0,1)  [dot] (top1) {};
	\node at (0,3)  [dot] (top2) {};
	
	\draw[dist] (intl) to (root);
	\draw[dist] (intr) to (root);
	
	\draw[->,generic] (top1) to  node[labl,pos=0.45] {\tiny 2+,0} (intl);
	\draw[->,generic] (top1) to node[labl,pos=0.45] {\tiny 2+,0} (intr);
	\draw[->,generic] (top2) to node[labl] {\tiny 2+,0} (intl); 
	\draw[->,generic] (top2) to  node[labl] {\tiny 2+,0} (intr); 
\end{tikzpicture}\;.
\end{equ}
(Here, the label $(2+,0)$ on an edge $e$ means that $a_e = 2+\kappa/2$ and $r_e = 0$.)
It is straightforward to verify that the assumptions of Theorem~\ref{theo:ultimate} are 
satisfied, so that one has the bound 
$\E |\bigl(\hat \Pi_0^{(\eps)}\<2>\bigr)(\phi_0^\lambda)|^2 \lesssim \lambda^{\alpha}$
where 
\begin{equ}
\alpha = \#\{{\rm  vertices~not ~adjacent~to~root}\} |s|- \sum_e a_e
= 2\cdot 3 - 8 - 2\kappa =  2|\<2>| +2\kappa\;,
\end{equ}
since the homogeneity of \<2> is $|\<2>| = -1-2\kappa$.
This bound holds uniformly over $\eps \in (0,1]$, so that it is indeed the required
first bound in \eqref{e:wantedKPZ}.

\begin{remark}
From now on, whenever we write $\CI_\lambda^\CG$ without specifying a collection
of kernels $K$, we really mean ``$\CI_\lambda^\CG(K)$ for a collection
of kernels $K$ satisfying the assumptions of 
Section~\ref{sec:bounds} uniformly over $\eps \in (0,1]$''.
\end{remark}

 We still need to 
obtain the second bound in \eqref{e:wantedKPZ}. This however can be obtained in exactly the
same way as soon as we note that, when considering the difference between $\Pi_0$ and 
$\hat \Pi^{(\eps)}_0$, we obtain a sum of expressions of the type \eqref{e:exprPi2},
but in each term some of the instances of $K_\eps'$ are replaced by $K'$
and exactly one instance is replaced by $K_\eps' - K'$. 
We then use the fact that $K'$ satisfies the same bound as $K_\eps'$, while $K_\eps' - K'$
satisfies the improved bound
\begin{equ}
\|K_\eps' - K'\|_{2+\kappa/2;p} \lesssim \eps^{\kappa/2}\;,
\end{equ}
as a consequence of \cite[Lem.~10.17]{Regularity}.
This is the reason for using labels $2+ \kappa/2$ in \eqref{e:labelledG} rather than
$2$, since although $\sup_{\eps \in (0,1]}\|K_\eps'\|_{2+\kappa/2;p} < \infty$,
one has $\|K_\eps' - K'\|_{2;p} \not\to 0$ as $\eps \to 0$.
This is the same for all of the symbols, so we only ever explicitly show how to obtain
the first bound in \eqref{e:wantedKPZ} with the understanding that the second bound then
follows in the same way.

\subsubsection[Term PsiIPsi]{The case \texorpdfstring{$\tau = \<11>$}{PsiIPsi}}

We now turn to $\tau = \<11>$. This time, one has $\hat \Pi_0^{(\eps)}\<11> = \Pi_0^{(\eps)}\<11>$ so that,
similarly to before, we have the identity
\begin{equ}[e:iden1]
\bigl(\hat \Pi_0^{(\eps)}\<11>\bigr)(\phi_0^\lambda)
= 
\begin{tikzpicture}[scale=0.35,baseline=0.3cm]
	\node at (0,-1)  [root] (root) {};
	\node at (0,1)  [dot] (int) {};
	\node at (-1.5,1)  [dot] (centre) {};
	\node at (-1.5,3.5)  [var] (left) {};
	\node at (0,3.5)  [var] (right) {};
	
	\draw[testfcn] (int) to  (root);
	
	\draw[keps] (left) to (centre);
	\draw[kernel1] (centre) to (int);
	\draw[keps] (right) to (int);
\end{tikzpicture}\;
+ \;
\begin{tikzpicture}[scale=0.35,baseline=0.3cm]
	\node at (0,-1)  [root] (root) {};
	\node at (0,1)  [dot] (int) {};
	\node at (-1,2.25) [dot] (cent) {};
	\node at (0,3.5) [dot] (top) {};
	
	\draw[testfcn] (int) to  (root);
	
	\draw[keps] (top) to[bend left=60] (int); 
	\draw[keps] (top) to (cent); 
	\draw[kernel1] (cent) to (int); 
\end{tikzpicture}\;
= \;
\begin{tikzpicture}[scale=0.35,baseline=0.3cm]
	\node at (0,-1)  [root] (root) {};
	\node at (0,1)  [dot] (int) {};
	\node at (-1.5,1)  [dot] (centre) {};
	\node at (-1.5,3.5)  [var] (left) {};
	\node at (0,3.5)  [var] (right) {};
	
	\draw[testfcn] (int) to  (root);
	
	\draw[keps] (left) to (centre);
	\draw[kernel1] (centre) to (int);
	\draw[keps] (right) to (int);
\end{tikzpicture}\;
- \;
\begin{tikzpicture}[scale=0.35,baseline=0.3cm]
	\node at (0.75,-1)  [root] (root) {};
	\node at (0,1)  [dot] (int) {};
	\node at (1.5,1) [dot] (cent) {};
	\node at (0.75,3) [dot] (top) {};
	
	\draw[testfcn] (int) to  (root);
	
	\draw[keps] (top) to (int); 
	\draw[keps] (top) to (cent); 
	\draw[kernel] (cent) to (root); 
\end{tikzpicture}\;.
\end{equ}
In order to see this, recall that the barred arrow represents a difference $K'(\bar z - z) - K'(-z)$, so that one
has the identity
\begin{equ}
\begin{tikzpicture}[scale=0.35,baseline=0.3cm]
	\node at (0,-1)  [root] (root) {};
	\node at (0,1)  [dot] (int) {};
	\node at (-1,2.25) [dot] (cent) {};
	\node at (0,3.5) [dot] (top) {};
	
	\draw[testfcn] (int) to  (root);
	
	\draw[keps] (top) to[bend left=60] (int); 
	\draw[keps] (top) to (cent); 
	\draw[kernel1] (cent) to (int); 
\end{tikzpicture}
\;=\;
\begin{tikzpicture}[scale=0.35,baseline=0.3cm]
	\node at (0,-1)  [root] (root) {};
	\node at (0,1)  [dot] (int) {};
	\node at (-1,2.25) [dot] (cent) {};
	\node at (0,3.5) [dot] (top) {};
	
	\draw[testfcn] (int) to  (root);
	
	\draw[keps] (top) to[bend left=60] (int); 
	\draw[keps] (top) to (cent); 
	\draw[kernel] (cent) to (int); 
\end{tikzpicture}
\;- \;
\begin{tikzpicture}[scale=0.35,baseline=0.3cm]
	\node at (0.75,-1)  [root] (root) {};
	\node at (0,1)  [dot] (int) {};
	\node at (1.5,1) [dot] (cent) {};
	\node at (0.75,3) [dot] (top) {};
	
	\draw[testfcn] (int) to  (root);
	
	\draw[keps] (top) to (int); 
	\draw[keps] (top) to (cent); 
	\draw[kernel] (cent) to (root); 
\end{tikzpicture}\;.
\end{equ}
The first term appearing on the right hand side of this expression vanishes
because the kernel $(K'* K_\eps')\cdot K'$ is odd under the substitution $(t,x) \mapsto (t,-x)$ (recall that we assumed that the mollifier $\rho$ 
is even under that substitution), 
so that it integrates to $0$, thus yielding \eqref{e:iden1}.

Since random variables belonging to Wiener chaoses of different order are orthogonal,
we obtain as before the bound
\begin{equ}
\E \bigl|\bigl(\hat \Pi_0^{(\eps)}\<11>\bigr)(\phi_0^\lambda)\bigr|^2
\le 2\;
\begin{tikzpicture}[scale=0.35,baseline=0.25cm]
	\node at (0,-1)  [root] (root) {};
	\node at (-2,1)  [dot] (intl) {};
	\node at (2,1)  [dot] (intr) {};
	\node at (0,1)  [dot] (top1) {};
	\node at (0,3)  [dot] (top2) {};
	\node at (-2,3)  [dot] (topl) {};
	\node at (2,3)  [dot] (topr) {};
	
	\draw[testfcn] (intl) to  (root);
	\draw[testfcn] (intr) to  (root);
	
	\draw[keps] (top1) to (intl);
	\draw[keps] (top1) to (intr);
	\draw[keps] (top2) to (topl);
	\draw[keps] (top2) to (topr);
	\draw[kernel1] (topl) to (intl);
	\draw[kernel1] (topr) to (intr);
\end{tikzpicture}
+ \left(
\begin{tikzpicture}[scale=0.35,baseline=0.3cm]
	\node at (0.75,-1)  [root] (root) {};
	\node at (0,1)  [dot] (int) {};
	\node at (1.5,1) [dot] (cent) {};
	\node at (0.75,3) [dot] (top) {};
	
	\draw[testfcn] (int) to  (root);
	
	\draw[keps] (top) to (int); 
	\draw[keps] (top) to (cent); 
	\draw[kernel] (cent) to (root); 
\end{tikzpicture}
\right)^2\;.
\end{equ}
Both terms separately can be bounded in the same way as before.  This time however
the first term is given by $\CI_\lambda^\CG$ for the graph 
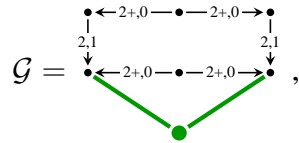
\begin{equ}
\CG = 
\begin{tikzpicture}[scale=0.4,baseline=0.3cm]
	\node at (0,-1)  [root] (root) {};
	\node at (-3,1)  [dot] (intl) {};
	\node at (3,1)  [dot] (intr) {};
	\node at (0,1)  [dot] (top1) {};
	\node at (0,3)  [dot] (top2) {};
	\node at (-3,3)  [dot] (topl) {};
	\node at (3,3)  [dot] (topr) {};
	
	\draw[dist] (intl) to (root);
	\draw[dist] (intr) to (root);
	
	\draw[->,generic] (top1) to  node[labl,pos=0.45] {\tiny 2+,0} (intl);
	\draw[->,generic] (top1) to node[labl,pos=0.45] {\tiny 2+,0} (intr);
	\draw[->,generic] (top2) to node[labl] {\tiny 2+,0} (topl); 
	\draw[->,generic] (top2) to  node[labl] {\tiny 2+,0} (topr); 
	\draw[->,generic] (topl) to node[labl] {\tiny 2,1} (intl); 
	\draw[->,generic] (topr) to  node[labl] {\tiny 2,1} (intr); 
\end{tikzpicture}\;,
\end{equ}
i.e.\ the two vertical edges have $r_e = 1$. Again, it is straightforward to verify
that Assumption~\ref{ass:mainGraph} is verified, so that the required bounds
follow.

\subsubsection[Term PsiIPsi]{The case \texorpdfstring{$\tau = \<21>$}{PsiIPsi}}
\label{sec:psiIpsi}

We now turn to $\tau = \<21>$ which is slightly trickier. One can verify from its recursive definition that
the structure group acts trivially on \<21>, so that 
one has similarly to before the identity
\begin{equ}[e:expr21]
\bigl(\hat \Pi_0^{(\eps)}\<21>\bigr)(\phi_0^\lambda)
=
\begin{tikzpicture}[scale=0.35,baseline=1.1cm]
	\node at (0.8,1.5)  [root] (root) {};
	\node at (-1.3,3)  [dot] (left1) {};
	\node at (-1.3,5)  [dot] (left2) {};
	\node at (0.8,3) [var] (variable2) {};
	\node at (0.8,4.3) [var] (variable3) {};
	\node at (0.8,5.7) [var] (variable4) {};
	
	\draw[testfcn] (left1) to (root);

	\draw[kernel] (left2) to  (left1);
	\draw[keps] (variable2) to  (left1); 
	\draw[keps] (variable3) to  (left2); 
	\draw[keps] (variable4) to (left2);
\end{tikzpicture}
+ 2\;
\begin{tikzpicture}[scale=0.35,baseline=1.1cm]
	\node at (0.8,1.5)  [root] (root) {};
	\node at (-1.3,3)  [dot] (left1) {};
	\node at (-1.3,5)  [dot] (left2) {};
	\node at (0.8,4) [dot] (variable2) {};
	\node at (0.8,5.7) [var] (variable3) {};
	
	\draw[testfcn] (left1) to (root);

	\draw[kernel] (left2) to  (left1);
	\draw[keps] (variable2) to  (left1); 
	\draw[keps] (variable3) to  (left2); 
	\draw[keps] (variable2) to (left2);
\end{tikzpicture}\;,
\end{equ}
where the second term comes from the product formula
\cite{Nualart}. This time, it turns out that when trying to
``na\"\i vely'' apply Theorem~\ref{theo:ultimate}, its conditions fail to be
satisfied for the second term. Denote however by $Q_\eps$ the kernel
\begin{equ}
Q_\eps(z) = \begin{tikzpicture}[scale=0.35,baseline=1.3cm]
	\node at (-1.3,3)  [var] (left1) {};
	\node at (-1.3,5)  [root] (left2) {};
	\node at (0.8,4) [dot] (variable2) {};

	\draw[kernel] (left2) to  (left1);
	\draw[keps] (variable2) to  (left1); 
	\draw[keps] (variable2) to (left2);
\end{tikzpicture}
= K(z) \int K_\eps(z-\bar z)K_\eps(-\bar z)\,d\bar z\;.
\end{equ}
It then follows by symmetry as above that $\int Q_\eps(z) \,dz = 0$. As a consequence,
for any $\eps > 0$, the distribution $\Ren Q_\eps(z)$ given by \eqref{e:defRen} with $I_{e,k} = 0$
is exactly \textit{the same}
as simple integration against $Q_\eps$, without any renormalisation. 
Furthermore, it follows easily from \cite[Sec.~10]{Regularity} that there is a limiting kernel $Q$
such that $\sup_{\eps \in (0,1]}\|Q_\eps\|_{3,p} < \infty$
and $\|Q_\eps - Q\|_{3+\kappa,p} \lesssim \eps^\kappa$.
Writing \tikz[baseline=-0.1cm] \draw[kernelBig] (0,0) -- (1,0); as a graphical notation for the
kernel $Q_\eps = \Ren Q_\eps$, we can rewrite \eqref{e:expr21} as
\begin{equ}
\bigl(\hat \Pi_0^{(\eps)}\<21>\bigr)(\phi_0^\lambda)
=
\begin{tikzpicture}[scale=0.35,baseline=1.1cm]
	\node at (0.8,1.5)  [root] (root) {};
	\node at (-1.3,3)  [dot] (left1) {};
	\node at (-1.3,5)  [dot] (left2) {};
	\node at (0.8,3) [var] (variable2) {};
	\node at (0.8,4.3) [var] (variable3) {};
	\node at (0.8,5.7) [var] (variable4) {};
	
	\draw[testfcn] (left1) to (root);

	\draw[kernel] (left2) to  (left1);
	\draw[keps] (variable2) to  (left1); 
	\draw[keps] (variable3) to  (left2); 
	\draw[keps] (variable4) to (left2);
\end{tikzpicture}
+ 2\;
\begin{tikzpicture}[scale=0.35,baseline=1.1cm]
	\node at (0.8,1.5)  [root] (root) {};
	\node at (-1.3,3)  [dot] (left1) {};
	\node at (-1.3,5)  [dot] (left2) {};
	\node at (0.8,5.7) [var] (variable3) {};
	
	\draw[testfcn] (left1) to (root);

	\draw[kernelBig] (left2) to  (left1);
	\draw[keps] (variable3) to  (left2); 
\end{tikzpicture}\;
\end{equ}
and  bound
$
\E \bigl|\bigl(\hat \Pi_0^{(\eps)}\<21>\bigr)(\phi_0^\lambda)\bigr|^2$ by a constant multiple of 
$
|\CI_\lambda^\CG| + |\CI_\lambda^{\bar \CG}|$ for graphs $\CG$ and $\bar \CG$ given by
\begin{equ}
\CG = 
\begin{tikzpicture}[scale=0.35,baseline=1.2cm]
	\node at (0,1.5)  [root] (root) {};
	\node at (-3,3)  [dot] (left1) {};
	\node at (-3,6)  [dot] (left2) {};
	\node at (3,3)  [dot] (right1) {};
	\node at (3,6)  [dot] (right2) {};
	\node at (0,4) [dot] (variable2) {};
	\node at (0,5) [dot] (variable3) {};
	\node at (0,7) [dot] (variable4) {};
	
	\draw[dist] (left1) to (root);
	\draw[dist] (right1) to (root);
	
	\draw[->,generic] (left2) to node[labl,pos=0.45] {\tiny 2,0} (left1);
	\draw[->,generic] (variable2) to node[labl] {\tiny 2+,0} (left1); 
	\draw[->,generic] (variable3) to  node[labl] {\tiny 2+,0} (left2); 
	\draw[->,generic] (variable4) to  node[labl] {\tiny 2+,0} (left2);
	
	\draw[->,generic] (right2) to node[labl,pos=0.45] {\tiny 2,0} (right1);
	\draw[->,generic] (variable2) to node[labl] {\tiny 2+,0} (right1); 
	\draw[->,generic] (variable3) to  node[labl] {\tiny 2+,0} (right2); 
	\draw[->,generic] (variable4) to  node[labl] {\tiny 2+,0} (right2);
\end{tikzpicture}
\;,\qquad \bar \CG = 
\begin{tikzpicture}[scale=0.35,baseline=1.2cm]
	\node at (0,1.5)  [root] (root) {};
	\node at (-3,3)  [dot] (left1) {};
	\node at (-3,6)  [dot] (left2) {};
	\node at (3,3)  [dot] (right1) {};
	\node at (3,6)  [dot] (right2) {};
	\node at (0,7) [dot] (variable3) {};
	
	\draw[dist] (left1) to (root);
	\draw[dist] (right1) to (root);
	
	\draw[->,generic] (left2) to node[labl,pos=0.45] {\tiny 3+,-1} (left1);
	\draw[->,generic] (variable3) to  node[labl] {\tiny 2+,0} (left2); 
	
	\draw[->,generic] (right2) to node[labl,pos=0.45] {\tiny 3+,-1} (right1);
	\draw[->,generic] (variable3) to  node[labl] {\tiny 2+,0} (right2); 
\end{tikzpicture}
\;.
\end{equ}
Again, Assumption~\ref{ass:mainGraph} can easily be checked for both of these graphs
so that, in view of the above comments, Theorem~\ref{theo:ultimate} applies and yields the desired bounds.

\subsubsection[Term IPsi22]{The case \texorpdfstring{$\tau = \<22>$}{IPsi22}}

Again, the structure group acts trivially on $\<22>$ and one has the identity
$\DeltaW M_0 \<22> = \bigl(\<22> - C_3^{(\eps)} \one\bigr) \otimes \one$. As a consequence,
we obtain the identity
\begin{equ}
\hat \Pi_0^{(\eps)}\<22> = \bigl(K' *  \Pi_0^{(\eps)}\<2>\bigr)^2 - 2C_3^{(\eps)}\;.
\end{equ}
When testing against the test function $\phi_0^\lambda$, it follows from the product formula 
and the definition of $C_3^{(\eps)}$ that
the Wiener chaos decomposition of this expression is given by
\begin{equ}
\bigl(\hat \Pi_0^{(\eps)}\<22>\bigr)(\phi_0^\lambda)
=
\begin{tikzpicture}[scale=0.35,baseline=.9cm]
	\node at (0,0)  [root] (root) {};
	\node at (0,2)  [dot] (int) {};
	\node at (-2,3)  [dot] (left) {};
	\node at (-1,5) [var] (variable3) {};
	\node at (-3,5) [var] (variable4) {};

	\node at (2,3)  [dot] (right) {};
	\node at (1,5) [var] (variable1) {};
	\node at (3,5) [var] (variable2) {};
	
	\draw[testfcn] (int) to (root);

	\draw[kernel] (left) to  (int);
	\draw[kernel] (right) to  (int);
	\draw[keps] (variable2) to  (right); 
	\draw[keps] (variable1) to  (right); 
	
	\draw[keps] (variable3) to  (left); 
	\draw[keps] (variable4) to (left);
\end{tikzpicture}
+ 4\;
\begin{tikzpicture}[scale=0.35,baseline=.9cm]
	\node at (0,0)  [root] (root) {};
	\node at (0,2)  [dot] (int) {};
	\node at (-2,3)  [dot] (left) {};
	\node at (-3,5) [var] (variable4) {};

	\node at (2,3)  [dot] (right) {};
	\node at (0,4) [dot] (top) {};
	\node at (3,5) [var] (variable2) {};
	
	\draw[testfcn] (int) to (root);

	\draw[kernel] (left) to  (int);
	\draw[kernel] (right) to  (int);
	\draw[keps] (variable2) to  (right); 
	\draw[keps] (top) to  (right); 
	
	\draw[keps] (top) to  (left); 
	\draw[keps] (variable4) to (left);
\end{tikzpicture}\;.
\end{equ}
Note that the term appearing in the Wiener chaos of order $0$ is cancelled out 
exactly by the renormalisation constant $C_3^{(\eps)}$, which is why it does not appear here.
Similarly to before, it is now straightforward to reduce ourselves to the 
situation of Theorem~\ref{theo:ultimate} and to verify that Assumption~\ref{ass:mainGraph}
holds for the two resulting labelled graphs.

\subsubsection[Term IPsi22]{The case \texorpdfstring{$\tau = \<211>$}{IPsi22}}

This time the structure group acts nontrivially on \<211> and it follows from 
\eqref{e:admissible4} combined with the definition of the renormalisation map $M^{(\eps)}$
that
\begin{equ}
\hat \Pi_0^{(\eps)}\<211> = \bigl(\bigl(K' *  \hat \Pi_0^{(\eps)}\<21>\bigr)(\cdot) - \bigl(K' *  \hat \Pi_0^{(\eps)}\<21>\bigr)(0)\bigr)(K'*\xi^{(\eps)}) - 2C_2^{(\eps)}\;.
\end{equ}
As a consequence, one has the identity
\begin{equ} 
\bigl(\hat \Pi_0^{(\eps)}\<211>\bigr)(\phi_0^\lambda)
=
\begin{tikzpicture}[scale=0.35,baseline=0.8cm]
	\node at (0,-0.8)  [root] (root) {};
	\node at (-2,1)  [dot] (left) {};
	\node at (-2,3)  [dot] (left1) {};
	\node at (-2,5)  [dot] (left2) {};
	\node at (0,1) [var] (variable1) {};
	\node at (0,3) [var] (variable2) {};
	\node at (0,4.3) [var] (variable3) {};
	\node at (0,5.7) [var] (variable4) {};
	
	\draw[testfcn] (left) to (root);

	\draw[kernel1] (left1) to   (left);
	\draw[kernel] (left2) to  (left1);
	\draw[keps] (variable2) to  (left1); 
	\draw[keps] (variable1) to   (left); 
	\draw[keps] (variable3) to   (left2); 
	\draw[keps] (variable4) to   (left2);
\end{tikzpicture}
\;+\;
\begin{tikzpicture}[scale=0.35,baseline=0.8cm]
	\node at (0,-0.8)  [root] (root) {};
	\node at (-2,1)  [dot] (left) {};
	\node at (-2,3)  [dot] (left1) {};
	\node at (-2,5)  [dot] (left2) {};
	\node at (0,2) [dot] (variable1) {};
	\node at (0,4.3) [var] (variable3) {};
	\node at (0,5.7) [var] (variable4) {};
	
	\draw[testfcn] (left) to (root);

	\draw[kernel1] (left1) to   (left);
	\draw[kernel] (left2) to  (left1);
	\draw[keps] (variable1) to  (left1); 
	\draw[keps] (variable1) to   (left); 
	\draw[keps] (variable3) to   (left2); 
	\draw[keps] (variable4) to   (left2);
\end{tikzpicture}
\;+2\;
\begin{tikzpicture}[scale=0.35,baseline=0.8cm]
	\node at (0,-0.8)  [root] (root) {};
	\node at (-2,1)  [dot] (left) {};
	\node at (-2,3)  [dot] (left1) {};
	\node at (-2,5)  [dot] (left2) {};
	\node at (0,1) [var] (variable1) {};
	\node at (0,4) [dot] (variable3) {};
	\node at (0,5.7) [var] (variable4) {};
	
	\draw[testfcn] (left) to (root);

	\draw[kernel1] (left1) to   (left);
	\draw[kernel] (left2) to  (left1);
	\draw[keps] (variable3) to  (left1); 
	\draw[keps] (variable1) to   (left); 
	\draw[keps] (variable3) to   (left2); 
	\draw[keps] (variable4) to   (left2);
\end{tikzpicture}
\;+2\;
\begin{tikzpicture}[scale=0.35,baseline=0.8cm]
	\node at (0,-0.8)  [root] (root) {};
	\node at (-2,1)  [dot] (left) {};
	\node at (0,3)  [dot] (left1) {};
	\node at (-2,5)  [dot] (left2) {};
	\node at (0,5) [var] (variable1) {};
	\node at (-2,3) [dot] (variable3) {};
	\node at (-0.5,6) [var] (variable4) {};
	
	\draw[testfcn] (left) to (root);

	\draw[kernel1] (left1) to   (left);
	\draw[kernel] (left2) to  (left1);
	\draw[keps] (variable3) to  (left); 
	\draw[keps] (variable1) to   (left1); 
	\draw[keps] (variable3) to   (left2); 
	\draw[keps] (variable4) to   (left2);
\end{tikzpicture}
\;+2\;
\begin{tikzpicture}[scale=0.35,baseline=0.8cm]
	\node at (0,-0.8)  [root] (root) {};
	\node at (-2,1)  [dot] (left) {};
	\node at (0,3)  [dot] (left1) {};
	\node at (-2,5)  [dot] (left2) {};
	\node at (-2,3) [dot] (variable3) {};
	\node at (0,5.7) [dot] (variable4) {};
	
	\draw[testfcn] (left) to (root);

	\draw[kernel1] (left1) to   (left);
	\draw[kernel] (left2) to  (left1);
	\draw[keps] (variable3) to  (left); 
	\draw[keps] (variable4) to   (left1); 
	\draw[keps] (variable3) to   (left2); 
	\draw[keps] (variable4) to   (left2);
\end{tikzpicture}
- 2C_2^{(\eps)}\;
\begin{tikzpicture}[scale=0.35,baseline=-0.1cm]
	\node at (0,-0.8)  [root] (root) {};
	\node at (0,1)  [dot] (left) {};
	
	\draw[testfcn] (left) to (root);
\end{tikzpicture}\;.
\end{equ}
At this stage we not again that the last two terms cancel each other out, except
for the fact that one of the arrows in the penultimate term is ``barred''. 
Using agin the notation \tikz[baseline=-0.1cm] \draw[kernelBig] (0,0) -- (1,0);
for the kernel $Q_\eps$, we can therefore rewrite this as
\begin{equ} 
\bigl(\hat \Pi_0^{(\eps)}\<211>\bigr)(\phi_0^\lambda)
=
\begin{tikzpicture}[scale=0.35,baseline=0.8cm]
	\node at (0,-0.8)  [root] (root) {};
	\node at (-2,1)  [dot] (left) {};
	\node at (-2,3)  [dot] (left1) {};
	\node at (-2,5)  [dot] (left2) {};
	\node at (0,1) [var] (variable1) {};
	\node at (0,3) [var] (variable2) {};
	\node at (0,4.3) [var] (variable3) {};
	\node at (0,5.7) [var] (variable4) {};
	
	\draw[testfcn] (left) to (root);

	\draw[kernel1] (left1) to   (left);
	\draw[kernel] (left2) to  (left1);
	\draw[keps] (variable2) to  (left1); 
	\draw[keps] (variable1) to   (left); 
	\draw[keps] (variable3) to   (left2); 
	\draw[keps] (variable4) to   (left2);
\end{tikzpicture}
\;+\;
\begin{tikzpicture}[scale=0.35,baseline=0.8cm]
	\node at (0,-0.8)  [root] (root) {};
	\node at (-2,1)  [dot] (left) {};
	\node at (-2,3)  [dot] (left1) {};
	\node at (-2,5)  [dot] (left2) {};
	\node at (0,4.3) [var] (variable3) {};
	\node at (0,5.7) [var] (variable4) {};
	
	\draw[testfcn] (left) to (root);

	\draw[kernelBig] (left1) to   (left);
	\draw[kernel] (left2) to  (left1);
	\draw[keps] (variable3) to   (left2); 
	\draw[keps] (variable4) to   (left2);
\end{tikzpicture}
\;-\;
\begin{tikzpicture}[scale=0.35,baseline=0.8cm]
	\node at (0,-0.8)  [root] (root) {};
	\node at (-2,1)  [dot] (left) {};
	\node at (0,3)  [dot] (left1) {};
	\node at (0,5)  [dot] (left2) {};
	\node at (-2,3) [dot] (variable1) {};
	\node at (-2,4.3) [var] (variable3) {};
	\node at (-2,5.7) [var] (variable4) {};
	
	\draw[testfcn] (left) to (root);

	\draw[kernel] (left1) to   (root);
	\draw[kernel] (left2) to  (left1);
	\draw[keps] (variable1) to  (left1); 
	\draw[keps] (variable1) to   (left); 
	\draw[keps] (variable3) to   (left2); 
	\draw[keps] (variable4) to   (left2);
\end{tikzpicture}
\;+2\;
\begin{tikzpicture}[scale=0.35,baseline=0.8cm]
	\node at (0,-0.8)  [root] (root) {};
	\node at (-2,1)  [dot] (left) {};
	\node at (-2,3)  [dot] (left1) {};
	\node at (-2,5)  [dot] (left2) {};
	\node at (0,1) [var] (variable1) {};
	\node at (0,5.7) [var] (variable4) {};
	
	\draw[testfcn] (left) to (root);

	\draw[kernel1] (left1) to   (left);
	\draw[kernelBig] (left2) to  (left1);
	\draw[keps] (variable1) to   (left); 
	\draw[keps] (variable4) to   (left2);
\end{tikzpicture}
\;+2\;
\begin{tikzpicture}[scale=0.35,baseline=0.8cm]
	\node at (0,-0.8)  [root] (root) {};
	\node at (-2,1)  [dot] (left) {};
	\node at (0,3)  [dot] (left1) {};
	\node at (-2,5)  [dot] (left2) {};
	\node at (0,5) [var] (variable1) {};
	\node at (-2,3) [dot] (variable3) {};
	\node at (-0.5,6) [var] (variable4) {};
	
	\draw[testfcn] (left) to (root);

	\draw[kernel1] (left1) to   (left);
	\draw[kernel] (left2) to  (left1);
	\draw[keps] (variable3) to  (left); 
	\draw[keps] (variable1) to   (left1); 
	\draw[keps] (variable3) to   (left2); 
	\draw[keps] (variable4) to   (left2);
\end{tikzpicture}
\;-2\;
\begin{tikzpicture}[scale=0.35,baseline=0.8cm]
	\node at (0,-0.8)  [root] (root) {};
	\node at (-2,1)  [dot] (left) {};
	\node at (0,3)  [dot] (left1) {};
	\node at (-2,5)  [dot] (left2) {};
	\node at (-2,3) [dot] (variable3) {};
	\node at (0,5.7) [dot] (variable4) {};
	
	\draw[testfcn] (left) to (root);

	\draw[kernel] (left1) to   (root);
	\draw[kernel] (left2) to  (left1);
	\draw[keps] (variable3) to  (left); 
	\draw[keps] (variable4) to   (left1); 
	\draw[keps] (variable3) to   (left2); 
	\draw[keps] (variable4) to   (left2);
\end{tikzpicture}\;.
\end{equ}
At this stage, we can once again reduce ourselves to the situation of Theorem~\ref{theo:ultimate} just
as above.

\subsubsection{Symbols containing \texorpdfstring{$\Eps$}{Eps}}
\label{sec:Eps}

We now turn to the proof of \eqref{e:wantedKPZ} for those symbols $\tau$ with $|\tau| < 0$
which contain at least one occurrence of the symbol $\Eps$. 

We first consider symbols of the type $\tau = \Eps^k(\Psi^{2k+2})$. Note that for $k = 0$, one
has $\tau = \<2>$, which has already been treated, so we assume that $k \ge 1$. 
Thanks to \eqref{e:Hermite}, the choice of renormalisation constant in the definition of
$M^{(\eps)}$, and the definition of the Wick product, one has the identity
\begin{equ}
\bigl(\hat \Pi_0^{(\eps)} \tau\bigr)(\phi_0^\lambda) = \bigl(\hat \Pi_0^{(\eps)} \Eps^k(\Psi^{2k+2})\bigr)(\phi_0^\lambda) = \eps^k \bigl((\hat \Pi_0^{(\eps)} \Psi)^{\diamond (2k+2)}\bigr)(\phi_0^\lambda)\;,
\end{equ}
which can also be written as
\begin{equ}[e:exprPi2eps]
\bigl(\hat \Pi_0^{(\eps)}\tau\bigr)(\phi_0^\lambda)
= \eps^k
\begin{tikzpicture}[scale=0.35,baseline=0.3cm]
	\node at (0,-1)  [root] (root) {};
	\node at (0,1)  [dot] (int) {};
	\node at (-2.5,3.5)  [var] (left) {};
	\node at (2.5,3.5)  [var] (right) {};
	\node at (0,3.5)  {$\cdots$};
	
	\draw[testfcn] (int) to  (root);
	
	\draw[keps] (left) to (int);
	\draw[keps] (right) to (int);
\draw [decorate,decoration={brace,amplitude=7pt}]
(-2.6,4) -- (2.6,4) node [midway,yshift=0.5cm] 
{\scriptsize $(2k+2)$ times};
\end{tikzpicture}\;.
\end{equ}
At this stage, we introduce the shorthand \tikz[baseline=-0.1cm] \draw[dots] (0,0) -- (1,0); for the kernel $N_\eps \eqdef \eps (K'*\rho_\eps)* (K'*\rho_\eps)(-\cdot)$, namely
\begin{equ}
\begin{tikzpicture}[scale=0.35,baseline=-0.1cm]
	\node at (-2,0)  [dot] (left) {};
	\node at (2,0)  [dot] (right) {};
	\draw[dots] (left) to (right);
\end{tikzpicture} = 
\eps\;
\begin{tikzpicture}[scale=0.35,baseline=-0.1cm]
	\node at (-3,0)  [dot] (left) {};
	\node at (0,0)  [dot] (cent) {};
	\node at (3,0)  [dot] (right) {};
	\draw[keps] (cent) to (left);
	\draw[keps] (cent) to (right);
\end{tikzpicture}
\;.
\end{equ}
(Since this kernel is symmetric, its orientation is irrelevant so we do not
draw any arrow on it.)
With this notation, we then obtain in a similar way to before the bound
\begin{equ}[e:boundSimpleMultiple]
\E \bigl|\bigl(\hat \Pi_0^{(\eps)}\tau\bigr)(\phi_0^\lambda)\bigr|^2
\le
\begin{tikzpicture}[scale=0.35,baseline=0.25cm]
	\node at (0,-1)  [root] (root) {};
	\node at (-3,1)  [dot] (intl) {};
	\node at (3,1)  [dot] (intr) {};
	\node at (0,3)  [dot] (top1) {};
	\node at (0,4.5)  [dot] (top2) {};
	
	\draw[testfcn,bend right=20] (intl) to  (root);
	\draw[testfcn,bend left=20] (intr) to  (root);
	
	\draw[keps,bend right=20] (top1) to (intl);
	\draw[keps,bend left=20] (top1) to (intr);
	\draw[keps,bend right=40] (top2) to (intl);
	\draw[keps,bend left=40] (top2) to (intr);
	\draw[dots] (intl) to node[labl] {\scriptsize $(2k)$} (intr) ;
\end{tikzpicture}\;,
\end{equ}
where we wrote
\tikz[baseline=-0.1cm] \draw[dots] (0,0) to node[labl]{\scriptsize $(2k)$} (1.8,0); as a shorthand
for $N_\eps^{2k}$ on the right.
We also note that, as a consequence of \cite[Lem.~10.17]{Regularity} and the scaling properties of 
$K'$, one has the bound 
\begin{equ}[e:boundNeps]
\|N_\eps\|_{\delta;p} \lesssim \eps^\delta\;,
\end{equ}\label{def_of_Neps}
for every $\delta \in (0,1]$ and every $p > 0$.
As a consequence, we are again in the setting of Theorem~\ref{theo:ultimate}, 
with a graph $\tilde \CG$ that is exactly the same as the graph $\CG$ in 
\eqref{e:labelledG}, except for an additional
edge with $a_e$ arbitrarily small connecting the left and right vertices. 
Since Assumption~\ref{ass:mainGraph} is an open condition\, any graph $\tilde \CG$ obtained
from another graph $\CG$ by the addition of some new edges with $a_e = \delta$ 
or the increase of the homogeneities of some edges by $\delta$ satisfies 
Assumption~\ref{ass:mainGraph} for $\delta$ sufficiently small, provided that 
the original graph $\CG$ satisfies it.
Combining this with \eqref{e:boundNeps}, it follows that one has the bound
\begin{equ}
\E \bigl|\bigl(\hat \Pi_0^{(\eps)}\tau\bigr)(\phi_0^\lambda)\bigr|^2
\lesssim \eps^\delta \lambda^{|\tau| + \delta}\;,
\end{equ} 
for some sufficiently small choice of $\delta$.
In particular, the bounds \eqref{e:wantedKPZ} are satisfied with $\hat \Pi_0 \tau = 0$.

Something similar happens for all other symbols containing at least one instance of $\Eps$.
Indeed, consider next $\tau = \Eps^{k}\bigl(\Psi^{2k+1}\CI'(\Eps^{\ell}\Psi^{2\ell+2})\bigr)$ with
$k, \ell \ge 0$,
which is the ``$\Eps$-decorated'' version of $\tau = \<21>$. As a consequence of 
\eqref{e:Hermite}, Proposition~\ref{prop:renormProd}, and the fact that
$\hat \Pi^{(\eps)}$ is again an admissible model by construction (see \cite[Sec.~8]{Regularity}),
we conclude that one has the identity
\begin{equ}[e:exprWick]
\bigl(\hat \Pi_0^{(\eps)}\tau\bigr)(z) = 
\eps^{k+\ell} (K' * \xi^{(\eps)})(z)^{\diamond (2k+1)} \bigl(K' * (K' * \xi^{(\eps)})^{\diamond (2\ell+2)}\bigr)(z)\;,
\end{equ}
where we use the symbol $\diamond$ to denote the Wick product (or rather Wick power in this
case), see \cite{Nualart}.
Similarly to above, the kernel $Q_\eps^{(m)} = Q_\eps N_\eps^m$ \label{defofqeps} is odd for 
every $m \ge 0$, so that it can again be identified with $\Ren Q_\eps^{(m)}$.
Furthermore, it is of order $3+\delta$ for any $\delta > 0$ and 
$\|Q_\eps^{(m)}\|_{3+\delta,p} \lesssim \eps^\delta$ for $\delta \in (0,1)$
provided that $m > 0$.
Combining this with \eqref{e:exprWick} we conclude that, for every 
sufficiently small exponent
$\delta, \bar \delta > 0$, one has again a bound of the type
\begin{equ}
\E \bigl|\bigl(\hat \Pi_0^{(\eps)}\tau\bigr)(\phi_0^\lambda)\bigr|^2
\le
|\CI_\lambda^\CG| + |\CI_\lambda^{\bar \CG}|\;,
\end{equ}
but this time the two labelled graphs $\CG$ and $\bar \CG$ are given by
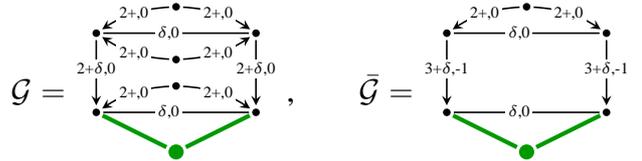
\begin{equ}
\CG = 
\begin{tikzpicture}[scale=0.35,baseline=1.2cm]
	\node at (0,1.5)  [root] (root) {};
	\node at (-3,3)  [dot] (left1) {};
	\node at (-3,6)  [dot] (left2) {};
	\node at (3,3)  [dot] (right1) {};
	\node at (3,6)  [dot] (right2) {};
	\node at (0,4) [dot] (variable2) {};
	\node at (0,5) [dot] (variable3) {};
	\node at (0,7) [dot] (variable4) {};
	
	\draw[dist] (left1) to (root);
	\draw[dist] (right1) to (root);
	
	\draw[generic] (left2) to node[labl,pos=0.45] {\tiny $\delta$,0} (right2);
	\draw[generic] (left1) to node[labl,pos=0.45] {\tiny $\delta$,0} (right1);

	\draw[->,generic] (left2) to node[labl,pos=0.45] {\tiny 2+$\delta$,0} (left1);
	\draw[->,generic,bend right=15] (variable2) to node[labl] {\tiny 2+,0} (left1); 
	\draw[->,generic,bend left=15] (variable3) to  node[labl] {\tiny 2+,0} (left2); 
	\draw[->,generic,bend right=15] (variable4) to  node[labl] {\tiny 2+,0} (left2);
	
	\draw[->,generic] (right2) to node[labl,pos=0.45] {\tiny 2+$\delta$,0} (right1);
	\draw[->,generic,bend left=15] (variable2) to node[labl] {\tiny 2+,0} (right1); 
	\draw[->,generic,bend right=15] (variable3) to  node[labl] {\tiny 2+,0} (right2); 
	\draw[->,generic,bend left=15] (variable4) to  node[labl] {\tiny 2+,0} (right2);
\end{tikzpicture}
\;,\qquad \bar \CG =
\begin{tikzpicture}[scale=0.35,baseline=1.2cm]
	\node at (0,1.5)  [root] (root) {};
	\node at (-3,3)  [dot] (left1) {};
	\node at (-3,6)  [dot] (left2) {};
	\node at (3,3)  [dot] (right1) {};
	\node at (3,6)  [dot] (right2) {};
	\node at (0,7) [dot] (variable3) {};
	
	\draw[dist] (left1) to (root);
	\draw[dist] (right1) to (root);
	
	\draw[generic] (left2) to node[labl,pos=0.45] {\tiny $\delta$,0} (right2);
	\draw[generic] (left1) to node[labl,pos=0.45] {\tiny $\delta$,0} (right1);
	\draw[->,generic] (left2) to node[labl,pos=0.45] {\tiny 3+$\delta$,-1} (left1);
	\draw[->,generic,bend right=15] (variable3) to  node[labl] {\tiny 2+,0} (left2); 
	
	\draw[->,generic] (right2) to node[labl,pos=0.45] {\tiny 3+$\delta$,-1} (right1);
	\draw[->,generic,bend left=15] (variable3) to  node[labl] {\tiny 2+,0} (right2); 
\end{tikzpicture}\;.
\end{equ}
Furthermore, again as a consequence of the bound \eqref{e:boundNeps} and the
corresponding bound for $Q_\eps^{(m)}$, it follows that
as soon as $k + \ell > 0$, at least one of the factors $\|K_e\|_{a_e;p}$ appearing
in Theorem~\ref{theo:ultimate} is bounded by $\eps^\delta$,
thus yielding the required bound.

\subsubsection{Additional bounds on \texorpdfstring{$\hat \Pi^{(\eps)}$}{Pieps}}

We now turn to the proof of the bound \eref{e:assumPi}. This bound is of course 
non-trivial only for symbols $\tau$ with $|\tau| < \bar \gamma$. The bound for $\tau = \<1>$
is very easy to obtain so we do not dwell on it. Regarding $\tau = \<20>$, we can write
it as in \eqref{e:exprPi2} as
\begin{equ}
\bigl(\hat \Pi^{(\eps)}_z\<20>\bigr)(\eta_z^\lambda) =
\begin{tikzpicture}[scale=0.35,baseline=-0.1cm]
	\node at (-1,0)  [root] (root) {};
	\node at (3,0)  [dot] (int) {};
	\node at (1,0)  [dot] (bot) {};
	\node at (4,-1.5)  [var] (left) {};
	\node at (4,1.5)  [var] (right) {};
	
	\draw[kernel] (int) to (bot);
	\draw[testfcn] (bot) to  (root);
	
	\draw[keps] (left) to (int);
	\draw[keps] (right) to (int);
\end{tikzpicture}\;.
\end{equ}
Since the test function $\eta$ integrates to
$0$, this is \textit{equal} to
\begin{equ}
\bigl(\hat \Pi^{(\eps)}_z\<20>\bigr)(\eta_z^\lambda) =
\begin{tikzpicture}[scale=0.35,baseline=-0.1cm]
	\node at (-1,0)  [root] (root) {};
	\node at (3,0)  [dot] (int) {};
	\node at (1,0)  [dot] (bot) {};
	\node at (4,-1.5)  [var] (left) {};
	\node at (4,1.5)  [var] (right) {};
	
	\draw[kernel1] (int) to (bot);
	\draw[testfcn] (bot) to  (root);
	
	\draw[keps] (left) to (int);
	\draw[keps] (right) to (int);
\end{tikzpicture}\;,
\end{equ}
so that we have the bound
\begin{equ}
\E\bigl|\bigl(\hat \Pi^{(\eps)}_z\<20>\bigr)(\eta_z^\lambda)\bigr|^2 
= 2 \;
\begin{tikzpicture}[scale=0.35,baseline=0.5cm]
	\node at (0,0)  [root] (root) {};
	\node at (-2,1)  [dot] (intl) {};
	\node at (2,1)  [dot] (intr) {};
	\node at (-2,3)  [dot] (tl) {};
	\node at (2,3)  [dot] (tr) {};
	\node at (0,2)  [dot] (top1) {};
	\node at (0,4)  [dot] (top2) {};
	
	\draw[testfcn] (intl) to  (root);
	\draw[testfcn] (intr) to  (root);
	
	\draw[kernel1] (tl) to (intl);
	\draw[kernel1] (tr) to (intr);

	\draw[keps] (top1) to (tl);
	\draw[keps] (top1) to (tr);
	\draw[keps] (top2) to (tl);
	\draw[keps] (top2) to (tr);
\end{tikzpicture}\;.
\end{equ}
At this point, we note that, as a consequence of \cite[Lem.~10.7]{Regularity}, we have the bound\begin{equ}
\$K_\eps'\$_{\alpha;p} \lesssim \eps^{\alpha-2}\;,
\end{equ}
for every $\alpha \in [1,2]$. For such values of $\alpha$, we can therefore write
\begin{equ}
\E\bigl|\bigl(\hat \Pi^{(\eps)}_z\<20>\bigr)(\eta_z^\lambda)\bigr|^2 
\lesssim \eps^{4(\alpha-2)} |\CI_\lambda^\CG|\;,\qquad
\CG = 
\begin{tikzpicture}[scale=0.5,baseline=0.8cm]
	\node at (0,0)  [root] (root) {};
	\node at (-2,1)  [dot] (intl) {};
	\node at (2,1)  [dot] (intr) {};
	\node at (-2,3)  [dot] (tl) {};
	\node at (2,3)  [dot] (tr) {};
	\node at (0,2)  [dot] (top1) {};
	\node at (0,4)  [dot] (top2) {};
	
	\draw[dist] (intl) to  (root);
	\draw[dist] (intr) to  (root);
	
	\draw[->,generic] (tl) to node[labl,pos=0.45] {\tiny 2,1} (intl);
	\draw[->,generic] (tr) to node[labl,pos=0.45] {\tiny 2,1} (intr);

	\draw[->,generic] (top1) to  node[labl] {\tiny $\alpha$,0} (tl);
	\draw[->,generic] (top1) to node[labl] {\tiny $\alpha$,0} (tr);
	\draw[->,generic] (top2) to node[labl] {\tiny $\alpha$,0} (tl);
	\draw[->,generic] (top2) to node[labl] {\tiny $\alpha$,0} (tr);
\end{tikzpicture}\;.
\end{equ}
One can now verify that as long as $\alpha > {3\over 2}$, the conditions
of Theorem~\ref{theo:ultimate} are satisfied, so that one has the bound
\begin{equ}
\E\bigl|\bigl(\hat \Pi^{(\eps)}_z\<20>\bigr)(\eta_z^\lambda)\bigr|^2 
\lesssim \eps^{4\alpha-8} \lambda^{8-4\alpha}\;.
\end{equ}
In particular, since $\bar \gamma < 1$, we can choose $\alpha$ such that $8-4\alpha = \bar \gamma + \kappa$,
so that the required bound \eqref{e:assumPi} follows for $\tau = \<20>$.

We now turn to $\tau = \<210>$.
Following the exact same procedure, combined with the steps from 
Section~\ref{sec:psiIpsi}, we see that in this case one has
\begin{equ}
\E\bigl|\bigl(\hat \Pi^{(\eps)}_z\<210>\bigr)(\eta_z^\lambda)\bigr|^2 
\lesssim \eps^{2(\alpha-2)} \bigl(|\CI_\lambda^\CG| + |\CI_\lambda^{\bar \CG}|\bigr)\;,
\end{equ}
where the graphs $\CG$ and $\bar \CG$ are given by
\begin{equ}
\CG = 
\begin{tikzpicture}[scale=0.35,baseline=1cm]
	\node at (0,-1)  [root] (root) {};
	\node at (-3,0.5)  [dot] (leftb) {};
	\node at (3,0.5)  [dot] (rightb) {};
	\node at (-3,3)  [dot] (left1) {};
	\node at (-3,6)  [dot] (left2) {};
	\node at (3,3)  [dot] (right1) {};
	\node at (3,6)  [dot] (right2) {};
	\node at (0,4) [dot] (variable2) {};
	\node at (0,5) [dot] (variable3) {};
	\node at (0,7) [dot] (variable4) {};
	
	\draw[dist] (leftb) to (root);
	\draw[dist] (rightb) to (root);
	\draw[->,generic] (left1) to node[labl,pos=0.45] {\tiny 2,1} (leftb);
	\draw[->,generic] (right1) to node[labl,pos=0.45] {\tiny 2,1} (rightb);
	
	\draw[->,generic] (left2) to node[labl,pos=0.45] {\tiny 2,0} (left1);
	\draw[->,generic] (variable2) to node[labl] {\tiny $\alpha$,0} (left1); 
	\draw[->,generic] (variable3) to  node[labl] {\tiny 2,0} (left2); 
	\draw[->,generic] (variable4) to  node[labl] {\tiny 2,0} (left2);
	
	\draw[->,generic] (right2) to node[labl,pos=0.45] {\tiny 2,0} (right1);
	\draw[->,generic] (variable2) to node[labl] {\tiny $\alpha$,0} (right1); 
	\draw[->,generic] (variable3) to  node[labl] {\tiny 2,0} (right2); 
	\draw[->,generic] (variable4) to  node[labl] {\tiny 2,0} (right2);
\end{tikzpicture}
\;,\qquad \bar \CG = 
\begin{tikzpicture}[scale=0.35,baseline=1cm]
	\node at (0,-1)  [root] (root) {};
	\node at (-3,0.5)  [dot] (leftb) {};
	\node at (3,0.5)  [dot] (rightb) {};
	\node at (-3,3)  [dot] (left1) {};
	\node at (-3,6)  [dot] (left2) {};
	\node at (3,3)  [dot] (right1) {};
	\node at (3,6)  [dot] (right2) {};
	\node at (0,7) [dot] (variable3) {};
	
	\draw[dist] (leftb) to (root);
	\draw[dist] (rightb) to (root);
	\draw[->,generic] (left1) to node[labl,pos=0.45] {\tiny 2,1} (leftb);
	\draw[->,generic] (right1) to node[labl,pos=0.45] {\tiny 2,1} (rightb);
	
	\draw[->,generic] (left2) to node[labl,pos=0.45] {\tiny $\alpha$+1,0} (left1);
	\draw[->,generic] (variable3) to  node[labl] {\tiny 2,0} (left2); 
	
	\draw[->,generic] (right2) to node[labl,pos=0.45] {\tiny $\alpha$+1,0} (right1);
	\draw[->,generic] (variable3) to  node[labl] {\tiny 2,0} (right2); 
\end{tikzpicture}
\;.
\end{equ}
Again, one can verify that the assumptions of Theorem~\ref{theo:ultimate}
hold provided that we choose $\alpha > {3\over 2}$ so that we then obtain the bound
\begin{equ}
\E\bigl|\bigl(\hat \Pi^{(\eps)}_z\<210>\bigr)(\eta_z^\lambda)\bigr|^2 
\lesssim \eps^{2\alpha-4} \lambda^{5-2\alpha}\;.
\end{equ}
Again, the required bound follows since $\bar \gamma < 1$.
The case $\tau = \<10>$ follows in a very similar way.
All other symbols in $\CU'$ of homogeneity below $1$
are just ``decorated'' versions of \<1>, \<20>, \<210>, 
or \<10> and can therefore be treated in exactly the same way as in Section~\ref{sec:Eps}.

\subsubsection{Bounds on \texorpdfstring{$\hat f^{(\eps)}$}{feps}}

It now only remains to show that the bounds \eqref{e:assumf} also hold. 
For this, we recall from \eqref{e:specialTau} that the only symbols 
$\tau$ such that $|\tau| < 0$ and 
$|\Eps^{j-1}(\tau)| > 0$ for some $j > 1$ are all of the form 
\begin{equ}[e:mytau]
\tau = \Psi^{2j-n} \CI'(\Eps^{k_1-1}\Psi^{2k_i})\cdots \CI'(\Eps^{k_\ell-1}\Psi^{2k_\ell})\;,
\end{equ}
with $n > 2$, $j \in \{\lceil n/2\rceil,\ldots,m\}$, and $k_i \in \{1,\ldots,m\}$.
In order to bound $\hat f^{(\eps)}_z(\EE^{j-1}_0(\tau))$, note first that, 
setting $\Psi^{(\eps)}(z) = (K' * \xi^{(\eps)})(z)$,
it is a straightforward calculation to show that one has the bounds
\begin{equ}[e:boundPsieps]
\E |D^k\Psi^{(\eps)}(z)|^2 \lesssim \eps^{-1-2|k|}\;,\qquad
|\E D^k\Psi^{(\eps)}(0)D^k\Psi^{(\eps)}(z)| \lesssim (|z| + \eps)^{-1-2|k|}\;,
\end{equ}
for every multiindex $k$. Let now $\{k_1,\ldots,k_m\}$ be a finite collection of 
such multiindices and set 
\begin{equ}
\Psi^{(\eps)}_{k_1,\ldots,k_m}(z) = D^{k_1}\Psi^{(\eps)}(z)\diamond\ldots\diamond D^{k_m}\Psi^{(\eps)}(z)\;.
\end{equ}
Combining this with \eqref{e:boundPsieps} and Lemma~\ref{lem:twoscale} below,
it is not difficult to see that
\begin{equ}
\E |K'* \Psi_{k_1,\ldots,k_m}^{(\eps)}(z)|^2 \lesssim \eps^{2-m-2\sum_{i=1}^m|k_i|}\;.
\end{equ}
In particular, setting $\Phi_\ell^{(\eps)}(z) = \bigl(K'*(\Psi^{(\eps)})^{\diamond \ell}\bigr)(z)$,
one has the bound
\begin{equ}[e:boundDerPhi]
\E |D^k\Phi_{\ell}^{(\eps)}(z)|^2 \lesssim \eps^{2-\ell-2|k|}\;.
\end{equ}
We now note that, for $\tau$ as in \eqref{e:mytau}, one has
\begin{equ}
\hat f^{(\eps)}_z(\EE^{j-1}_0(\tau)) = \eps^{j-1+|k|-\ell} \Psi^{(\eps)}(z)^{\diamond (2j-n)}\Phi_{2k_1}^{(\eps)}(z)\cdots \Phi_{2k_\ell}^{(\eps)}(z)\;.
\end{equ}
Combining \eqref{e:boundDerPhi} and \eqref{e:boundPsieps} with the generalised Leibniz
rule and the equivalence of moments for random variables
belonging to a Wiener chaos of finite order, we conclude that
\begin{equ}
\E |D^m\hat f^{(\eps)}_z(\EE^{j-1}_0(\tau))| \lesssim \eps^{{n\over 2}-1-|m|}\;.
\end{equ}
The bound \eqref{e:assumf} now follows immediately.
\end{proof}

\subsection{Behaviour of the renormalisation constants}
\label{sec:logs}

The goal of this section is to provide precise asymptotic results on the 
behaviour of the renormalisation constants $C_\tau^{(\eps)}$ for $\tau \in \Bad$ 
appearing in the construction of our model. We have the following convergence result.

\begin{theorem}\label{theo:logs}
Let $C_2^{(\eps)}$ and $C_3^{(\eps)}$ be as in \eref{e:defRenormConst} and let $C_\tau^{(\eps)}$ be as in \eqref{e:defCtaueps}. Then, there exists
a constant $c \in \R$ depending both on the choice of $K$ and of the mollifier $\rho$ such that
\begin{equ}[e:logsCancel]
\lim_{\eps \to 0} \bigl(C_3^{(\eps)} + 4C_2^{(\eps)}\bigr) = c\;.
\end{equ}
Furthermore, for every $\tau \in \Bad \setminus\{\<22>,\<211>\}$, there exists a constant $c_\tau \in \R$ such that $\lim_{\eps \to 0} C_\tau^{(\eps)} = c_\tau$, and these constants are
independent of the choice of kernel $K$.
\end{theorem}

\begin{remark}
The statement \eqref{e:logsCancel} is non-trivial since in general both 
of these constants diverge logarithmically as
$\eps \to 0$, see \cite{KPZ}. 
Note furthermore that although it is very similar, this theorem does 
\textit{not} follow immediately from \cite[Lem.~6.5]{KPZ} because 
here we consider space-time regularisations of the noise.
\end{remark}

For the remainder of this section, it turns out to be more convenient to work
with the rescaled kernel 
\begin{equ}
K_{\eps,\rho}(z) \eqdef \bigl(\rho * \CS_\eps^{(1)} K\bigr)(z)\;,
\end{equ}
where the scaling operator $\CS_\eps^{(\alpha)}$ is defined by
\begin{equ}
\bigl(\CS_\eps^{(\alpha)} K\bigr)(t,x) = \eps^{\alpha} K(\eps^2 t, \eps x)\;.
\end{equ}
This is because in the rescaled variables, our kernels will turn out to converge to
non-trivial limits, which is something that would not be easily seen in the original
variables.
Similarly to before, $K_{\eps,\rho}'$ denotes the spatial derivative
of $K_{\eps,\rho}$. A simple change of variables then shows that
\eqref{e:defRenormConst} is still valid if we interpret 
\tikz[baseline=-0.1cm] \draw[keps] (0,0) -- (1,0); as an instance of the rescaled
kernel $K_{\eps,\rho}'$ instead of the kernel $\rho_{\eps} * K'$ and 
\tikz[baseline=-0.1cm] \draw[kernel] (0,0) -- (1,0); as an instance of
$(\CS_\eps^{(1)} K)' = \CS_\eps^{(2)} (K')$ instead of $K'$.
We make use of these interpretations for the remainder of this section.

Before we turn to the proof of Theorem~\ref{theo:logs}, we provide
a number of useful technical results. In order to state our first result,
we introduce the family of norms
\begin{equ}
\|F\|_{\alpha,\beta} = \sup_{|z| \le 1} |z|^\alpha |F(z)| + \sup_{|z| \ge 1} |z|^\beta |F(z)|\;,
\end{equ}
and we denote by $\CB_{\alpha,\beta}$ the Banach space consisting of the functions
$F \colon \R^{d+1} \to \R$ such that $\|F\|_{\alpha,\beta} < \infty$.
Here, for $z = (t,x)$, we denoted by $|z| = |x| + \sqrt{|t|}$ its parabolic norm.

\begin{remark}\label{rem:Keps}
It is straightforward to show that $K_{\eps,\rho}$ and $K_{\eps,\rho}'$
belong to $\CB_{0,1}$ and $\CB_{0,2}$ respectively and that, for every $\kappa > 0$,
they converge to limits in  $\CB_{0,1-\kappa}$ and $\CB_{0,2-\kappa}$ respectively. 
These limits are given by $P_\rho$ and $P_\rho'$ respectively, where $P_\rho = P*\rho$.
\end{remark}

Our first preparatory result shows how convolution
acts in these spaces.

\begin{lemma}\label{lem:twoscale}
Suppose that for $j=1,2$, $F_j$ are functions on $\R^{d+1}$ with parabolic scaling such that $F_i \in \CB_{\alpha_i,\beta_i}$
with $\alpha_i < d+2$, $i=1,2$ and $\beta_1 + \beta_2 > d+2$. Then there exists $C > 0$ such that
\begin{equ}[e:boundconvol]
\|F_1\ast F_2\|_{\alpha,\beta} \le C \|F_1\|_{\alpha_1,\beta_1}\|F_2\|_{\alpha_2,\beta_2}\;,
\end{equ}
with $\alpha = 0 \vee (\alpha_1+\alpha_2 -d-2)$ and $\beta = (\beta_1+\beta_2 -d-2)\wedge 
\beta_1\wedge \beta_2$.
\end{lemma}

\begin{proof}
The condition $\alpha_i<d+2$, $i=1,2$ is required or the integral defining 
$F_1*F_2$ diverges at small scales.  Similarly, we 
need $\beta_1+\beta_2 >d+2$ for the integral to  converge at large scales. 
By bilinearity, we can (and will from now on) assume that 
$\|F_j\|_{\alpha_j,\beta_j} = 1$ for $j \in \{1,2\}$.

Let first $|z| \le 1$ and write
\begin{equ}[e:convol]
\bigl(F_1*F_2\bigr)(z) = \int_{\R^{d+1}} F_1(y)\,F_2(z-y)\,dy\;.
\end{equ}
We now break the domain of integration into four regions $\{A_i\}_{i=1}^4$ 
and we bound it separately in each of them. We set
\begin{equs}
A_1 &= \{y\,:\, |y| \le 2|z|\;\&\; |y| \le |z-y|\}\;, \\
A_2 &= \{y\,:\, |y| \le 2|z|\;\&\; |y| > |z-y|\}\;, \\
A_3 &= \{y\,:\, |y| \in (2|z|,2)\}\;, \\
A_4 &= \{y\,:\, |y| > 2\}\;.
\end{equs}
For $y \in A_1$, since $|z| \le |y| + |z-y|$, we have $|z-y| \ge |z|/2$,
so that $|F_1(y)F_2(z-y)| \le |z|^{-\alpha_2} |y|^{-\alpha_1}$. Integrating
this bound over $\{|y| \le 2|z|\}$ yields a bound proportional to $|z|^{d+2-\alpha_1-\alpha_2}$. Exchanging the roles of $y$ and $z-y$, we obtain the same
bound for the integral over $A_2$. 
For $y \in A_3$, have  $|z-y| \ge |y| - |z| \ge |y|/2$ and $|z-y| \le 3$, so that 
$|F_1(y)F_2(z-y)| \lesssim |y|^{-\alpha_1-\alpha_2}$. Integrating this bound over 
$A_3$ yields this time a bound proportional to $1 + |z|^{d+2-\alpha_1-\alpha_2}$.
Finally, on $A_4$, we also have $|z-y| \ge |y|/2$, but we additionally have
$|y| \ge 2$, so that this time $|F_1(y)F_2(z-y)| \lesssim |y|^{-\beta_1-\beta_2}$.
Since $\beta_1 + \beta_2 > d+2$ by assumption, this is integrable over $|y| \ge 2$,
so that we obtain a bound proportional to $1$, thus completing the required
bound on $\bigl|\bigl(F_1*F_2\bigr)(z)\bigr|$ for $|z| \le 1$.

For $|z| \ge 1$, we break the domain of integration for \eqref{e:convol} into
five regions $\{B_i\}_{i=1}^5$, namely
\begin{equs}
B_1 &= \{y\,:\, |y| \le 1/2\}\;, \\
B_2 &= \{y\,:\, |z-y| \le 1/2\}\;, \\
B_3 &= \{y\,:\, |y| \le 2|z|\;\&\; |y| \le |z-y|\} \setminus B_1\;, \\
B_4 &= \{y\,:\, |y| \le 2|z|\;\&\; |y| > |z-y|\} \setminus B_2\;, \\
B_5 &= \{y\,:\, |y| > 2|z|)\}\;. 
\end{equs} 
On $B_1$, we have $|z-y| \ge |z| - |y| \ge |z|/2$ so that, since furthermore
$|z| \ge 1$, one has 
\begin{equ}[e:boundB1]
|F_1(y)F_2(z-y)| \lesssim |z|^{-\beta_2}|y|^{-\alpha_1}\;.
\end{equ}
Integrating this over $B_1$ yields a bound of the order $|z|^{-\beta_2}$
since we assumed that $\alpha_1 < d+2$. In the case of $B_2$,
we similarly obtain a bound of the order $|z|^{-\beta_1}$.
On $B_3$, we have instead
$|F_1(y)F_2(z-y)| \lesssim |z|^{-\beta_2}|y|^{-\beta_1}$,
which we integrate
over $|y| \in (1/2,2|z|]$, so that we obtain a bound of the order
of $|z|^{-\beta_2} \bigl(1 + |z|^{d+2-\beta_1}\bigr)$.
In the same way, the integral over $B_4$ yields a bound of the
order of $|z|^{-\beta_1} \bigl(1 + |z|^{d+2-\beta_2}\bigr)$.
Finally, for $y \in B_5$, we have $|z-y| \ge |y| - |z| \ge |y|/2$
so that $|F_1(y)F_2(z-y)| \lesssim |y|^{-\beta_1-\beta_2}$, thus yielding
a bound of the order $|z|^{d+2-\beta_1-\beta_2}$. Collecting all of these
bounds completes the proof.
\end{proof}

We also need a slightly stronger conclusion in a special case. In order to 
formulate this, we introduce the family of norms
\begin{equ}
\|F\|_{\alpha,\beta;1} = \sup_{|z| \le 1} |z|^\alpha |F(z)| + \sup_{|z| \ge 1} |z|^\beta \bigl(|F(z)|+ |z| \,|\nabla_x F(z)|+ |z|^2 \,|\d_t F(z)|\bigr)\;,
\end{equ}
and we denote by $\CB_{\alpha,\beta;1}$ the Banach space consisting of the functions
$F \colon \R^{d+1} \to \R$ such that $\|F\|_{\alpha,\beta;1} < \infty$.

\begin{lemma}\label{lem:gainOne}
Let $F_j$ as in Lemma~\ref{lem:twoscale}, but with $\beta_1 > d+2 > \beta_2 > 0$, $\alpha_i < d+2$,
and such that additionally $\int F_1(z)\,dz = 0$ and $\|F_2\|_{\alpha_1,\beta_1;1} < \infty$. 
Then, one has the stronger conclusion 
$\beta = (\beta_1+\beta_2-d-2) \wedge (\beta_2+1)$.
\end{lemma}

\begin{proof}
We only need to consider $|z| \ge 2$ say and, as before we want to estimate the
integral
\begin{equ}[e:convol2]
\bigl(F_1*F_2\bigr)(z) = \int_{\R^{d+1}} F_1(y)\,\bigl(F_2(z-y) - F_2(z)\bigr)\,dy\;.
\end{equ}
The reason why this identity holds is of course that we assumed that $F_1$ 
integrates to $0$. This time, we break the integral into three regions. 

First, we consider the case $|y| \le |z|/2$. In this case, as a simple consequence of our bounds
on the derivatives of $F_2$, one has
\begin{equ}
\bigl|F_2(z-y) - F_2(z)\bigr| \lesssim |y|\, |z|^{-\beta_2-1}\;.
\end{equ}
On the other hand, one has 
\begin{equ}
\int_{|y| \le |z|} |y|\,|F_1(y)|\,dy \lesssim |z|^{0 \vee (d+3-\beta_1)}\;.
\end{equ}
Combining the two yields a bound of the required form. For the integral over the
region $|y| \ge 2|z|$, we use the ``brutal'' bound
\begin{equ}
\bigl|F_2(z-y) - F_2(z)\bigr| \lesssim |z|^{-\beta_2}\;,
\end{equ}
so that this integral is bounded by $|z|^{-\beta_2} \int_{|y| \ge |z|} |F_1(y)|\,dy$.
Since we assumed that $\beta_1 > d+2$, this integral converges and is of order
$|z|^{d+2-\beta_1}$ thus yielding the required bound. 
Finally, in the region $|z|/2 \le |y| \le 2|z|$, we bound $|F_1(y)|$ by
$|z|^{-\beta_1}$. Since $\beta_2 < d+2$, the integral of 
$\bigl|F_2(z-y) - F_2(z)\bigr|$ over that region
can be bounded by $|z|^{d+2-\beta_2}$, thus again yielding the correct bound. 
\end{proof}

\begin{remark}
Lemmas~\ref{lem:twoscale} and \ref{lem:gainOne} immediately extend to the case of arbitrary scalings
by replacing each instance of $d+2$ by the scaling dimension of the underlying space. 
\end{remark}

Before we turn to the proof of Theorem~\ref{theo:logs}, we define a kernel
$P_\eps$ by
\begin{equ}[e:defPeps]
P_\eps(z) = 
\begin{tikzpicture}[scale=0.35,baseline=1.3cm]
	\node at (-1.3,3)  [var] (left1) {};
	\node at (-1.3,5)  [root] (left2) {};
	\node at (0.8,4) [dot] (variable2) {};

	\draw[keps] (variable2) to  (left1); 
	\draw[keps] (variable2) to (left2);
\end{tikzpicture}
= \int K_{\eps,\rho}'(z-\bar z)K_{\eps,\rho}'(-\bar z)\,d\bar z\;.
\end{equ}
We then have the following result.

\begin{lemma}\label{lem:Keps}
With $P_\eps$ as above, define kernels $R_\eps$, $\tilde R_\eps$ through the identities
\begin{equ}
2P_\eps(z) = K_{\eps,\rho}(z) + K_{\eps,\rho}(-z) + R_\eps^{(1)}(z) + \bigl(\CS_\eps^{(1)} R_\eps^{(2)}\bigr)(z)\;,\quad 
\CS_\eps^{(2)} K' = K_{\eps,\rho}' + \tilde R_\eps\;.
\end{equ}
Then, $R_\eps^{(1)}$, $R_\eps^{(2)}$ and $\tilde R_\eps$ satisfy the bounds
\begin{equ}
\|R_\eps^{(1)}\|_{0,2} + \|R_\eps^{(2)}\|_{0,4}  + \|\tilde R_\eps\|_{2,3} \le C \;,
\end{equ}
for some $C$ independent of $\eps \in (0,1]$. Furthermore, for every $\kappa > 0$,
these kernels converge in $\CB_{0,2-\kappa}$, $\CB_{0,4}$ and $\CB_{2,3-\kappa}$ respectively as $\eps \to 0$. In the case of $R_\eps^{(1)}$, the limit is $0$ and in the case of 
$\tilde R_\eps$ it is independent of the choice of $K$.
\end{lemma}

\begin{proof}
The claim for $\tilde R_\eps$ is straightforward to show. 
Regarding $R_\eps$, an explicit calculation shows that 
if we denote by $P$ the heat kernel, one has the identity
\begin{equ}
2\int P'(z-\bar z)P'(-\bar z)\,d\bar z = P(z) + P(-z)\;.
\end{equ}
Since $K$ is compactly supported and agrees with $P$ in some neighbourhood
of the origin, this immediately implies that there exists a smooth compactly supported 
function $R$ such that
\begin{equ}
2\int K'(z-\bar z)K'(-\bar z)\,d\bar z = K(z) + K(-z) + R(z)\;.
\end{equ}
Convolving with $\rho^{(2)}$ and then rescaling, we conclude that
\begin{equ}
2P_\eps(z) = \bigl(\rho^{(2)} * \CS_\eps^{(1)}K\bigr)(z)
+ \bigl(\rho^{(2)} * \CS_\eps^{(1)}K\bigr)(-z) + \CS_\eps^{(1)} \bigl(\rho_\eps^{(2)}*R\bigr)(z)\;,
\end{equ}
so that we can set
\begin{equ}
R_\eps^{(2)} = \rho_\eps^{(2)}*R\;,\quad
R_\eps^{(1)}(z) = \bigl((\rho^{(2)} - \rho) * \CS_\eps^{(1)}K\bigr)(z)
+ \bigl((\rho^{(2)} - \rho) * \CS_\eps^{(1)}K\bigr)(-z)\;.
\end{equ}
The required bounds then follow easily.
\end{proof}

\begin{lemma}\label{lem:defRenormConstTweak}
Let $\tilde C_2^{(\eps)}$ and $\tilde C_3^{(\eps)}$ be defined by
the identities
\begin{equ}[e:defRenormConstTweak]
C_2^{(\eps)} = \;
\begin{tikzpicture}[baseline=0.6cm,scale=0.35]
\node at (0,0) [root] (0) {};
\node at (0,2) [dot] (1) {};
\node at (-2,2) [dot] (2) {};
\node at (0,4) [dot] (4) {};
\node at (2,3) [dot] (5) {};

\draw[keps] (1) to (0);
\draw[keps] (2) to (0);
\draw[keps] (5) to (1);
\draw[keps] (4) to (1);
\draw[keps] (5) to (4);
\draw[keps] (2) to (4);
\end{tikzpicture} + \tilde C_2^{(\eps)}\;,\qquad
C_3^{(\eps)} = \;
\begin{tikzpicture}[baseline=0.6cm,scale=0.35]
\node at (0,0) [root] (0) {};
\node at (-2,2) [dot] (1) {};
\node at (2,2) [dot] (2) {};
\node at (0,2) [dot] (3) {};
\node at (0,4) [dot] (4) {};

\draw[keps] (1) to (0);
\draw[keps] (2) to (0);
\draw[keps] (3) to (1);
\draw[keps] (3) to (2);
\draw[keps] (4) to (1);
\draw[keps] (4) to (2);
\end{tikzpicture} + \tilde C_3^{(\eps)}\;.
\end{equ}
Then both $\tilde C_2^{(\eps)}$ and $\tilde C_3^{(\eps)}$ converge to finite limits 
as $\eps \to 0$, and these limits do not depend on the choice of the cutoff kernel $K$.
\end{lemma}

\begin{proof}
Comparing \eqref{e:defRenormConstTweak} to \eqref{e:defRenormConst} and writing
\tikz[baseline=-0.1cm] \draw[Deps] (0,0) -- (1,0); for the kernel 
$D_\eps \eqdef \CS^{(2)}_\eps K' - K_{\eps,\rho}'$, 
we have
\begin{equ}
\tilde C_2^{(\eps)} = 
\begin{tikzpicture}[baseline=0.6cm,scale=0.35]
\node at (0,0) [root] (0) {};
\node at (0,2) [dot] (1) {};
\node at (-2,2) [dot] (2) {};
\node at (0,4) [dot] (4) {};
\node at (2,3) [dot] (5) {};

\draw[keps] (1) to (0);
\draw[keps] (2) to (0);
\draw[keps] (5) to (1);
\draw[Deps] (4) to (1);
\draw[keps] (5) to (4);
\draw[keps] (2) to (4);
\end{tikzpicture}
\;+\;
\begin{tikzpicture}[baseline=0.6cm,scale=0.35]
\node at (0,0) [root] (0) {};
\node at (0,2) [dot] (1) {};
\node at (-2,2) [dot] (2) {};
\node at (0,4) [dot] (4) {};
\node at (2,3) [dot] (5) {};

\draw[Deps] (1) to (0);
\draw[keps] (2) to (0);
\draw[keps] (5) to (1);
\draw[kernel] (4) to (1);
\draw[keps] (5) to (4);
\draw[keps] (2) to (4);
\end{tikzpicture}\;.
\end{equ}
At this point, we note that $K_{\eps,\rho}'$, $\CS^{(2)}_\eps K'$, and 
$D_\eps$ converge in $\CB_{0,2-\kappa}$, $\CB_{0,2-\kappa}$, and  $\CB_{0,3-\kappa}$ 
respectively, and that these limits do not depend on the choice of cutoff $K$. 
The claim for $\tilde C_2^{(\eps)}$ now follows by repeatedly 
applying Lemma~\ref{lem:twoscale}. The constant $\tilde C_3^{(\eps)}$ can be dealt with 
in a very similar fashion. 
\end{proof}

We now have finally all the ingredients required for the proof of Theorem~\ref{theo:logs}.

\begin{proof}[of Theorem~\ref{theo:logs}]
We first prove that \eqref{e:logsCancel} holds. 
Since we also need the kernel $K_{\eps,\rho}$ in this proof, we use for it the graphical
notation \tikz[baseline=-0.1cm] \draw[kbase] (0,0) to (1,0); As a consequence of Lemmas~\ref{lem:defRenormConstTweak}, \ref{lem:Keps}, and \ref{lem:twoscale},
we have the identities
\begin{equs}[e:exprCeps]
4C_2^{(\eps)} &= \;
\begin{tikzpicture}[baseline=0.6cm,scale=0.35]
\node at (0,0) [root] (0) {};
\node at (0,2) [dot] (1) {};
\node at (0,4) [dot] (4) {};

\draw[keps] (1) to (0);
\draw[kbase] (4) to[bend left=60] (0);
\draw[kbase] (4) to[bend right=60] (1);
\draw[keps] (4) to (1);
\end{tikzpicture}
+
\begin{tikzpicture}[baseline=0.6cm,scale=0.35]
\node at (0,0) [root] (0) {};
\node at (0,2) [dot] (1) {};
\node at (0,4) [dot] (4) {};

\draw[keps] (1) to (0);
\draw[kbase] (0) to[bend left=60] (4);
\draw[kbase] (1) to[bend right=60] (4);
\draw[keps] (4) to (1);
\end{tikzpicture}
+
\begin{tikzpicture}[baseline=0.6cm,scale=0.35]
\node at (0,0) [root] (0) {};
\node at (0,2) [dot] (1) {};
\node at (0,4) [dot] (4) {};

\draw[keps] (1) to (0);
\draw[kbase] (4) to[bend left=60] (0);
\draw[kbase] (1) to[bend left=60] (4);
\draw[keps] (4) to (1);
\end{tikzpicture}
+
\begin{tikzpicture}[baseline=0.6cm,scale=0.35]
\node at (0,0) [root] (0) {};
\node at (0,2) [dot] (1) {};
\node at (0,4) [dot] (4) {};

\draw[keps] (1) to (0);
\draw[kbase] (0) to[bend left=60] (4);
\draw[kbase] (4) to[bend left=60] (1);
\draw[keps] (4) to (1);
\end{tikzpicture}
 + (\ldots)\;,\\
4C_3^{(\eps)} &= 2\;
\begin{tikzpicture}[baseline=0.4cm,scale=0.35]
\node at (0,0) [root] (0) {};
\node at (-1.5,2) [dot] (1) {};
\node at (1.5,2) [dot] (2) {};

\draw[keps] (1) to (0);
\draw[keps] (2) to (0);
\draw[kbase] (1) to (2);
\draw[kbase] (1) to[bend left=60] (2);
\end{tikzpicture}
+ 2\;
\begin{tikzpicture}[baseline=0.4cm,scale=0.35]
\node at (0,0) [root] (0) {};
\node at (-1.5,2) [dot] (1) {};
\node at (1.5,2) [dot] (2) {};

\draw[keps] (1) to (0);
\draw[keps] (2) to (0);
\draw[kbase] (2) to (1);
\draw[kbase] (1) to[bend left=60] (2);
\end{tikzpicture} + (\ldots)\;,
\end{equs}
where $(\ldots)$ denotes an expression that converges to a finite limit as $\eps \to 0$. This can easily be shown in a way similar to the proof of 
Lemma~\ref{lem:defRenormConstTweak}. For example, one of the additional
terms appearing in the right hand side of $C_2^{(\eps)}$ is given by 
\begin{equ}[e:oneterm]
\bigl((R_\eps^{(1)} \cdot K_{\eps,\rho}')*K_{\eps,\rho}'*P_\eps\bigr)(0)
+ \bigl((R_\eps^{(2)} \cdot K_\eps')*K_\eps'*K_\eps'*K_\eps'(-\cdot)\bigr)(0)
\end{equ}
To show that this converges to a finite limit, one uses the fact that,
by Remark~\ref{rem:Keps} and Lemma~\ref{lem:Keps}, 
$R_\eps^{(1)} \cdot K_{\eps,\rho}'$, $K_{\eps,\rho}'$ and $P_\eps$ converge 
as $\eps \to 0$
in $\CB_{0,4-\kappa}$, $\CB_{0,2-\kappa}$, and $\CB_{0,1-\kappa}$ respectively, for 
every $\kappa > 0$. It then suffices to take $\kappa$ sufficiently small and to apply 
Lemma~\ref{lem:twoscale} twice to show that the first term in \eqref{e:oneterm}
converges to a finite limit. Regarding the second term of \eqref{e:oneterm}, 
both $R_\eps^{(2)} \cdot K_\eps'$ and $K_\eps'$ converge to limits in 
$\CB_{2+\kappa,3}$ for any $\kappa > 0$
so that its convergence can again be reduced to repeated applications of
Lemma~\ref{lem:twoscale}. The other terms appearing in the remainder
terms of \eqref{e:exprCeps} can be dealt with in an analogous way.

At this stage, we perform an integration by parts for the integration variable 
represented by the top-left vertex in the first term for $C_3^{(\eps)}$. This yields the
exact identity
\begin{equ}
\begin{tikzpicture}[baseline=0.4cm,scale=0.35]
\node at (0,0) [root] (0) {};
\node at (-1.5,2) [dot] (1) {};
\node at (1.5,2) [dot] (2) {};

\draw[keps] (1) to (0);
\draw[keps] (2) to (0);
\draw[kbase] (1) to (2);
\draw[kbase] (1) to[bend left=60] (2);
\end{tikzpicture}
= -2\;
\begin{tikzpicture}[baseline=0.4cm,scale=0.35]
\node at (0,0) [root] (0) {};
\node at (-1.5,2) [dot] (1) {};
\node at (1.5,2) [dot] (2) {};

\draw[kbase] (1) to (0);
\draw[keps] (2) to (0);
\draw[keps] (1) to (2);
\draw[kbase] (1) to[bend left=60] (2);
\end{tikzpicture}\;,
\end{equ}
where the factor $2$ comes from the fact that the derivative of $(K_{\eps,\rho})^2$
(the two arrows linking the two top vertices) equals $2K_{\eps,\rho} K_{\eps,\rho}'$. 
Inserting this into the above expression for $C_3^{(\eps)}$ yields 
\begin{equ}
C_3^{(\eps)} = -\;
\begin{tikzpicture}[baseline=0.4cm,scale=0.35]
\node at (0,0) [root] (0) {};
\node at (-1.5,2) [dot] (1) {};
\node at (1.5,2) [dot] (2) {};

\draw[kbase] (1) to (0);
\draw[keps] (2) to (0);
\draw[keps] (1) to (2);
\draw[kbase] (1) to[bend left=60] (2);
\end{tikzpicture}
+ {1\over 2}\;
\begin{tikzpicture}[baseline=0.4cm,scale=0.35]
\node at (0,0) [root] (0) {};
\node at (-1.5,2) [dot] (1) {};
\node at (1.5,2) [dot] (2) {};

\draw[keps] (1) to (0);
\draw[keps] (2) to (0);
\draw[kbase] (2) to (1);
\draw[kbase] (1) to[bend left=60] (2);
\end{tikzpicture} + (\ldots)\;.
\end{equ}
We now note that the first term in this expression is identical
to the first term appearing in the expression
for $4C_2^{(\eps)}$. As a consequence, we have
\begin{equ}[e:finalExprLogs]
C_3^{(\eps)} + 4C_2^{(\eps)}
= 
\begin{tikzpicture}[baseline=0.6cm,scale=0.35]
\node at (0,0) [root] (0) {};
\node at (0,2) [dot] (1) {};
\node at (0,4) [dot] (4) {};

\draw[keps] (1) to (0);
\draw[kbase] (0) to[bend left=60] (4);
\draw[kbase] (1) to[bend right=60] (4);
\draw[keps] (4) to (1);
\end{tikzpicture}
+
\begin{tikzpicture}[baseline=0.6cm,scale=0.35]
\node at (0,0) [root] (0) {};
\node at (0,2) [dot] (1) {};
\node at (0,4) [dot] (4) {};

\draw[keps] (1) to (0);
\draw[kbase] (4) to[bend left=60] (0);
\draw[kbase] (1) to[bend left=60] (4);
\draw[keps] (4) to (1);
\end{tikzpicture}
+
\begin{tikzpicture}[baseline=0.6cm,scale=0.35]
\node at (0,0) [root] (0) {};
\node at (0,2) [dot] (1) {};
\node at (0,4) [dot] (4) {};

\draw[keps] (1) to (0);
\draw[kbase] (0) to[bend left=60] (4);
\draw[kbase] (4) to[bend left=60] (1);
\draw[keps] (4) to (1);
\end{tikzpicture}
+{1\over 2}\;
\begin{tikzpicture}[baseline=0.4cm,scale=0.35]
\node at (0,0) [root] (0) {};
\node at (-1.5,2) [dot] (1) {};
\node at (1.5,2) [dot] (2) {};

\draw[keps] (1) to (0);
\draw[keps] (2) to (0);
\draw[kbase] (2) to (1);
\draw[kbase] (1) to[bend left=60] (2);
\end{tikzpicture} + (\ldots)\;.
\end{equ}
It is therefore sufficient to show that the four terms appearing on the right hand
side of this expression all converge to finite limits
as $\eps \to 0$.

To bound the first two terms, we use the easily shown fact that the kernel $K_{\eps,\rho}'(z)K_{\eps,\rho}(-z)$ converges to $P_\rho'(z) P_\rho(-z)$ (where we set $P_\rho = P * \rho$)
 $\CB_{0,\beta}$ for every $\beta > 0$. 
The fact that these terms converge to finite limits independent of the choice of $K$ 
then immediately follows 
by applying Lemma~\ref{lem:twoscale} twice.
A virtually identical argument allows to deal with the fourth term.
Concerning the third term appearing in the right hand side of \eqref{e:finalExprLogs},
we note that, by Remark~\ref{rem:Keps} and Lemma~\ref{lem:twoscale}, the kernel
$F_\eps \eqdef (K_{\eps,\rho}K_{\eps,\rho}')*K_{\eps,\rho}'$ converges to a limit in
$\CB_{0,2-\kappa}$ for any $\kappa > 0$, and is supported in $\{(t,x)\,:\, t > -C\}$, for some fixed constant $C>0$. Since the kernel
$K_{\eps,\rho}$ also has the same support property and converges in $\CB_{0,1-\kappa}$, the product $F_\eps(z)K_{\eps,\rho}(-z)$ converges in $\CB_{0,3-\kappa}$
and is supported in $\{(t,x)\,:\, |t| \le C\}$. It is straightforward to conclude
that such a function is absolutely integrable for $\kappa$ small enough,
and the claim then follows. 

It remains to show that the constants $C_\tau^{(\eps)}$ have finite limits for all
$\tau \in \Bad \setminus\{\<22>,\<211>\}$, where $\Bad$ was defined in \eqref{e:tauBB}.
Let us first consider
elements $\tau$ of the form 
\begin{equ}
\tau = \Eps^\ell\bigl(\Psi^{2\ell} \CI'(\Eps^m(\Psi^{2m+2}))\CI'(\Eps^n(\Psi^{2n+2}))\bigr)\;, 
\end{equ}
with $\ell + m + n > 0$, which is essentially a ``decorated'' version of $\<22>$.
By the definition \eqref{e:defCtaueps} of $C_\tau^{(\eps)}$ combined with the definitions
of $\Wick$ and $\PPi^{(\eps)}$, we have the identity
\begin{equ}
C_\tau^{(\eps)} = \eps^{\ell+m+n} \E \bigl((\Psi^{(\eps)})^{\diamond (2\ell)}\Phi_{2m+2}^{(\eps)}\Phi_{2n+2}^{(\eps)}\bigr)(0)\;,
\end{equ}
where we used the notations $\Psi^{(\eps)} = K' * \xi^{(\eps)}$
and $\Phi_\ell^{(\eps)} = (K' * (\Psi^{(\eps)})^{\diamond \ell})$ 
as in \eqref{e:boundPsieps} and \eqref{e:boundDerPhi}.
Using graphical notations similar to before and the properties of the Wick product, 
the expectation appearing in this expression 
is given by all possible ways of performing pairwise contractions of all 
nodes of the type \tikz[baseline=-3] \node [var] {};
without ever contracting two nodes belonging to the same
``group'' in the following graph:
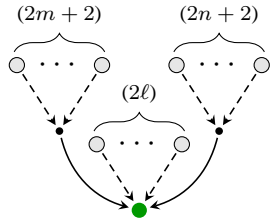
\begin{equ}
\begin{tikzpicture}[scale=0.35,baseline=0.3cm]
\def\x{6}
\def\hf{3}
\def\w{1.6}
	\node at (0,1)  [dot] (int) {};
	\node at (-\w,3.5)  [var] (left) {};
	\node at (\w,3.5)  [var] (right) {};
	\node at (0,3.5)  {$\cdots$};
	
	\draw[keps] (left) to (int);
	\draw[keps] (right) to (int);
\draw [decorate,decoration={brace,amplitude=7pt}]
(-\w-0.1,4) -- (\w+0.1,4) node [midway,yshift=0.5cm] 
{\scriptsize $(2m+2)$};
	\node at (\x,1)  [dot] (int2) {};
	\node at (\x-\w,3.5)  [var] (left2) {};
	\node at (\x+\w,3.5)  [var] (right2) {};
	\node at (\x,3.5)  {$\cdots$};
	
	\draw[keps] (left2) to (int2);
	\draw[keps] (right2) to (int2);
\draw [decorate,decoration={brace,amplitude=7pt}]
(\x-\w-0.1,4) -- (\x+\w+0.1,4) node [midway,yshift=0.5cm] 
{\scriptsize $(2n+2)$};

	\node at (\hf,-2)  [root] (root) {};
	\draw[kernel,bend left=30] (int2) to (root);
	\draw[kernel,bend right=30] (int) to (root);

	\node at (\hf-\w,0.5)  [var] (left3) {};
	\node at (\hf+\w,0.5)  [var] (right3) {};
	\node at (\hf,0.5)  {$\cdots$};
	\draw[keps] (left3) to (root);
	\draw[keps] (right3) to (root);
\draw [decorate,decoration={brace,amplitude=7pt}]
(\hf-\w-0.1,1) -- (\hf+\w+0.1,1) node [midway,yshift=0.5cm] 
{\scriptsize $(2\ell)$};
\end{tikzpicture}
\end{equ}
It is clear that such a pairing can exist only when no such group is larger than the two 
others combined, i.e.\ when $m \le \ell + n$, $n \le \ell + m$, and
$\ell \le m+n+2$. If one of these conditions fails, one has $C_\tau^{(\eps)} = 0$ and the
statement is trivial. If they are satisfied on the other hand, 
one obtains with the same graphical notations as in \eqref{e:boundSimpleMultiple}
the identity
\begin{equ}[e:formulaCtau]
C_\tau^{(\eps)} =  \eps^{\ell+m+n} C_{\ell,m,n}
\begin{tikzpicture}[scale=0.35,baseline=-0.25cm]
	\node at (0,1)  [dot] (int) {};
	\node at (6,1)  [dot] (int2) {};

	\node at (3,-2)  [root] (root) {};
	\draw[kernel,bend left=30] (int2) to (root);
	\draw[kernel,bend right=30] (int) to (root);

	\draw[dots,bend right=30] (int2) to node[labl] {\scriptsize $(a)$} (root);
	\draw[dots,bend left=30] (int) to node[labl] {\scriptsize $(b)$} (root);
	\draw[dots,bend left=30] (int) to node[labl] {\scriptsize $(c)$} (int2);
\end{tikzpicture}
\;,
\end{equ}
where the integer values $a$, $b$ and $c$ are related to $\ell$, $m$ and $n$
by $a+b = 2\ell$, $a+c = 2m+2$, $b+c = 2n+2$, and the combinatorial factor 
$C_{\ell,m,n}$ is given by
\begin{equ}
C_{\ell,m,n} = {(2m+2)!(2n+2)!(2\ell)! \over a!b!c!}\;.
\end{equ}
The above conditions on $\ell$, $m$, $n$ precisely guarantee that $a$, $b$ and $c$ 
are positive. In order to show that $C_\tau^{(\eps)}$ converges to a limit as $\eps \to 0$,
we note first that as before we can perform a change of variables such that
one actually has
\begin{equ}[e:formulaCtauEps]
C_\tau^{(\eps)} =  C_{\ell,m,n}
\begin{tikzpicture}[scale=0.35,baseline=-0.25cm]
	\node at (0,1)  [dot] (int) {};
	\node at (6,1)  [dot] (int2) {};

	\node at (3,-2)  [root] (root) {};
	\draw[kernel,bend left=30] (int2) to (root);
	\draw[kernel,bend right=30] (int) to (root);

	\draw[dots,bend right=30] (int2) to node[labl] {\scriptsize $(a)$} (root);
	\draw[dots,bend left=30] (int) to node[labl] {\scriptsize $(b)$} (root);
	\draw[dots,bend left=30] (int) to node[labl] {\scriptsize $(c)$} (int2);
\end{tikzpicture}
\;,
\end{equ}
provided that we now interpret \tikz[baseline=-0.1cm] \draw[dots] (0,0) to node[labl]{\scriptsize $(k)$} (1.5,0); as $P_\eps^k$ and \tikz[baseline=-0.1cm] \draw[kernel] (0,0) -- (1,0); as $\CS_\eps^{(2)}K'$.
As a consequence of Lemma~\ref{lem:Keps}, combined with the properties of the
scaling operator and the definition of $K$, the kernel $P_\eps$ converges to 
a limit $P_0$ in $\CB_{0,1-\kappa}$ for every $\kappa > 0$. Similarly, the kernel
$\CS_\eps^{(2)}K'$ converges to $P'$ (the spatial derivative of the 
heat kernel $P$) in $\CB_{2,2-\kappa}$ for every $\kappa > 0$.
In all cases, these limits are independent of the choice of kernel $K$.

Write $\tilde P_\eps^{(a)} = P_\eps^a \CS_\eps^{(2)}K'$ as a shorthand.
As a consequence of the above, the kernels $\tilde P_\eps^{(a)}$, $\tilde P_\eps^{(b)}$,
and $P_\eps^c$ converge in $\CB_{2,2+a-\kappa}$, $\CB_{2,2+b-\kappa}$ and $\CB_{0,c-\kappa}$
respectively. We now distinguish between two different cases. First, we consider the case
$c = 0$. In this case we see from \eqref{e:formulaCtauEps} that 
\begin{equ}
C_\tau^{(\eps)} =  C_{\ell,m,n} \int \tilde P_\eps^{(a)}(z)\,dz
\int \tilde P_\eps^{(b)} (z)\,dz\;.
\end{equ}
Since the kernels $\tilde P_\eps^{(a)}$ and $\tilde P_\eps^{(a)}$ are odd under the
substitution $x \mapsto -x$, we have $C_\tau^{(\eps)} = 0$ in this case so the claim
is trivial.
In the case $c > 0$, we obtain from \eqref{e:formulaCtauEps} the identity
\begin{equ}
C_\tau^{(\eps)} =  C_{\ell,m,n} \bigl(\tilde P_\eps^{(a)} * \tilde P_\eps^{(a)}(-\cdot) * P_\eps^c \bigr)(0)\;.
\end{equ}
To show that this converges, note first that as a consequence of
Lemma~\ref{lem:twoscale}, $\tilde P_\eps^{(a)} * \tilde P_\eps^{(a)}(-\cdot)$ converges in 
$\CB_{1,\beta}$ to some limit $\tilde P^{(a,b)}$ for every $\beta < (1+a+b) \wedge (2+a) \wedge (2+b)$.
There are now three cases. If $a = b = 0$, then $\tilde P^{(a,b)} \in \CB_{1,1-\kappa}$. 
In this case one has $\ell = 0$ and $c = m + n + 2 \ge 3$, so that 
$P_\eps^c$ converges in $\CB_{0,3-\kappa}$. Lemma~\ref{lem:twoscale} then 
implies that the convolution
converges in $\CB_{0,0}$, so that $C_\tau^{(\eps)}$ converges.
If $a > b = 0$, then $\tilde P^{(a,b)} \in \CB_{1,2-\kappa}$. In this case, since $b=0$ and $b+c = 2n+2$, one has $c \ge 2$ so that $P_\eps^{c}$ 
converges in $\CB_{0,2-\kappa}$. This does again allow us to apply Lemma~\ref{lem:twoscale}
to show that $C_\tau^{(\eps)}$ converges in this case. The case $b > a = 0$ is of 
course identical. In the last case when 
both $a$ and $b$ are strictly positive, one has $\tilde P^{(a,b)} \in \CB_{1,3-\kappa}$.
Since we assumed $c > 0$, this again allows us to apply Lemma~\ref{lem:twoscale} to cover
this last case as well. 

We now turn to
\begin{equ}
\tau = \Eps^\ell\bigl(\Psi^{2\ell+1} \CI'(\Eps^m(\Psi^{2m+1}\CI'(\Eps^n(\Psi^{2n+2}))))\bigr)\;,
\end{equ}
the ``decorated'' version of $\<211>$. In this case, an argument virtually identical to above
shows that one has
\begin{equ}[e:formulaCtauEps2]
C_\tau^{(\eps)} =  C_{\ell,m,n}
\begin{tikzpicture}[scale=0.35,baseline=-0.25cm]
	\node at (0,1)  [dot] (int) {};
	\node at (6,1)  [dot] (int2) {};

	\node at (3,-2)  [root] (root) {};
	\draw[kernel,bend right=30] (root) to (int2);
	\draw[kernel,bend right=30] (int) to (root);

	\draw[dots,bend right=30] (int2) to node[labl] {\scriptsize $(a)$} (root);
	\draw[dots,bend left=30] (int) to node[labl] {\scriptsize $(b)$} (root);
	\draw[dots,bend left=30] (int) to node[labl] {\scriptsize $(c)$} (int2);
\end{tikzpicture}
\;,
\end{equ}
but this time the constants $a$, $b$, $c$ satisfy
\begin{equ}[e:condabc]
a+c = 2\ell +1 \;,\qquad a+b = 2m+1\;,\qquad b+c = 2n+2\;.
\end{equ}
As before one has $C_\tau^{(\eps)} = 0$ when $c = 0$ so that we can assume $c > 0$.
As before, we then have
\begin{equ}
C_\tau^{(\eps)} =  C_{\ell,m,n} \bigl(\tilde P^{(a,b)}_\eps * P_\eps^c\bigr)(0)\;,
\end{equ}
this time with $\tilde P^{(a,b)}_\eps = \tilde P_\eps^{(a)} * \tilde P_\eps^{(b)}$
which, as before, converges to a limit $\tilde P^{(a,b)}$ in $\CB_{1,\beta}$
for every $\beta < (1+a+b) \wedge (2+a) \wedge (2+b)$.
The case $a = b = 0$ is impossible since one has $a + b \ge 1$, so assume first
$a > b = 0$. As before, this implies that $c \ge 2$, so that this case is covered
by Lemma~\ref{lem:twoscale} as above. The case where $a,b,c> 0$ is also covered in exactly
the same way as above.
This time however, the case $b > a = 0$ is not the same
as the case $a > b =0$ since the conditions \eqref{e:condabc} are no longer symmetric 
under $a \leftrightarrow b$. If $c \ge 2$, then this case is covered in the same way as before.

However, it can happen this time that $a = 0$ and $c = 1$, which is not covered by
Lemma~\ref{lem:twoscale} anymore. Our assumptions then imply that $b \ge 3$, so that 
$\tilde P_\eps^{(b)}$ is integrable.
We furthermore exploit the fact that $\tilde P_\eps^{(b)}$ is odd, so that it
actually integrates to $0$ and we are in the setting of Lemma~\ref{lem:gainOne}
with $\alpha_1 = \alpha_2 = 2$, $\beta_1 = 5-\kappa$, and $\beta_2 = 2-\kappa$.
This shows that in this case $\tilde P_\eps^{(a,b)}$ converges to
$\tilde P^{(a)}$ not only in $\CB_{1,\beta}$ for $\beta < 2$, but also for all
$\beta < 3$. Lemma~\ref{lem:twoscale} now applies to show that the convolution
with $P_\eps$ converges in $\CB_{0,0}$, thus yielding the required convergence
and concluding the proof.
\end{proof}

\section{Main convergence results}

We are now ready to collect the various results from the previous sections in
order to prove the main convergence results of this article.

\subsection{Weak asymmetry regime}

We have the 
following result, which allows us to identify solutions driven by the
model $\Z$ with the Hopf-Cole solutions to the KPZ equation. 

\begin{proposition}\label{prop:CH}
Let $\gamma, \eta$ be as in Theorem~\ref{theo:FP}
and let $H \in \CD^{\gamma,\eta}$ be the solution to \eref{e:abstrFP} given by
Theorem~\ref{theo:FP} for the model $\Z$ given by Theorem~\ref{theo:convModel},
and with initial condition $h_0 \in \CC^\eta$.
Then, there exists a constant $c$ depending only on the choice of cutoff kernel $K$ 
such that the function $h(t,x) = \bigl(\CR H\bigr)(t,x) - \lambda^3 c t$ 
is almost surely equal to $h_{\Hopf}^{(\lambda)}$ with $\lambda = \hat a_1$.
\end{proposition}


In order to prove this result, we give an alternative construction of the
model $\Z$. This will allow us to obtain Proposition~\ref{prop:CH} as an
essentially immediate consequence of \cite[Thm~4.7]{Etienne}.
To formulate this preliminary result,
we define $\tilde M^{(\eps)}$ exactly as $M^{(\eps)}$, but this time
with $C_\tau = 0$ for every $\tau$ of the form \eqref{e:tauBB} with 
$\ell + m + n > 0$.
Using the same notations as above, we then have the following result:

\begin{proposition}\label{prop:convCH}
Let $\xi^{(\eps)}$ be given by \eref{e:xiEps} and consider the sequence
of models on $\CT$ given by
\begin{equ}
\tilde \Z_\eps = \tilde M^{(\eps)} \LL_0(\xi^{(\eps)})\;,
\end{equ}
with $\LL_0$ defined in Section~\ref{sec:canonical}.
Then, one has $\tilde \Z_\eps \to \Z$ in $\MM_0$ in probability, where $\Z$ is the same
(random) model as in Theorem~\ref{theo:convModel}.
\end{proposition}

\begin{remark}
Note that in the statement of Proposition~\ref{prop:convCH}, we consider
the lift $\LL_0$ instead of the lift $\LL_\eps$. Since we furthermore 
set $C_\tau = 0$ for every formal expression $\tau$ containing the symbol $\Eps$, 
the model $\tilde \Z_\eps$ yields $0$ when applied to any formal expression
that includes a power of $\Eps$.
\end{remark}

\begin{proof}
By the combined definitions of $\LL_0$ and $\tilde M^{(\eps)}$ (in particular the 
fact that $C_\tau = 0$ for every $\tau$ of the form \eqref{e:tauBB} with 
$\ell + m + n > 0$), the model $\tilde \Z_\eps = (\tilde \Pi^{(\eps)},\tilde f^{(\eps)})$
satisfies $\tilde \Pi^{(\eps)}_z\tau = 0$ for every symbol $\tau$ that contains at least one
occurrence of $\Eps$. Therefore, any limiting model $\tilde \Pi$ must satisfy
$\tilde \Pi_z \tau = 0$ for such symbols, which is indeed the case for $\hat \Pi$.

Regarding the symbols $\tau$ not containing $\Eps$, we see from the definition
of $\LL_\eps$ in Section~\ref{sec:canonical} that both $\LL_\eps(\xi^{(\eps)})$ and $\LL_0(\xi^{(\eps)})$
act in exactly the same way on these symbols. Furthermore, the map
$\DeltaW$ appearing in \eqref{e:newModel} is the same for the constructions of 
$M^{(\eps)}$ and $\tilde M^{(\eps)}$, and the maps $M_0$ 
(also appearing in \eqref{e:newModel}) coincide on all elements not containing the symbol
$\Eps$. Therefore, we have $\tilde \Pi^{(\eps)}_z \tau = \hat \Pi^{(\eps)}_z \tau$
for every $\tau$ not containing $\Eps$. The claim (including that the models
$\tilde \Z_\eps$ converge in $\MM_0$) immediately follows from the fact 
that, $\tilde f^{(\eps)}$ is uniquely determined from $\tilde \Pi^{(\eps)}$ by the
condition that our models are admissible and satisfy $\tilde f_z^{(\eps)}(\EE_\ell^k(\tau)) = 0$
for every $\tau$.
\end{proof}

\begin{proof}[of Proposition~\ref{prop:CH}]
By Proposition~\ref{prop:convCH}, $h$ is the limit in probability of $h_\eps$, where 
$h_\eps = \CR^{(\eps)} H_\eps$, with $H_\eps$ the solution to the fixed point problem
associated to the model $\hat \Z_\eps$ and $\CR^{(\eps)}$ the corresponding
reconstruction operator. (Note that $\hat \Z_\eps$ is a model in $\MM_0$ and the convergence
takes place there. As a consequence, we can take an initial condition in $\CC^\eta$ even
for $\eps \neq 0$.)

However, we know from Proposition~\ref{prop:solutionRenormalised} that
$h_\eps$ is the classical strong solution to the semilinear PDE
\begin{equ}[e:approxKPZeps]
\d_t h_\eps = \d_x^2 h_\eps + \hat a_1 (\d_x h_\eps)^2 + \hat \xi^{(\eps)} - \hat a_1 C_0^{(\eps)} - c_\eps\;,
\end{equ}
where the constant $c_\eps$ is given by
\begin{equ}
c_\eps = 2 \hat a_1^3 \bigl(4C_{2}^{(\eps)} + C_{3}^{(\eps)}\bigr)\;.
\end{equ}
This constant converges to a finite limit of the form $\hat a_1^3 c_0$ with $c_0 \in \R$ 
depending in general both on the mollifier $\rho$ and the (arbitrary) choice of kernel $K$ 
by Theorem~\ref{theo:logs}.
In particular, a simple application of the chain rule shows that 
$Z_\eps = \exp(\hat a_1 h_\eps)$ is
the mild solution to 
\begin{equ}[e:approxSHE]
\d_t Z_\eps = \d_x^2 Z_\eps + \hat a_1 Z_\eps \,\xi^{(\eps)} - \hat a_1\bigl(\hat a_1 C_0^{(\eps)} + c_\eps\bigr) Z_\eps\;.
\end{equ} 
It was recently shown in \cite[Thm~4.7]{Etienne} (but see also \cite{KPZ}) 
that, for every $T>0$, the family $Z_\eps$ converges in probability in $\CC^\eta([0,T]\times S^1)$
to a limit $Z$ and that, provided that the renormalisation constant $c_\eps$ 
is suitably chosen (of the form $\hat a_1^3 \hat  c_0$ for some $\hat c_0$ depending
only on the choice of mollifier), 
this limit is almost surely equal to the solution to the stochastic heat equation \eqref{eq:SHE}
with $\lambda = \hat a_1$. This shows that the limit of \eqref{e:approxSHE} 
is given by
\begin{equ}
Z = \exp\bigl(\hat a_0^3(\hat c_0- c_0)t\bigr) Z^{(\hat a_1)}\;.
\end{equ}
Since we know that $Z^{(\hat a_1)}$ remains strictly positive \cite{MR1462228},
this implies in particular that
$h_\eps - \hat a_0^3 (\hat c_0 - c_0) t$ converges in probability to $h_{\Hopf}^{(\lambda)}$,
thus proving the claim with $c = \hat c_0 - c_0$.
The fact that $c$ depends only on $K$ and not on the choice of mollifier $\rho$ is a simple 
consequence of the fact that neither the limiting model $\Z$ nor the Hopf-Cole solution
depend on $\rho$. (But the limiting model $\Z$ does depend on the choice of $K$,
this is why there is no ``canonical'' value for $c$.)
\end{proof}

We are now ready to collect all of these results to prove the main 
convergence result of this article.

\begin{proof}[of Theorem  \ref{theo:weakAsym}]
Writing $\hat h_\eps = h_\eps - (\eps^{-1}\hat \lambda + c) t$, we first note
that $\hat h_\eps$ solves the equation
\begin{equ}
\d_t \hat h_\eps = \d_x^2 \hat h_\eps + {1\over \eps}F\bigl(\sqrt \eps \d_x h_\eps\bigr)  - \eps^{-1}\hat \lambda - c + \xi^{(\eps)}\;.
\end{equ}
Define now coefficients $\hat a_k$ implicitly by imposing the identity between polynomials
\begin{equ}
F(x) = \sum_{k=0}^m \hat a_k H_{2k}(x,C_0)\;,
\end{equ}
where $H_k(x,c)$ denotes the $k$th generalised Hermite polynomial as 
in \eqref{e:Hermite}.
One can check that the coefficients $\hat a_k$ are then given by
\begin{equ}
\hat a_k = {1\over k!} \int F^{(k)}(x)\,\mu_0(dx)\;.
\end{equ}
As a consequence of Proposition~\ref{prop:solutionRenormalised}, it then follows
that, provided that the constant $c$ is suitably chosen and that we set $\hat \lambda = \hat a_0$, 
one has $\hat h_\eps = \CR H$, where $H$ solves the fixed point problem
\eqref{e:abstrFP} for the renormalised model $\hat \Z_\eps$ considered in 
Theorem~\ref{theo:convModel}.
The (local in time) convergence of $h_\eps$ to a limit $h$ now follows 
by combining the convergence of $\hat \Z_\eps$ given in
Theorem~\ref{theo:convModel} with Theorem~\ref{theo:FP}.
The identification of the limit as the Hopf-Cole solution (provided that the constant
$c$ is suitably chosen) is given by Proposition~\ref{prop:CH}.
Since we know that the Hopf-Cole solutions are global, we immediately obtain
convergence over any fixed time interval from the last statement of Theorem~\ref{theo:FP}.
\end{proof}

\subsection{Intermediate disorder regime}

We now prove Theorem \ref{theo:convDisorder}. 
Let us first consider the special case where $F$ is a polynomial, so that $\tilde F = 0$. In this case,
we can rewrite the nonlinearity of \eref{e:rewritten} as
\begin{equ}
\sum_{k=0}^{2p-1} a_{p+k}\eps^{k \over 2p-1} \eps^{p + k - 1}  (\d_x h_\eps)^{2(p+k)}\;,
\end{equ}
which suggests that we should set $F_\eps(x) =  \sum_{k=0}^{2p-1} a_{p+k}\eps^{k \over 2p-1} x^{2(p+k)}$ and
define coefficients $\hat a_k^{(\eps)}$ as before by
\begin{equ}
\hat a_k^{(\eps)} = {1\over k!} \int F_\eps^{(k)}(x)\,\mu_0(dx)\;.
\end{equ}
In this case, one has in particular $\hat a_2^{(\eps)} \to \lambda$, with $\lambda$ as in 
\eqref{e:defLambdaIntermediate}, as well as $\hat a_0^{(\eps)} = \eps C_\eps$, with $C_\eps$
as in \eqref{e:defLambdaIntermediate}. One also has $\hat a_k^{(\eps)} \to \hat a_k$ for some $\hat a_k$
proportional to $a_p$ for $k \le p$,
and $\hat a_k^{(\eps)} \to 0$ for $k \in \{p,\ldots,3p-1\}$. With these notations at hand, we consider the
fixed point problem
\begin{equs}
H_\eps &= \CP \one_+ \Bigl(\Xi + \sum_{j=1}^{3p-1} \hat a_j^{(\eps)} \CQ_{\le 0}\hat\Eps^{j-1} \bigl(\CQ_{\le 0}(\DD H_\eps)^{2j}\bigr) + 
\eps^{-{3p-1\over 2p-1}} \tilde F \bigl(\eps^{p\over 2p-1} \CR \DD H_\eps\bigr)\one\Bigr)\\
&\qquad + P h_0^{(\eps)}\;,\label{e:abstrFPWeak}
\end{equs}
where $\CR$ denotes the reconstruction operator. 
Note now that if $H_\eps$ solves this fixed point equation and belongs to $\CD^{\gamma,\eta}_\eps$, 
then $\DD H_\eps$ is necessarily of the form
\begin{equ}
\DD H_\eps = \CI'(\Xi) + U'\;, 
\end{equ}
with $U' \in \CD^{\gamma,\eta}_\eps$ and $U'$ taking values in the subspace of $\CT$ spanned by $\one$
and elements with strictly positive homogeneity. In particular, by \eqref{e:defgammaetaeps}
and \cite[Def.~6.2]{Regularity}, $\CR U'$ is a continuous function
such that
\begin{equ}[e:boundRU']
\bigl|\bigl(\CR U'\bigr)(t,x)\bigr| \lesssim (\eps^2 + |t|)^{\eta-1\over 2} \|H_\eps\|_{\gamma,\eta;\eps}\;.
\end{equ}
It is also straightforward to show that 
\begin{equ}[e:wantedNoise]
\bigl(\Pi_z^{(\eps)} \CI'(\Xi)\bigr)(z) = |(K'*\xi^{(\eps)})(z)| \lesssim \eps^{-{1\over 2}-\kappa}\;,
\end{equ}
for any $\kappa > 0$, 
uniformly over compact domains. This shows that, for $\eta > {1\over 2}-\kappa$, the map
\begin{equ}
H_\eps \mapsto \eps^{{1\over 2}+\kappa} \CR \DD H_\eps\;,
\end{equ}
is locally Lipschitz continuous from $\CD^{\gamma,\eta}_\eps$ into $\CC$ (the space of continuous functions on the 
compact domain $D$ endowed with the supremum norm), uniformly over models $(\Pi,\Gamma) \in \MM_\eps$ with $\$\Pi\$_\eps$
bounded and furthermore satisfying \eqref{e:wantedNoise}.
Combining this with the fact that $|\tilde F(u) - \tilde F(v)| \lesssim |u-v|(|u|^{6p-1}+ |v|^{6p-1})$ for $u$ and $v$ 
bounded, we conclude that, provided that $\kappa < 1/(12p^2)$, the map
\begin{equ}
H_\eps \mapsto \tilde F^{(\eps)}(H_\eps) \eqdef \eps^{-{3p-1\over 2p-1}} \tilde F \bigl(\eps^{p\over 2p-1} \CR \DD H_\eps\bigr)\;,
\end{equ}
is locally Lipschitz continuous from $\CD^{\gamma,\eta}_\eps$ into $\CC$ (the space of continuous functions on the 
compact domain $D$ endowed with the supremum norm), with both norm and Lipschitz constant bounded uniformly over 
$\eps \in (0,1]$, $H_\eps$ in bounded balls of $\CD^{\gamma,\eta}_\eps$, and models in $\MM_\eps$ with bounded norm
satisfying \eqref{e:wantedNoise} for a fixed proportionality constant.
As a matter of fact, both the norm and the Lipschitz constant of $\tilde F^{(\eps)}$ are bounded by $\eps^\theta$ for some 
$\theta > 0$. Since the map $u \mapsto P*\one_+ u$, where $P$ denotes the heat kernel, maps $\CC$ into
$\CD^{\gamma,\eta}_\eps$ with norm bounded uniformly in $\eps$ and behaving like $T^\theta$ for some $\theta> 0$, where
$T$ is the local existence time under consideration, we can proceed as in the proof of Theorem~\ref{theo:FP}
to conclude that \eqref{e:abstrFPWeak} admits local solutions with a local existence time uniform over initial conditions
and models as just discussed.

As in the proof of Theorem~\ref{theo:FP}, one shows that as $\eps \to 0$, assuming that 
$\|\Pi^{(\eps)};\Pi\$_\eps \to 0$ for some model $\Pi$ and that the bound \eqref{e:wantedNoise}
holds uniformly over $\eps \in (0,1]$, one has $\|H_\eps;H\|_{\gamma,\eta;\eps}\to 0$,
where $H$ solves the fixed point problem
\begin{equ}[e:abstrFPLimit]
H_\eps = \CP \one_+ \Bigl(\Xi + \sum_{j=1}^{p} \hat a_j \CQ_{\le 0}\hat\Eps^{j-1} \bigl(\CQ_{\le 0}(\DD H_\eps)^{2j}\bigr)\Bigr) + P h_0\;.
\end{equ}
We now conclude exactly as before, noting that if we take for $\Pi^{(\eps)}$ the model
$\Z_\eps$ considered in Theorem~\ref{theo:convModel} then, as a consequence of 
Proposition~\ref{prop:solutionRenormalised}, $\CR H_\eps$ 
is precisely equal to $h_\eps - (C_\eps + c_\eps)t$, for the same constant $C_\eps$ as 
in the statement and some constant $c_\eps$ converging to a limit $c \in \R$.

\appendix

\section{A bound on generalised convolutions}
\label{sec:bounds}

In this section we obtain an estimate which allows us to bound the kind of generalised
convolutions of kernels appearing in the construction of quite general models built from Gaussian (and other) processes.

The basic ingredients are the following:  A finite directed multigraph $\CCG = (\CCV,\CCE)$ with edges 
$e \in \CCE$ labelled by pairs $(a_e, r_e) \in \R_+ \times \ZZ$, and  kernels 
$K_e \colon \R^d \setminus\{0\}\to \R$ which are compactly supported in the ball of radius $1$ around the origin. 
By multigraph we mean that we allow (a finite number of) multiple edges between vertices.  However, we will not allow edges from a vertex
to itself (loops).  We will always assume that every vertex has either an outgoing or incoming edge.
The exponent $a_e$ describes the singularity of the kernel $K_e$ at the origin 
in the sense that we assume that, for every $p>0$ and every edge $e \in \CCE$, 
the quantity $\|K_e\|_{a_e;p}$ is finite, where
\begin{equ}[e:boundK]
\|K\|_{\alpha;p} \eqdef \sup_{\|x\|_\s \le 1\atop |k|_\s < p} \|x\|_\s^{\alpha +|k|_\s}|D^k K(x)| <\infty\;.
\end{equ}
The constant $r_e$ will be used to allow for a 
renormalisation of the singularity.  The kernels are otherwise assumed
to be smooth.  If $r_e < 0$, then we will in addition be given a collection of real numbers $\{I_{e,k}\}_{|k|_\s < |r_e|}$ used to identify a Schwartz distribution associated to the singularity   (see \eqref{e:defRen}).

We will always consider the situation
where $\CCG$ contains a finite number $M\ge 1$ (typically $M = 2$) 
of distinguished edges $e_{\star,1},\ldots,e_{\star,M}$ 
connecting a distinguished vertex $0 \in \CCV$ to $M$ distinct vertices
$v_{\star,1},\ldots,v_{\star,M}$,
and all with label $(a_e,r_e)=(0,0)$.
In other words, the graphs we 
consider will always be of the following type:
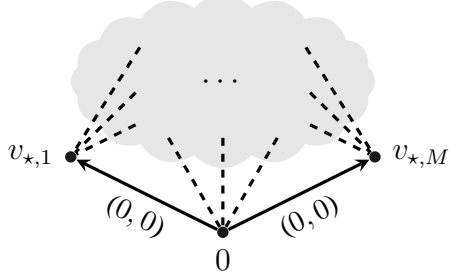
\begin{center}
\begin{tikzpicture}
	[dot/.style={circle,fill=black!90,inner sep=0pt, minimum size=1.5mm},
	empty/.style={inner sep=0pt, minimum size=1.5mm},
   	spine/.style={very thick}]
	\begin{scope}
		  \node[cloud, fill=gray!20, cloud puffs=16, cloud puff arc= 100,
		        minimum width=4cm, minimum height=2.3cm, aspect=1] at (0,2) {$\cdots$};

		  \node at (0,0)  [dot,label=below:$0$] (zero) {};
		  \node at (-2,1)  [dot,label=left:$v_{\star,1}$] (left) {};
		  \node at (2,1)  [dot,label=right:$v_{\star,M}$] (right) {};
		\draw[->,spine] (zero) to node [sloped,below] {\small$(0,0)$} (left);
		\draw[->,spine] (zero) to node [sloped,below] {\small$(0,0)$} (right);

		  \foreach \y in {2.6,2,1.5} {
		  	\node at (-1,\y)  (nl) {};
			\draw[dashed,spine] (nl) to node [left] {} (left);
		  	\node at (1,\y)   (nr) {};
			\draw[dashed,spine] (nr) to node [left] {} (right);
		  	}

		  \foreach \y in {-.8,0,.8} {
		  	\node at (\y,1.4)  (tp) {};
			\draw[dashed,spine] (tp) to node [left] {} (zero);
		  	}

	\end{scope}
\end{tikzpicture}
\end{center}
We will use the notation $\CCV_\star \subset \CCV$ for the set consisting
of the special
vertex $0$, plus the vertices $v_{\star,i}$, and we write $\CCV_0 = \CCV \setminus \{0\}$.

  Given a directed edge $e \in \CCE$, we write 
$e_\pm$ for the 
two vertices so that $e = (e_-,e_+)$ is directed from $e_-$ to $e_+$. In cases where there is more than one edge connecting $e_-$ to $e_+$ we will always \textit{assume that at most one can have nonzero renormalization $r_e$, and in that case $r_e$ must be positive}.  Then we may identify the multigraph with a graph  $(\CCV,\hat\CCE)$ where the multi edges from $e_-$ to $e_+$ are concatenated to one edge whose
 label $(\hat a_e,r_e)$ is simply the sum of the labels on the original multi edges.  The rest of the assumptions are most easily stated in 
terms of these new labels on the \textit{resulting directed graph} $(\CCV,\hat\CCE)$, \textit{although the application will be to the generalized convolution on the original graph}. We will also sometimes make the abuse of notation
that identifies $e$ with the set $\{e_-, e_+\}$ even though our edges are directed.
\label{star}
A subset $\bar\CCV \subset \CCV$ has \textit{outgoing edges} $\CCE^\uparrow (\bar \CCV)  =  
 \{ e\in \CCE : e\cap \bar\CCV = e_-\}$, \textit{incoming edges} $\CCE^\downarrow (\bar \CCV)  =   \{ e\in \CCE : e\cap \bar\CCV = e_+\}$, \textit{internal edges} $\CCE_0 (\bar \CCV)  =   \{ e\in \CCE : e\cap \bar\CCV = e\}$, and
\textit{incident edges} $\CCE (\bar \CCV)  =   \{ e\in \CCE : e\cap \bar\CCV \neq \emptyset\}$.
We will also use
$
\CCE_+ (\bar \CCV)  =   \{ e\in \CCE(\bar \CCV) :  r_e>0\}$ to denote the edges with positive renormalization, $\CCE_+^\uparrow = \CCE_+ \cap \CCE^\uparrow$ and $\CCE_+^\downarrow = \CCE_+ \cap \CCE^\downarrow$. 

%
\begin{assumption}\label{ass:mainGraph}
The resulting directed graph $(\CCV,\hat\CCE)$ with labels $(\hat a_e,r_e)$ satisfies:  
No edge containing the vertex $0$ may have $r_e > 0$; no edge with $r_e \neq 0$
connects two elements in $\CCV_\star$ and $0 \in e \Rightarrow r_e = 0$; no more than one edge with 
negative renormalization $r_e<0$ may emerge from the same vertex; and
\begin{claim}
\item[1.] For all $e\in \hat\CCE$, one has $ \hat a_e + (r_e \wedge 0) < |\s|$; 
\item[2.] For every subset $\bar \CCV \subset \CCV_0$ of cardinality at least $3$,  
\begin{equ}[e:assEdges]
\sum_{e  \in \hat\CCE_0(\bar \CCV)} \hat a_e < (|\bar \CCV| - 1)|\s|\;;
\end{equ}
\item[3.]  For every subset $\bar \CCV \subset \CCV$ containing $0$ of cardinality at least $2$,  
\begin{equ}[e:assEdges2]
\sum_{e  \in \hat \CCE_0 (\bar \CCV) } \hat a_e 
+\sum_{e  \in \hat\CCE^{\uparrow}_+ (\bar \CCV)  } (\hat a_e+ r_e -1)  - \sum_{e  \in \hat\CCE^{\downarrow}_+ (\bar \CCV)  }  r_e< (|\bar \CCV| - 1)|\s|\;;
\end{equ}
\item[4.]  For every non-empty subset $\bar \CCV \subset \CCV\setminus \CCV_\star$,
\begin{equ}[e:assEdges3]
 \sum_{e\in \hat\CCE(\bar\CCV)\setminus \hat\CCE^{\downarrow}_+(\bar \CCV) }  \hat a_e 
 +\sum_{e\in \hat \CCE^{\uparrow}_+(\bar\CCV)}  r_e
- \sum_{e \in \hat \CCE^\downarrow_+(\bar \CCV)} (r_e-1)
> |\bar \CCV||\s|\;.
\end{equ}
\end{claim}
\end{assumption}

Next we describe the renormalization procedure.
If $r_e < 0$, then, in a way reminiscent of \cite{BP,Hepp,Zimmermann}, 
we associate to $K_e$ the distribution, 
\begin{equ}[e:defRen]
\bigl(\Ren K_e\bigr)(\phi) = \int K_e(x) \Bigl(\phi(x) - \sum_{|k|_\s < |r_e|} {x^k \over k!} D^k\phi(0)\Bigr)\,dx + \sum_{|k|_\s < |r_e|} {I_{e,k} \over k!} D^k\phi(0).
\end{equ}   \label{defofren}
Note that  Assumption~\ref{ass:mainGraph}.1 and  \eref{e:boundK} imply that the integral in the definition of $\Ren K_e$ converges, so that
this definition actually makes sense.

Of course, if $\int |K_e(x)| |x|^k dx <\infty$ for $|k|_\s < |r_e|$ and $I_{e,k} = \int K_e(x) x^k dx$, then one just has
$\bigl(\Ren K_e\bigr)(\phi) = \int K_e(x)\phi(x)\,dx$.  
For $r_e \ge 0$, we just define $\bigl(\Ren K_e\bigr)(\phi) = \int K_e(x)\phi(x)\,dx$.

\eqref{e:defRen} defines a distributional ``kernel'' $\hat K_e$ for $r_e < 0$ acting on smooth $\phi$ on $\R^d\times\R^d$ by
\begin{equ}
 \hat K_e (  \phi)
\eqdef\tfrac12\int  
\Ren K_e( \phi_{z})\,dz\;,
\end{equ} where $\phi_{ z}(\bar z)\eqdef\phi((z+\bar z)/2,(z-\bar z)/2)$.
Of course if $\hat K_e$ is a function $\hat K_e ( x_{e_-}, x_{e_+})$ we will have   $\hat K_e(\phi)=\int \int \hat K_e ( x_{e_-}, x_{e_+}) \phi(x_{e_-},x_{e_+})\,dx_{e_-}\,dx_{e_+}$ and if $\phi(x_{e_-},x_{e_+}) = \phi_1(x_{e_-})\phi_2(x_{e_+})$, $\hat K_e(\phi)= \int \int \hat K_e ( x_{e_-}, x_{e_+}) \phi_1(x_{e_-})\phi_2(x_{e_+})\,dx_{e_-}\,dx_{e_+}$.

For $r_e \ge 0$, we define
\begin{equ}[e:defKhatn2]
\hat K_e(x_{e_-}, x_{e_+}) = 
K_e(x_{e_+}-x_{e_-}) - \sum_{|j|_\s < r_e} {x_{e_+}^j \over j!} D^j K_e(-x_{e_-})\;,
\end{equ}

\begin{remark}
In principle, one may encounter situations where more sophisticated renormalization procedures
are required. For the purpose of the present article however, the procedure described here is
sufficient.
\end{remark}

For a smooth test function $\phi$, let
$
\phi_\lambda(x) = \lambda^{-|\s|} \phi(x/\lambda)$.  The key quantity of interest is the generalized convolution
\begin{equ}\label{genconv}
\CI^\CCG(\phi_\lambda,K) \eqdef  \int_{(\R^d)^{\CCV_0}} \prod_{e\in \CCE}{\hat K}_e(x_{e_-}, x_{e_+})\prod_{i=1}^M \phi_\lambda(x_{v_{\star,i}}) \,dx\,
.
\end{equ}
It is not obvious that the right hand side of 
\eqref{genconv} even makes sense, but actually it is not so hard to see that our 
conditions imply that the distributions $\Ren K_e$, $r_e<0$ are only acting on the smooth parts of the other kernels.  The fact that it does make sense is part of the following statement,  which is the main 
 result of this section.

\begin{theorem}\label{theo:ultimate1}  Let $\CCG=(\CCV,\CCE)$ be a finite directed multigraph with labels $\{a_e,r_e\}_{e\in \CCE}$ and kernels
$\{K_e\}_{e\in \CCE}$ with  the resulting graph satisfying Assumption~\ref{ass:mainGraph} and its preamble.
Then, there exist  $C, p<\infty$ depending only on the 
structure of the graph $(\CCV,\CCE)$ and the labels $r_e$ such that 
\begin{equ}\label{genconvBound}
\CI^\CCG(\phi_\lambda,K) 
\le C\lambda^{\tilde \alpha}\prod_{e\in \CCE} \|K_e\|_{a_e;p}\;,
\end{equ}
for $0<\lambda\le 1$, where \begin{equ}\tilde \alpha = |\s||\CCV\setminus \CCV_\star| - \sum_{e\in \CCE} a_e. \end{equ}
In particular, the generalized convolution
in \eqref{genconv} is well-defined, and if $K_{e,m}\to K_e$ pointwise on $x\in \R^d\setminus\{0\}$ as $m\to \infty$,  for each $e$, and satisfying \eqref{e:boundK} uniformly in $m$, then $\CI^\CCG(\phi_\lambda,K_m)\to \CI^\CCG(\phi_\lambda,K)$.
\end{theorem}

Note in particular, that the bound \eqref{genconvBound} for genuine distributions, i.e. kernels $K_e$ with non-integrable singularities at $0$ and $r_e<0$, follows immediately once we prove the bound for regularizations of the kernels,
but with the norms on the right hand side independent of the regularization.  This has the consequence 
that within the proof, we can assume without loss of generality that all the kernels are smooth on all of $\R^d$.
The theorem will be proved in subsections \ref{decomp}-\ref{lastsubsection}.

\subsection{Decomposition}\label{decomp}

To simplify notations in what follows, we will start by enhancing the set of edges in our graph to include any $(v,w) \in \CCV^2$ for which there is not already one, or several, edges in $\CCE$.  To all such new directed edges we simply assign the 
kernel $\hat K_{(v,w)} \equiv 1$, so that, since every vertex of the original graph had either and incoming or outgoing edge, \eqref{genconv} is unaffected, and the fact that these new kernels do not have compact support is
irrelevant. These new edges necessarily come with $a_e=r_e= 0$.  We will abuse notation somewhat by henceforth referring to this enhanced graph as $\CCG=(\CCV, \CCE)$.

Now define a sequence of kernels $\{K_e^{(n)}\}_{n \ge 0}$ through the following
\begin{lemma}\label{lem:kernelsRenorm}
If $K_e$ are as above, then there exist $\{K_e^{(n)}\}_{n \ge 0}$ satisfying:
\begin{enumerate}
\item  $K_e(x) = \sum_{n \ge 0} K_e^{(n)}(x)$ for all $x \neq 0$;
\item  $\bigl(\Ren K_e\bigr)(\phi) = \sum_{n \ge 0} \int K_e^{(n)}(x)\,\phi(x)\,dx$ for smooth test functions $\phi$;
\item $K_e^{(n)}$ is supported in the annulus
$2^{-(n+2)}\le  \|x\|_\s \le  2^{-n}$;
\item for some $C<\infty$
\begin{equation}\label{e:boundKn}\sup_{|k| \le p, n\ge 0} 2^{-(a_e + |k|_\s)n}
|D^k K_e^{(n)}(x)| \le C \| K_e\|_{a_e;p}\, ;
\end{equation}
\item if $r_e < 0$, then $
\int P(x) K_e^{(n)}(x)\,dx = 0$
for all $n > 0$ and all polynomials $P$ with scaled degree strictly less than $|r_e|$.
\end{enumerate}
\end{lemma}\begin{proof} 
We first treat the case $r_e \ge 0$. Let $\psi \colon \R \to [0,1]$ be a
smooth function supported on $[3/8,1]$ and such that $\sum_{n \in \ZZ} \psi(2^{n} x) = 1$ for every $x \neq 0$, and let
\begin{equ}[psisupn]
\Psi^{(n)}(x) = \psi(2^n x)\;,
\end{equ}
so that $\Psi^{(n)}$ is supported in $2^{-(n+2)}\le  \|x\|_\s \le  2^{-n}$,
satisfies \eref{e:boundKn} with $a_n$ replaced by $0$, and sums up to $1$.
We also use the shorthands $\Psi^{(\le N)}(x) = \sum_{n\le N}\Psi^{(n)}(x)$ and $\Psi^{(-)}(x) = \Psi^{(\le 0)}(x)$.

Let $K_e^{(0)}(x) = \Psi^{(-)}(x)\, K_e(x)$ and
$K_e^{(n)}(x) = \Psi^{(n)}(x)\, K_e(x)$ for $n>0$. 
As a consequence of \eref{e:boundK}, and the fact that $|D^k\Psi^{(n)}(x) |\lesssim \|x\|_\s^{-|k|_\s}$, it is then straightforward to verify that 
$K_e^{(n)}$ does indeed satisfy the claimed properties.

In the case $r_e < 0$, the situation is a little less straightforward since then 2.\
doesn't follow from 1.\ and 4., and since we then also want to impose 5. In order to achieve 
this, we first note that it is possible
to find functions $\eta_k \colon \R^d \to \R$ which are supported in the 
annulus $\{x\,:\, \|x\|_\s \in [1/4,1/2]\}$ 
and are such that $\int x^\ell \eta_k(x) \,dx = \delta_{k,\ell}$ for
every $\ell$ with $|\ell|_\s < |r_e|$. 
We also set
\begin{equ}
\eta_k^{(n)}(x_1,\ldots,x_d) = 2^{n(|\s| + |k|_\s)}\eta_k(2^{n\s_1}x_1,\ldots, 2^{n\s_d}x_d)\;.
\end{equ}
We then set ~$I_{e,k}^{(0)} \eqdef I_{e,k} - \int x^k \Psi^{(-)}(x)\, K_e(x) \,dx$,
\begin{equs}
K^{(0)}_e(x) \eqdef \Psi^{(-)}(x)\, K_e(x) + \sum_{|k|_\s < |r_e|} \eta_k^{(0)}(x) I_{e,k}^{(0)}\;,\end{equs}
and recursively for $n > 0$, ~~$I_{e,k}^{(n)} \eqdef I_{e,k}^{(n-1)} - \int x^k \Psi^{(n)}(x)\, K_e(x) \,dx$,
\begin{equs}
K^{(n)}_e(x) &\eqdef \Psi^{(n)}(x)\, K_e(x) + \sum_{|k|_\s < |r_e|} \bigl(\eta_k^{(n)}(x) I_{e,k}^{(n)} - \eta_k^{(n-1)}(x) I_{e,k}^{(n-1)}\bigr)\;.
\end{equs}
With this definition, it is then straightforward to verify that 1 is satisfied due to the fact that the additional terms form a 
telescopic sum. 4 is satisfied since $\Psi^{(n)}$
satisfies \eref{e:boundKn} with $a_n$ replaced by $0$.  Finally,
as a consequence of the definition of the coefficients $I_{e,k}^{(n)}$ one
has $
 \int \sum_{j=0}^nK_e^{(j)}(x) x^k dx=I_{e,k} 
$ for $|k|_\s < |r_e|$ which proves 2 in the limit as $n\to\infty$ by 1.
\end{proof}

\begin{definition}\label{d:Khat} For  $\n \in \N^3$  define  $\hat K_e^{(\n )}(x,y)$ as  follows:
{\em If} $r_e \le 0$, then  $\hat K_e^{(\n)} = 0$ unless 
$\n = (k,0,0)$ in which case  $\hat K_e^{(\n)}(x,y) = K_e^{(k)}(y-x)$ with $K_e^{(k)}$
given by Lemma~\ref{lem:kernelsRenorm};
{\em if} $r_e > 0$, then\begin{equ}[e:defKhatn]
\hat K_e^{(k,p,m)}(x,y) = \Psi^{(k)}(y-x)\Psi^{(p)}(x)\Psi^{(m)}(y) 
\Bigl(K_e(y-x) - \sum_{|j|_\s < r_e} {y^j \over j!} D^j K_e(-x)\Bigr)\;,
\end{equ}
where the functions $\Psi^{(k)}$ are defined in \eref{psisupn}.\end{definition}

For $\n \colon \CCE \to \N^3$, let
\begin{equ}\label{def:kayhat}
\hat K^{(\n)}(x) = \prod_{e \in \CCE} \hat K_e^{(\n_e)}(x_{e_-},x_{e_+})
\end{equ}
so that if $K_e$ are smooth on all of $\R^d$, 
\begin{equ}
\CI^\CCG(\phi_\lambda,K) =\sum_{\n } \int_{(\R^d)^{\CCV_0}} \hat K^{(\n )}(x)\,\prod_{i=1}^M \phi_\lambda(x_{v_{\star,i}}) \,dx\;.
\end{equ}
For $\lambda \in (0,1]$, let
\begin{equ}
\CN_\lambda\eqdef\{ 
\n \colon \CCE \to \N^3 : 2^{-|\n_{e_{\star,i}}|} \le \lambda,  i=1,\ldots,M\}
\end{equ}
where $
e_{\star,i} = (0, v_{\star,i})$ and $|\n_{e_{\star,i}}| =m$ from above since by assumption $r_{e_{\star,i}}=0$.
Let
\begin{equ}[e:bigsum]
\CI_\lambda^\CCG(K) \eqdef \sum_{\n \in \CN_\lambda} \int_{(\R^d)^{\CCV_0}} \hat K^{(\n )}(x)\,dx\;.
\end{equ}

\begin{remark}  The main reason to add all the extra edges with $K_e=1$ is that $\n$ now completely determines the
distance (up to a factor $4$) between any two coordinates $x_v$ and $x_w$ of $x\in(\R^d)^{\CCV_0}$.
\end{remark}

 Theorem \ref{theo:ultimate1} follows from 
 
\begin{lemma}\label{theo:ultimate}
Under the same assumptions as Theorem \ref{theo:ultimate1}, there exist  $C, p<\infty$ depending only on the 
structure of the graph $(\CCV,\CCE)$ and the labels $r_e$ such that \begin{equ}\label{e:wantedBound2}
|\CI_\lambda^\CCG(K)| \le C \lambda^\alpha \prod_{e\in \CCE} \|K_e\|_{a_e;p}\;,\qquad \lambda \in (0,1]\;,
\end{equ}
where $\alpha =|\s| |\CCV_0| - \sum_{e \in \CCE} a_e$.
\end{lemma}

To see that Lemma \ref{theo:ultimate} implies Theorem \ref{theo:ultimate1} for smooth kernels, we use the fact that the rescaled test function  can be viewed as just another kernel  $K_{e_{*,i}}(v_{*,i}):= \phi_\lambda(v_{*,i})$ with $a_e= 0$ and $\|K_{e_{*,i}}\|_{a_e;p}=\lambda^{-|\s|}$.

To see that it suffices to prove Theorem \ref{theo:ultimate1} for smooth kernels, we argue as follows.  Given a labelled graph $\CCG$, let $p$ be given by the theorem.  Given singular kernels $K_e$ with $\|K_e\|_{a_e;p}<\infty$, $e\in\CCE$, let
$K_{e,m}$ be smooth kernels with $\|K_{e,m} - K_e\|_{a_e;p}\to 0$ as $m\to \infty$ for each $e$.  By the 
multilinearity it is not hard to see that the real numbers $\CI^\CCG(\phi_\lambda,K_m)$, $m=1,2,\ldots$ form a  Cauchy sequence, and therefore have a unique limit,
which, in addition, satisfies the bound \eqref{genconvBound}.

The lemma will be proved in Subsections~\ref{trees}--\ref{lastsubsection}. Throughout this section, 
the symbol $\sim$ denotes a bound from above and below, with proportionality
constants that only depend on $|\CCV|$.  Note that all constructions are finite so, for example, 
the constants appearing in inductive proofs are allowed to get worse at each stage, and no effort has been made
to optimize the dependence on the size of the graph $\CCG$.  Note that we can reduce ourselves to the case where all 
$\|K_e\|_{\alpha;p} =1$  by multilinearity, so we will not follow these norms in the sequel.

\subsection{Multiscale clustering} 
\label{trees} 

It turns out to be convenient to think of the integral in \eqref{e:bigsum} as over $x\in (\R^d)^{\CCV}$, with 
$x_0=0$ and we will use this convention throughout the proof.
Since our kernels are smooth,
the set of $x\in (\R^d)^{\CCV}$ where any  two different 
$\|x_v-x_w\|_\s$ coincide can be ignored in the integral in \eqref{e:bigsum}. 
To other points $x\in (\R^d)^{\CCV}$ 
we will associate a labelled rooted binary tree $T$ whose leaves are the $v\in \CCV$.

We will use the terminology {\it node}
instead of {\it vertex} to distinguish the nodes of this tree from the vertices $\CCV$ 
of the original graph, and denote them
by $\v$, $\w$, etc.   
A {\it leaf} is a node of degree 1. An {\it inner node} is one of degree at least 2.  
A rooted tree comes with a partial order, $\v \ge \w$  means that $\w$ belongs
to the shortest path connecting $\v$ to the root.  In genealogical terms, $\w$ is an ancestor of $\v$.  For any two nodes $\v$ and $\w$,
we write $\v \wedge \w$ for the unique node such that for any node $\u$ satisfying
$\u \le \v$ and $\u \le \w$,
one necessarily has $\u \le (\v \wedge \w)$, i.e., $\v\wedge \w$ is the 
most recent common ancestor of $\v$ and $\w$.  We will furthermore impose that
every inner node has exactly two descendants, that only the inner nodes are labelled, by natural numbers, 
and that the labelling $\ell$ of the inner nodes respects the partial order in the sense 
 that $\ell_\v \ge \ell_\w$ whenever $\v \ge \w$.  Note that the leaves of the tree will sometimes be denoted $v,w$ since they are also elements of 
 $\CCV$.

The way the tree is constructed  is as follows:  
First consider the complete undirected weighted graph with 
vertices $v\in \CCV$, and edge weight $\|x_v-x_w\|_\s$ assigned to the edge $(v,w)$, $v,w\in \CCV$.
A minimal spanning tree can be constructed, for example, by Kruskal's algorithm \cite{MR0078686}: Choose first the
edge of minimal weight, then successively add the edge with the smallest weight which is not in the tree already,
as long as adding it does not create a loop, in which case, it is skipped and we attempt to add the next smallest weight.  Since the edge weights can be strictly ordered, there is no ambiguity in this definition.  The binary
tree $T$ with leaves $v\in \CCV$ simply records the order in which edges were added to the minimal spanning
tree:  At the stage when the edge $(v,w)$ is  added to the minimal spanning tree, the branch containing $v$ is joined to the branch containing $w$.

Now for each node $\v$ we let 
\begin{equ}
\ell_{\v} = \max_{v\wedge w=\v}\lfloor -\log_2 \|x_v-x_w\|_\s\rfloor \,. 
\end{equ}
From the construction, if $\v \ge \w$, then $\ell_\v \ge \ell_\w$.

Given a set of vertices $\CCV$,  denote by $\CCT(\CCV)$  the set of rooted labelled binary
trees $(T,\ell)$ as above, with an order preserving labelling $\ell$,
which have $\CCV$ as their set of leaves. 
From the construction, a  generic $x\in (\R^d)^{\CCV}$ corresponds to an element $(T,\ell)$ of $\CCT(\CCV)$.
The downside of course is that  we can only partially read off the edge lengths $\|x_v-x_w\|_\s$ from $(T,\ell)$.
More precisely, for any two leaves $v, w \in \CCV$, one has
$
\|x_v - x_{w}\|_\s \sim 2^{-\ell_{v \wedge w}}$,
however the constants of proportionality can be quite poor.  In particular, it is not hard to see that
\begin{equ}[e:distancexx]
2^{-\ell_{v \wedge w}}\le \|x_v - x_{w}\|_\s \le |\CCV| 2^{-\ell_{v \wedge w}}\;,
\end{equ}
and that the upper bound cannot really be improved (for example, place the points co-linearly, with the largest gap at one end.) 
In applications such as cladograms, this renders such constructions essentially worthless, however, in our application, it only 
means that the resulting constant $C$ on the right hand side of \eqref{e:wantedBound2} will depend badly  on the size of the vertex set $\CCV$.  Since in any subcritical stochastic PDE there are only finitely many 
universal objects to control, the resulting bound suffices for our purposes.


\begin{definition} 
For $c= \log|\CCV|+2$, let  $\CN(T,\ell)$ consist of all functions $\n \colon \CCE \to \N^3$  such that for every edge $e = (v,w)$ with $r_e \le 0$, one has $\n_e = (k,0,0)$ with
$|k - \ell_{v\wedge w}| \le c$, and for every edge $e = (v,w)$ with $r_e > 0$, one has $\n_e = (k,p,m)$ with
$|k - \ell_{v\wedge w}| \le c$, $|p - \ell_{v\wedge 0}| \le c$, and
$|m - \ell_{w\wedge 0}| \le c$.
\end{definition}

 If $\n \colon \CCE\to \N^3$ 
is such that $\int_{(\R^d)^{\CCV}} \hat K^{(\n)}(x)\,dx$  is non-vanishing, then the  support of $ \hat K^{(\n)}$ is non-empty.  From \eref{e:bigsum},  
 $x \in \R^{\CCV}$ is in that support only if it belongs to the support of 
 $\hat K_e^{(\n_e)}(x_{e_-},x_{e_+})$ for every $e \in \CCE$.  Let $(T,\ell) \in \CCT(\CCV_0)$ be the tree associated to 
 $x \in \R^{\CCV}$.  
If
 $r_e \le 0$, then from Definition \ref{d:Khat}, we have that $\n_e = (m,0,0)$ and 
  $\hat K_e^{(\n_e)}(x_{e_-},x_{e_+})=K^{(m)} (x_{e_+}-x_{e_-})\neq0$. From 3 of Lemma~\ref{lem:kernelsRenorm} we have
  $\|x_{e_+}-x_{e_-}\|_\s \in [2^{-m-2},2^{-m}]$, and then from \eqref{e:distancexx}, we have 
  $|m - \ell_{v\wedge w}| \le c$.  
 If $r_e > 0$, the kernel $\hat K_e^{(\n_e)}$ with 
$\n_e = (k,p,m)$ is given by
\eref{e:defKhatn}, so that for $x$ to belong to its support we must have 
$\|x_{e_+}-x_{e_-}\|_\s \in [2^{-k-2},2^{-k}]$, $\|x_{e_+}\|_\s \in [2^{-p-2},2^{-p}]$, as well as
$\|x_{e_-}\|_\s\in [2^{-m-2},2^{-m}]$, 
which in the same way implies $|k - \ell_{v\wedge w}| \le c$, $|p - \ell_{v\wedge 0}| \le c$, and
$|m - \ell_{w\wedge 0}| \le c$.
Hence we have

\begin{lemma}\label{lem:treelike}  For 
every $\n\colon \CCE \to \N^3$ such that $\int_{(\R^d)^{\CCV_0}} \hat K^{(\n )}(x)\,dx$ from \eref{e:bigsum} is non-vanishing, there exists 
an element $(T,\ell) \in \CCT(\CCV)$ with $\n\in \CN(T,\ell)$. \qed
\end{lemma}

Denote by $\CCT_\lambda(\CCV)$,
 the subset of those labelled trees in $\CCT(\CCV)$ with the property that 
$2^{-\ell_{v\wedge w}} \le \lambda$ for any two leaves $v,w \in \CCV_\star$ (as defined on page~\pageref{star}).
As a consequence of Lemma~\ref{lem:treelike},
we can turn the sum over
$\CN_\lambda$ appearing in the definition of $\CI_\lambda^\CCG(K)$  into a sum over $\CCT_\lambda(\CCV)$:
\begin{equ}\label{eq:a13}
|\CI_\lambda^\CCG(K)| \lesssim \sum_{(T,\ell) \in \CCT_\lambda(\CCV)} \sum_{n \in \CN(T,\ell)} \Bigl|  \int_{(\R^d)^{\CCV}} \hat K^{(\n)}(x)\,dx\Bigr|\;.
\end{equ}

In order to bound the right hand side we will use the following construction.
Consider a rooted binary tree $T$ with a fixed distinguished inner node
$v_\star$ (in particular it has at least one inner node). 
We will denote by $T^\circ$ the set of inner nodes of $T$.
Since the tree is binary, every node of the subtree $T^\circ \subset T$
has exactly two children (in $T$), so that $T^\circ$, together 
with its partial order, actually determines the full tree $T$.
We then consider the set $\CN_\lambda(T^\circ)$ of all integer labelings $\ell \colon T^\circ \to \N$ which preserve the partial order of the tree $T^\circ$ as above and are such that
$2^{-\ell_{\v_\star}} \le \lambda$.  Finally, given a function
$\eta \colon T^\circ \to \R$, we write
\begin{equ}
\CI_\lambda(\eta) = \sum_{\ell \in \CN_\lambda(T^\circ)}\prod_{\v \in T^\circ} 2^{-\ell_\v \eta_\v}\;.
\end{equ}
Setting $|\eta| = \sum_{\v\in T^\circ} \eta_\v$,
we then have the following bound

 \begin{lemma}\label{theo:sumTree}
Assume that $\eta$ satisfies the following two properties: 
\begin{enumerate}
\item For every $\v \in T^\circ$, one has $\sum_{\u \ge \v} \eta_\u > 0$.
\item For every $\v \in T^\circ$ such that $\v\le \v_\star$, one has 
$\sum_{\u \not\ge \v} \eta_\u < 0$, provided that this sum contains at least one term.
\end{enumerate}
Then, one has $\CI_\lambda(\eta) \lesssim \lambda^{|\eta|}$, uniformly over $\lambda \in (0,1]$.
\end{lemma}

\begin{remark}
Since the order on $T^\circ$ is only partial, $\u \not \ge \v$ is different from
$\u < \v$. The latter would only consider the nodes between $\v$ and the root,
while the former also includes the subtrees dangling from these nodes. 
Note also that the second condition above is empty (and therefore automatically satisfied) 
in the special case when $\v_\star$ also happens to be the root.
\end{remark}

\begin{remark}
As will be evident from the proof, the first condition is necessary for
$\CI_\lambda(\eta)$ to even be finite. Regarding the second condition, if it fails,
then for every $\v$ with $\v_\star \ge \v$ such that $\sum_{\u \not\ge \v} \eta_\u = \alpha > 0$,
the upper bound for $\CI_\lambda(\eta)$ is larger by a factor $\lambda^{-\alpha}$.
If $\sum_{\u \not\ge \v} \eta_\u = 0$, one loses a factor $|\log \lambda|$.
\end{remark}

\begin{proof}[of Lemma \ref{theo:sumTree}]
The proof goes by induction on the size of $T^\circ$. If  $|T^\circ|=1$, it
consists of only the node $\v_\star$. Condition $1$ implies
that $\eta_{\v_\star} > 0$ and one has 
$\CI_\lambda(\eta) = \sum_{2^{-\ell} \le \lambda} 2^{-\ell \eta_{\v_\star}}
\sim \lambda^{\eta_{\v_\star}} = \lambda^{|\eta|}$ as required.

If $|T^\circ| > 1$, we distinguish between two different cases. 
In the first case, $T^\circ$ contains at least one extremal node\footnote{We call leaves of $T^\circ$ ``extremal nodes'' in order not to confuse them with the leaves of $T$.} $v$ 
which is different from the distinguished node $\v_\star$.
In this case, one has $\eta_\v > 0$ by the first assumption (since $v$ is
extremal, the only $\u$ with $\u \ge \v$ is $\v$ itself).
Denote now by $\bar T^\circ$ the tree obtained by erasing the leaf $\v$ and by
$\bar \eta \colon \bar T^\circ \to \R$ the function obtained by setting
$\bar \eta_\u = \eta_\u$ for every node 
$\u$ which is not the parent $\v_\uparrow$
of $\v$. 
We also set $\bar \eta_{\v_\uparrow} = \eta_{\v_\uparrow} + \eta_\v$, which ensures that $\bar \eta$
still satisfies conditions 1 and 2.
One then has
\begin{equs}
\CI_\lambda(\eta) &= \sum_{\ell \in \CN_\lambda(T^\circ)}\prod_{\w \in T^\circ} 2^{-\ell_\w \eta_\w}
= \sum_{\ell \in \CN_\lambda(\bar T^\circ)}\sum_{m \ge \ell_{\v_\uparrow}} 2^{-m\eta_\v }\prod_{\w \in \bar T^\circ} 2^{-\ell_\w \eta_]w}  \\
&\sim \sum_{\ell \in \CN_\lambda(\bar T^\circ)}2^{-\ell_{\v_\uparrow}\eta_\v }\prod_{\w \in \bar T^\circ} 2^{-\ell_\w \eta_\w} = \sum_{\ell \in \CN_\lambda(\bar T^\circ)}\prod_{\w \in \bar T^\circ} 2^{-\ell_\w \bar \eta_\w} = \CI_\lambda(\bar \eta)\;.
\end{equs}
By the induction hypothesis the required upper and lower bounds 
hold.

On the other hand, it may happen that the only extremal node of $T^\circ$ is $\v_\star$ itself.
In this case, the tree $T^\circ$ has a total order and, if $|T^\circ| = k \ge 2$, one can
label its nodes $\v_1\le \ldots\le \v_k = \v_\star$.
Denoting the corresponding values of $\eta$ by $\eta_1,\ldots,\eta_k$, 
we see that in this case our assumptions are equivalent to  
the fact that, for each $j \in \{1,\ldots,k\}$, one has
$\sum_{i \ge j} \eta_i > 0$ and, if $j > 1$, $\sum_{i < j} \eta_i < 0$.
In this case, we define $\bar T$ to be the tree where we remove the root $v_1$
and take $v_2$ as our new root. Similarly, to above, we define $\bar \eta$ on $\bar T^\circ$
by setting it equal to $\eta$ except on $\v_2$ where we set $\bar \eta_{\v_2} = \eta_{\v_2} + \eta_{\v_1}$. We then have the bound
\begin{equs}
\CI_\lambda(\eta) &= \sum_{\ell \in \CN_\lambda(T^\circ)}\prod_{\w \in T^\circ} 2^{-\ell_\w \eta_\w}
= \sum_{\ell \in \CN_\lambda(\bar T^\circ)}\sum_{0 \le m \le \ell_{\v_2}} 2^{-m\eta_{\v_1} }\prod_{\w \in \bar T^\circ} 2^{-\ell_\w \eta_\w}  \\
&\sim \sum_{\ell \in \CN_\lambda(\bar T^\circ)}2^{-\ell_{\v_2}\eta_{\v_1} }\prod_{\w \in \bar T^\circ} 2^{-\ell_\w \eta_\w} = \sum_{\ell \in \CN_\lambda(\bar T^\circ)}\prod_{\w \in \bar T^\circ} 2^{-\ell_\w \bar \eta_\w} = \CI_\lambda(\bar \eta)\;,
\end{equs}
as above. Again, we note that $\bar \eta$ does again satisfy our assumptions,
so the claim follows from the inductive hypothesis. Since we have exhausted all
possibilities, this concludes the proof.
\end{proof}

\subsection{General form of the bound}

Given a labelled tree $(T,\ell)\in \CCT(\CCV)$, denote by $\CD(T,\ell)$ the subset of $(\R^d)^{\CCV}$ such that
$\|x_v - x_w\|_\s \le  |\CCV| 2^{-\ell_{v\wedge w}}$ for all $v,w \in \CCV$. 
As usual, we use the convention that $x_0 = 0$.

\begin{lemma}\label{l:12} Suppose $\tilde K^{(\n)}(x)$ is a function such for each $\n \in \CN(T,\ell)$,
\begin{equ}\label{eel}\supp \tilde K^{(\n)} \subset  \CD(T,\ell)\end{equ} 
and \begin{equ}\label{eek}
\int_{(\R^d)^{\CCV}} \hat K^{(\n)}(x)\,dx = \int_{(\R^d)^{\CCV}}  \tilde K^{(\n)}(x)\,dx,
\end{equ}  then
\begin{equ}[e:bestBoundILambda]
|\CI_\lambda(K)| \lesssim \sum_{(T,\ell) \in \CCT_\lambda(\CCV)} \Bigl(\prod_{v \in T^\circ} 2^{-\ell_v |\s|}\Bigr) \sup_{\n \in \CN(T,\ell)} \sup_{x}|\tilde K^{(\n)}(x)|\;,
\end{equ}
where $T^\circ$ denotes the set of interior nodes of $T$.
\end{lemma}

\begin{proof}From \eqref{eq:a13}, using \eref{eek} and \eref{eel},
\begin{equ}
|\CI_\lambda(K)| \lesssim \sum_{T \in \CCT_\lambda(\CCV)} \sum_{n \in \CN(T,\ell)} \Bigl|  \int_{\CD(T,\ell)} \tilde K^{(\n)}(x)\,dx\Bigr|\;.
\end{equ}
We claim that the Lebesgue measure of 
$\CD(T,\ell)$ is bounded from above by some fixed constant multiple of $\prod_{v \in T^\circ} 2^{-\ell_v |\s|}$
.
To prove this, for each interior vertex $\v \in T^\circ$, we choose two 
elements $v_-, v_+ \in \CCV$ so that $v_- \wedge v_+ = \v$. 
The collection of edges $\{(v_-, v_+)\,:\, \v\in T^\circ\}$  forms
a spanning tree of $\CCV$ and 
\begin{equ}
\CD(T,\ell) \subset \bigl\{x\,:\, \|x_{v_-} - x_{v_+}\|_\s \le|\CCV| 2^{-\ell_{v}}\quad \forall  v \in T^\circ\bigr\}\;.
\end{equ}
The claim  follows by integrating over these
coordinates one by one. 
\end{proof}

\begin{remark}  
One could in principle simply choose $\tilde K^{(\n)} =  \hat K^{(\n)}$ in the lemma. 
It turns out that the 
resulting bound fails to take into account some cancellations and is not good enough for our purposes. 
The strategy of proof will then be to build, for each $n \in \CN(T,\ell)$, a function
$\tilde K^{(\n)}$  such that $\sup_{x}|\tilde K^{(\n)}(x)|$
can be bounded in a sharp way yielding a bound
of the desired homogeneity in $\lambda$.
\end{remark}

\subsection{Na\"\i ve bound}

Define  $\eta\colon T^\circ \to \R$ by
$\eta(v) =|\s|+ \sum_{e \in \hat\CCE} \eta_e(v)$, where
\begin{equs}[e:defEtaE7]
\eta_e(v)   &= - \hat{a}_e \one_{e_\uparrow} (v)
+ r_e  (\one_{e_+\wedge 0} (v) 
- \one_{e_\uparrow} (v)\bigr) \one_{r_e >0, e_+\wedge 0 > e_\uparrow} \\
&\quad + (1-r_e -  \hat{a}_e) (\one_{e_-\wedge 0} (v)
- \one_{e_\uparrow} (v)\bigr) \one_{r_e >0, e_-\wedge 0 > e_\uparrow}\;,
\end{equs}
with $\one_v(v) = 1$ and $\one_v(w) = 0$ for $w \neq v$.
We then have the following bound for the functions $\hat K^{(\n)}$:
%
%

\begin{lemma} \label{lem:simpleBoundKhat}
Assume the $\hat K^{(\n)}$ are given by Definition \ref{d:Khat} with $K_e$ satisfying \eref{e:boundK} and $\Psi^{(n)}$ given by \eref{psisupn}.
Then
 $\eta$ defined by \eqref{e:defEtaE7}  is such that
\begin{equ}[e:wantedProdBoundKHat]
\Bigl(\prod_{v \in T^\circ} 2^{-\ell_v |\s|}\Bigr)\sup_x |\hat K^{(\n)}(x)| \lesssim \prod_{v \in T^\circ} 2^{-\ell_v \eta(v)}\;,
\end{equ}
uniformly over all $\n \in \CN(T,\ell)$.
\end{lemma}

\begin{proof}
Due to the multiplicative 
structure of both sides of this inequality, it holds as soon as we are
able to prove the bound
\begin{equ}[e:wantedSingleBound]
\sup_x \left| \prod_{\tilde{e} = (e_-,e_+)} \hat K_{\tilde{e}}^{(\n_{\tilde{e}})}(x_{e_-},x_{e_+})\right| \lesssim \prod_{v \in T^\circ} 2^{-\ell_v \eta_e(v)}\;.
\end{equ}

Note here that the product on the right hand side actually only involves
at most two terms as a consequence of \eref{e:defEtaE7}.
This bound in turn follows trivially from \eref{e:boundKn} for those
edges ${\tilde{e}}$ for which all $r_{\tilde{e}} \le 0$. 

For the multiedges with some $r_{\tilde{e}_0} > 0$, and $\n_{\tilde{e}_0} = (k,p,m)$ we will estimate in two different ways:  If  $2m > k $,  then we use \cite[Prop.~A.1]{Regularity} to bound the next term in the Taylor expansion \eref{e:defKhatn} for that particular edge $\tilde{e}_0$.  The other multiedges $\tilde{e}=(e_-,e_+)$ all have $r_{\tilde{e}}=0$ and so produce multiplicative factors $K_{\tilde{e}}^{(k_{\tilde{e}})}(x_{e_+}-x_{e_-})$.  From the definition we must have $| k_{\tilde{e}}- k|\le 1$ or the product simply vanishes.  This gives
\begin{equ}[e:firstBB]
\sup_x \left| \prod_{\tilde{e} = (e_-,e_+)} \hat K_{\tilde{e}}^{(\n_{\tilde{e}})}(x)\right| \lesssim 2^{-r_e m + (\hat{a}_e + r_e) k}
\;.
\end{equ}
Since $\n \in \CN(T,\ell)$,
the index $n_e = (k,p,m)$ satisfies
$
|k - \ell_{e_\uparrow}| \le c$, $
|p - \ell_{e_-\wedge 0}| \le c$, and  $
|m - \ell_{e_+\wedge 0}| \le c$,
which gives
\begin{equ}[e:firstBB2]
\sup_x \left| \prod_{\tilde{e} = (e_-,e_+)} \hat K_{\tilde{e}}^{(\n_{\tilde{e}})}(x)\right| \lesssim  2^{-r_e \ell_{e_+\wedge 0} + (\hat{a}_e + r_e) \ell_{e_\uparrow}} \;.
\end{equ}
On the other hand, if $2m \le k $, then we simply bound the terms
appearing in \eref{e:defKhatn} separately, which yields \begin{equ}[e:secondBB]
\sup_x \left|\prod_{\tilde{e} = (e_-,e_+)} \hat K_{\tilde{e}}^{(\n_{\tilde{e}})}(x)\right| \lesssim 2^{\hat{a}_e k} + \sum_{|j|_\s < r_e} 2^{-m |j|_\s + (\hat{a}_e + |j|_\s) p}\;.
\end{equ}
Then it is mostly straightforward to check that \eqref{e:wantedSingleBound} holds for \eqref{e:defEtaE7}.
The only non-obvious point is that in the case $e_-\wedge 0 >e_\uparrow$, we have 
$2^{\hat{a}_e k } \le 2^{ \hat{a}_e p}$ and $p\ge m$ so $2^{\hat{a}_e k} + \sum_{|j|_\s < r_e} 2^{-m |j|_\s + (\hat{a}_e + |j|_\s) p} \lesssim 2^{(r_e-1)( p-m ) + \hat{a}_e  p}$.
\end{proof}

The problem with this bound is that it is not the case in general that the function
$\eta$ satisfies the assumptions of Lemma~\ref{theo:sumTree}. This is because of the possible
presence of edges $e$ with $\hat{a}_e > |\s|$, which can cause the first assumption of 
Lemma~\ref{theo:sumTree} to fail. The purpose of the next subsection is to obtain an
improved bound which deals with such a situation.

\subsection{Improved bound} 

Let $\Amin \subset \CCE$ be the subset of those edges $e$ such that
the following two properties hold.
\begin{claim}
\item One has $r_e < 0$.  
\item The element $e_\uparrow \eqdef e_- \wedge e_+ \in T$ is such that 
if $\{u,v\}$ are such that $u \wedge v = e_\uparrow$, then $\{u,v\} = \{e_-,e_+\}$.
\end{claim}
In graphical terms, edges $e \in \Amin$ are those giving rise to the situation
where the subtree of $T$ below $e_\uparrow$ consists only of the node $e_\uparrow$
and the leaves $e_-$ and $e_+$:
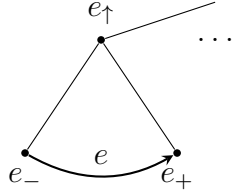
\begin{center}
\begin{tikzpicture}
	\node at (0,0)  [dot] (variable1) [label=below:{$e_-$}] {};
	\node at (2,0)  [dot] (variable2) [label=below:{$e_+$}] {};
	\node at (1,1.5)  [dot] (variable3) [label=above:{$e_\uparrow$}] {};
	\node at (2.5,1.5)  {$\ldots$};

	\draw (variable1) -- (variable3) -- (variable2);
	\draw (variable3) -- (2.5,2);
	\draw[thick,->] (variable1) to[bend right=30] node [above] {$e$} (variable2);
\end{tikzpicture}
\end{center}

We now build a function $\tilde K^{(\n)}$ as follows.
First, given any edge $e =(e_-,e_+)$ and any $r > 0$, we define an operator $\Op_e^r$ acting
on sufficiently smooth functions $V \colon \R^{\CCV} \to \R$ by
\begin{equ}
\bigl(\Op_e^r V\bigr)(x) = V(x) - \sum_{|k|_\s < r} {(x_{e_+} - x_{e_-})^k \over k!} \bigl(D_{e_+}^k V\bigr)(P_e(x))\;,
\end{equ}
where $D_{e_+}$ denotes differentiation with respect to the coordinate $x_{e_+}$ and
the function $P_e \colon \R^{\CCV} \to\R^{\CCV}$ is given by
\begin{equ}
\bigl(P_e(x)\bigr)_v = 
\left\{\begin{array}{cl}
	x_v & \text{if $v \neq e_+$,} \\
	x_{e_-} & \text{otherwise.}
\end{array}\right.
\end{equ}
We then further note that, as an immediate consequence of \eqref{def:kayhat}, the kernel 
$\hat K^{(\n)}$ factors naturally as
\begin{equ}
\hat K^{(\n)}(x) = \hat G^{(\n)}(x) \prod_{e \in \Amin} \hat K^{(\n_e)}_e(x_{e_-},x_{e_+})\;,\qquad
\hat G^{(\n)}(x) = \prod_{e  \not\in \Amin} \hat K^{(\n_e)}_e(x_{e_-},x_{e_+})\;.
\end{equ}
With these notations at hand, and writing $\Amin = \{e^{(1)},\ldots,e^{(k)}\}$ for 
some $k \ge 0$, we then define the kernel $\tilde K^{(\n)}$ by
\begin{equ}[e:defKtilde]
\tilde K^{(\n)}(x) = \bigl(\Op_{e^{(k)}}^{r_{e^{(k)}}}\cdots \Op_{e^{(1)}}^{r_{e^{(1)}}}\hat G^{(\n)}\bigr)(x) \prod_{e \in \Amin} \hat K^{(\n_e)}_e(x_{e_-},x_{e_+})\;.
\end{equ}

We can easily verify that one does indeed have the identity \eqref{eek}
 because $ \hat K^{(\n)}(x)$ and $\tilde K^{(\n)}(x)$ differ by a number of terms
that are all of the form
$
J(x)\, (x_{e_+}-x_{e_-})^k \hat K^{(\n_{e})}_{e}(x_{e_-}-x_{e_+})$
where $e \in \Amin$, $|k|_\s < r_e$, and where $J$  is some smooth 
function depending on $e$ and $k$ that does not depend 
on the variable $x_{e_+}$. Integrating over $x_{e_+}$ and using
the fact that $\hat K^{(\n_{e})}_{e}$ annihilates polynomials of degree less than $r_{e}$ by assumption,
we conclude that \eref{eek} holds as claimed.  

Now define $\tilde \eta(v) = |\s| + \sum_{e \in \hat\CCE} \tilde \eta_e(v)$ where
\begin{equ}[tildeetas]
\tilde \eta_e(v) = \eta_e(v) +  |r_e|  \one_{e \in \Amin} \bigl(\one_{e_\uparrow} (v) - \one_{e_\Uparrow} (v) \bigr)\;.
\end{equ} where $\eta_e(v)$ is given in \eqref{e:defEtaE7}.
Here  $e_\Uparrow \in T^\circ$ denotes
the ancestor of $e_- \wedge e_+$, i.e.\ the element of the form $w \wedge e_-$ with $w \not \in e$ 
which is furthest from the root. Note that there is at least one such $w$ as long as $e\in \Amin$,
since either $0$ or $v_{\star,1}$ is a candidate. (If $e\in \Amin$ contains $0$, it must be $e_-$, since $e_+\neq 0$ by
 the assumption that $e_+=0$ implies $r_e\ge 0$. But then $e_+\neq v_{\star,1}$ since $r_{(0, v_{\star,1})} =0$ by assumption.)

\begin{lemma}
The kernels $\tilde K^{(\n)}$ defined in \eqref{e:defKtilde} satisfy the bound
\begin{equ}[e:wantedBoundK]
\Bigl(\prod_{v \in T^\circ} 2^{-\ell_v |\s|}\Bigr) \sup_x |\tilde K^{(\n)}(x)| \lesssim \prod_{v \in T^\circ} 2^{-\ell_v \tilde \eta_v}\;,
\end{equ}
uniformly over all $\n \in \CN(T,\ell)$.
\end{lemma}

\begin{remark}\label{rem:complete}
Recalling Lemma~\ref{theo:sumTree}, and keeping in mind that the summation over labelled trees with 
vertex set $\CCV$ can be absorbed into a $(\CCV,\CCE)$ dependent constant, we see that the proof of Theorem~\ref{theo:ultimate}
is complete as soon as we show that $\tilde \eta$ does indeed satisfy
the conditions of Lemma~\ref{theo:sumTree}, applied to the binary tree  $T^\circ$, and is such that
\begin{equ}[e:wantedEst]
|\tilde \eta| = |\s| |\CCV_0| - \sum_{e \in \hat\CCE} \hat{a}_e\;.
\end{equ}
\end{remark}

\begin{proof}
Write $\d \Amin$ for the set
of all functions $k \colon \Amin \to \N^d$ with $|k_e|_\s<|r|$ but $|k_e+e_i|_\s \ge |r|$ for some $e_i\in \N^d$ with $|e_i|=1$.
For such a  $k$, we write $\D^k$ for the differential operator
in $(\R^d)^\Amin$ given by
$
\D^k = \prod_{e \in \Amin} D_{x_{e_+}}^{k_e}$.
With these notations at hand, it then follows
from the construction of $\tilde K^{(\n)}$ and the generalized Taylor's formula \cite[Prop~A.1]{Regularity} 
that there are explicitly described positive measures $\CQ_x^{k,e}$ on $\R^d$ with \begin{equ}[e:boundMassQ]
\CQ_x^{k,e}(\R^d) \lesssim \|x_{e_+}-x_{e_-}\|_\s^{|k_e|_\s}\;,
\end{equ}such that 
one has the identity
\begin{equ}[e:defGhat]
\tilde K^{(\n)}(x) = \Bigl(\prod_{e \in \Amin} \hat K_e^{(\n_e)}(x_{e_-},x_{e_+})\Bigr)
\sum_{k \in \d \Amin} \int_{(\R^d)^\Amin} \D^{k} \hat G^{(\n)}(x|y) \prod_{e \in \Amin}\CQ^{k,e}_x(dy_e)\;,
\end{equ}
where we introduced the notation $x|y$ for the element in $(\R^d)^{\CCV_0}$
which is obtained by setting
\begin{equ}
(x|y)_v = 
\left\{\begin{array}{cl}
	y_e & \text{if there is $e \in \Amin$ such that $v = e_-$,} \\
	x_v & \text{otherwise.}
\end{array}\right.
\end{equ}
This definition makes sense thanks to our assumption that there are not multiple edges $e\in A^-$ 
emerging from the same vertex, and because we are using smooth approximations to the distributional
kernels.

Furthermore, it follows similarly to before that if all multiedges $\tilde{e}$ connecting $e_-$ to $e_+$ have $r_{\tilde{e}} \le 0$ then, 
for every such multiindex
$k$,  one has the bound
\begin{equ}
\sup_x \Big|D_{e_{\pm}}^k \prod_{\tilde{e} =(e_-,e_+)} \hat K_{ \tilde{e}}^{(\n_{ \tilde{e}})}(x)\Big| 
\lesssim 2^{\ell_{e_\uparrow} |k|_\s} \prod_{v \in T^\circ} 2^{-\ell_v \eta_e(v)}\;,
\end{equ}where $e=(e_-,e_+)$,
uniformly over $n \in \CN_c(T,\ell)$. If some $r_{\tilde{e}} > 0$ on the other hand, 
one obtains the bound
\begin{equ}
\sup_x \Big|D_{e_{\pm}}^k \prod_{\tilde{e} =(e_-,e_+)}\hat K_{ \tilde{e}}^{(\n_{ \tilde{e}})}(x)\Big| \lesssim \Bigl(2^{\ell_{e_\uparrow} |k|_\s} + 2^{\ell_{e_\pm \wedge 0} |k|_\s}\Bigr) \prod_{v \in T^\circ} 2^{-\ell_v \eta_e(v)}\;.
\end{equ}
Combining this with the bound \eref{e:boundMassQ}, the definition of
$\eta$, and the fact that one has $|k_e|_\s \ge |r_e|$ for every edge $e \in \Amin$,
we conclude that the function $\tilde K^{(\n)}$ satisfies 
\begin{equ}[e:boundKtilde]
\sup_x |\tilde K^{(\n)}(x)| \lesssim \Bigl(\prod_{v \in T^\circ} 2^{-\ell_v \eta(v)}\Bigr)
\Bigl(\prod_{e \in \Amin} 2^{(\ell_{e_\Uparrow}-\ell_{e_\uparrow}) |r_e|} \Bigr)\;,
\end{equ}
which is precisely the required bound.
\end{proof}

\begin{remark}
By the definition of the set of edges $\Amin$,
for every $e \in \Amin$ and every $w \not \in e$, one always the property
that $e_\pm \wedge w < e_\uparrow$, so that the exponent appearing in the
second factor above is always negative. In other words, our choice of the set $\Amin$
guarantees that the bound
\eref{e:boundKtilde} is always an improvement over \eref{e:wantedProdBoundKHat}.
\end{remark}

\subsection{Putting everything together}\label{lastsubsection}

By Remark~\ref{rem:complete} the following lemma, which is the final
statement of this section, completes the proof of Theorem~\ref{theo:ultimate}.

\begin{lemma}
The function $\tilde \eta$ given in \eref{tildeetas} satisfies the identity \eref{e:wantedEst} and
the assumptions of Lemma~\ref{theo:sumTree} (applied to the tree $T$) as well as the identity \eref{e:wantedEst}.
\end{lemma}

\begin{proof}  To verify that assumption~1 of Lemma~\ref{theo:sumTree} holds, 
we choose an arbitrary element $v \in T^\circ$ 
and we consider the set $L_v \subset \CCV$ of all the leaves $u \in T$ with $u\ge v$. 
Note that one always has $|L_v| \ge 2$, and we will treat the case
$|L_v| = 2$ separately.

If $|L_v| = 2$, then there exists an edge $e$ such that $L_v = e$ and 
$v = e_\uparrow$. In this case, assumption~1 of Lemma~\ref{theo:sumTree}
requires that $\tilde \eta(v) > 0$. We have
$
\tilde \eta(v)= |\s| - \hat{a}_e  + |r_e| \one_{r_e <0}
$ which is indeed positive by
 Assumption~\ref{ass:mainGraph}.1.

We now turn to the case $|L_v| > 2$. Since $e_\uparrow > e_\Uparrow$, one
always has $\sum_{u\ge v} \bigl(\one_{e_\uparrow} (u) - \one_{e_\Uparrow} (u)\bigr) \ge 0$. From the definitions \eref{e:defEtaE7} of $\eta$ and \eref{tildeetas} of $\tilde\eta$  we have  $\sum_{u \ge v} \tilde \eta(u)\ge  \sum_{u \ge v} \eta(u)$. 
By checking all cases $e_-\in L_v$,
$e_-\not\in L_v$,  $e_+\in L_v$,
$e_+\not\in L_v$,  $0\in L_v$,
$0\not\in L_v$ separately, we find that $\sum_{e\in\hat\CCE(L_v)} \sum_{u\ge v} \eta_e(u)$ is given by
\begin{equ}
\sum_{e\in\hat\CCE_0(L_v)} -\hat{a}_e 
+ \one_{0\in L_v} \bigg( 
\sum_{e\in\hat\CCE^\downarrow(L_v)\cap\hat\CCE_+(L_v)}  r_e 
+ \sum_{e\in\hat\CCE^\uparrow(L_v)\cap\hat\CCE_+(L_v)} (-\hat{a}_e- r_e +1)\bigg)\;.
\end{equ}
Note the cancellation which appears in the special
case when all three $e_-,e_+,0\in L_v$.
Now points 2 and 3 of Assumption~\ref{ass:mainGraph} with the choice $\bar\CCV=L_v$
imply that $\sum_{u\ge v}\eta(u) >0$, which concludes the proof that assumption~1 holds.
%
%
 
We now turn to the second condition appearing in Lemma~\ref{theo:sumTree}.
In our case, we choose for the distinguished node $v_\star$ the most recent
common ancestor between the elements of $\CCV_\star$.
The reason for this choice is that this node encodes the 
largest scale appearing in the multiscale clustering which is still guaranteed to be smaller than 
the scale $\lambda$ fixed by the test function.
We then fix an arbitrary node $v \in T^\circ$ such that $v_\star \ge v$.
Denoting by $U_v = \{u \in T^\circ\,:\, u \not \ge v\}$,
the situation is the following, where $U_v$ contains all the nodes lying in the shaded region: 
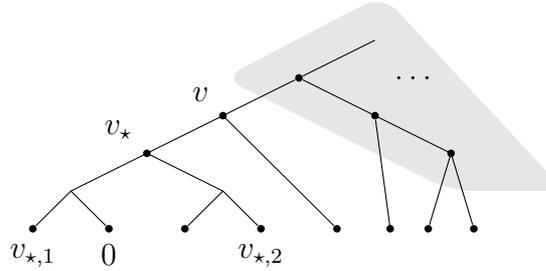
\begin{center}
\begin{tikzpicture}
	\node at (0,0)  [dot] (variable1) [label=below:{$v_{\star,1}$}] {};
	\node at (1,0)  [dot] (variable2) [label=below:{$0$}] {};
	\node at (2,0)  [dot] (variable3)  {};
	\node at (3,0)  [dot] (variable4) [label=below:{$v_{\star,2}$}] {};
	\node at (1.5,1)  [dot] (vv) [label=above left:{$v_{\star}$}] {};
	\node at (2.5,1.5)  [dot] (uu) [label=above left:{$v$}] {};
	\node at (3.5,2)  [dot] (root)  {};
	\node at (5.8,0)  [dot] (l1) {};
	\node at (5.2,0)  [dot] (l2) {};
	\node at (4,0)  [dot] (l3) {};
	\node at (4.7,0)  [dot] (l4) {};
	\node at (5.5,1)  [dot] (bif1) {};
	\node at (4.5,1.5)  [dot] (bif2) {};
	\node at (5,2)  {$\ldots$};
	
	\fill[rounded corners=8pt,very nearly transparent,fill=black] (4.6,3.1) -- (2.5,2) -- (5.5,0.5) -- (7,0.5) -- cycle;

	\draw (variable1) -- (0.5,0.5) -- (variable2);
	\draw (0.5,0.5) -- (1.5,1) -- (2.5,0.5) -- (variable4);
	\draw (2.5,0.5) -- (variable3);

	\draw (root) -- (4.5,2.5);
	\draw (bif2) -- (4.7,0);
	\draw (1.5,1) -- (root) -- (5.5,1) -- (l1);
	\draw (5.5,1) -- (l2);
	\draw (2.5,1.5) -- (l3);
\end{tikzpicture}
\end{center}

Note again that similarly to before, one has
$\sum_{u \in U_v} \tilde \eta(u) \le \sum_{u \in U_v} \eta(u)$,
so that we can restrict ourselves to the verification of the second 
condition for the function $\eta$. Denoting by $\bar \CCV$ the set
of leaves attached to $U_v$, one has
$\bar \CCV \subset \CCV \setminus \CCV_\star$.   
By checking the three cases directly we have 
\begin{equs}
&\sum_{e} \sum_{u\in U_v}\eta_e(u) =-
\sum_{e\in \hat\CCE\setminus \hat\CCE^{\downarrow} }  \hat{a}_e - \sum_{e\in \hat\CCE^\uparrow\cap \hat\CCE_+}  r_e  
+ \sum_{e\in \hat\CCE^\downarrow\cap \hat\CCE_+}  (r_e  -1 )
\end{equs}
with the obvious notation that  $\hat\CCE=\hat\CCE(\bar\CCV)$, $\hat\CCE^\uparrow=\hat\CCE^{\uparrow} (\bar\CCV)$, etc.
Furthermore, 
the cardinality of $U_v$ is exactly equal to $\bar \CCV$ so we have 
\begin{equ}
\sum_{u \in U_v} \eta(u) = |\s|\,|\bar \CCV| -
\sum_{e\in \hat\CCE\setminus \hat\CCE^{\downarrow}}  \hat{a}_e - \sum_{e\in \hat\CCE^\uparrow\cap\hat\CCE_+}  r_e  
+ \sum_{e\in \hat\CCE^\downarrow\cap\hat\CCE_+}  (r_e  -1 )
\;,
\end{equ}
so that 
the condition $\sum_{u \in U_v} \eta(u) < 0$ is satisfied as a consequence
of Assumption~\ref{ass:mainGraph}.4. 

Finally to see that it satisfies the identity \eref{e:wantedEst}, note that similar to before,  termwise cancellation gives us $\sum_{v\in T}\tilde \eta_e(v) =-\hat{a}_e$.  Hence $|\eta|= \sum_{v\in T^\circ} (|\s| + \sum_{e\in \hat\CCE} \tilde\eta_e(v) = |\s||T^\circ| -\sum_{e\in \hat\CCE} \hat{a}_e$.  
Since $T^\circ$ is  a binary tree, $|T^\circ| = \#$ of leaves $ -1= |\CCV_0|$.
\end{proof}

\section{Notes on renormalisation}

Recall that given a map $M \colon \CT_\ex \to \CT_\ex$ as in 
Section~\ref{sec:renormOp}, the map $\hDeltaM \colon \CT_+ \to \CT_+ \otimes \CT_+$ 
is uniquely defined by the relations
\begin{equs}[e:defDeltaM]
\bigl(\CA\hat M \CA \otimes \hat M\bigr) \Deltap &= (1 \otimes \CM) (\Deltap \otimes 1)\hDeltaM\;,\\
(M \otimes \hat M)\Delta &= (1\otimes \CM)(\Delta \otimes 1)\DeltaM\;,\\
\hat M\J_k &= \CM \bigl(\J_k \otimes 1\bigr) \DeltaM \CQ_{>|k|}\;,
\end{equs}
where $\CA \colon \CT_+ \to \CT_+$ is the antipode of the Hopf algebra $\CT_+$ defined
as in \cite[Thm~8.16]{Regularity}, and $\Deltap$ is its coproduct given in \cite[Equ.~8.9]{Regularity} and $\CQ_{>\alpha}$ projects onto elements of homogeneity 
greater than $\alpha - 2$ (the number $2$ being the gain of homogeneity given by $\J$). 
In the sequel, we also denote by $\CQ_{\ge\alpha}$ the projection
onto elements of homogeneity at least $\alpha - 2$, so that 
$\CQ_{\ge\alpha} + \CQ_{<\alpha} = I$. The motivation for the definitions \eqref{e:defDeltaM}
is that if $(\Pi,\Gamma)$ is an admissible model and one defines $(\Pi^M,\Gamma^M)$ by
\begin{equ}[e:renormModel]
\Pi_x^M = (\Pi_x \otimes f_x)\DeltaM\;,\qquad
\gamma_{xy}^M = (\gamma_{xy} \otimes f_y)\hDeltaM\;,
\end{equ}
(with $\Gamma_{xy} \tau = (1\otimes\gamma_{xy})\Delta \tau$ and similarly for $\Gamma_{xy}^M$)
then $(\Pi^M,\Gamma^M)$ does satisfy all the algebraic identities required 
for an admissible model.

In \cite{Regularity}, the renormalisation group associated to a regularity structure generated by
noises, products and abstract integration maps was defined as the set of maps $M$ 
preserving the noises and $\one$, commuting with the abstract integration maps
and multiplication by $X^k$, such that furthermore both $\DeltaM$ and $\hDeltaM$
are ``upper triangular'' in the sense that, if we write
\begin{equ}
\DeltaM \tau = \tau^{(1)}\otimes \tau^{(2)}\;,\qquad
\hDeltaM \bar \tau = \bar \tau^{(1)}\otimes \bar \tau^{(2)}\;,
\end{equ}
with an implicit summation suppressed in the notation, then one
has $|\tau^{(1)}| \ge |\tau|$ and $|\bar \tau^{(1)}| \ge |\bar \tau|$.
This property was absolutely crucial since this is what guarantees that 
if we use $\DeltaM$ and $\hDeltaM$ to renormalise a model as in \eqref{e:renormModel}, 
then $(\Pi^M,\Gamma^M)$ also satisfies the \textit{analytical} properties required
to be a model. In this section, we show that one only ever needs to
verify that $\DeltaM$ is upper triangular, as this then automatically implies the
same for $\hDeltaM$. 

Throughout this section, we consider a general regularity structure generated
by a number of ``noise symbols'' $\Xi$, a multiplication operation, as well
as a number of abstract integration operators. In other words, every basis
vector of $\CT$ is assumed to be generated from the vectors $\Xi_i$, 
$X_i$ or $\one$ by multiplication and / or abstract integration. 
The structure considered in this
article is of this type since $\CE^k$ can be considered as an integration
operator of order $|k|$. Our main result can be summarised as follows.

\begin{theorem}\label{theo:triangular}
Let $(\CT,\CG)$ be a regularity structure as above and let $M\colon \CT \to \CT$ be
a linear map preserving $\Xi_i$, $X^k$, and commuting with the abstract 
integration maps and with multiplication by $X^k$. Let $\DeltaM$ and $\hDeltaM$
be given by \eqref{e:defDeltaM}. If $\DeltaM$ is upper triangular, then so is $\hDeltaM$.
\end{theorem}

%
%

In order to prove Theorem~\ref{theo:triangular}, we first derive a number of identities
involving the operators $\DeltaM$ and $\hDeltaM$.
We  first note that a simple calculation using the coassociativity of $\Deltap$ and the 
properties of the antipode $\CA$ yields
\begin{equ}
\bigl((1 \otimes \CM) (\Deltap \otimes 1)\bigr)^{-1} = (1 \otimes \CM)(1\otimes \CA \otimes 1) (\Deltap \otimes 1)\;,
\end{equ}
so that the first identity can be rewritten somewhat more explicitly as
\begin{equ}[e:hatDM]
\hDeltaM = (1 \otimes \CM)\bigl((1\otimes \CA)\Deltap\CA\hat M \CA \otimes \hat M\bigr) \Deltap\;.
\end{equ}
Similarly, the second identity is equivalent to
\begin{equ}[e:DM]
\DeltaM = (1 \otimes \CM)\bigl((1\otimes \CA)\Delta M \otimes \hat M\bigr) \Delta\;.
\end{equ}
Throughout this section, we will make use of the following notation. Given
a map $\sigma\colon \{1,\ldots,n\} \to \{1,\ldots,k\}$, we write
$\CM^\sigma$ for the map
\begin{equ}
\CM^\sigma \colon \bigotimes_{j=1}^n \tau_j \mapsto 
\bigotimes_{i=1}^k \Bigl(\prod_{j \in \sigma^{-1}(i)} \tau_j\Bigr)\;.
\end{equ}
For a surjection $\sigma$, we also use the notation
$\sigma = (\sigma^{-1}(1))\cdots (\sigma^{-1}(k))$, so that for example
\begin{equ}
\CM^{(2)(1,4)(2,5)}(\tau_1 \otimes \cdots\otimes \tau_5)
= \tau_2 \otimes (\tau_1 \tau_4) \otimes (\tau_2\tau_5)\;.
\end{equ}
It will also sometimes be convenient to use for the above example the 
alternative notation
\begin{equ}
\CM^{(2)(1,4)(2,5)} = \CM^{(2)}\otimes \CM^{(1,4)} \otimes \CM^{(2,5)}\;.
\end{equ}

Recall also that the antipode $\CA$ is automatically an antihomomorphism 
of coalgebras \cite{Sweedler}, so that 
\begin{equ}
\Deltap \CA = \CM^{(2)(1)}(\CA \otimes \CA)\Deltap\;.
\end{equ}
With all of these notations at hand, we first claim that on has the following.

\begin{lemma}\label{lem:MMM}
The identity
\begin{equ}[e:MMM]
\CM^{(1,3)(2)}(\hat M \otimes \hDeltaM \CA)\Deltap
= (1\otimes \CA) \Deltap \hat M\;,
\end{equ}
holds true.
\end{lemma}

\begin{proof}
In view of \eqref{e:renormModel}, it is natural to test both sides of \eqref{e:MMM}
against $f_y \otimes \gamma_{xy}$. Since $\gamma_{xy} = f_x^{-1} \circ f_y$, the right hand
side is then equal to $(f_y \circ \gamma_{xy}^{-1})\hat M = f_x \hat M = f_x^M$.
The left hand side on the other hand is equal to 
\begin{equ}
(f_y \otimes \gamma_{xy} \otimes f_y)(\hat M \otimes \hDeltaM \CA)\Deltap
= (f_y^M \otimes (\gamma_{xy}^M)^{-1})\Deltap = f_y^M \circ (\gamma_{xy}^M)^{-1} = f_x^M\;,
\end{equ}
as required since $f_x$ and $f_y$ are arbitrary multiplicative functionals.
More directly, it follows from \eqref{e:hatDM} that
\begin{equs}
\CM^{(1,3)(2)}&\bigl(\hat M \otimes \hDeltaM \CA\bigr)\Deltap \\
& = 
\CM^{(1,3,4)(2)}\bigl(\hat M \otimes \bigl((1\otimes \CA)\Deltap \CA \hat M \CA \otimes \hat M\bigr)\Deltap \CA\bigr)\Deltap \\
& = 
\CM^{(1,2,4)(3)}\bigl(\hat M \otimes \bigl(\hat M \CA \otimes (1\otimes \CA)\Deltap \CA \hat M\bigr)\Deltap \bigr)\Deltap \\
& = 
\CM^{(1,2,4)(3)}\bigl(\bigl(\hat M \otimes \hat M \CA\bigr)\Deltap \otimes (1\otimes \CA)\Deltap \CA \hat M\bigr)\Deltap \\
& = 
\CM^{(2)(1)}(1\otimes \CA)\Deltap \CA \hat M
 = (1 \otimes \CA)\Deltap \hat M\;,
\end{equs}
as required.
\end{proof}

\begin{lemma}\label{lem:silly2}
One has the identity
\begin{equ}
\bigl((1\otimes \CA)\Delta\otimes 1\bigr)\DeltaM
=
\CM^{(1)(3,4)(2,5)}
\bigl((1\otimes \Deltap)\DeltaM \otimes (\CA \otimes 1)\hDeltaM\bigr)\Delta\;.
\end{equ}
\end{lemma}

\begin{proof}
It follows from the definitions of $\DeltaM$ and $\hDeltaM$ that the right 
hand side is given by
\begin{equs}
{}&\CM^{(1)(3,4)(2,5)}
\bigl((1\otimes \Deltap)\DeltaM \otimes (\CA \otimes 1)\hDeltaM\bigr)\Delta \\
&= \CM^{(1)(3,4)(2,5,6)}
\Big((1\otimes \Deltap\CM)\bigl((1\otimes \CA)\Delta M \otimes \hat M\bigr)\Delta \\
&\qquad \otimes \bigl((\CA \otimes \CA)\Deltap \CA \hat M \CA \otimes \hat M\bigr)\Deltap\Big)\Delta \\
&= \CM^{(1)(3,5)(2,4,6)}
\bigl((1\otimes \Deltap\CM)\bigl((1\otimes \CA)\Delta M \otimes \hat M\bigr)\Delta \otimes \bigl(\Deltap \hat M \CA \otimes \hat M\bigr)\Deltap\bigr)\Delta \\
&= \CM^{(1)(3,5,7)(2,4,6,8)}
\bigl(\bigl((1\otimes \Deltap\CA)\Delta M \otimes \Deltap\hat M\bigr)\Delta \otimes \bigl(\Deltap \hat M \CA \otimes \hat M\bigr)\Deltap\bigr)\Delta \\
&= \CM^{(1)(3,5,7)(2,4,6,8)}
\bigl((1\otimes \Deltap\CA)\Delta M \otimes \Deltap\hat M \otimes \Deltap \hat M \CA \otimes \hat M\bigr)(\Delta \otimes \Deltap)\Delta \\
&= \CM^{(1)(3,5,7)(2,4,6,8)}
\bigl((1\otimes \Deltap\CA)\Delta M \otimes (\Deltap\hat M \otimes \Deltap \hat M \CA)\Deltap \otimes \hat M\bigr)(1 \otimes \Deltap)\Delta \\
&= \CM^{(1)(3,5)(2,4,6)}
\bigl((1\otimes \Deltap\CA)\Delta M \otimes \Deltap\hat M \CM (1 \otimes \CA)\Deltap \otimes \hat M\bigr)(1 \otimes \Deltap)\Delta \\
&= \CM^{(1)(3)(2,4)}
\bigl((1\otimes \Deltap\CA)\Delta M \otimes \hat M\bigr)\Delta  \\
&= (1\otimes 1 \otimes \CM)
\bigl((1\otimes (\CA \otimes \CA)\Deltap)\Delta M \otimes \hat M\bigr)\Delta \\
&= (1\otimes 1 \otimes \CM)
\bigl(((1\otimes \CA)\Delta \otimes \CA)\Delta M \otimes \hat M\bigr)\Delta \\
&= ((1\otimes \CA)\Delta \otimes 1)(1 \otimes \CM)
\bigl((1 \otimes \CA)\Delta M \otimes \hat M\bigr)\Delta
= ((1\otimes \CA)\Delta \otimes 1) \DeltaM
\end{equs}
as required.
\end{proof}

\begin{lemma}\label{lem:silly}
One has the identity
\begin{equ}[e:silly]
\J_k(\tau)\otimes \one - \one \otimes \J_k(\tau) 
= \sum_\ell \Bigl(1 \otimes {(-X)^\ell \over \ell!} \CM\Bigr)\Bigl((1\otimes \CA)\Deltap \J_{k+\ell} \otimes 1\Bigr)\Delta\tau\;.
\end{equ}
\end{lemma}

\begin{proof}
It follows from the recursive definition of $\Deltap$ that
\begin{equs}
\sum_\ell &\Bigl(1 \otimes {(-X)^\ell \over \ell!} \CM\Bigr)\Bigl((1\otimes \CA)\Deltap \J_{k+\ell} \otimes 1\Bigr)\Delta\tau \\
&= 
\sum_{\ell,m} \Bigl(1 \otimes {(-X)^\ell \over \ell!} \CM\Bigr)\Bigl(\Bigl(\J_{k+\ell+m}\otimes {X^m \over m!}\CA\Bigr)\Delta  \otimes 1\Bigr)\Delta\tau \\
&\qquad + \sum_\ell \one \otimes {(-X)^\ell \over \ell!} \CM \bigl(\CA \J_{k+\ell} \otimes 1\bigr)\Delta \tau\;.
\end{equs}
It now follows from the defining property of $\CA$, 
followed by the binomial identity and the
comodule property of $\Delta$ and $\Deltap$ (see the statement 
and proof of \cite[Thm~8.16]{Regularity}) that  
\begin{equs}
\sum_\ell {(-X)^\ell \over \ell!} \CM &\bigl(\CA \J_{k+\ell} \otimes 1\bigr)\Delta \tau \\
&=
- \sum_{\ell,m} {(-X)^\ell \over \ell!} \CM \bigl(\CM\bigl(\J_{k+\ell+m} \otimes {X^m \over m!}\CA\bigr)\Delta \otimes 1\bigr)\Delta \tau \\
&=
- \CM \bigl(\CM\bigl(\J_{k} \otimes \CA\bigr)\Delta \otimes 1\bigr)\Delta \tau \\
&=
- \CM \bigl(\CM\bigl(\J_{k} \otimes \CA\bigr) \otimes 1\bigr) (1 \otimes \Deltap\bigr)\Delta \tau \\
&=
- \CM \bigl(\J_{k} \otimes \CM(\CA \otimes 1)\Deltap\bigr) \Delta \tau
=
- \J_{k}(\tau)\;.
\end{equs}
Here, we used the fact that $\CM(\CA \otimes 1)\Deltap = \one \one^*$ and $(1\otimes \one^*)\Delta \tau = \tau$.
Similarly, it follows from the binomial identity followed by the comodule property that
\begin{equs}
\sum_{\ell,m} \Bigl(1 \otimes {(-X)^\ell \over \ell!} \CM\Bigr)&\Bigl(\Bigl(\J_{k+\ell+m}\otimes {X^m \over m!}\CA\Bigr)\Delta  \otimes 1\Bigr)\Delta\tau \\
&= \bigl(1 \otimes \CM\bigr)\bigl((\J_{k}\otimes \CA)\Delta  \otimes 1\bigr)\Delta\tau\\
&= \bigl(\J_k \otimes \CM\bigl(\CA \otimes 1)\Deltap \bigr)\Delta\tau = \J_k(\tau) \otimes \one\;,
\end{equs}
which concludes the proof of \eqref{e:silly}.
\end{proof}

We now have all the ingredients required to obtain a recursive 
characterisation of $\hDeltaM$ from which Theorem~\ref{theo:triangular} can
then easily be derived. 

\begin{proposition}\label{prop:ridiculous}
The map $\hDeltaM$ satisfies the identity
\begin{equs}
\hDeltaM \J_k(\tau) &= \bigl(\J_k \otimes 1\bigr) \DeltaM \tau \label{e:ridiculous}
 - \sum_\ell \Bigl({X^\ell\over \ell!} \CM^{(2,3)}\otimes \CM^{(1,4)}\Bigr)  \\
 &\qquad\times \Bigl(\bigl(1\otimes \CA\bigr)\Deltap \CM \bigl(\J_{k+\ell} \otimes 1\bigr)\DeltaM \CQ_{\le|k+\ell|} \otimes \hDeltaM\Bigr)\Delta \tau\;.
\end{equs}
\end{proposition}

\begin{proof}
We apply the ``swapping'' operator $\CM^{(2)(1)}\colon \tau \otimes \bar \tau \mapsto \bar \tau \otimes \tau$
to \eqref{e:silly} and then apply the map $(1\otimes \CM)(\hDeltaM \otimes \hat M)$ to both
sides. This yields the identity
\begin{equs}
\hDeltaM \J_k \tau &= \one \otimes \hat M \J_k \tau \\
&\quad + \sum_\ell \Bigl({(-X)^\ell \over \ell!} \otimes \CM\Bigr) \bigl(\hDeltaM \CM^{(2,3)}\otimes \CM^{(1)}\bigr) \bigl((\hat M \otimes \CA)\Deltap \J_{k+\ell} \otimes 1\bigr)\Delta \tau \\
&=  \one \otimes \hat M \J_k \tau \\
&\quad + \sum_\ell \Bigl({(-X)^\ell \over \ell!}\CM^{(2,4)} \otimes \CM^{(1,3,5)}\Bigr) \bigl((\hat M \otimes \hDeltaM \CA)\Deltap \J_{k+\ell} \otimes \hDeltaM\bigr)\Delta \tau \\
&=  \one \otimes \hat M \J_k \tau \label{e:basic}\\
&\quad + \sum_\ell \Bigl({(-X)^\ell \over \ell!}\CM^{(2,3)} \otimes \CM^{(1,4)}\Bigr) \bigl((1 \otimes \CA)\Deltap \hat M \J_{k+\ell} \otimes \hDeltaM\bigr)\Delta \tau\;,
\end{equs}
where we made use of Lemma~\ref{lem:MMM} to obtain the last identity.
We furthermore use the definition of $\hat M$ which leads to the identity
\begin{equ}[e:iden1']
\bigl((1 \otimes \CA)\Deltap \hat M \J_{k+\ell} \otimes \hDeltaM\bigr)\Delta \tau
= 
\bigl((1 \otimes \CA)\Deltap \CM (\J_{k+\ell}  \otimes 1)\DeltaM \CQ_{>|k+\ell|} \otimes \hDeltaM\bigr)\Delta \tau
\end{equ}

Noting that
\begin{equ}[e:iden2]
\one \otimes \hat M \J_k(\tau) = \one \otimes \CM (\J_k \otimes 1)\DeltaM \tau\;,
\end{equ}
we then make use of Lemma~\ref{lem:silly} which yields
\begin{equs}
\one \otimes \CM (\J_k \otimes 1)\DeltaM \tau &= 
(\J_k \otimes 1)\DeltaM \tau - \sum_\ell \Bigl({(-X)^\ell \over \ell!} \CM^{(2,3)} \otimes \CM^{(1,4)}\Bigr)\\
&\qquad \times\Bigl(\bigl((1\otimes \CA)\Deltap \J_{k+\ell} \otimes 1\bigr)\Delta \otimes 1\Bigr)\DeltaM\tau\;. \label{e:smallIden}
\end{equs}
(Here we used the fact that the left hand side, and therefore also the right hand side,
of \eqref{e:silly} is symmetric under the map $\tau \otimes \bar \tau \mapsto \bar \tau \otimes \tau$.)
At this stage we use the fact that, thanks to Lemma~\ref{lem:silly2}, one has
\begin{equ}
\CM^{(1)(3,4)(2,5)}
\bigl((1\otimes (1\otimes \CA)\Deltap)\DeltaM \otimes \hDeltaM\bigr)\Delta
= \bigl(\Delta\otimes 1\bigr)\DeltaM \;.
\end{equ}
(just compose both sides with $1\otimes \CA \otimes 1$), which then yields the identity
\begin{equs}
\CM^{(2,3)(1,4)} &\bigl((1\otimes \CA)\Deltap \CM \bigl(\J_k \otimes 1\bigr)\DeltaM \otimes \hDeltaM\bigr)\Delta \\
&= 
\CM^{(2,4,5)(1,3,6)} \bigl(\bigl((1\otimes \CA)\Deltap\J_k \otimes (1\otimes \CA)\Deltap\bigr)\DeltaM \otimes \hDeltaM\bigr)\Delta \\
&= \CM^{(2,3)(1,4)}\bigl((1\otimes \CA)\Deltap\J_k \otimes 1\otimes 1\bigr)\\
&\qquad \times \CM^{(1)(3,4)(2,5)} \bigl(\bigl(1 \otimes (1\otimes \CA)\Deltap\bigr)\DeltaM \otimes \hDeltaM\bigr)\Delta \\
&= \CM^{(2,3)(1,4)}\bigl(\bigl((1\otimes \CA)\Deltap\J_k \otimes 1\bigr)\Delta\otimes 1\bigr)\DeltaM\;.
\end{equs}
Combining this with \eqref{e:smallIden}, \eqref{e:iden2}, \eqref{e:iden1'} and \eqref{e:basic}
finally leads to the required identity.
\end{proof}

We can now finally turn to the proof of Theorem~\ref{theo:triangular}.

\begin{proof}[of Theorem~\ref{theo:triangular}]
We proceed by induction. Assume that the statement holds for all the elements in $\CT_+$ appearing in the
description of $\Delta\tau$, then we claim that the statement also holds for $\J_k(\tau)$
as well as for  $\EE_k^\ell(\tau)$.
Since the algebraic properties of $\EE_k^\ell$ are the same as those of
$\J_k$, we only consider the latter.
For the first term in \eqref{e:ridiculous}, this follows from the upper triangular structure of $\DeltaM$.
Regarding the second term, it follows from the induction hypothesis that the quantity
\begin{equ}
\bigl(\CQ_{<|k+\ell|} \otimes \hDeltaM\bigr)\Delta \tau\;,
\end{equ}
is necessarily a linear combination of expressions of the form $\tau^{(1)} \otimes \tau^{(2)} \otimes \tau^{(3)}$
with $|\tau^{(1)}| + |\tau^{(2)}| \ge |\tau|$ and $|\tau^{(1)}| < |k+\ell|-2$. In particular, one
has $|\tau^{(2)}| \ge |\tau| + 2 - |k+\ell|$. 
It now suffices to note that, with $\tau^{(i)}$ as just defined for any fixed $\ell$, the 
second term in \eqref{e:ridiculous} always consists of linear combinations of terms of the form
\begin{equ}
X^\ell \tau^{(2)} \sigma^{(1)} \otimes \sigma^{(2)}\tau^{(3)}\;,
\end{equ}
with the $\tau^{(i)}$ as above and some $\sigma^{(i)}$ in $\CT_+$.
Since the $\sigma^{(i)}$ belong to $\CT_+$, they have positive homogeneity, so that the homogeneity
of the first factor is at least $|\tau^{(2)}| + |\ell| = |\tau| + 2 - |k+\ell| + |\ell| = |\J_k(\tau)|$,
thus concluding the proof.
\end{proof}

\section{Symbolic index}

In this appendix, we collect some of the most used symbols of the article, together
with their meaning and the page where they were first introduced.

\begin{center}
\renewcommand{\arraystretch}{1.1}
\begin{tabular}{lll}\toprule
Symbol & Meaning & Page\\
\midrule
$|z| = \|z\|_\s$ & Parabolic distance & \pageref{standing}\\
$|k|=|k|_\s$ & Parabolic length of a multi-index & \pageref{parlm} \\
$\xi$ & Space-time white noise & 2\\
$F$ & Nonlinearity & \pageref{e:growth}\\
$\CU$ & Symbols used to describe solution $h$ & \pageref{e:rel}\\
$\CU'$ & Symbols used to describe derivative $\partial_x h$ & \pageref{e:rel}\\
$\CV$ & Symbols used to describe right hand side & \pageref{e:rel}\\
$\CV_{\ell,k}$ &  & \pageref{struck}\\
$\CW$ & Symbols used to describe equation & \pageref{e:rel}\\
$\CW_+$ & Symbols used to describe structure group & \pageref{struck}\\
$\CT$ & Linear span of $\CW$ & \pageref{e:rel}\\
$\Xi$ & Abstract symbol for noise & \pageref{e:rel}\\
$X=(X_0,X_1)$ & Abstract symbol for $z=(t,x)$ &  \pageref{e:rel}\\
$|\tau|$ & Homogeneity of $\tau\in \CT$ & \pageref{e:rel}\\
$\CI, \CI'$ & Abstract integration maps & \pageref{e:rel}\\
$\CE^k$ & Abstract symbol of multiplication by $\eps^k$ on $\CT_\ex$ & \pageref{e:recursion}\\
$\CG$ & Structure group & \pageref{defofCG}\\
$\hat\CE^k$ & Abstract symbol on $\CD^\gamma$ & \pageref{e:defEpshat}\\
$ \J_\ell(\tau)$ & Formal symbol representing $D^{(\ell)}K \star \Pi_z\tau)(z)$ & \pageref{e:genT+}\\
$\EE_\ell^{k}(\tau)$ & Formal symbol representing $\eps^kD^{(\ell)}\Pi_z\tau)(z)$ & \pageref{e:genT+}\\
$\CT_+$ & Free commutative algebra generated by $X$, $ \J_\ell(\tau)$, $\EE_\ell^{k}(\tau)$ & \pageref{e:genT+}\\
$\bar\CT$ & Subspace of $\CT$ generated by powers of $X$ & \pageref{e:genT+}\\
$\Delta$ &  Linear map $\CT \to \CT\otimes \CT_+$ used to build structure group  & \pageref{Delta}\\
$\DD$ &  Abstract spatial derivative &  \pageref{DD}\\
$\CG_+$& Multiplicative linear functionals on $\CT_+$ & \pageref{defofCGplus}\\
\bottomrule
\end{tabular}

\newpage
\begin{tabular}{lll}\toprule
Symbol & Meaning & Page\\
\midrule
$\CR$ & Reconstruction operator & \pageref{reco}\\
$\CT_\ex$ & Extended regularity structure & \pageref{sec:extended}\\
$\CU'_\ex$ & Symbols used to describe powers of $\d_x h$ & \pageref{defofUprimeex}\\
$\bar \CW_+$&&  \pageref{defofwbarplus} \\
 $\CW_\ex$ & & \pageref{defofwex}\\
$\CB$ & Smooth compactly supported functions $\phi \colon \R^2 \to \R$ & \pageref{eq:models}\\
$\MM$ & Admissible models &  \pageref{admissible}\\
$\gamma_{z\bar z}$ & &  \pageref{defofgammazbarz}\\
$\$\Pi\$$, $\$\Pi;\bar \Pi\$$ & Norm on models  &  \pageref{triple}\\
$|\tau|_\alpha$ & (Euclidean) norm of projection of $\tau$ onto $\CT_\alpha$ & \pageref{eq:distH0}\\
$\|H\|_\gamma$,  $\|H;\bar H\|_\gamma $ & (H\"older) norm on $\CD_\gamma$ & 
\pageref{e:distH}\\
$\LL_\eps(\zeta)$ &  Lift of smooth $\zeta$ to $\CT_\ex$ & 
\pageref{sec:canonical}\\
$ \CS'$ & Dual of Schwartz space &\pageref{defofschwartz}
\\
$\CT_{<\gamma}$& Elements of $\CT$ of homogeneity strictly less than $\gamma$& \pageref{defofCTsubleg}\\
$\MM_\eps$ & $\eps$-models  & \pageref{epsilonmodels}\\
$\CK$ & Operator on $\CD^\gamma$ defined from kernel $K$    & \pageref{eq:integrationop}\\
$\CP$ & Operator defined from the heat kernel & \pageref{eq:integrationop}\\
$\MM_\ex$& Admissible models on $\CT_\ex$ & \pageref{cemex}\\
$\CD^{\gamma,\eta}$, $ \|\cdot\|_{\gamma,\eta}$ & Modelled distributions with blow-up at $t=0$ &
\pageref{def:cdgammaeta} \\
 $\CD^{\gamma,\eta}_\eps$, $\|\cdot\|_{\gamma,\eta;\eps}$ & Modelled distributions with short range $\eps$ &
\pageref{def:cep} \\
$\Wick$ &Wick renormalization &\pageref{defWick}\\
$\H_k$& $k$th Hermite polynomial &\pageref{defofhermite}  \\
$\RWick$ & Reconstruction operator for Wick renormalized model & \pageref{rwick} \\
$ C^\W$ & Wick renormalization constant &\pageref{seew}\\
$N_\eps$& Kernel appearing in the model bounds &\pageref{def_of_Neps}\\
$\Ren $ &Kernel renormalization& \pageref{defofren}\\
$Q_\eps^{(m)}$ & Kernel appearing in the model bounds &\pageref{defofqeps} \\
$\CQ_{\le 0}$ & Projection onto $\bigoplus_{\alpha \le 0}\CT_\alpha$ in $\CT_\ex$& \pageref{defofCQ}\\
\bottomrule
\end{tabular}

\end{center}

\endappendix

\bibliographystyle{Martin}
\bibliography{refs}

\def\cprime{$'$} \def\polhk#1{\setbox0=\hbox{#1}{\ooalign{\hidewidth
  \lower1.5ex\hbox{`}\hidewidth\crcr\unhbox0}}}
\begin{thebibliography}{BPRS93}
\expandafter\ifx\csname url\endcsname\relax
  \def\url#1{\texttt{#1}}\fi
\expandafter\ifx\csname urlprefix\endcsname\relax\def\urlprefix{URL }\fi
\expandafter\ifx\csname href\endcsname\relax
  \def\href#1#2{#2}\fi
\expandafter\ifx\csname burlalt\endcsname\relax
  \def\burlalt#1#2{\href{#2}{\texttt{#1}}}\fi

\bibitem[ACQ11]{MR2796514}
\textsc{G.~Amir}, \textsc{I.~Corwin}, and \textsc{J.~Quastel}.
\newblock Probability distribution of the free energy of the continuum directed
  random polymer in {$1+1$} dimensions.
\newblock \emph{Comm. Pure Appl. Math.} \textbf{64}, no.~4, (2011), 466--537.
\newblock \burlalt{arXiv:1003.0443}{http://arxiv.org/abs/1003.0443}.
\newblock \burlalt{doi:10.1002/cpa.20347}{http://dx.doi.org/10.1002/cpa.20347}.

\bibitem[AKQ10]{AKQ}
\textsc{T.~{Alberts}}, \textsc{K.~{Khanin}}, and \textsc{J.~{Quastel}}.
\newblock {Intermediate Disorder Regime for Directed Polymers in Dimension
  1+1}.
\newblock \emph{Physical Review Letters} \textbf{105}, no.~9, (2010), 090603.
\newblock
  \burlalt{doi:10.1103/PhysRevLett.105.090603}{http://dx.doi.org/10.1103/PhysRevLett.105.090603}.

\bibitem[BA07]{ben-artzi}
\textsc{M.~Ben-Artzi}.
\newblock Lectures on {V}iscous {H}amilton-{J}acobi {E}quations, 2007.
\newblock \urlprefix\url{http://www.ma.huji.ac.il/~mbartzi/}.

\bibitem[BC14]{MR3152785}
\textsc{A.~Borodin} and \textsc{I.~Corwin}.
\newblock Macdonald processes.
\newblock \emph{Probab. Theory Related Fields} \textbf{158}, no. 1-2, (2014),
  225--400.
\newblock
  \burlalt{doi:10.1007/s00440-013-0482-3}{http://dx.doi.org/10.1007/s00440-013-0482-3}.

\bibitem[BG97]{MR1462228}
\textsc{L.~Bertini} and \textsc{G.~Giacomin}.
\newblock Stochastic {B}urgers and {KPZ} equations from particle systems.
\newblock \emph{Comm. Math. Phys.} \textbf{183}, no.~3, (1997), 571--607.
\newblock
  \burlalt{doi:10.1007/s002200050044}{http://dx.doi.org/10.1007/s002200050044}.

\bibitem[BP57]{BP}
\textsc{N.~N. Bogoliubow} and \textsc{O.~S. Parasiuk}.
\newblock \"{U}ber die {M}ultiplikation der {K}ausalfunktionen in der
  {Q}uantentheorie der {F}elder.
\newblock \emph{Acta Math.} \textbf{97}, (1957), 227--266.
\newblock
  \burlalt{doi:10.1007/BF02392399}{http://dx.doi.org/10.1007/BF02392399}.

\bibitem[BPRS93]{MR1317994}
\textsc{L.~Bertini}, \textsc{E.~Presutti}, \textsc{B.~R{\"u}diger}, and
  \textsc{E.~Saada}.
\newblock Dynamical fluctuations at the critical point: convergence to a
  nonlinear stochastic {PDE}.
\newblock \emph{Teor. Veroyatnost. i Primenen.} \textbf{38}, no.~4, (1993),
  689--741.
\newblock \burlalt{doi:10.1137/1138062}{http://dx.doi.org/10.1137/1138062}.

\bibitem[BQS11]{MR2784327}
\textsc{M.~Bal{\'a}zs}, \textsc{J.~Quastel}, and
  \textsc{T.~Sepp{\"a}l{\"a}inen}.
\newblock Fluctuation exponent of the {KPZ} / stochastic {B}urgers equation.
\newblock \emph{J. Amer. Math. Soc.} \textbf{24}, no.~3, (2011), 683--708.
\newblock \burlalt{arXiv:0909.4816}{http://arxiv.org/abs/0909.4816}.
\newblock
  \burlalt{doi:10.1090/S0894-0347-2011-00692-9}{http://dx.doi.org/10.1090/S0894-0347-2011-00692-9}.

\bibitem[CT15]{Corwin:2015aa}
\textsc{I.~Corwin} and \textsc{L.-C. Tsai}.
\newblock {KPZ} equation limit of higher-spin exclusion processes.
\newblock \emph{ArXiv e-prints} (2015).
\newblock \burlalt{arXiv:1505.04158}{http://arxiv.org/abs/1505.04158}.

\bibitem[DT13]{Dembo:2013aa}
\textsc{A.~Dembo} and \textsc{L.-C. Tsai}.
\newblock Weakly asymmetric non-simple exclusion process and the
  {K}ardar-{P}arisi-{Z}hang equation.
\newblock \emph{ArXiv e-prints} (2013).
\newblock \burlalt{arXiv:1302.5760}{http://arxiv.org/abs/1302.5760}.

\bibitem[FH14]{Book}
\textsc{P.~Friz} and \textsc{M.~Hairer}.
\newblock \emph{A Course on Rough Paths}.
\newblock Universitext. Springer, New York, 2014,  xiv+251.
\newblock With an introduction to {R}egularity {S}tructures.
\newblock
  \burlalt{doi:10.1007/978-3-319-08332-2}{http://dx.doi.org/10.1007/978-3-319-08332-2}.

\bibitem[FQ14]{Funaki:2014aa}
\textsc{T.~Funaki} and \textsc{J.~Quastel}.
\newblock Kpz equation, its renormalization and invariant measures.
\newblock \emph{ArXiv e-prints} (2014).
\newblock \burlalt{arXiv:1407.7310}{http://arxiv.org/abs/1407.7310}.

\bibitem[GJ14]{MR3176353}
\textsc{P.~Gon{\c{c}}alves} and \textsc{M.~Jara}.
\newblock Nonlinear fluctuations of weakly asymmetric interacting particle
  systems.
\newblock \emph{Arch. Ration. Mech. Anal.} \textbf{212}, no.~2, (2014),
  597--644.
\newblock \burlalt{arXiv:1309.5120}{http://arxiv.org/abs/1309.5120}.
\newblock
  \burlalt{doi:10.1007/s00205-013-0693-x}{http://dx.doi.org/10.1007/s00205-013-0693-x}.

\bibitem[GP15]{Gubinelli:2015cr}
\textsc{M.~Gubinelli} and \textsc{N.~Perkowski}.
\newblock {KPZ} reloaded.
\newblock \emph{ArXiv e-prints} (2015).
\newblock \burlalt{arXiv:1508.03877}{http://arxiv.org/abs/1508.03877}.

\bibitem[Hai13]{KPZ}
\textsc{M.~Hairer}.
\newblock Solving the {KPZ} equation.
\newblock \emph{Ann. of Math. (2)} \textbf{178}, no.~2, (2013), 559--664.
\newblock \burlalt{arXiv:1109.6811}{http://arxiv.org/abs/1109.6811}.
\newblock
  \burlalt{doi:10.4007/annals.2013.178.2.4}{http://dx.doi.org/10.4007/annals.2013.178.2.4}.

\bibitem[Hai14]{Regularity}
\textsc{M.~Hairer}.
\newblock A theory of regularity structures.
\newblock \emph{Invent. Math.} \textbf{198}, no.~2, (2014), 269--504.
\newblock \burlalt{arXiv:1303.5113}{http://arxiv.org/abs/1303.5113}.
\newblock
  \burlalt{doi:10.1007/s00222-014-0505-4}{http://dx.doi.org/10.1007/s00222-014-0505-4}.

\bibitem[Hai15a]{Notes}
\textsc{M.~Hairer}.
\newblock Introduction to regularity structures.
\newblock \emph{Braz. J. Probab. Stat.} \textbf{29}, no.~2, (2015), 175--210.
\newblock \burlalt{arXiv:1401.3014}{http://arxiv.org/abs/1401.3014}.
\newblock
  \burlalt{doi:10.1214/14-BJPS241}{http://dx.doi.org/10.1214/14-BJPS241}.

\bibitem[Hai15b]{Hairer:2015aa}
\textsc{M.~Hairer}.
\newblock Regularity structures and the dynamical $\phi^4_3$ model.
\newblock \emph{ArXiv e-prints} (2015).
\newblock \burlalt{arXiv:1508.05261}{http://arxiv.org/abs/1508.05261}.

\bibitem[Hep69]{Hepp}
\textsc{K.~Hepp}.
\newblock On the equivalence of additive and analytic renormalization.
\newblock \emph{Comm. Math. Phys.} \textbf{14}, (1969), 67--69.
\newblock
  \burlalt{doi:10.1007/BF01645456}{http://dx.doi.org/10.1007/BF01645456}.

\bibitem[HHZ95]{hhz}
\textsc{T.~Halpin-Healy} and \textsc{Y.-C. Zhang}.
\newblock Kinetic roughening phenomena, stochastic growth, directed polymers
  and all that.
\newblock \emph{Physics Reports} \textbf{254}, no. 4--6, (1995), 215--414.
\newblock
  \burlalt{doi:10.1016/0370-1573(94)00087-J}{http://dx.doi.org/10.1016/0370-1573(94)00087-J}.

\bibitem[HL15]{Hairer:2015ab}
\textsc{M.~{Hairer}} and \textsc{C.~{Labb{\'e}}}.
\newblock {Multiplicative stochastic heat equations on the whole space}.
\newblock \emph{ArXiv e-prints} (2015).
\newblock \burlalt{arXiv:1504.07162}{http://arxiv.org/abs/1504.07162}.

\bibitem[HP15]{Etienne}
\textsc{M.~Hairer} and \textsc{{\'E}.~Pardoux}.
\newblock A {W}ong-{Z}akai theorem for stochastic {PDE}s.
\newblock \emph{J. Math. Soc. Japan} \textbf{67}, no.~4, (2015), 1551--1604.
\newblock \burlalt{arXiv:1409.3138}{http://arxiv.org/abs/1409.3138}.
\newblock
  \burlalt{doi:10.2969/jmsj/06741551}{http://dx.doi.org/10.2969/jmsj/06741551}.

\bibitem[HPP13]{HPP13}
\textsc{M.~Hairer}, \textsc{{\'E}.~Pardoux}, and \textsc{A.~Piatnitski}.
\newblock Random homogenization of a highly oscillatory singular potential.
\newblock \emph{SPDEs: Anal. Comp.} \textbf{1}, no.~4, (2013), 571--605.
\newblock \burlalt{arXiv:1303.1955}{http://arxiv.org/abs/1303.1955}.
\newblock
  \burlalt{doi:10.1007/s40072-013-0018-y}{http://dx.doi.org/10.1007/s40072-013-0018-y}.

\bibitem[HS15]{HaoShen}
\textsc{M.~{Hairer}} and \textsc{H.~{Shen}}.
\newblock A central limit theorem for the {KPZ} equation.
\newblock \emph{ArXiv e-prints} (2015).
\newblock \burlalt{arXiv:1507.01237}{http://arxiv.org/abs/1507.01237}.

\bibitem[KPZ86]{KPZOrig}
\textsc{M.~Kardar}, \textsc{G.~Parisi}, and \textsc{Y.-C. Zhang}.
\newblock Dynamic scaling of growing interfaces.
\newblock \emph{Phys. Rev. Lett.} \textbf{56}, no.~9, (1986), 889--892.
\newblock
  \burlalt{doi:10.1103/PhysRevLett.56.889}{http://dx.doi.org/10.1103/PhysRevLett.56.889}.

\bibitem[Kru56]{MR0078686}
\textsc{J.~B. Kruskal, Jr.}
\newblock On the shortest spanning subtree of a graph and the traveling
  salesman problem.
\newblock \emph{Proc. Amer. Math. Soc.} \textbf{7}, (1956), 48--50.

\bibitem[KS91]{krug-spohn}
\textsc{J.~Krug} and \textsc{H.~Spohn}.
\newblock Kinetic roughening of growing surfaces.
\newblock In \emph{Solids Far from Equilibrium},  479--583. Cambridge
  University Press, 1991.

\bibitem[Mey92]{MR1228209}
\textsc{Y.~Meyer}.
\newblock \emph{Wavelets and operators}, vol.~37 of \emph{Cambridge Studies in
  Advanced Mathematics}.
\newblock Cambridge University Press, Cambridge, 1992,  xvi+224.
\newblock Translated from the 1990 French original by D. H. Salinger.
\newblock
  \burlalt{doi:10.1017/CBO9780511623820}{http://dx.doi.org/10.1017/CBO9780511623820}.

\bibitem[MFQR]{MQR}
\textsc{G.~Moreno-Flores}, \textsc{J.~Quastel}, and \textsc{D.~Remenik}.
\newblock In preparation.

\bibitem[Mue91]{Carl}
\textsc{C.~Mueller}.
\newblock On the support of solutions to the heat equation with noise.
\newblock \emph{Stochastics and Stochastic Reports} \textbf{37}, no.~4, (1991),
  225--245.
\newblock
  \burlalt{doi:10.1080/17442509108833738}{http://dx.doi.org/10.1080/17442509108833738}.

\bibitem[Nua06]{Nualart}
\textsc{D.~Nualart}.
\newblock \emph{The {M}alliavin calculus and related topics}.
\newblock Probability and its Applications (New York). Springer-Verlag, Berlin,
  second ed., 2006,  xiv+382.
\newblock
  \burlalt{doi:10.1007/3-540-28329-3}{http://dx.doi.org/10.1007/3-540-28329-3}.

\bibitem[PP12]{PP12}
\textsc{{\'E}.~Pardoux} and \textsc{A.~Piatnitski}.
\newblock Homogenization of a singular random one-dimensional {PDE} with
  time-varying coefficients.
\newblock \emph{Ann. Probab.} \textbf{40}, no.~3, (2012), 1316--1356.
\newblock \burlalt{arXiv:0912.5277}{http://arxiv.org/abs/0912.5277}.
\newblock \burlalt{doi:10.1214/11-AOP650}{http://dx.doi.org/10.1214/11-AOP650}.

\bibitem[QV13]{Diffusive}
\textsc{J.~Quastel} and \textsc{B.~Valk{\'o}}.
\newblock Diffusivity of lattice gases.
\newblock \emph{Arch. Ration. Mech. Anal.} \textbf{210}, no.~1, (2013),
  269--320.
\newblock \burlalt{arXiv:1211.3716}{http://arxiv.org/abs/1211.3716}.
\newblock
  \burlalt{doi:10.1007/s00205-013-0651-7}{http://dx.doi.org/10.1007/s00205-013-0651-7}.

\bibitem[Spo91]{spohn}
\textsc{H.~Spohn}.
\newblock \emph{Large Scale Dynamics of Interacting Particles}.
\newblock Texts and Monographs in Physics. Springer-Verlag, 1991.

\bibitem[Swe69]{Sweedler}
\textsc{M.~E. Sweedler}.
\newblock \emph{Hopf algebras}.
\newblock Mathematics Lecture Note Series. W. A. Benjamin, Inc., New York,
  1969,  vii+336.

\bibitem[Wal86]{MR876085}
\textsc{J.~B. Walsh}.
\newblock An introduction to stochastic partial differential equations.
\newblock In \emph{\'{E}cole d'\'et\'e de probabilit\'es de {S}aint-{F}lour,
  {XIV}---1984}, vol. 1180 of \emph{Lecture Notes in Math.},  265--439.
  Springer, Berlin, 1986.

\bibitem[Zim69]{Zimmermann}
\textsc{W.~Zimmermann}.
\newblock Convergence of {B}ogoliubov's method of renormalization in momentum
  space.
\newblock \emph{Comm. Math. Phys.} \textbf{15}, (1969), 208--234.
\newblock
  \burlalt{doi:10.1007/BF01645676}{http://dx.doi.org/10.1007/BF01645676}.

\end{thebibliography}

\end{document}